\pdfoutput=1
\newif\iffull
\newif\ifnotes

\fulltrue
\notesfalse

\iffull
 \documentclass[10pt,letter]{article}
 \usepackage[\ifnotes notes=true \else notes=xxx \fi]{dtrt}
 \usepackage{fullpage}
\else
 \documentclass[11pt,letter]{article}
 \usepackage[notes=xxx]{dtrt}
 \usepackage{fullpage}
\fi
\usepackage[backend=bibtex8,style=alphabetic,maxnames=99]{biblatex}
\bibliography{references}

\usepackage{amsthm}
\usepackage[utf8]{inputenc}
\usepackage{times}
\usepackage{amsmath}
\usepackage{mathtools}
\usepackage{comment}
\usepackage{xspace}
\usepackage{paralist}
\usepackage{enumitem}
  \setlist[itemize]{leftmargin=*}
  \setlist[enumerate]{leftmargin=*}
\usepackage{makecell}
\usepackage{tikz}
	\usetikzlibrary{matrix,shapes,arrows,positioning,chains,calc}
\usepackage{hypcap}
\usepackage{mleftright}
\usepackage{bbm}
\usepackage{colortbl}
\usepackage{tabu}
\usepackage{multirow}
\usepackage{mathrsfs}
\usepackage{adjustbox}
\usepackage{microtype}
\usepackage{etoolbox}
\usepackage{graphicx}
\usepackage{breakcites}
\usepackage{booktabs}
\usepackage{mdframed}
\usepackage{soul}
\usepackage{amssymb}
\usepackage[small,margin=2mm]{caption}

\makeatletter
\def\grd@save@target#1{%
  \def\grd@target{#1}}
\def\grd@save@start#1{%
  \def\grd@start{#1}}
\tikzset{
  grid with coordinates/.style={
    to path={%
      \pgfextra{%
        \edef\grd@@target{(\tikztotarget)}%
        \tikz@scan@one@point\grd@save@target\grd@@target\relax
        \edef\grd@@start{(\tikztostart)}%
        \tikz@scan@one@point\grd@save@start\grd@@start\relax
        \draw[minor help lines] (\tikztostart) grid (\tikztotarget);
        \draw[major help lines] (\tikztostart) grid (\tikztotarget);
        \grd@start
        \pgfmathsetmacro{\grd@xa}{\the\pgf@x/1cm}
        \pgfmathsetmacro{\grd@ya}{\the\pgf@y/1cm}
        \grd@target
        \pgfmathsetmacro{\grd@xb}{\the\pgf@x/1cm}
        \pgfmathsetmacro{\grd@yb}{\the\pgf@y/1cm}
        \pgfmathsetmacro{\grd@xc}{\grd@xa + \pgfkeysvalueof{/tikz/grid with coordinates/major step}}
        \pgfmathsetmacro{\grd@yc}{\grd@ya + \pgfkeysvalueof{/tikz/grid with coordinates/major step}}
        \foreach \x in {\grd@xa,\grd@xc,...,\grd@xb}
        \node[anchor=north] at (\x,\grd@ya) {\pgfmathprintnumber{\x}};
        \foreach \y in {\grd@ya,\grd@yc,...,\grd@yb}
        \node[anchor=east] at (\grd@xa,\y) {\pgfmathprintnumber{\y}};
      }
    }
  },
  minor help lines/.style={
    help lines,
    step=\pgfkeysvalueof{/tikz/grid with coordinates/minor step}
  },
  major help lines/.style={
    help lines,
    line width=\pgfkeysvalueof{/tikz/grid with coordinates/major line width},
    step=\pgfkeysvalueof{/tikz/grid with coordinates/major step}
  },
  grid with coordinates/.cd,
  minor step/.initial=.2,
  major step/.initial=1,
  major line width/.initial=2pt,
}
\makeatother

\newcommand{\doclearpage}{%
\iffull
\clearpage
\else
\fi
}

\newcommand{\keywords}[1]{\bigskip\par\noindent{\footnotesize\textbf{Keywords\/}: #1}}

\newcommand{\sunderline}[1]{\underline{\smash{#1}}}

\newcommand{\defemph}[1]{\textbf{\emph{#1}}}

\newcommand{\mhl}[1]{\text{\hl{$#1$}}}

\newcommand{\FormatAuthor}[3]{
\begin{tabular}{c}
#1 \\ {\small\texttt{#2}} \\ {\small #3}
\end{tabular}
}

\theoremstyle{plain} %
\newtheorem{theorem}{Theorem}[section]

\newtheorem{lemma}[theorem]{Lemma}

\newtheorem{claim}[theorem]{Claim}
\newtheorem{corollary}[theorem]{Corollary}
\newtheorem{definition}[theorem]{Definition}

\newtheorem{construction}[theorem]{Construction}

\newtheorem*{uclaim}{Claim}

\theoremstyle{definition} %
\newtheorem{remark}[theorem]{Remark}

\theoremstyle{remark} %

\newtheoremstyle{soundnessstyle} %
    {\topsep}                    %
    {\topsep}                    %
    {}                           %
    {0em}                        %
    {\scshape}                   %
    {.}                          %
    {.5em}                       %
    {}  %
\theoremstyle{soundnessstyle}

\newcommand{\secref}[1]{Section~\protect\ref{#1}}
\newcommand{\appref}[1]{Appendix~\protect\ref{#1}}
\newcommand{\defref}[1]{Definition~\protect\ref{#1}}

\newcommand{\figref}[1]{Figure~\protect\ref{#1}}

\newcommand{\thmref}[1]{Theorem~\protect\ref{#1}}

\newcommand{\lemref}[1]{Lemma~\protect\ref{#1}}
\newcommand{\clmref}[1]{Claim~\protect\ref{#1}}
\newcommand{\corref}[1]{Corollary~\protect\ref{#1}}

\newcommand{\eqnref}[1]{Equation~\protect\ref{#1}}

\newcommand{\remref}[1]{Remark~\protect\ref{#1}}
\newcommand{\stepref}[1]{Step~\protect\ref{#1}}

\newcommand{\equnref}[1]{Equation~\protect\ref{#1}}
\newcommand{\CB}{\allowbreak}

\newcommand{\pair}[2]{(#1 ,\CB #2)}

\DeclareMathOperator{\poly}{poly}
\DeclareMathOperator{\polylog}{polylog}

\newcommand{\Ot}[1]{\tilde{O}(#1)}

\DeclareMathOperator{\Expectation}{\mathbb{E}}

\newcommand{\yes}{\mathsf{yes}}
\newcommand{\no}{\mathsf{no}}

\newcommand{\Bits}{\{0,1\}}

\newcommand{\Naturals}{\mathbb{N}}

\newcommand{\DefineEqual}{:=}

\newcommand{\SetCardinality}[1]{|#1|}

\newcommand{\BitSize}[1]{|#1|}

\newcommand{\Proof}{\pi\xspace}

\newcommand{\Randomness}{r}

\newcommand{\FormatComplexityClass}[1]{\mathbf{#1}}
\newcommand{\NTIME}{\FormatComplexityClass{NTIME}}
\newcommand{\BPP}{\FormatComplexityClass{BPP}}
\newcommand{\NP}{\FormatComplexityClass{NP}}

\newcommand{\sharpP}{\FormatComplexityClass{\#P}}
\newcommand{\PSPACE}{\FormatComplexityClass{PSPACE}}
\newcommand{\NEXP}{\FormatComplexityClass{NEXP}}

\newcommand{\IPCP}{\FormatComplexityClass{IPCP}}
\newcommand{\PZKIPCP}{\FormatComplexityClass{PZK\mbox{-}IPCP}}

\newcommand{\AM}{\FormatComplexityClass{AM}}
\newcommand{\coAM}{\FormatComplexityClass{coAM}}

\newcommand{\Domain}{D}
\newcommand{\Range}{R}
\newcommand{\SubDomain}{\tilde{\Domain}}

\newcommand{\Restrict}[2]{#1|_{#2}}

\newcommand{\Field}{\mathbb{F}}
\newcommand{\SubField}{\mathbb{K}}
\newcommand{\FieldSize}{q}

\newcommand{\Multiplicative}[1]{#1^{\times}}

\newcommand{\VariableX}{X}
\newcommand{\VariableY}{Y}
\newcommand{\VariableZ}{Z}

\newcommand{\Machine}{M}

\newcommand{\Lagrange}[1]{I_{#1}}

\NewDocumentCommand{\IndividualDegree}{m o}{\IfValueTF{#2}{\mathrm{deg}_{#2}(#1)}{\mathrm{deg}(#1)}}

\newcommand{\PolyA}{P}

\newcommand{\PolynomialRing}[3]{#1[#3_{1,\dots,#2}]}
\newcommand{\PolynomialRingIndOne}[4]{#1[#3_{1,\dots,#2}^{\leq #4}]}
\newcommand{\PolynomialRingIndOneXY}[7]{#1[#3_{1,\dots,#2}^{\leq #6},#5_{1,\dots,#4}^{\leq #7}]}

\newcommand{\pST}{\; \middle\vert \;}
\newcommand{\MakeDistribution}[1]{\mathcal{#1}}
\newcommand{\Distribution}{\MakeDistribution{D}}

\newcommand{\Relation}{\mathscr{R}}
\newcommand{\Language}{\mathscr{L}}
\newcommand{\Witnesses}[2]{#1\vert_{#2}}

\newcommand{\GetLanguage}[1]{\mathrm{Lan}(#1)}

\newcommand{\Instance}{\mathbbmss{x}}
\newcommand{\Witness}{\mathbbmss{w}}
\newcommand{\InstanceSize}{n}

\newcommand{\DeciderMachine}{M}
\newcommand{\DeciderTime}{T}

\newcommand{\Prover}{P\xspace}
\newcommand{\Verifier}{V\xspace}
\newcommand{\Simulator}{S\xspace}

\newcommand{\IOPize}[1]{#1}

\newcommand{\IOPVerifier}{\IOPize{\Verifier}}

\newcommand{\SCSymbol}{\mathrm{SC}}
\newcommand{\SCSubset}{H}
\newcommand{\SCVars}{m}
\newcommand{\SCDegree}{d}
\newcommand{\SCPoly}{F}
\newcommand{\SCSum}{a}
\newcommand{\SCInput}{\Field,\SCVars,\SCDegree,\SCSubset,\SCSum}

\newcommand{\SCRelation}{\Relation_{\SCSymbol}}
\newcommand{\RandPoly}{R}
\newcommand{\MaskedPoly}{Q}
\newcommand{\AnsTable}[1]{\mathsf{ans}_{#1}}

\newcommand{\IPSCProver}{\Prover_{\mathrm{IP}}}
\newcommand{\IPSCVerifier}{\Verifier_{\mathrm{IP}}}
\newcommand{\IPSCSimulator}{\Simulator_{\mathrm{IP}}}
\newcommand{\SCStrength}{b}
\newcommand{\CodeSimAlgorithm}{\mathcal{A}}
\newcommand{\EmptyVector}{\bot}
\newcommand{\ListSize}{\ell}
\newcommand{\SlowSimulator}{\Simulator_{\mathrm{slow}}}

\newcommand{\Strong}[1]{{#1^{\star}}}
\newcommand{\SubsetSize}{\lambda}
\newcommand{\SSCVars}{k}
\newcommand{\SSCSubset}{G}
\newcommand{\StrongRandPoly}{Z}
\newcommand{\AuxRandPoly}{A}
\newcommand{\SoundnessSet}{I}

\newcommand{\LD}[1]{\hat{#1}}
\newcommand{\RLD}[1]{\dot{#1}}
\newcommand{\Layer}[1]{V_{#1}}
\newcommand{\Add}[1]{\mathrm{add}_{#1}}
\newcommand{\Mult}[1]{\mathrm{mul}_{#1}}
\newcommand{\LDAdd}[1]{\LD{\Add{}}_{#1}}
\newcommand{\LDMult}[1]{\LD{\Mult{}}_{#1}}
\newcommand{\GKRVars}{m}
\newcommand{\GKRVarsInput}{m'}
\newcommand{\Depth}{D}
\newcommand{\Size}{S}
\newcommand{\Circuit}{C}
\newcommand{\WireOracles}{F}

\newcommand{\Interact}[2]{\langle #1,#2 \rangle}

\DeclareMathOperator{\rank}{rank}

\newcommand{\Malicious}[1]{\tilde{#1}}
\newcommand{\Simulated}[1]{#1_{\mathrm{sim}}}

\newcommand{\View}{\mathrm{View}}
\newcommand{\FView}[1]{v({#1})}
\newcommand{\IOPView}[2]{\View\;\Interact{#1}{#2}}

\newcommand{\IPCPView}[2]{\View\;\Interact{#1}{#2}}

\newcommand{\Characteristic}[1]{\mathrm{char}(#1)}

\newcommand{\RMSubDomain}{H}
\newcommand{\RMVars}{m}
\newcommand{\RMDegree}{d}

\newcommand{\TQBFLanguage}{\Language_{\mathrm{TQBF}}}
\newcommand{\Formula}{\phi}
\newcommand{\QFormula}{\Phi}
\newcommand{\Quantifier}{\mathrm{Q}}

\newcommand{\NumVars}{n}
\newcommand{\NumClauses}{c}

\newcommand{\Clause}{K}

\newcommand{\OracleSATRelation}{\Relation_{\mathrm{O3SAT}}}

\newcommand{\FormatComplexityFunction}[1]{\mathsf{#1}}

\newcommand{\SoundnessError}{\FormatComplexityFunction{\varepsilon}}
\newcommand{\ProofLength}{\FormatComplexityFunction{l}}
\newcommand{\QueryComplexity}{\FormatComplexityFunction{q}}

\newcommand{\NumRounds}{\FormatComplexityFunction{k}}
\newcommand{\QueryBound}{\FormatComplexityFunction{b}}
\newcommand{\AnyBound}{\FormatComplexityFunction{*}}

\newcommand{\Tree}{T}
\newcommand{\VertexSet}{V}
\newcommand{\EdgeSet}{E}
\newcommand{\Root}[1][]{r_{#1}}
\newcommand{\TreeDepth}[2][]{\mathsf{depth}_{#1}(#2)}
\newcommand{\TreeDegree}[2][]{\mathsf{out}_{#1}(#2)}

\newcommand{\Graph}{G}
\newcommand{\GraphDepth}[2][]{\mathsf{depth}_{#1}(#2)}
\newcommand{\GraphOutDegree}[2][]{\mathsf{out}_{#1}(#2)}
\newcommand{\GraphInDegree}[2][]{\mathsf{in}_{#1}(#2)}
\newcommand{\GraphMaxInDegree}[1]{\mathsf{in}(#1)}
\newcommand{\Width}[1]{\mathsf{width}(#1)}

\newcommand{\SPFreeProjection}[1][]{\rho_{#1}}
\newcommand{\SPSumProjection}[1][]{\sigma_{#1}}
\newcommand{\SPArity}[1]{\mathsf{arity}(#1)}
\newcommand{\SPMaxArity}{\mathsf{arity}}

\newcommand{\aodeg}{t}

\newcommand{\SPValueL}[2]{#2[#1]}

\newcommand{\LDSPValueL}[2]{\expandafter\LD#2[#1]}

\newcommand{\RLDSPValueL}[2]{\expandafter\RLD#2[#1]}
\newcommand{\SPFormula}{\mathcal{F}}
\newcommand{\SPFormulaTuple}{(\Field, \SPSubset, \InternalVertexSetDegree, \SPLeafDegree, \Tree, \SPPoly)}
\newcommand{\SPCircuit}{\mathcal{C}}
\newcommand{\SPCircuitTuple}{(\Field, \SPSubset, \InternalVertexSetDegree, \SPLeafDegree, \Graph, \SPPoly)}
\newcommand{\SPSubset}{H}

\newcommand{\SPDegree}{\delta}
\newcommand{\InternalVertexSetDegree}{\SPDegree_{\mathsf{in}}}
\newcommand{\SPLeafDegree}{\SPDegree_{\mathsf{lf}}}
\newcommand{\SPVars}[1][]{m_{#1}}
\newcommand{\SPPoly}[1][]{C_{#1}}
\newcommand{\SPLeaf}[1][]{\mathbbmss{x}_{#1}}
\newcommand{\SPAuxLeafM}[1][]{\tilde{\SPAuxInput}_{#1}}
\newcommand{\SPAuxLeaf}[1][]{\SPAuxInput_{#1}}

\newcommand{\SPInput}{\mathbbmss{x}}
\newcommand{\SPAuxInput}{\mathbbmss{z}}
\newcommand{\SPOutput}{\mathbbmss{y}}

\newcommand{\Labels}[1]{L_{#1}}
\newcommand{\SPFELanguage}{\Language_{\mathrm{SPFE}}}
\newcommand{\SPFSRelation}{\Relation_{\mathrm{SPFS}}}
\newcommand{\SPCELanguage}{\Language_{\mathrm{SPCE}}}
\newcommand{\SPCSRelation}{\Relation_{\mathrm{SPCS}}}
\newcommand{\ZeroPoly}[1]{\mathbb{Z}_{#1}}
\newcommand{\MaxSpace}[1]{\mathsf{space}(#1)}

\newcommand{\Curve}[1]{A_{#1}}
\newcommand{\ComposedCurve}[1]{B_{#1}}

\begin{document}
\title{%
A Zero Knowledge Sumcheck and its Applications
}
\author{%
\makebox[\linewidth][c]{
\begin{tabular}[h!]{ccc}
\FormatAuthor{Alessandro Chiesa}{alexch@berkeley.edu}{UC Berkeley}
&
\FormatAuthor{Michael A. Forbes}{miforbes@csail.mit.edu}{Simons Institute for the Theory of Computing}
&
\FormatAuthor{Nicholas Spooner}{nick.spooner@berkeley.edu}{University of Toronto and UC Berkeley}
\end{tabular}
}
}
\iffull
  \date{\today}
\else
  \date{}
\fi
\maketitle
\begin{abstract}

Many seminal results in Interactive Proofs (IPs) use algebraic techniques based on low-degree polynomials, the study of which is pervasive in theoretical computer science. Unfortunately, known methods for endowing such proofs with zero knowledge guarantees do not retain this rich algebraic structure.

In this work, we develop algebraic techniques for obtaining zero knowledge variants of proof protocols in a way that leverages and preserves their algebraic structure. Our constructions achieve unconditional (perfect) zero knowledge in the Interactive Probabilistically Checkable Proof (IPCP) model of Kalai and Raz \cite{KalaiR08} (the prover first sends a PCP oracle, then the prover and verifier engage in an Interactive Proof in which the verifier may query the PCP).

Our main result is a zero knowledge variant of the sumcheck protocol \cite{LundFKN92} in the IPCP model. The sumcheck protocol is a key building block in many IPs, including the protocol for polynomial-space computation due to Shamir \cite{Shamir92}, and the protocol for parallel computation due to Goldwasser, Kalai, and Rothblum \cite{GoldwasserKR15}. A core component of our result is an algebraic commitment scheme, whose hiding property is guaranteed by algebraic query complexity lower bounds \cite{AaronsonW09,JumaKRS09}.  This commitment scheme can then be used to considerably strengthen our previous work \cite{BenSassonCFGRS16} that gives a sumcheck protocol with much weaker zero knowledge guarantees, itself using algebraic techniques based on algorithms for polynomial identity testing \cite{RazS05,BogdanovW04}.

We demonstrate the applicability of our techniques by deriving zero knowledge variants of well-known protocols based on algebraic techniques. First, we construct zero knowledge IPCPs for $\NEXP$ starting with the Multi-prover Interactive Proofs of Babai, Fortnow, and Lund \cite{BabaiFL91}. This result is a direct application of our zero knowledge sumcheck and our algebraic commitment scheme, augmented with the use of `randomized' low-degree extensions.

We also construct protocols in a more restricted model where the prover and verifier engage in a standard Interactive Proof with oracle access to a uniformly random low-degree polynomial (soundness holds with respect to \emph{any} oracle). In this setting we achieve zero knowledge variants of the protocols of Shamir and of Goldwasser, Kalai, and Rothblum.

\keywords{zero knowledge; sumcheck; algebraic query complexity; probabilistically checkable and interactive proofs}

\end{abstract}

\iffull
\clearpage
\setcounter{tocdepth}{2}
{\footnotesize \tableofcontents}
\clearpage
\fi

\clearpage
\section{Introduction}
\label{sec:introduction}

The notion of \emph{Interactive Proofs} (IPs) \cite{BabaiM88,GoldwasserMR89} is fundamental in Complexity Theory and Cryptography. An Interactive Proof for a language $\Language$ is a protocol between a probabilistic polynomial-time \emph{verifier} and a resource-unbounded \emph{prover} that works as follows: given a common input $\Instance$, the prover and verifier exchange some number of messages and then the verifier either accepts or rejects. If $\Instance$ is in $\Language$ then the verifier always accepts; if instead $\Instance$ is not in $\Language$ then the verifier rejects with high probability, regardless of the prover's actions. The seminal results of Lund, Fortnow, Karloff, and Nisan~\cite{LundFKN92} and Shamir~\cite{Shamir92} demonstrate the surprising expressiveness of Interactive Proofs, in particular showing that every language decidable in polynomial space has an Interactive Proof.

Research on IPs has recently focused on new and more refined goals, motivated by the paradigm of \emph{delegation of computation}, in which a resource-limited verifier receives the help of a resource-rich prover to check the output of an expensive (but tractable) computation. In this setting bounding the complexity of the honest prover is important. While every IP protocol has a polynomial-space prover, this prover may run in superpolynomial time, even if the protocol is for a tractable language. Recent work has focused on \emph{doubly-efficient} IPs, where the prover is efficient (it runs in polynomial time) and the verifier is highly efficient (it runs in, say, quasilinear time). Doubly-efficient IPs can be achieved, with various tradeoffs, for many types of computation: languages decidable by uniform circuits of polylogarithmic depth \cite{GoldwasserKR15}; languages decidable in polynomial time and bounded-polynomial space \cite{ReingoldRR16}; and languages decidable by conjunctions of polynomially-many `local' conditions \cite{GoldreichR17}.

A key building block in \emph{all} of these protocols is the \emph{sumcheck protocol} \cite{LundFKN92}, which is an Interactive Proof for claims of the form ``$\sum_{\vec{\alpha} \in H^{m}} F(\vec{\alpha}) = 0$'', where $H$ is a subset of a finite field $\Field$ and $F$ is an $m$-variate polynomial over $\Field$ of small individual degree. The use of sumcheck imbues the aforementioned protocols with an algebraic structure, where the verifier \emph{arithmetizes} a boolean problem into a statement about low-degree polynomials, which can then be checked via the sumcheck protocol. This algebraic structure is not only elegant, but also very useful. Indeed, this structure is crucial not only for highly-efficient software and hardware systems for delegating computation \cite{CormodeMT12,ThalerRMP12,Thaler13,Thaler15,WahbyHGSW16,WahbyKBSTWW17} but also for a diverse set of compelling theoretical applications such as memory delegation \cite{ChungKLR11}, refereed delegation \cite{CanettiRR13}, IPs of proximity \cite{RothblumVW13}, and many others.

Despite these subsequent works demonstrating the flexibility of sumcheck to accommodate additional desirable properties, the \emph{zero knowledge} \cite{GoldwasserMR89} properties of sumcheck have not been explored. This is surprising because zero knowledge, the ability of the prover to establish the validity of an assertion while revealing no insight into its proof, is highly desirable for the cryptographic applications of Interactive Proofs. Unfortunately, achieving zero knowledge is nontrivial because the sumcheck protocol reveals the results of intermediate computations, in particular the partial sums $\sum_{\vec{\alpha} \in H^{m-i}} F(c_{1}, \ldots, c_{i}, \vec{\alpha})$ for $c_{1}, \ldots, c_{i} \in \Field$ chosen by the verifier. These partial sums are in general $\sharpP$-hard to compute so they convey significant additional knowledge to the verifier.

The goal of this work is to enlarge the existing algebraic toolkit based on low-degree polynomials and use these tools to provide a native extension of the sumcheck protocol that is zero knowledge, and to explore applications of such an extension. As we discuss shortly, however, we cannot expect to do so within the model of Interactive Proofs.

\subsection{Prior techniques for achieving zero knowledge}
\label{sec:achieving-zk}

We briefly describe why existing methods fall short of our goal, which is making the sumcheck protocol zero knowledge in an algebraic way. A seminal result in cryptography says that if one-way functions exist then every language having an IP also has a \emph{computational} zero knowledge IP \cite{GoldwasserMR89,ImpagliazzoY87,BenOrGGHKMR88}; this assumption is `minimal' in the sense that if one-way functions do not exist then computational zero knowledge IPs capture only ``average-case'' $\BPP$ \cite{Ostrovsky91,OstrovskyW93}. While powerful, such results are unsatisfactory from our perspective. First, cryptography adds significant efficiency overheads, especially when used in a non-blackbox way as these results do. Second, the results rely on transformations that erase all the algebraic structure of the underlying protocols. While these limitations can be mitigated by using cryptography that leverages some of the underlying structure \cite{CramerD98}, the costs incurred by the use of cryptography remain significant. Ideally, we wish to \emph{avoid} intractability assumptions.

Unfortunately, this is impossible to achieve under standard complexity assumptions, because Interactive Proofs that are \emph{statistical} zero knowledge are limited to languages in $\AM \cap \coAM$ \cite{Fortnow87,AielloH91}. Such languages (conjecturally) do not even include $\NP$, so that we cannot even hope to achieve a `zero knowledge sumcheck protocol' (which would give $\sharpP$).

The quest for zero knowledge without relying on intractability assumptions led to the formulation of \emph{Multi-prover Interactive Proofs} (MIPs) \cite{BenOrGKW88}, where the verifier exchanges messages with two or more non-communicating provers. Groundbreaking results establish that MIPs are very powerful: all (and only) languages decidable in non-deterministic exponential time have MIPs \cite{BabaiFL91} and, in fact, even \emph{perfect} zero knowledge MIPs \cite{BenOrGGHKMR88,DworkFKNS92}. Similar results hold even for the related model of \emph{Probabilistically Checkable Proofs} (PCPs) \cite{FortnowRS88,BabaiFLS91,FGLSS96,AroraS98,AroraLMSS98}, where the prover outputs a proof string that the verifier can check by reading only a few randomly-chosen locations. Namely, all (and only) languages decidable in non-deterministic exponential time have PCPs \cite{BabaiFLS91} and, in fact, even \emph{statistical} zero knowledge PCPs \cite{KilianPT97,IshaiMSX15}.

However, while information-theoretic, the aforementioned works rely on transformations that, once again, discard the rich algebraic structure of the underlying protocols. Thus, zero knowledge in this setting continues to be out of reach of simple and elegant algebraic techniques.

\subsection{Our goal: algebraic techniques for zero knowledge}
\label{sec:our-focus}

Our goal is to develop information-theoretic techniques for achieving zero knowledge in a way that leverages, and preserves, the algebraic structure of the sumcheck protocol and other protocols that build on it. An additional goal is to preserve the simplicity and elegance of these foundational protocols.

Yet, as discussed, we cannot hope to do so with Interactive Proofs, and so we work in another model. We choose to work in a model that combines features of both Interactive Proofs and PCPs: the \emph{Interactive PCP} (IPCP) model of Kalai and Raz \cite{KalaiR08}. The prover first sends to the verifier a long string as a PCP oracle, after which the prover and verifier engage in an Interactive Proof. The verifier is free at any point to query the PCP oracle at locations of its choice, and the verifier only pays for the number of queries it makes, so that exponentially-large PCP oracles are allowed.

Kalai and Raz \cite{KalaiR08} show that the IPCP model has efficiency advantages over both PCPs and IPs (individually). Goyal, Ishai, Mahmoody, and Sahai \cite{GoyalIMS10} construct efficient zero knowledge IPCPs, but their techniques mirror those for zero knowledge PCPs and, in particular, are not algebraic.

One can think of the IPCP model as lying somewhere `in between' the IP and MIP models. Indeed, it is equivalent to a (2-prover) MIP where one of the provers is stateless (its answers do not depend on the verifier's prior messages or queries). This means that soundness is easier to achieve for an IPCP than for an MIP. Zero knowledge, however, is more difficult for an IPCP than for an MIP, because the stateless prover cannot choose which queries it will answer.

A significant advantage of the IPCP model over the MIP model is that one can easily compile (public-coin) IPCPs into cryptographic proofs via transformations that preserve zero knowledge, while only making a black-box use of cryptography. For example, using collision-resistant functions one can obtain public-coin interactive arguments by extending ideas of \cite{Kilian92,IshaiMSX15}; also, using random oracles one can obtain publicly-verifiable non-interactive arguments via \cite{BenSassonCS16} (extending the Fiat--Shamir paradigm \cite{FiatS86} and Micali's ``CS proofs'' \cite{Micali00}). In contrast, known transformations for MIPs yield private-coin arguments \cite{BitanskyC12}, or do not preserve zero knowledge \cite{KalaiRR14}.

\subsection{Main result: a zero knowledge sumcheck}
\label{sec:intro-zk-sumcheck}

Our main result is a zero knowledge analogue of the sumcheck protocol \cite{LundFKN92}, a key building block in many protocols. We now informally state and discuss this result, and in the next sub-section we discuss its applications.

The goal of the sumcheck protocol is to efficiently verify claims of the form ``$\sum_{\vec{\alpha} \in H^{m}} \SCPoly(\vec{\alpha}) = 0$'', where $H$ is a subset of a finite field $\Field$ and $\SCPoly$ is an $m$-variate polynomial over $\Field$ of low individual degree. As the sumcheck protocol is often used in situations where the polynomial $\SCPoly$ is only implicitly defined (including this paper), it is helpful to adopt the viewpoint of Meir~\cite{Meir13}, regarding the sumcheck protocol as a \emph{reduction} from the summation ``$\sum_{\vec{\alpha} \in H^{m}}\SCPoly(\vec{\alpha}) = 0$'' to an evaluation ``$\SCPoly(\vec{c}) = b$''; the latter can be checked directly by the verifier or by another protocol. The verifier in this reduction does not need any information about $\SCPoly$, aside from knowing that $\SCPoly$ has small individual degree. The completeness property of the reduction is that if the summation claim is true, then so is the evaluation claim with probability one. Its soundness property is that if the summation claim is false, then so is the evaluation claim with high probability.

The theorem below states the existence of sumcheck protocol in the above sense that works in the IPCP model and is \emph{zero knowledge}, which means that a verifier does not learn any information beyond the fact that $\SCPoly$ sums to $0$ on $H^{m}$ (and, in our case, a single evaluation of $\SCPoly$). As usual, this means that we establish an efficient procedure for simulating the interaction of a (possibly malicious) verifier with the prover, where the simulator only uses the knowledge of the sum of $\SCPoly$ on $H^{m}$ (and a single evaluation of $\SCPoly$) but otherwise has no actual access to the prover (or $\SCPoly$). This interaction is a random variable depending on the randomness of both the verifier and the prover, and we show that the simulated interaction perfectly replicates this random variable.

\begin{theorem}[Informal version of \thmref{thm:strong-sumcheck-ipcp}]
\label{thm:strong-sumcheck-ipcp-intro}
There exists an IPCP for sumcheck with the following zero knowledge guarantee: the view of any probabilistic polynomial-time verifier in the protocol can be perfectly and efficiently simulated by a simulator that makes only a single query to $\SCPoly$. Moreover, we do not require the full power of the IPCP model: the honest prover's PCP consists only of a random multi-variate polynomial over $\Field$ of small individual degree.
\end{theorem}

Our result significantly strengthens the IPCP for sumcheck of \cite{BenSassonCFGRS16} (co-authored by this paper's authors), which is only zero knowledge with respect to a simulator which requires \emph{unrestricted} query access to $\SCPoly$. Namely, in order to simulate a verifier's view, their simulator must make a number of queries to $\SCPoly$ that equals the number of queries to the PCP oracle made by the verifier. This zero knowledge guarantee is weaker than the above because a malicious verifier can make an \emph{arbitrarily-large} (but polynomial) number of queries to the PCP oracle. However, this weaker guarantee suffices in some cases, such as in the previous work \cite{BenSassonCFGRS16}.

Perhaps more damaging is that when using this `weakly zero knowledge' sumcheck protocol recursively (as required in applications), we would incur an \emph{exponential blowup}: each simulated query recursively requires further queries to be simulated. In contrast, the `strongly zero knowledge' sumcheck protocol that we achieve only requires the simulator to make a single query to $\SCPoly$ regardless of the malicious verifier's runtime, both providing a stronger zero knowledge guarantee and avoiding any exponential blow-up.

An important property of our sumcheck protocol (which also holds for \cite{BenSassonCFGRS16}), is that it suffices for the honest prover to send as an oracle a uniformly random polynomial of a certain arity and degree. This brings the result `closer to IP', in the sense that while IPCP captures all of $\NEXP$, only languages in $\PSPACE$ have IPCPs (with perfect completeness) where the honest prover behaves in this way. The same property holds for some of our applications.

\parhead{Algebraic commitments}
As detailed in \secref{sec:techniques}, a key ingredient of our result is a commitment scheme based on algebraic techniques. That is, the prover wishes to commit to a value $b\in\Field$, and to do so sends a PCP oracle to the verifier. To then reveal (decommit) $b$, the prover and verifier engage in an Interactive Proof. We show that the sumcheck protocol naturally yields such a commitment scheme, where the PCP oracle is simply a random low-degree polynomial $\RandPoly$ such that $\sum_{\vec{\alpha} \in \SCSubset^{\SCVars}} \RandPoly(\vec{\alpha}) = b$. The soundness guarantee of the sumcheck protocol shows that this commitment scheme is binding, so that the prover cannot ``reveal'' a value other than the correct $b$. To establish the hiding property, which states that the verifier cannot learn anything about $b$ before the prover reveals it, we leverage lower bounds on the algebraic query complexity of polynomial summation, previously studied for completely different reasons \cite{AaronsonW09,JumaKRS09}.

As our commitments are themselves defined by low degree polynomials, they are `transparent' to low degree testing. That is, in various protocols the prover sends to the verifier the evaluation table of a low-degree polynomial as a PCP oracle, and the verifier ensures that this evaluation table is (close to) low degree via low degree testing. In our zero knowledge setting, we need the prover to hide the evaluation table under a commitment (to be revealed selectively), and yet still allow the verifier to check that the underlying evaluation table represents a low-degree polynomial. Our commitments naturally have this property due to their algebraic structure, which we exploit in our applications discussed below.

Overall, the methods of this paper not only significantly deviate from traditional methods for achieving zero knowledge but also further illustrate the close connection between zero knowledge and Algebraic Complexity Theory, the theory of efficient manipulations of algebraic circuits and low-degree polynomials. This connection was first seen in our prior work developing the `weakly zero knowledge sumcheck' of \cite{BenSassonCFGRS16}, used here as a subroutine. Indeed, this subroutine derives its zero-knowledge guarantee from an efficient algorithm for adaptively simulating random low-degree polynomials \cite{BenSassonCFGRS16,BogdanovW04}. This algorithm itself relies on deterministic algorithms for polynomial identity testing of certain restricted classes of algebraic circuits \cite{RazS05}. We believe that it is an exciting research direction to further investigate this surprising connection, and to further broaden the set of information-theoretic algebraic techniques that are useful towards zero knowledge.

\subsection{Applications: delegating computation in zero knowledge}
\label{sec:intro-applicaqtions}

The original sumcheck protocol (without zero knowledge) has many applications, including to various methods of delegating computation. Our zero knowledge sumcheck protocol can be used to obtain zero knowledge analogues of foundational results in this area: we achieve natural zero knowledge extensions of the first construction of PCPs/MIPs \cite{BabaiFL91,BabaiFLS91}, Shamir's protocol \cite{Shamir92}, and doubly-efficient Interactive Proofs for low-depth circuits \cite{GoldwasserKR15}.

\subsubsection{Delegating non-deterministic exponential time}
\label{sec:intro-nexp}

One of the earliest and most influential applications of the sumcheck protocol is the construction of Multi-prover Interactive Proofs for $\NEXP$ due to Babai, Fortnow, and Lund \cite{BabaiFL91}; the same construction also demonstrated the power of low-degree testing as a tool for checking arbitrary computations, another highly influential insight. The subsequent improvements by Babai, Fortnow, Levin, and Szegedy \cite{BabaiFLS91} led to the formulation and the study of \emph{Probabilistically-Checkable Proofs} \cite{BabaiFLS91,FGLSS96} and then the celebrated PCP Theorem \cite{AroraS98,AroraLMSS98}.

We show how, by using our zero knowledge sumcheck protocol, we can obtain a zero knowledge analogue of the classical constructions of \cite{BabaiFL91,BabaiFLS91}. 

\begin{theorem}[Informal version of \thmref{thm:pzk-for-nexp}]
$\NEXP$ has perfect zero knowledge Interactive PCPs.
\end{theorem}

Our construction extends the protocol of \cite{BabaiFL91,BabaiFLS91}, which can be viewed as an IPCP that is later `compiled' into an MIP or a PCP. This protocol reduces the $\NEXP$-complete problem of determining the satisfiability of a `succinct' 3CNF, to testing whether there exists a low-degree polynomial satisfying an (exponentially large) set of constraints. A polynomial satisfying these constraints is necessarily a low-degree extension of a satisfying assignment, and thus implies the existence of such an assignment. The prover sends such a polynomial as an oracle, which the verifier then low-degree tests. That the constraints are satisfied can be checked using the sumcheck protocol.

To make this protocol zero knowledge we need to ensure that the oracle hides the original witness. We achieve this by sending not the low-degree extension itself but an algebraic commitment to it. The zero knowledge sumcheck then reduces the problem of checking the constraint on this witness to a single evaluation point of the low degree extension. However, this itself is not zero knowledge as evaluations of the low-degree extension can reveal information about the witness, especially if this evaluation is over the interpolating set $H^m$. Thus, our construction exploits the fact that the sumcheck protocol works for \emph{any} low-degree extension of the witness, and not just the one of minimal degree. Thus, the prover will instead send (the commitment to) a randomly sampled extension of the witness of slightly higher (but still constant) individual degree. The evaluations of this polynomial will (outside the interpolating set $H^m$) be $O(1)$-wise independent. As the sumcheck reduction will reduce to an evaluation point outside $H^m$ with high probability, the prover can then decommit the evaluation at this point to complete the sumcheck protocol without revealing any non-trivial information. We discuss this construction in more detail in \secref{sec:techniques-nexp}.

\subsubsection{Delegating polynomial space}
\label{sec:intro-pspace}

The above result shows that a powerful prover can convince a probabilistic polynomial-time verifier of $\NEXP$ statements in perfect zero knowledge in the IPCP model. Now we turn our attention to protocols where the honest prover need not be a $\NEXP$ machine. One such protocol, due to Shamir \cite{Shamir92}, provides an Interactive Proof for the $\PSPACE$-complete True Quantified Boolean Formula (TQBF) problem. This protocol is a more sophisticated application of the sumcheck protocol because sumcheck is applied \emph{recursively}.

We aim to obtain zero knowledge analogues for these types of more complex protocols as well but now, to tackle the greater complexity, we proceed in two steps. First, we design a generic framework called \emph{sum-product circuits} (which we believe to be of independent interest) that can express in a unified way a large class of `sumcheck-based Interactive Proofs', such as Shamir's protocol. Second, we show how to use our zero knowledge sumcheck protocol to obtain zero knowledge analogues of these, and thus also for Shamir's protocol. We discuss this further in \secref{sec:techniques-sum-product-circuits}.

As before, the resulting protocols are within the IPCP model. However, a key feature of these protocols that differentiates them from our result for $\NEXP$ is that the prover \emph{does not need the full power of IPCPs}: it suffices for the honest prover to send a PCP oracle that is the evaluation table of a random low-degree polynomial. Of course, soundness will continue to hold against any malicious prover that uses the PCP oracle in arbitrary ways.

\begin{theorem}[Informal version of \thmref{thm:pzk-for-TQBF}]
$\PSPACE$ has perfect zero knowledge Interactive PCPs, where the honest prover sends a random low-degree polynomial as the oracle.
\end{theorem}

As discussed above, any language having an IPCP where the honest prover sends a random low-degree polynomial (and achieves perfect completeness) can be decided in $\PSPACE$. In contrast, in the general case, deciding languages having IPCPs is $\NEXP$-hard. This result shows that moreover, all that is required to achieve unconditional zero knowledge for $\PSPACE$ is the ability to send a uniformly random polynomial as an oracle.

\subsubsection{Delegating low-depth circuits}
\label{sec:intro-gkr}

The doubly-efficient Interactive Proofs for low-depth circuits due to Goldwasser, Kalai, and Rothblum \cite{GoldwasserKR15} are another landmark result that makes a recursive use of the sumcheck protocol. Their construction can be viewed as a `scaled down' version of Shamir's protocol where the prover is efficient and the verifier is highly efficient.

We obtain a zero knowledge analogue of this protocol; again it suffices for the honest prover to send a random low-degree polynomial as the oracle (and soundness holds against any oracle). We do so by showing how the computation of low-depth circuits can be reduced to a corresponding sum-product circuit, by following the arithmetization in \cite{GoldwasserKR15}; then we rely on our zero knowledge results for sum-product circuits, mentioned above.

The protocol of \cite{GoldwasserKR15} is an IP for delegating certain \emph{tractable} computations: the evaluation of log-space uniform $\mathbf{NC}$ (circuits of polynomial size and polylogarithmic depth). The prover runs in polynomial time, while the verifier runs in quasilinear time and logarithmic space. But what does achieving zero knowledge mean in this case? If the simulator can run in polynomial time, then it can trivially simulate the verifier's view by simply running the honest prover. We thus need to consider a fine-grained notion of zero knowledge, by analyzing in more detail the overhead incurred by the simulator with respect to the malicious verifier's running time. This reckoning is similar to \cite{BermanRV17}, who study zero knowledge for Interactive Proofs of Proximity, and is a relaxation of \emph{knowledge tightness} \cite[Section 4.4.4.2]{Goldreich01}.

Concretely, we show that the running time of our simulator is a fixed (and small) polynomial in the verifier's running time, with only a polylogarithmic dependence on the size of the circuit. For example, the view of a malicious verifier running in time, say, $O(n^{2})$ can be simulated in time $\Ot{n^{6}}$. If the circuit has size $O(n^{8})$, then the zero knowledge guarantee is meaningful because the simulator is not able to evaluate the circuit.

\begin{theorem}[Informal version of \thmref{thm:gkr-space-uniform}]
Log-space uniform $\mathbf{NC}$ has perfect zero knowledge Interactive PCPs, where
the honest prover sends a random low-degree polynomial as the oracle, and
the verifier runs in quasilinear time and logarithmic space.
The simulator overhead is a small polynomial: the view of a verifier running in time $T$ can be simulated in time $T^{3} \cdot \polylog(n)$.
\end{theorem}

An interesting open problem is whether the simulator overhead can be reduced to $T \cdot \polylog(n)$, as required in the (quite strict) definition given in \cite[Section 4.4.4.2]{Goldreich01}. It seems that our techniques are unlikely to achieve this because they depend on solving systems of linear equations in $\Omega(T)$ variables.

Finally, a property in \cite{GoldwasserKR15} that has been very useful in subsequent work is that the verifier only needs to query a single point in the low-degree extension of its input. In this case, the verifier runs in polylogarithmic time and logarithmic space. Our zero knowledge analogue retains these properties. Additionally, in this setting the size of the circuit is subexponential in the running time of the verifier. Our zero knowledge guarantee then implies that we obtain zero knowledge under the standard (not fine-grained) definition.

\doclearpage
\section{Techniques}
\label{sec:techniques}

We summarize the techniques underlying our contributions. We begin in \secref{sec:techniques-nexp} by recalling the protocol of Babai, Fortnow, and Lund, in order to explain its sources of information leakage and how one could prevent them via algebraic techniques. This discussion motivates the goal of an \emph{algebraic} commitment scheme, described in \secref{sec:technniques-algebraic-commitment}. Then in \secref{sec:techniques-strong-zksc} we explain how to use this tool to obtain our main result, a zero knowledge sumcheck protocol.

The rest of the section is then dedicated to explaining how to achieve our other applications, which involve achieving zero knowledge for \emph{recursive} uses of the sumcheck protocol. First we explain in \secref{sec:techniques-recursion} what are the challenges that arise with regard to zero knowledge in recursive invocations of the sumcheck protocol, such as in the protocol of Shamir. Then in \secref{sec:techniques-sum-product-circuits} we describe the framework of sum-product circuits, and the techniques within it that allows us to achieve zero knowledge for the protocols of Shamir and of Goldwasser, Kalai, and Rothblum.

\subsection{An algebraic approach for zero knowledge in the BFL protocol}
\label{sec:techniques-nexp}

We recall the protocol of Babai, Fortnow, and Lund \cite{BabaiFL91} (`BFL protocol'), in order to explain its sources of information leakage and how one could prevent them via algebraic techniques. These are the ideas that underlie our algebraic construction of an unconditional (perfect) zero knowledge IPCP for $\NEXP$ (see \secref{sec:intro-nexp}).

\parhead{The BFL protocol, and why it leaks}
The $\mathrm{O3SAT}$ problem is the following $\NEXP$-complete problem: given a boolean formula $B$, does there exist a boolean function $A$ such that 
\begin{equation*}
	B(z, b_{1}, b_{2}, b_{3}, A(b_{1}), A(b_{2}), A(b_{3}))=0 \quad \text{for all } z \in \Bits^{r}, b_{1}, b_{2}, b_{3} \in \Bits^{s}
	\;\;\text{?}
\end{equation*}
The BFL protocol constructs an IPCP for $\mathrm{O3SAT}$ and later converts it to an MIP. Only the first step is relevant for us.

In the BFL protocol, the honest prover first sends a PCP oracle $\LD{A} \colon \Field^{s} \to \Field$ that is the unique multilinear extension (in some finite field $\Field$) of a valid witness $A \colon \Bits^{s} \to \Bits$. The verifier must check that
\begin{inparaenum}[(a)]
	\item $\LD{A}$ is a boolean function on $\Bits^{s}$, and
	\item $\LD{A}$'s restriction to $\Bits^{s}$ is a valid witness for $B$.
\end{inparaenum}
To do these checks, the verifier arithmetizes $B$ into an arithmetic circuit $\LD{B}$, and reduces the checks to conditions that involve $\LD{A}$, $\LD{B}$, and other low-degree polynomials. A technique of \cite{BabaiFLS91} allows the verifier to `bundle' all of these conditions into a low-degree polynomial $f$ such that (with high probability over the choice of $f$) the conditions hold if and only if $f$ sums to $0$ on $\Bits^{r+3s+3}$. The verifier checks that this is the case via a sumcheck protocol with the prover. The soundness of the sumcheck protocol depends on the PCP oracle being the evaluation of a low-degree polynomial; the verifier checks this using a low-degree test.

We see that the BFL protocol is \emph{not} zero knowledge for two reasons:
\begin{inparaenum}[(i)]
\item the verifier has oracle access to $\LD{A}$ and, in particular, to the witness $A$;
\item the prover's messages during the sumcheck protocol leak further information about $A$ (namely, hard-to-compute partial sums of $f$, which itself depends on $A$).
\end{inparaenum}

\parhead{A blueprint for zero knowledge}
We now describe the `blueprint' for an approach to achieve zero knowledge in the BFL protocol. The prover does not send $\LD{A}$ directly but instead a \emph{commitment} to it. After this, the prover and verifier engage in a sumcheck protocol with suitable zero knowledge guarantees; at the end of this protocol, the verifier needs to evaluate $f$ at a point of its choice, which involves evaluating $\LD{A}$ at three points. Now the prover reveals the requested values of $\LD{A}$, without leaking any information beyond these, so that the verifier can perform its check. We explain how these ideas motivate the need for certain algebraic tools, which we later obtain and use to instantiate our approach.

\parhead{(1) Randomized low-degree extension}
Even if the prover reveals only three values of $\LD{A}$, these may still leak information about $A$. We address this problem via a \emph{randomized low-degree extension}. Indeed, while the prover in the BFL protocol sends the \emph{unique} multilinear extension of $A$, one can verify that \emph{any} extension of $A$ of sufficiently low degree also works. We exploit this flexibility as follows: the prover randomly samples $\LD{A}$ in such a way that any three evaluations of $\LD{A}$ do not reveal any information about $A$. Of course, if any of these evaluations is within $\Bits^{s}$, then no extension of $A$ has this property. Nevertheless, during the sumcheck protocol, the prover can ensure that the verifier chooses only evaluations outside of $\Bits^{s}$ (by aborting if the verifier deviates), which incurs only a small increase in the soundness error. With this modification in place, it suffices for the prover to let $\LD{A}$ be a random degree-$4$ extension of $A$: by a dimensionality argument, any $3$ evaluations outside of $\Bits^{s}$ are now independent and uniformly random in $\Field$. Remarkably, we are thus able to reduce a claim about $A$ to a claim which contains \emph{no information} about $A$.

\parhead{(2) Low-degree testing the commitment}
The soundness of the sumcheck protocol relies on $f$ having low degree, or at least being close to a low-degree polynomial. This in turn depends on the PCP oracle $\LD{A}$ being close to a low-degree polynomial. If the prover sends $\LD{A}$, the verifier can simply low-degree test it. However, if the prover sends a commitment to $\LD{A}$, then it is not clear what the verifier should do. One option would be for the prover to reveal \emph{more} values of $\LD{A}$ (in addition to the aforementioned three values), in order to enable the verifier to conduct its low-degree test on $\LD{A}$. The prover would then have to ensure that revealing these additional evaluations is `safe' by increasing the amount of independence among values in $\LD{A}$ (while still restricting the verifier to evaluations outside of $\Bits^{s}$), which would lead to a blowup in the degree of $\LD{A}$ that is proportional to the number of queries that the verifier wishes to make. For the low-degree test to work, however, the verifier \emph{must} see more evaluations than the degree. In sum, this circularity is inherent. To solve this problem, we will design an `algebraic' commitment scheme that is \emph{transparent to low-degree tests}: the verifier can perform a low-degree test on the commitment itself (without the help of the prover), which will ensure access to a $\LD{A}$ that is low-degree. We discuss this further in \secref{sec:technniques-algebraic-commitment}.

\parhead{(3) Sumcheck in zero knowledge}
We need a sumcheck protocol where the prover's messages leak little information about $f$. The prior work in \cite{BenSassonCFGRS16} achieves an IPCP for sumcheck that is `weakly' zero knowledge: any verifier learns at most one evaluation of $f$ for each query it makes to the PCP oracle. If the verifier could evaluate $f$ by itself, as was the case in that paper, this guarantee would suffice for zero knowledge. In our setting, however, the verifier \emph{cannot} evaluate $f$ by itself because $f$ is (necessarily) hidden behind the algebraic commitment.

One approach to compensate would be to further randomize $\LD{A}$ by letting $\LD{A}$ be a random extension of $A$ of some well-chosen degree $d$. We are limited to $d$ of polynomial size because the honest verifier's running time is $\Omega(d)$. But this means that a polynomial-time malicious verifier, participating in the protocol of \cite{BenSassonCFGRS16} and making $d^{2}$ queries to the PCP oracle, could learn information about $A$.

We resolve this by relying on more algebraic techniques, achieving an IPCP for sumcheck with a much stronger zero knowledge guarantee (see \thmref{thm:strong-sumcheck-ipcp-intro}): any malicious verifier that makes polynomially-many queries to the PCP oracle learns only a \emph{single} evaluation of $f$. This suffices for zero knowledge in our setting: learning one evaluation of $f$ implies learning only three evaluations of $\LD{A}$, which can be made `safe' if $\LD{A}$ is chosen to be a random extension of $A$ of high-enough degree. Our sumcheck protocol uses as building blocks both our algebraic commitment scheme and the \cite{BenSassonCFGRS16} sumcheck; we summarize its construction in \secref{sec:techniques-strong-zksc}.

\begin{remark}
Kilian, Petrank, and Tardos \cite{KilianPT97} construct PCPs for $\NEXP$ that are statistical zero knowledge, via a combinatorial construction that makes black-box use of the PCP Theorem. Our modification of the BFL protocol achieves a perfect zero knowledge IPCP via algebraic techniques, avoiding the use of the PCP Theorem.
\end{remark}

\subsection{Algebraic commitments from algebraic query complexity lower bounds}
\label{sec:technniques-algebraic-commitment}

We describe how the sumcheck protocol can be used to construct an information-theoretic commitment scheme that is `algebraic', in the IPCP model. (Namely, an algebraic \emph{interactive locking scheme}; see \remref{rem:gims-comparison} below.) The prover commits to a message by sending to the verifier a PCP oracle that perfectly hides the message; subsequently, the prover can reveal positions of the message by engaging with the verifier in an Interactive Proof, whose soundness guarantees statistical binding. A key algebraic property that we rely on is that the commitment is `transparent' to low-degree tests.

\parhead{Committing to an element}
We first consider the simple case of committing to a single element $a$ in $\Field$. Let $k$ be a security parameter, and set $N \DefineEqual 2^{k}$. Suppose that the prover samples a random $B$ in $\Field^{N}$ such that $\sum_{i=1}^{N} B_{i} = a$, and sends $B$ to the verifier as a commitment. Observe that any $N-1$ entries of $B$ do not reveal any information about $a$, and so any verifier with oracle access to $B$ that makes less than $N$ queries cannot learn any information about $a$. However, as $B$ is unstructured it is not clear how the prover can convince the verifier that $\sum_{i=1}^{N} B_{i} = a$.

Instead, we can consider imbuing $B$ with additional structure by providing its low-degree extension. That is, the prover thinks of $B$ as a function from $\Bits^{k}$ to $\Field$, and sends its unique multilinear extension $\LD{B} \colon \Field^{k} \to \Field$ to the verifier. Subsequently, the prover can reveal $a$ to the verifier, and then engage in a sumcheck protocol for the claim ``$\sum_{\vec{\beta} \in \Bits^{k}} \LD{B}(\vec{\beta}) = a$'' to establish the correctness of $a$. The soundness of the sumcheck protocol protects the verifier against cheating provers and hence guarantees that this scheme is binding.

However, giving $B$ additional structure calls into question the hiding property of the scheme. Indeed, surprisingly a result of \cite{JumaKRS09} shows that this new scheme is \emph{not} hiding (in fields of characteristic different than 2): it holds that $\LD{B}(2^{-1}, \ldots, 2^{-1}) = a \cdot 2^{-k}$ for any choice of $B$, so the verifier can learn $a$ with only a single query to $\LD{B}$!

Sending an extension of $B$ has created a new problem: querying the extension outside of $\Bits^{k}$, the verifier can learn information that may require many queries to $B$ to compute. Indeed, this additional power is precisely what underlies the soundness of the sumcheck protocol. To resolve this, we need to understand what the verifier can learn about $B$ given some low-degree extension $\LD{B}$. This is precisely the setting of \emph{algebraic query complexity} \cite{AaronsonW09}.

A natural approach is to let $\LD{B}$ be chosen uniformly at random from the set of degree-$d$ extensions of $B$ for some $d > 1$. It is not hard to see that if $d$ is very large (say, $\SetCardinality{\Field}$) then $2^{k}$ queries are required to determine the summation of $\LD{B}$ on $H^{m}$. But we need $d$ to be small to achieve soundness. A result of \cite{JumaKRS09} shows that $d = 2$ suffices: given a random multiquadratic extension $\LD{B}$ of $B$, one needs $2^{k}$ queries to $\LD{B}$ to determine $\sum_{\vec{\beta} \in \Bits^{k}} \LD{B}(\vec{\beta})$.

\parhead{Committing to a polynomial}
The prover in our zero knowledge protocols needs to commit not just to a single element but to the evaluation of an $m$-variate polynomial $Q$ over $\Field$ of degree $d_{Q}$. We extend our ideas to this setting.

Letting $K$ be a subset of $\Field$ of size $d_{Q}+1$, the prover samples a random $B^{\vec{x}}$ in $\Field^{N}$ such that $\sum_{i=1}^{N} B^{\vec{x}}_{i} = Q(\vec{x})$ for each $\vec{x} \in K^{m}$. We can view all of these strings as a single function $B \colon K^{m} \times \Bits^{k} \to \Field$, and as before we consider its unique low-degree extension $\LD{B} : \Field^{m} \times \Field^{k} \to \Field$; viewed as a polynomial, $\LD{B}(\vec{\VariableX},\vec{\VariableY})$ has degree at most $d_{Q}$ in $\vec{\VariableX}$ and is multilinear in $\vec{\VariableY}$. Observe that since $\sum_{\vec{\beta} \in \Bits^{k}} \LD{B}(\vec{\VariableX}, \vec{\beta})$ is a polynomial of individual degree $d_{Q}$ that agrees with $Q$ on $K^{m}$, it must equal $Q$. The binding property of the commitment scheme is clear: the prover can decommit to $Q(\vec{\alpha})$ for any $\vec{\alpha} \in \Field^{m}$ by using the sumcheck protocol as before. We are left to argue the hiding property.

It is not difficult to see that we run into the same issue as in the single-value case: we have $\LD{B}(\vec{\alpha}, 2^{-1}, \ldots, 2^{-1}) = Q(\vec{\alpha}) \cdot 2^{-k}$ for any $\vec{\alpha} \in \Field^{m}$. We resolve this by again choosing a \emph{random} extension $\LD{B}$ of degree $d > 1$ in $\vec{\VariableY}$ (and degree $d_{Q}$ in $\vec{\VariableX}$). Yet, arguing the hiding property now requires a \emph{stronger} statement than the one proved in \cite{JumaKRS09}. Not only do we need to know that the verifier cannot determine $Q(\vec{\alpha})$ for a particular $\vec{\alpha} \in \Field^{m}$, but we need to know that the verifier cannot determine $Q(\vec{\alpha})$ for \emph{any} $\vec{\alpha} \in \Field^{m}$, or even \emph{any linear combination of any such values}. We prove that this stronger guarantee holds in the same parameter regime: if $d > 1$ then $2^{k}$ queries are both necessary and sufficient.

\parhead{Transparency to low-degree tests}
Recall that a key algebraic property we required from our commitment scheme is that the verifier can perform a low-degree test on the committed polynomial without the assistance of the prover. Our commitment scheme naturally has this property. If the PCP oracle $\Malicious{B} \colon \Field^{m} \times \Field^{k} \to \Field$ is low-degree, then it is a commitment to the low-degree $Q \colon \Field^{m} \to \Field$ defined as $Q(\vec{\VariableX}) \DefineEqual \sum_{\vec{\beta} \in \Bits^{k}} \Malicious{B}(\vec{\VariableX}, \vec{\beta})$. In fact, even if $\Malicious{B} \colon \Field^{m} \times \Field^{k} \to \Field$ is merely \emph{close} to a low-degree $\LD{B} \colon \Field^{m} \times \Field^{k} \to \Field$, then we can still regard $\Malicious{B}$ as a commitment to the low-degree $Q \colon \Field^{m} \to \Field$ defined as $Q(\vec{\VariableX}) \DefineEqual \sum_{\vec{\beta} \in \Bits^{k}} \LD{B}(\vec{\VariableX}, \vec{\beta})$, because the verifier can check claims of the form ``$Q(\vec{\alpha}) = a$'' by obtaining the value of $\LD{B}$ it needs at the end of the sumcheck protocol via self-correction on $\Malicious{B}$.

\parhead{Beyond the boolean hypercube}
For efficiency reasons analogous to those in \cite{BabaiFLS91,GoldwasserKR15}, instead of viewing $B$ as a function from $K^{m} \times \Bits^{k}$ to $\Field$, we view $B$ as a function from $K^{m} \times H^{k'}$ to $\Field$ for a subset $H$ of $\Field$ of size $\Field^{\Omega(1)}$ and $k' \DefineEqual \log N/\log \SetCardinality{H}$. This requires us to extend our claims about the algebraic query complexity of polynomial summation to arbitrary sets $H$. We show that if $d > 2(\SetCardinality{H} - 1)$, then $\SetCardinality{H}^{k'} = N$ queries are necessary to determine $Q(\vec{\alpha})$ for any $\vec{\alpha}$ (or any linear combination of these). See \secref{sec:algebraic-query-complexity} for details.

\parhead{Decommitting in zero knowledge}
To use our commitment scheme in zero knowledge protocols, we must ensure that, in the decommitment phase, the verifier cannot learn any information beyond the value $a \DefineEqual Q(\vec{\alpha})$ for a chosen $\vec{\alpha}$. To decommit, the prover sends the value $a$ and has to convince the verifier that the claim ``$\sum_{\vec{\beta} \in \Bits^{k}} \LD{B}(\vec{\alpha}, \vec{\beta}) = a$'' is true. However, if the prover and verifier simply run the sumcheck protocol on this claim, the prover leaks partial sums $\sum_{\vec{\beta} \in \Bits^{k-i}} \LD{B}(\vec{\alpha}, c_{1}, \ldots, c_{i}, \vec{\beta})$ for $c_{1}, \ldots, c_{i} \in \Field$ chosen by the verifier, which could reveal additional information about $Q$. Instead, the prover and verifier run on this claim the IPCP for sumcheck of \cite{BenSassonCFGRS16}, whose `weak' zero knowledge guarantee ensures that this cannot happen. (Thus, in addition to the commitment, the honest prover also sends the evaluation of a random low-degree polynomial as required by the IPCP for sumcheck of \cite{BenSassonCFGRS16}.)

\begin{remark}[comparison with \cite{GoyalIMS10}]
\label{rem:gims-comparison}
Goyal, Ishai, Mahmoody, and Sahai \cite{GoyalIMS10} define and construct \emph{interactive locking schemes}, information-theoretic commitment schemes in the IPCP model. Their scheme is combinatorial, and we do not know how to use it in our setting (it is not clear how to low-degree test the committed message without disrupting zero knowledge). Putting this difference aside, their construction and our construction are incomparable. On the one hand, we achieve perfect hiding while they only achieve statistical hiding. On the other hand, their scheme is `oracle efficient' (any query to the oracle can be computed statelessly in polynomial time) while our scheme is not.
\end{remark}

\subsection{A zero knowledge sumcheck protocol}
\label{sec:techniques-strong-zksc}

We summarize the ideas behind our main result, a zero knowledge sumcheck protocol (see \thmref{thm:strong-sumcheck-ipcp-intro}). This result not only enables us to modify the BFL protocol to achieve zero knowledge (as discussed above), but also to modify the Shamir and GKR protocols to achieve zero knowledge (as discussed below). The two building blocks underlying our sumcheck protocol are our algebraic commitments (see \secref{sec:technniques-algebraic-commitment} above) and the IPCP for sumcheck of \cite{BenSassonCFGRS16}. We now cover necessary background and then describe our protocol.

\parhead{Previous sumcheck protocols}
The sumcheck protocol \cite{LundFKN92} is an IP for claims of the form ``$\sum_{\vec{\alpha} \in H^{m}} F(\vec{\alpha}) = 0$'', where $H$ is a subset of a finite field $\Field$ and $F$ is an $m$-variate polynomial over $\Field$ of small individual degree. The protocol has $m$ rounds: in round $i$, the prover sends the univariate polynomial $g_{i}(\VariableX_{i}) := \sum_{\vec{\alpha} \in H^{m-i}} F(c_{1}, \ldots, c_{i-1}, \VariableX_{i}, \vec{\alpha})$; the verifier checks that $\sum_{\alpha_{i} \in H} g_{i}(\alpha_{i}) = g_{i-1}(c_{i-1})$ and replies with a uniformly random challenge $c_{i} \in \Field$. After round $m$, the verifier outputs the claim ``$F(c_{1}, \dots, c_{m}) = g_{m}(c_{1}, \dots, c_{m})$''. If $F$ is of sufficiently low degree and does not sum to $\SCSum$ over the space, then the output claim is false with high probability. Note that the verifier does not need access to $F$.

The IPCP for sumcheck of \cite{BenSassonCFGRS16} modifies the above protocol as follows. The prover first sends a PCP oracle that equals the evaluation of a random `masking' polynomial $R$; the verifier checks that $R$ is (close to) low degree. After that the prover and verifier conduct an Interactive Proof. The prover sends $z \in \Field$ that allegedly equals $\sum_{\vec{\alpha} \in H^{m}} R(\vec{\alpha})$, and the verifier responds with a uniformly random challenge $\rho \in \Field^{*}$. The prover and verifier now run the (standard) sumcheck protocol to reduce the claim ``$\sum_{\vec{\alpha} \in H^{m}} \rho F(\vec{\alpha}) + R(\vec{\alpha}) = \rho \SCSum + z$'' to a claim ``$\rho F(\vec{c}) + R(\vec{c})=b$'' for random $\vec{c} \in \Field^{m}$. The verifier queries $R$ at $\vec{c}$ and then outputs the claim $``F(\vec{c})=\frac{b-R(\vec{c})}{\rho}$''. If $\sum_{\vec{\alpha} \in H^{m}} F(\vec{\alpha}) \neq a$ then with high probability over the choice of $\rho$ and the verifier's messages in the sumcheck protocol, this claim will be false.

A key observation is that if the verifier makes no queries to $R$, then the prover's messages are identically distributed to the sumcheck protocol applied to a uniformly random polynomial $Q$. When the verifier does make queries to $R$, simulating the resulting conditional distribution involves techniques from Algebraic Complexity Theory, as shown in \cite{BenSassonCFGRS16}. Given $Q$, the verifier's queries to $R(\vec{\alpha})$ for $\vec{\alpha} \in \Field^{m}$ are identically distributed to $Q(\vec{\alpha}) - \rho F(\vec{\alpha})$. Thus the simulator need only make at most one query to $F$ for every query to $R$. That is, any verifier making $q$ queries to $R$ learns no more than it would learn by making $q$ queries to $F$ alone.

As discussed, this zero knowledge guarantee does not suffice for the applications that we consider: when a sumcheck protocol is used as a subroutine of another protocol, $F$ may itself be recursively defined in terms of large sums which the verifier cannot evaluate on its own. The verifier does, however, have oracle access to $R$, and so can learn enough information about $F$ to break zero knowledge. 

\parhead{Our sumcheck protocol}
The zero knowledge guarantee that we aim for is the following: any polynomial-time verifier learns no more than it would by making \emph{one} query to $F$, regardless of its number of queries to the PCP oracle.

The main idea to achieve this guarantee is the following. The prover sends a PCP oracle that is an \emph{algebraic commitment} $Z$ to the aforementioned masking polynomial $R$. Then, as before, the prover and verifier run the sumcheck protocol to reduce the claim ``$\sum_{\vec{\alpha} \in H^{m}} \rho F(\vec{\alpha}) + R(\vec{\alpha}) = \rho a + z$'' to a claim ``$\rho F(\vec{c}) + R(\vec{c})=b$'' for random $\vec{c} \in \Field^{m}$.

We now face two problems. First, the verifier cannot simply query $R$ at $\vec{c}$ and then output the claim $``F(\vec{c})=\frac{b-R(\vec{c})}{\rho}$'', since the verifier only has oracle access to the commitment $Z$ of $R$. Second, the prover could cheat the verifier by having $Z$ be a commitment to an $R$ that is far from low degree, which allows cheating in the sumcheck protocol.

The first problem is addressed by the fact that our algebraic commitment scheme has a decommitment sub-protocol that is zero knowledge: the prover can reveal $R(\vec{c})$ in such a way that no other values about $R$ are also revealed as a side-effect. As discussed, this relies on the protocol of \cite{BenSassonCFGRS16}, used a subroutine (for the second time).

The second problem is taken care of by the fact that our algebraic commitment scheme is `transparent' to low-degree tests: the verifier simply performs a low-degree test on $Z$, which by self-correction gives the verifier oracle access to a low-degree $Z'$ that \emph{is} a commitment to a low-degree $R$.

Overall, the only value that a malicious verifier can learn is $F(\vec{c})$ for $\vec{c} \in \Field^{m}$ of its choice.

\begin{remark}
Our sumcheck protocol `leaks' a single evaluation of $\SCPoly$. We believe that this limitation is inherent: the honest verifier always outputs a true claim about one evaluation of $\SCPoly$, which it cannot do without learning that evaluation. Either way, this guarantee is strong enough for applications: we ensure that learning a single evaluation of $\SCPoly$ does not harm zero knowledge, either because it carries no information or because the verifier can evaluate $\SCPoly$ itself.
\end{remark}

\subsection{Challenges: handling recursion}
\label{sec:techniques-recursion}

We have so far discussed the ideas behind our main result (a zero knowledge sumcheck protocol) and how to use it to achieve a natural zero knowledge analogue of the classical MIP/PCP construction for $\NEXP$ \cite{BabaiFL91,BabaiFLS91}. Other applications require additional ideas to overcome challenges that arise when the sumcheck protocol is used \emph{recursively}.

\parhead{Shamir's protocol}
Consider the goal of achieving zero knowledge for Shamir's protocol for $\PSPACE$ \cite{Shamir92}. This protocol reduces checking any $\PSPACE$ computation to checking that:
\begin{equation*}
	\sum_{x_{1} \in \Bits} \prod_{x_{2} \in \Bits} \ldots \sum_{x_{n-1} \in \Bits} \prod_{x_{n} \in \Bits} \LD{\Formula}(x_{1}, \ldots, x_{n}) = 0
\end{equation*}
where $\LD{\Formula}$ is the (efficiently computable) arithmetization over a finite field $\Field$ of a certain boolean formula $\Formula$. Since Shamir's protocol is similar to the sumcheck protocol, a natural starting point would be to try to merely adapt the techniques that `worked' in the case of the sumcheck protocol. However, the similarity between the two protocols is only superficial (e.g., it lacks the useful linear structure present in the sumcheck protocol). An accurate way to compare the two is to view Shamir's protocol as a \emph{recursive} application of the sumcheck protocol, as shown by Meir \cite{Meir13}.

For example, the TQBF problem is \emph{downward self-reducible} \cite{TrevisanV07}: for $i \in \{1,\dots,n\}$ let
\begin{equation*}
G_{i}(\VariableX_{1}, \ldots, \VariableX_{i}) \DefineEqual \sum_{x_{i+1} \in \Bits} \prod_{x_{i+2} \in \Bits} \ldots \sum_{x_{n-1} \in \Bits} \prod_{x_{n} \in \Bits} \LD{\Formula}(\VariableX_{1}, \ldots, \VariableX_{i}, x_{i+1}, \ldots, x_{n}) \enspace.
\end{equation*}
From this one obtains the recurrence
\begin{equation}
\label{eqn:shamir-recurrence}
G_{i}(\VariableX_{1}, \ldots, \VariableX_{i}) = \sum_{x_{i+1} \in \Bits} G_{i+2}(\VariableX_{1}, \ldots, \VariableX_{i}, x_{i+1}, 0) \cdot G_{i+2}(\VariableX_{1}, \ldots, \VariableX_{i}, x_{i+1}, 1) \enspace.
\end{equation}
In other words, with oracle access to $G_{i+2}$, one can use a (small) sumcheck to compute $G_{i}$. This suggests a recursive approach: if we could check evaluations of $G_{i+2}$ in zero knowledge, then maybe we could use this as a subprotocol to check evaluations of $G_{i}$ also in zero knowledge.

\parhead{GKR's protocol}
The recursive structure is perhaps more evident in the doubly-efficient Interactive Proof of Goldwasser, Kalai, and Rothblum \cite{GoldwasserKR15} (`GKR protocol'). Its barebones sub-protocol checks a more complex arithmetic expression: the output of a layered arithmetic circuit. Fix an input $x$ to the circuit, and let $\Layer{i}(j) \colon [S] \to \Field$ be the value of the $j$-th gate in layer $i$ ($S$ is the number of gates in a layer). For some subset $\SPSubset \subseteq \Field$ and sufficiently large $\GKRVars$, one views $\Layer{i}$ as a function from $\SPSubset^{\GKRVars}$ to $\Field$ by imposing some ordering on $\SPSubset^{\GKRVars}$. One can relate $\Layer{i-1}$ to $\Layer{i}$ as follows:
\begin{equation}
\label{eqn:gkr-recurrence}
\Layer{i-1}(\vec{z}) = \sum_{\vec{\omega}_{1},\vec{\omega}_{2} \in \SPSubset^{\GKRVars}} \Add{i}(\vec{z}, \vec{\omega}_{1}, \vec{\omega}_{2}) \cdot \big(\Layer{i}(\vec{\omega}_{1}) + \Layer{i}(\vec{\omega}_{2})\big)
+ \Mult{i}(\vec{z}, \vec{\omega}_{1}, \vec{\omega}_{2}) \cdot \big(\Layer{i}(\vec{\omega}_{1}) \cdot \Layer{i}(\vec{\omega}_{2})\big)
\end{equation}
where $\Add{i}(\vec{z}, \vec{\omega}_{1}, \vec{\omega}_{2})$ is $1$ if the $\vec{z}$-th gate in layer $i-1$ is an addition gate whose inputs are the $\vec{\omega}_{1}$-th and $\vec{\omega}_{2}$-th gates in layer $i$, and $\Mult{i}$ is defined similarly for multiplication gates.

We again see a recursive structure: a function defined as the summation over some product space of a polynomial whose terms are functions of the same form; this allows to check $\Layer{i-1}$ given a protocol for checking $\Layer{i}$.

The use of recursion in the GKR protocol is even more involved: the barebones protocol relies on the verifier having oracle access to low-degree extensions of $\Add{i}$ and $\Mult{i}$. For very restricted classes of circuits, the verifier can efficiently `implement' these oracles; however, for the class of circuits that is ultimately supported by the protocol this requires a further sub-protocol that delegates the evaluation of these oracles to the prover, and this is done by \emph{composing} multiple instances of the GKR protocol. To achieve zero knowledge we also have to tackle this form of recursion.

\parhead{The leakage of recursion}
By now the central role of recursion in applications of the sumcheck protocol is clear. There are two main sources of leakage that we need to overcome in order to achieve zero knowledge in such applications.
\begin{enumerate}

  \item Checking evaluations of $G_{i+2}$, $V_{i}$, or $\Add{i}$ and $\Mult{i}$, even in zero knowledge, leaks the evaluations themselves. The verifier, however, is not able to compute these itself (else it would not need to delegate), which means that information is leaked.

  \item The number of claims can grow exponentially: a claim about $G_{i}$ (resp. $V_{i-1}$) is reduced to \emph{two} claims about $G_{i+2}$ (resp. $V_{i}$). There are standard techniques that leverage interaction to reduce multiple claims about a low-degree polynomial to a single one, but we need to replace these with zero knowledge equivalents.

\end{enumerate}
We tackle both issues by devising a general framework that captures their shared algebraic structure, solving these problems within this framework, and then recovering the protocols of Shamir and GKR as special cases.

\subsection{Sum-product circuits}
\label{sec:techniques-sum-product-circuits}

We introduce the notion of \emph{sum-product circuits} and show that the sumcheck protocol naturally gives rise to algebraic Interactive Proofs for checking the value of such circuits. We then explain how to achieve zero knowledge variants of these by building on the techniques discussed in \secref{sec:techniques-nexp}. We recover zero knowledge variants of the protocols of Shamir and GKR as special cases of this approach.

Sum-product circuits are an abstract way of encoding `sum-product expressions'. A sum-product expression is either a polynomial over some finite field $\Field$ represented by a small arithmetic circuit or a polynomial of the form
\begin{equation}
\label{eqn:sum-product-exp}
\sum_{\vec{\beta} \in H^{m}} C\big(\vec{\VariableX}, \vec{\beta}, P_{1}(\vec{\VariableX}, \vec{\beta}), \ldots, P_{n}(\vec{\VariableX}, \vec{\beta})\big)
\end{equation}
where $C$ is a low-degree `combiner' polynomial represented by a small arithmetic circuit, and $P_{1}, \ldots, P_{n}$ are sum-product expressions. Both \eqnref{eqn:shamir-recurrence} (for Shamir's protocol) and \eqnref{eqn:gkr-recurrence} (for GKR's protocol) are of this form.

Like a standard arithmetic circuit, a sum-product circuit is a directed acyclic graph associated with a field $\Field$ in which we associate to each vertex a \emph{value}, which in our case is the sum-product expression that it computes. Each internal vertex is labeled by a combiner polynomial, and there is an edge from $u$ to $v$ if the sum-product expression of $v$ appears in that of $u$. For example, the above expression would correspond to a vertex labeled with $C$, with outgoing edges to the vertices corresponding to $P_{1}, \ldots, P_{n}$. An input to the circuit is a labeling of the leaf vertices with small arithmetic circuits. We now spell this out in a little more detail.

\begin{definition}[Informal version of \defref{def:sum-product-circuit}]
A \defemph{sum-product circuit} $\SPCircuit$ is a rooted directed acyclic graph where each internal vertex is labeled with an arithmetic circuit $\SPPoly[v]$ over a finite field $\Field$. An \emph{input} $\SPInput$ to $\SPCircuit$ labels each leaf $v$ with a polynomial $\SPLeaf[v]$ over $\Field$. The value of a vertex $v$ on input $\SPInput$ is a multivariate polynomial $\SPValueL{\SPInput}{v}$ over $\Field$ defined as follows: if $v$ is a leaf vertex then $\SPValueL{\SPInput}{v}$ equals $\SPLeaf[v]$; if instead $v$ is an internal vertex then, for a chosen integer $m$,
\begin{equation}
	\label{eqn:vertex-value-informal}
	\SPValueL{\SPInput}{v}(\vec{\VariableX}) \DefineEqual \sum_{\vec{\beta} \in \SPSubset^{\SPVars}} \SPPoly[v]\big(\vec{\VariableX},\vec{\beta},
	\SPValueL{\SPInput}{u_{1}}(\vec{\VariableX}, \vec{\beta}),\ldots,\SPValueL{\SPInput}{u_{\aodeg}}(\vec{\VariableX}, \vec{\beta})
	\big) \enspace.
\end{equation}
The \defemph{value} of $\SPCircuit$ on input $\SPInput$ is denoted $\SPValueL{\SPInput}{\SPCircuit}$ and equals the value of the root vertex $\Root$ (and we require that $\SPValueL{\SPInput}{\SPCircuit} \in \Field$).
\end{definition}

We next describe an Interactive Proof that works for any sum-product circuit. The protocols of Shamir \cite{Shamir92} and of GKR \cite{GoldwasserKR15} can be viewed as this protocol applied to specific sum-product circuits (computing $G_{0}$ and $V_{0}$ respectively). After that, we explain how to modify the Interactive Proof to obtain a corresponding zero knowledge IPCP for any sum-product circuit, which allows us to derive our zero knowledge variants of these two protocols. 

A significant advantage of working with sum-product circuits is that they are easy to compose. For example, we can view the composition of the GKR protocol with itself as a \emph{composition of sum-product circuits}. We can then apply our zero knowledge IPCP to the resulting circuit and directly obtain a zero knowledge analogue of the full GKR protocol.

\subsubsection{Delegating the evaluation of a sum-product circuit}

We explain how to use the sumcheck protocol to obtain an Interactive Proof for checking the value of a sum-product circuit. The protocol is recursively defined: to prove that $\SPValueL{\SPInput}{\SPCircuit} = a$ (i.e., that $\SPValueL{\SPInput}{r} = a$), it suffices to show that the values of $r$'s children $u_{1}, \ldots, u_{t}$ satisfy \eqnref{eqn:vertex-value-informal} where the left-hand side is $a$. The sumcheck protocol interactively reduces this claim to a new claim ``$\SPPoly[v]\big(\vec{c},\SPValueL{\SPInput}{u_{1}}( \vec{c}),\ldots,\SPValueL{\SPInput}{u_{\aodeg}}(\vec{c})\big) = b$'' for $\vec{c} \in \Field^{m}$ chosen uniformly at random by the verifier and $b \in \Field$ chosen by the prover. The prover sends $h_{1} \DefineEqual \SPValueL{\SPInput}{u_{1}}(\vec{c}), \ldots, h_{\aodeg} \DefineEqual \SPValueL{\SPInput}{u_{\aodeg}}(\vec{c})$, reducing this new claim to the set of claims ``$h_{i} = \SPValueL{\SPInput}{u_{i}}(\vec{c})$'' for $i = 1, \ldots, \aodeg$ and ``$\SPPoly[v](\vec{c}, h_{1}, \ldots, h_{\aodeg}) = b$''. The latter can be checked by the verifier directly and the rest can be recursively checked via the same procedure. Eventually the protocol reaches the leaf vertices, which are labeled with small arithmetic circuits that the verifier can evaluate on its own.

One technicality is that, as defined, the degree of the polynomial at a vertex may be exponentially large, and so the prover would have to send exponentially-large messages in the sumcheck protocol. To avoid this, we use a well-known interactive sub-protocol for \emph{degree reduction} \cite{Shen92,GoldwasserKR15}. Since for all $\vec{x}$ the value $\SPValueL{\SPInput}{v}(\vec{x})$ depends only on $\SPValueL{\SPInput}{u_{1}}(\vec{x},\vec{\beta}), \ldots, \SPValueL{\SPInput}{u_{t}}(\vec{x},\vec{\beta})$ for $\vec{\beta} \in \SPSubset^{\SPVars}$, we can safely replace each $\SPValueL{\SPInput}{u_{i}}$ with the unique degree-$(\SetCardinality{\SPSubset}-1)$ extension $\LDSPValueL{\SPInput}{u_{i}}$ of its evaluation over $\SPSubset^{\SPVars}$. The degree of the summand in \eqnref{eqn:vertex-value-informal} is now at most $\delta \SetCardinality{H}$, where $\delta$ is the \emph{total} degree of $\SPPoly[v]$. Now that the sumcheck protocol is only invoked on low-degree polynomials, efficiency is recovered.

Another technicality is that since a sum-product circuit is a directed acyclic graph (as opposed to a tree), it is possible that a single vertex $v$ will have many claims about it. If each such claim reduces to many claims about other vertices, the number of claims to check could grow exponentially. This is in fact the case in both Shamir's and GKR's protocols. To avoid this blowup, the verifier checks a random linear combination of the claims about each vertex $v$. It is not difficult to see that soundness is preserved, and the number of claims per vertex is reduced to one.

\subsubsection{Achieving zero knowledge}

The Interactive Proof for sum-product circuits that we have described above is not zero knowledge. First, the sumcheck protocol, which is used to reduce claims about parent vertices to claims about child vertices, leaks information in the form of partial sums of the summand polynomial, as usual. Second, in order to reduce a claim about the root to claims about its children, the prover must provide evaluations of the polynomials of the children. These may be hard for the verifier to compute (indeed, if the verifier could compute both of these on its own then there would be no need to recurse). We use the ideas discussed in \secref{sec:techniques-nexp} to resolve both of these issues, obtaining a zero knowledge variant in the IPCP model (where the honest prover sends a random low-degree polynomial as the oracle).

We resolve the first issue by using our zero knowledge sumcheck protocol. Its zero knowledge guarantee states that the protocol reveals only one value of the summand function, which can be computed via one query to each of the $\LDSPValueL{\SPInput}{u_{i}}$, which are precisely the $h_{i}$'s sent by the prover. We are left to ensure that $h_{i}$'s do not leak information.

As in our modification of the BFL protocol, rather than taking the unique degree-$(\SetCardinality{\SPSubset}-1)$ extension $\LDSPValueL{\SPInput}{v}(\vec{\VariableX})$ of $\SPValueL{\SPInput}{v}(\vec{\VariableX})$, we will instead take a \emph{random} degree-$(\SetCardinality{\SPSubset}+\delta)$ extension $\RLDSPValueL{\SPInput}{v}$, where $\delta$ depends only on the circuit structure (in all of our protocols, $\delta$ is a small constant). This ensures that the few evaluations actually revealed by the prover are uniformly random in $\Field$. The prover sends, for each vertex $v$, the evaluation of a random polynomial $R_{v}$, which defines the random low-degree extension as $\RLDSPValueL{\SPInput}{v}(\vec{\VariableX}) \DefineEqual \LDSPValueL{\SPInput}{v}(\vec{\VariableX}) + \ZeroPoly{\SPSubset^{\SPVars}}(\vec{\VariableX}) \cdot R_{v}(\vec{\VariableX})$ where $\ZeroPoly{\SPSubset^{\SPVars}}$ is a degree-$\SetCardinality{\SPSubset}$ polynomial that is zero on $\SPSubset^{\SPVars}$ and nonzero on $(\Field - \SPSubset)^{\SPVars}$. The prover cannot simply send $R_{v}$, however, because the verifier could then query it in order to `derandomize' $\RLDSPValueL{\SPInput}{v}$. Instead, the prover sends a commitment to $R_{v}$ using our algebraic commitment scheme. The decommitment is performed `implicitly' during the sumcheck for vertex $v$. See \secref{sec:pzk-sum-product-evaluation} for details.

Finally, recall that in order to avoid a blowup in the number of claims we have to check, the verifier checks a random linear combination of the claims about any given vertex; this is a linear operation. Also, to avoid a blowup in the degree, we take the low-degree extension, which is also a linear operation. Both of these operations are `compatible' with sumcheck, and thus zero knowledge is straightforwardly maintained.

\doclearpage
\section{Roadmap}
\label{sec:roadmap}

After providing formal definitions in \secref{sec:preliminaries}, the rest of the paper is organized as summarized by the table below. The shaded boxes denote some previous results that we rely on.

\begin{center}
\newcommand{\ResultHeader}[1]{\textbf{#1}}
\tikzstyle{headbox} = [rectangle, align=center, draw=black, fill=white, font=\footnotesize, anchor=south west]
\tikzstyle{backbox} = [rectangle, minimum width=12cm, align=center, draw=black, fill=white, font=\footnotesize, anchor=south west]
\tikzstyle{vtext} = [align=center, font=\footnotesize, rotate=90]
\tikzstyle{newresult} = [rectangle, rounded corners, minimum width=3cm, minimum height=1cm, align=center, draw=black, fill=white, font=\footnotesize]
\tikzstyle{oldresult} = [rectangle, rounded corners, minimum width=3cm, minimum height=1cm, align=center, draw=black, fill=gray!15, font=\footnotesize]
\tikzstyle{arrow} = [thick,->,>=stealth]
\begin{tikzpicture}[node distance=1.5cm]
\node (backbox1) [backbox,  minimum height=2cm] at (-2,-1) {};
\node (backbox2) [backbox,  minimum height=4cm] at (-2,1) {};
\node (backbox3) [backbox,  minimum height=4cm] at (-2,5) {};
\node (headbox1) [headbox,  minimum height=8cm, minimum width=0.5cm] at (-3,1) {};
\node (headbox2) [headbox,  minimum height=4cm, minimum width=0.5cm] at (-2.5,1) {};
\node (headbox3) [headbox,  minimum height=4cm, minimum width=0.5cm] at (-2.5,5) {};
\node (headbox4) [headbox,  minimum height=2cm, minimum width=1cm] at (-3,-1) {};
\node (text1) [vtext] at (-2.7,5) {\textbf{Interactive Probabilistically Checkable Proofs}};
\node (text2) [vtext] at (-2.2,3) {sum-only computations};
\node (text3) [vtext] at (-2.2,7) {sum-product alternations};
\node (text4) [vtext] at (-2.7,0) {\textbf{Algebraic}};
\node (text5) [vtext] at (-2.2,0) {\textbf{Complexity}};
\node (raz) [oldresult] at (0,0) {\ResultHeader{\cite{RazS05}} \\ derandomize PIT for sums \\ of products of univariates};
\node (scd) [oldresult] at (4,0) {\ResultHeader{\cite{BenSassonCFGRS16}} \\ succinct constraint detection \\ for multi-variate low-degree \\ polynomials and their sums};
\node (aqc) [newresult] at (8,0) {\ResultHeader{\S\ref{sec:algebraic-query-complexity}: \thmref{thm:sum-query-lower-bound}} \\ lower bounds for \\ algebraic query complexity \\ of polynomial summation};
\node (pzkcomm) [newresult] at (8,2) {\ResultHeader{\S\ref{sec:strong-zk-sumcheck}} \\ perfectly-hiding \\ statistically-binding \\ algebraic commitment};
\node (strong) [newresult] at (5.5,4) {\ResultHeader{\S\ref{sec:strong-zk-sumcheck}: \thmref{thm:strong-sumcheck-ipcp}} \\ strong PZK sumcheck};
\node (weak) [oldresult] at (4,2) {\ResultHeader{\cite{BenSassonCFGRS16}} \\ weak PZK sumcheck};
\node (sharp) [oldresult] at (1,4) {\ResultHeader{\cite{BenSassonCFGRS16}} \\ PZK analogue of LFKN's \\ protocol for $\sharpP$};
\node (eval) [newresult] at (1,6) {\ResultHeader{\S\ref{sec:pzk-sum-product-evaluation}: \thmref{thm:pzk-for-SPCE}} \\ PZK for sum-product \\ circuit \underline{evaluation}};
\node (sat) [newresult] at (4.5,6) {\ResultHeader{\S\ref{sec:pzk-sum-product-satisfaction}: \thmref{thm:pzk-for-SPCS}} \\ PZK for sum-product \\ circuit \underline{satisfaction}};
\node (pspace) [newresult] at (-0.35,8) {\ResultHeader{\S\ref{sec:zk-pspace}: \thmref{thm:pzk-for-TQBF}} \\ PZK analogue \\ of Shamir's protocol \\ for $\PSPACE$};
\node (gkr) [newresult] at (2.8,8) {\ResultHeader{\S\ref{sec:zk-gkr}: \thmref{thm:gkr-space-uniform}} \\ PZK analogue \\ of GKR's protocol \\ for low-depth circuits};
\node (nexp) [newresult] at (8,8) {\ResultHeader{\S\ref{sec:zk-nexp}: \thmref{thm:pzk-for-nexp}} \\ PZK analogue \\ of BFLS's protocol \\ for $\NEXP$};
\draw [arrow] (aqc) -- (pzkcomm);
\draw [arrow] (scd) -- (weak);
\draw [arrow] (raz) -- (scd);
\draw [arrow] (pzkcomm) -- (strong);
\draw [arrow] (weak) -- (strong);
\draw [arrow] (weak) -- (sharp);
\draw [arrow] (strong) -- (eval);
\draw [arrow] (strong) -- (nexp);
\draw [arrow] (eval) -- (sat);
\draw [arrow] (eval) -- (pspace);
\draw [arrow] (eval) -- (gkr);
\draw [arrow] (sat) -- (nexp);
\end{tikzpicture}
\end{center}

\doclearpage
\section{Preliminaries}
\label{sec:preliminaries}

\subsection{Basic notations}
\label{sec:basic-notations}

For $n \in \Naturals$ we denote by $[n]$ the set $\{1,\ldots,n\}$. For $m,n \in \Naturals$ we denote by $m+[n]$ the set $\{m+1,\ldots,m+n\}$. For a set $X$, $n \in \Naturals$, $I \subseteq [n]$, and $\vec{x} \in X^{n}$, we denote by $\vec{x}_{I}$ the vector $\big(x_{i}\big)_{i \in I}$ that is $\vec{x}$ restricted to the coordinates in $I$.

\parhead{Functions, distributions, fields}
We use $f \colon \Domain \to \Range$ to denote a function with domain $\Domain$ and range $\Range$; given a subset $\SubDomain$ of $\Domain$, we use $\Restrict{f}{\SubDomain}$ to denote the restriction of $f$ to $\SubDomain$. Given a distribution $\Distribution$, we write $x \gets \Distribution$ to denote that $x$ is sampled according to $\Distribution$. We denote by $\Field$ a finite field and by $\Field_{\FieldSize}$ the field of size $\FieldSize$. Arithmetic operations over $\Field_{q}$ take time $\polylog q$ and space $O(\log q)$.

\parhead{Polynomials}
We denote by $\PolynomialRing{\Field}{m}{\VariableX}$ the ring of polynomials in $m$ variables over $\Field$. Given a polynomial $P$ in $\PolynomialRing{\Field}{m}{\VariableX}$, $\IndividualDegree{P}[\VariableX_{i}]$ is the degree of $P$ in the variable $\VariableX_{i}$. The \emph{individual degree} of a polynomial is its maximum degree in any variable, $\max_{1 \leq i \leq m}{\IndividualDegree{P}[\VariableX_{i}]}$; we always refer to the individual degree unless otherwise specified. We denote by $\PolynomialRingIndOne{\Field}{m}{\VariableX}{d}$ the subspace consisting of $P \in \PolynomialRing{\Field}{m}{\VariableX}$ with individual degree at most $d$.

\parhead{Languages and relations}
We denote by $\Language$ a language consisting of \emph{instances} $\Instance$, and by $\Relation$ a (binary ordered) relation consisting of pairs $(\Instance,\Witness)$, where $\Instance$ is the \emph{instance} and $\Witness$ is the \emph{witness}. We denote by $\GetLanguage{\Relation}$ the language corresponding to $\Relation$, and by $\Witnesses{\Relation}{\Instance}$ the set of witnesses in $\Relation$ for $\Instance$ (if $\Instance \not\in \GetLanguage{\Relation}$ then $\Witnesses{\Relation}{\Instance} \DefineEqual \emptyset$). As always, we assume that $\BitSize{\Witness}$ is bounded by some computable function of $\InstanceSize \DefineEqual \BitSize{\Instance}$; in fact, we are mainly interested in relations arising from nondeterministic languages: $\Relation \in \NTIME(\DeciderTime)$ if there exists a $\DeciderTime(\InstanceSize)$-time machine $\DeciderMachine$ such that $\DeciderMachine(\Instance,\Witness)$ outputs $1$ if and only if $(\Instance,\Witness) \in \Relation$. Throughout, we assume that $\DeciderTime(\InstanceSize) \geq \InstanceSize$.

\parhead{Low-degree extensions}
Let $\Field$ be a finite filed, $H$ a subset of $\Field$, and $m$ a positive integer. The \emph{low-degree extension} (LDE) of a function $f \colon H^{m} \to \Field$ is denoted $\LD{f}$ and is the unique polynomial in $\PolynomialRingIndOne{\Field}{m}{\VariableX}{\SetCardinality{H}-1}$ that agrees with $f$ on $H^{m}$. In particular, $\LD{f} \colon \Field^{m} \to \Field$ is defined as follows:
\begin{equation*}
\LD{f}(\vec{\VariableX})
\DefineEqual
  \sum_{\vec{\beta} \in H^{m}}
  \Lagrange{H^{m}}(\vec{\VariableX}, \vec{\beta})
  \cdot
  f(\vec{\beta})
\enspace,
\end{equation*}
where $\Lagrange{H^{m}}(\vec{\VariableX}, \vec{\VariableY}) \DefineEqual \prod_{i=1}^{m} \sum_{\omega \in H} \prod_{\gamma \in H \setminus \{\omega\}} \frac{(\VariableX_{i} - \gamma)(\VariableY_{i} - \gamma)}{(\omega - \gamma)^{2}}$ is the unique polynomial in $\PolynomialRingIndOne{\Field}{m}{\VariableX}{\SetCardinality{H}-1}$ such that, for all $(\vec{\alpha},\vec{\beta}) \in H^{m} \times H^{m}$, $\Lagrange{H^{m}}(\vec{\alpha},\vec{\beta})$ equals $1$ when $\vec{\alpha}=\vec{\beta}$ and equals $0$ otherwise. Note that $\Lagrange{H^{m}}(\vec{\VariableX}, \vec{\VariableY})$ can be generated and evaluated in time $\poly(\SetCardinality{H}, m, \log \SetCardinality{\Field})$ and space $O(\log \SetCardinality{\Field} + \log m)$, so $\LD{f}(\vec{\alpha})$ can be evaluated in time $\SetCardinality{H}^{m} \cdot \poly(\SetCardinality{H}, m, \log \SetCardinality{\Field})$ and space $O(m \cdot \log \SetCardinality{\Field})$.

\subsection{Sampling partial sums of random low-degree polynomials}
\label{sec:partial-sums}

Let $\Field$ be a finite field, $\SCVars,\SCDegree$ positive integers, and $\SCSubset$ a subset of $\Field$, and recall that $\PolynomialRingIndOne{\Field}{\SCVars}{\VariableX}{\SCDegree}$ is the subspace of $\PolynomialRing{\Field}{\SCVars}{\VariableX}$ consisting of those polynomials with individual degrees at most $\SCDegree$. Given $Q \in \PolynomialRingIndOne{\Field}{\SCVars}{\VariableX}{\SCDegree}$ and $\vec{\alpha} \in \Field^{\leq\SCVars}$ (vectors over $\Field$ of length at most $\SCVars$), we define $Q(\vec{\alpha}) \DefineEqual \sum_{\vec{\gamma} \in \SCSubset^{\SCVars - \SetCardinality{\vec{\alpha}}}} Q(\vec{\alpha}, \vec{\gamma})$, i.e., the answer to a query that specifies only a prefix of the variables is the sum of the values obtained by letting the remaining variables range over $\SCSubset$.

In \secref{sec:strong-zk-sumcheck} we rely on the fact, formally stated below and proved in \cite{BenSassonCFGRS16}, that one can efficiently sample the distribution $\RandPoly(\vec{\alpha})$, where $\RandPoly$ is uniformly random in $\PolynomialRingIndOne{\Field}{\SCVars}{\VariableX}{\SCDegree}$ and $\vec{\alpha} \in \Field^{\leq\SCVars}$ is fixed, \emph{even conditioned on any polynomial number of (consistent) values for $\RandPoly(\vec{\alpha}_{1}),\dots,\RandPoly(\vec{\alpha}_{\ListSize})$} (with $\vec{\alpha}_{1},\dots,\vec{\alpha}_{\ListSize} \in \Field^{\leq\SCVars}$). More precisely, the sampling algorithm runs in time that is only $\poly(\log \SetCardinality{\Field}, \SCVars, \SCDegree, \SetCardinality{\SCSubset}, \ListSize)$, which is much faster than the trivial running time of $\Omega(\SCDegree^{\SCVars})$ achieved by sampling $\RandPoly$ explicitly. This ``succinct'' sampling follows from the notion of \emph{succinct constraint detection} studied in \cite{BenSassonCFGRS16} for the case of partial sums of low-degree polynomials.

\begin{corollary}[\cite{BenSassonCFGRS16}]
\label{cor:efficient-poly-simulator}
There exists a probabilistic algorithm $\CodeSimAlgorithm$ such that, for every finite field $\Field$, positive integers $\SCVars,\SCDegree$, subset $\SCSubset$ of $\Field$, subset $S = \{(\alpha_{1},\beta_{1}), \dots, (\alpha_{\ListSize}, \beta_{\ListSize})\} \subseteq \Field^{\leq\SCVars} \times \Field$, and $(\alpha,\beta) \in \Field^{\leq\SCVars} \times \Field$,
\begin{equation*}
\Pr\Big[
\CodeSimAlgorithm(\Field,\SCVars,\SCDegree,\SCSubset,S,\alpha) = \beta
\Big]
=
\Pr_{\RandPoly \gets \PolynomialRingIndOne{\Field}{\SCVars}{\VariableX}{\SCDegree}}
\left[
\RandPoly(\alpha) = \beta
\pST
\begin{array}{c}
\RandPoly(\alpha_{1}) = \beta_{1} \\
\vdots \\
\RandPoly(\alpha_{\ListSize}) = \beta_{\ListSize}
\end{array}
\right]\enspace.
\end{equation*}
Moreover $\CodeSimAlgorithm$ runs in time $\SCVars(\SCDegree \ListSize \SetCardinality{\SCSubset} + \SCDegree^{3}\ListSize^{3}) \cdot \poly(\log \SetCardinality{\Field}) = \ListSize^{3} \cdot \poly(\SCVars, \SCDegree, \SetCardinality{\SCSubset}, \log \SetCardinality{\Field})$.
\end{corollary}

\subsection{Interactive probabilistically checkable proofs}
\label{sec:ipcp}

An \emph{Interactive Probabilistically Checkable Proof} (Interactive PCP, IPCP) \cite{KalaiR08} is a Probabilistically Checkable Proof \cite{BabaiFLS91,FeigeGLSS91,AroraS98,AroraLMSS98} followed by an Interactive Proof \cite{Babai85,GoldwasserMR89}. Namely, the prover $\Prover$ and verifier $\Verifier$ interact as follows: $\Prover$ sends to $\Verifier$ a probabilistically checkable proof $\Proof$; afterwards, $\Prover$ and $\Verifier^{\Proof}$ engage in an interactive proof. Thus, $\Verifier$ may read a few bits of $\Proof$ but must read subsequent messages from $\Prover$ in full. An \emph{IPCP system} for a relation $\Relation$ is thus a pair $\pair{\Prover}{\Verifier}$, where $\Prover,\Verifier$ are probabilistic interactive algorithms working as described, that satisfies naturally-defined notions of perfect completeness and soundness with a given error $\SoundnessError(\cdot)$; see \cite{KalaiR08} for details.

We say that an IPCP has $\NumRounds$ rounds if this ``PCP round'' is followed by a $\NumRounds$-round interactive proof. (Though note that \cite{BenSassonCFGRS16} counts the PCP round towards round complexity.) Beyond round complexity, we also measure how many bits the prover sends and how many the verifier reads: the \emph{proof length} $\ProofLength$ is the length of $\Proof$ in bits plus the number of bits in all subsequent prover messages; the \emph{query complexity} $\QueryComplexity$ is the number of bits of $\Proof$ read by the verifier plus the number of bits in all subsequent prover messages (since the verifier must read all of those bits).

In this work, we do not count the number of bits in the verifier messages, nor the number of random bits used by the verifier; both are bounded from above by the verifier's running time, which we do consider. Overall, we say that a language $\Language$ (resp., relation $\Relation$) belongs to the complexity class $\IPCP[\SoundnessError,\NumRounds,\ProofLength,\QueryComplexity]$ if there is an IPCP system for $\Language$ (resp., $\Relation$) in which:
\begin{inparaenum}[(1)]
  \item the soundness error is $\SoundnessError(\InstanceSize)$;
  \item the number of rounds is at most $\NumRounds(\InstanceSize)$;
  \item the proof length is at most $\ProofLength(\InstanceSize)$;
  \item the query complexity is at most $\QueryComplexity(\InstanceSize)$.
\end{inparaenum}
We sometimes also specify the time and/or space complexity of the (honest) prover algorithm and/or (honest) verifier algorithm.

Finally, an IPCP is \emph{non-adaptive} if the verifier queries are non-adaptive, i.e., the queried locations depend only on the verifier's inputs; it is \emph{public-coin} if each verifier message is chosen uniformly and independently at random, and all of the verifier queries happen after receiving the last prover message. \emph{All of the IPCPs discussed in this paper are both non-adaptive and public-coin.}

\subsection{Zero knowledge for Interactive PCPs}
\label{sec:zk}

We define the notion of zero knowledge for IPCPs that we consider: \emph{perfect zero knowledge via straightline simulators}. This notion is quite strong not only because it unconditionally guarantees perfect simulation of the verifier's view but also because straightline simulation typically implies desirable properties. We first provide context and then definitions.

At a high level, zero knowledge requires that the verifier's view can be efficiently simulated without the prover. Converting the informal statement into a mathematical one involves many choices, including choosing which verifier class to consider (e.g., the honest verifier? all polynomial-time verifiers?), the quality of the simulation (e.g., is it identically distributed to the view? statistically close to it? computationally close to it?), the simulator's dependence on the verifier (e.g., is it non-uniform? or is the simulator universal?), and others. The definition below considers the case of perfect simulation via universal simulators against verifiers making a bounded number of queries to the proof oracle.

Moreover, in the case of universal simulators, one distinguishes between a non-blackbox use of the verifier, which means that the simulator takes the verifier's code as input, and a blackbox use of it, which means that the simulator only accesses the verifier via a restricted interface; we consider this latter case. Different models of proof systems call for different interfaces, which grant carefully-chosen ``extra powers'' to the simulator (in comparison to the prover) so to ensure that efficiency of the simulation does not imply the ability to efficiently decide the language. For example: in ZK IPs, the simulator may rewind the verifier; in ZK PCPs, the simulator may adaptively answer oracle queries. In ZK IPCPs (our setting), the natural definition would allow a blackbox simulator to rewind the verifier \emph{and also} to adaptively answer oracle queries. The definition below, however, considers only simulators that are straightline \cite{FeigeS89,DworkS98}, that is they do not rewind the verifier, because our constructions achieve this stronger notion.

We are now ready to define the notion of perfect zero knowledge via straightline simulators for IPCPs \cite{GoyalIMS10}.

\begin{definition}
	\label{def:ipcp-view}
	Let $A,B$ be algorithms and $x,y$ strings. We denote by $\IOPView{B(y)}{A(x)}$ the \defemph{view} of $A(x)$ in an IPCP protocol with $B(y)$, i.e., the random variable $(x,r,s_{1},\dots,s_{n},t_{1},\dots,t_{m})$ where $x$ is $A$'s input, $r$ is $A$'s randomness, $s_{1},\dots,s_{n}$ are $B$'s messages, and $t_{1},\dots,t_{m}$ are the answers to $A$'s queries to the proof oracle sent by $B$.
\end{definition}

Straightline simulators in the context of IPs were used in \cite{FeigeS89}, and later defined in \cite{DworkS98}. The definition below considers this notion in the context of IPCPs, where the simulator also has to answer oracle queries by the verifier. Note that since we consider the notion of perfect zero knowledge, the definition of straightline simulation needs to allow the efficient simulator to work even with inefficient verifiers \cite{GoyalIMS10}.

\begin{definition}
\label{def:ipcp-access}
We say that an algorithm $B$ has \defemph{straightline access} to another algorithm $A$ if $B$ interacts with $A$, without rewinding, by exchanging messages with $A$ and also answering any oracle queries along the way. We denote by $B^{A}$ the concatenation of $A$'s random tape and $B$'s output. (Since $A$'s random tape could be super-polynomially large, $B$ cannot sample it for $A$ and then output it; instead, we restrict $B$ to not see it, and we prepend it to $B$'s output.)
\end{definition}

\begin{definition}
\label{def:zk-ipcp}
An IPCP system $\pair{\Prover}{\Verifier}$ for a relation $\Relation$ is perfect zero knowledge (via straightline simulators) \sunderline{against unbounded queries} (resp., \sunderline{against query bound $\QueryBound$}) with simulator overhead $s \colon \Naturals \times \Naturals \to \Naturals$ if there exists a simulator algorithm $\Simulator$ such that for every algorithm (resp., $\QueryBound$-query algorithm) $\Malicious{\Verifier}$ and instance-witness pair $(\Instance,\Witness) \in \Relation$, $\Simulator^{\Malicious{\Verifier}} (\Instance)$ and $\IPCPView{\Prover(\Instance,\Witness)}{\Malicious{\Verifier}(\Instance)}$ are identically distributed. Moreover, $\Simulator$ must run in time $O(s(\BitSize{\Instance}, \QueryComplexity_{\Malicious{\Verifier}}(\BitSize{\Instance})))$, where $\QueryComplexity_{\Malicious{\Verifier}}(\cdot)$ is $\Malicious{\Verifier}$'s query complexity.

The case of a language $\Language$ is similar: the quantification is for all $\Instance \in \Language$ and the view to simulate is $\IPCPView{\Prover(\Instance)}{\Malicious{\Verifier}(\Instance)}$.
\end{definition}

\begin{remark}
Throughout this paper, an algorithm is \defemph{$\QueryBound$-query} if it makes \sunderline{strictly fewer than} $\QueryBound$ queries to its oracle. This is because all of our results will be of a `query threshold' character, i.e. if the verifier makes $\QueryBound$ queries it learns some information, but any verifier making strictly fewer queries learns nothing.
\end{remark}

\begin{remark}
The standard definition of zero knowledge allows the simulator overhead $s$ to be any fixed polynomial.
\end{remark}

\begin{remark}
The definition above places a strict bound on the running time of the simulator. This is in contrast to most zero knowledge results, which can only bound its \emph{expected} running time.
\end{remark}

We say that a language $\Language$ (resp., relation $\Relation$) belongs to the complexity class $\PZKIPCP[\SoundnessError,\NumRounds,\ProofLength,\QueryComplexity,\QueryBound,s]$ if there is an IPCP system for $\Language$ (resp., $\Relation$), with the corresponding parameters, that is perfect zero knowledge with query bound $\QueryBound$; also, it belongs to the complexity class $\PZKIPCP[\SoundnessError,\NumRounds,\ProofLength,\QueryComplexity,\AnyBound,s]$ if the same is true with unbounded queries. In this paper we only consider zero knowledge against bounded queries. (Note that, even in this case, one can `cover' all polynomial-time malicious verifiers by setting $\QueryBound$ to be superpolynomial in the input size.)

\begin{remark}
Kalai and Raz \cite{KalaiR08} give a general transformation for IPCPs that reduces the verifier's query complexity $\QueryComplexity$ to $1$. The transformation preserves our zero knowledge guarantee, with a small increase in the simulator overhead.
\end{remark}

\subsection{Sumcheck protocol and its zero knowledge variant}
\label{sec:sumcheck-protocol}

The sumcheck protocol \cite{LundFKN92} is a fundamental building block of numerous results in complexity theory and cryptography. We rely on it in \secref{sec:strong-zk-sumcheck}, so we briefly review it here. The protocol consists of an Interactive Proof for a claim of the form ``$\sum_{\alpha_{1},\dots,\alpha_{\SCVars} \in \SCSubset} \SCPoly(\alpha_{1},\dots,\alpha_{\SCVars})=\SCSum$'', where $\SCPoly$ is an $\SCVars$-variate polynomial of individual degree $\SCDegree$ with coefficients in a finite field $\Field$, $\SCSubset$ is a subset of $\Field$, and $\SCSum$ is an element of $\Field$. The prover and verifier receive $(\Field,\SCVars,\SCDegree,\SCSubset,\SCSum)$ as input; in addition, the prover receives $\SCPoly$ as input while the verifier has only oracle access to $\SCPoly$. In the $i$-th round, the prover sends the univariate polynomial $\SCPoly_{i}(\VariableX) \DefineEqual \sum_{\alpha_{i+1},\dots,\alpha_{\SCVars} \in \SCSubset} \SCPoly(c_{1},\dots,c_{i-1},\VariableX,\alpha_{i+1},\dots,\alpha_{\SCVars})$, and the verifier replies with a uniformly random element $c_{i} \in \Field$ and checks that $\SCPoly_{i-1}(c_{i-1}) = \sum_{\alpha \in \SCSubset} \SCPoly_{i}(\alpha)$ (defining $\SCPoly_{0}(c_{0})$ to be the element $\SCSum$). At the end of the interaction, the verifier also checks that $\SCPoly_{\SCVars}(c_{\SCVars})=\SCPoly(c_{1},\dots,c_{\SCVars})$, by querying $\SCPoly$ at the random location $(c_{1},\dots,c_{\SCVars})$. This interactive proof is public-coin, and has $\SCVars$ rounds, communication complexity $\poly(\log \SetCardinality{\Field}, \SCVars)$, and soundness error $\frac{\SCVars\SCDegree}{\SetCardinality{\Field}}$. The prover runs in time $\poly(\log \SetCardinality{\Field},\SetCardinality{\SCSubset}^{\SCVars})$ and space $\poly(\log \SetCardinality{\Field}, \SCVars, \SetCardinality{\SCSubset})$ and the verifier runs in time $\poly(\log \SetCardinality{\Field}, \SCVars, \SCDegree, \SetCardinality{\SCSubset})$ and space $O(\log \SetCardinality{\Field} \cdot \SCVars)$.

The sumcheck protocol is \emph{not} zero knowledge, because the prover reveals partial sums of $\SCPoly$ to the verifier. If we assume the existence of one-way functions, the protocol can be made computational zero knowledge by leveraging the fact that it is public-coin \cite{GoldwasserMR89,ImpagliazzoY87,BenOrGGHKMR88} (in fact, if we further assume the hardness of certain problems related to discrete logarithms then more efficient transformations are known \cite{CramerD98}); moreover, there is strong evidence that assuming one-way functions is necessary \cite{Ostrovsky91,OstrovskyW93}. Even more, achieving statistical zero knowledge for sumcheck instances would cause unlikely complexity-theoretic collapses \cite{Fortnow87,AielloH91}.

Nevertheless, \cite{BenSassonCFGRS16} have shown that, in the \emph{Interactive PCP} model (see \secref{sec:ipcp}), a simple variant of the sumcheck protocol is \emph{perfect zero knowledge}. The variant is as follows: the prover sends a proof oracle containing the evaluation of a random $\SCVars$-variate polynomial $\AuxRandPoly$ of individual degree $\SCDegree$, conditioned on summing to $0$ on $\SCSubset^{\SCVars}$; the verifier replies with a random element $\rho \in \Field$; then the prover and verifier engage in a sumcheck protocol for the claim ``$\sum_{\vec{\alpha} \in \SCSubset^{\SCVars}} \rho\SCPoly(\vec{\alpha})+\AuxRandPoly(\vec{\alpha})=\SCSum$'', with the verifier accessing $\AuxRandPoly$ via self-correction (after low-degree testing it). The proof oracle thus consists of $\SetCardinality{\Field}^{\SCVars}$ field elements, and the verifier accesses only $\poly(\log \SetCardinality{\Field}, \SCVars, \SCDegree)$ of them.

The auxiliary polynomial $\AuxRandPoly$ acts as a ``masking polynomial'', and yields the following zero knowledge guarantee: there exists a polynomial-time simulator algorithm that perfectly simulates the view of any malicious verifier, provided it can query $\SCPoly$ in as many locations as the \emph{total} number of queries that the malicious verifier makes to either $\SCPoly$ or $\AuxRandPoly$.

\doclearpage
\section{Algebraic query complexity of polynomial summation}
\label{sec:algebraic-query-complexity}

We have described in \secref{sec:technniques-algebraic-commitment} an algebraic commitment scheme based on the sumcheck protocol and lower bounds on the algebraic query complexity of polynomial summation. The purpose of this section is to describe this construction in more detail, and then provide formal statements for the necessary lower bounds.

We begin with the case of committing to a single element $a \in \Field$. The prover chooses a uniformly random string $B \in \Field^{N}$ such that $\sum_{i=1}^{N} B_{i} = a$, for some $N \in \Naturals$. Fixing some $\SCDegree \in \Naturals$, $\SSCSubset \subseteq \Field$ and $\SSCVars \in \Naturals$ such that $\SetCardinality{\SSCSubset} \leq \SCDegree+1$ and $\SetCardinality{\SSCSubset}^{\SSCVars} = N$, the prover views $B$ as a function from $\SSCSubset^{\SSCVars}$ to $\Field$ (via an ordering on $\SSCSubset^{\SSCVars}$) and sends the evaluation of a degree-$d$ extension $\LD{B} \colon \Field^{\SSCVars} \to \Field$ of $B$. The verifier tests that $\LD{B}$ is indeed (close to) a low-degree polynomial but (ideally) cannot learn any information about $a$ without reading \emph{all} of $B$ (i.e., without making $N$ queries). Subsequently, the prover can decommit to $a$ by convincing the verifier that $\sum_{\vec{\beta} \in \SSCSubset^{\SSCVars}} \LD{B}(\vec{\beta}) = a$ via the sumcheck protocol.

To show that the above is a commitment scheme, we must show both binding and hiding. Both properties depend on the choice of $\SCDegree$. The binding property follows from the soundness of the sumcheck protocol, and we thus would like the degree $d$ of $\LD{B}$ to be as small as possible. A natural choice would be $\SCDegree = 1$ (so $\SetCardinality{\SSCSubset} = 2$), which makes $\LD{B}$ the unique multilinear extension of $B$. However (as discussed in \secref{sec:technniques-algebraic-commitment}) this choice of parameters does not provide any hiding: it holds that $\sum_{\beta \in \Bits^{k}} B(\beta) = \LD{B}(2^{-1}, \ldots, 2^{-1}) \cdot 2^{k}$ (as long as $\Characteristic{\Field} \neq 2$). We therefore need to understand how the choice of $d$ affects the number of queries to $\LD{B}$ required to compute $a$. This is precisely the setting of \emph{algebraic query complexity}, which we discuss next.

The algebraic query complexity (defined in \cite{AaronsonW09} to study `algebrization') of a function $f$ is the (worst-case) number of queries to some low-degree extension $\LD{B}$ of a string $B$ required to compute $f(B)$. This quantity is bounded from above by the standard query complexity of $f$, but it may be the case (as above) that the low-degree extension confers additional information that helps in computing $f$ with fewer queries. The usefulness of this information depends on parameters $\SCDegree$ and $\SSCSubset$ of the low-degree extension. Our question amounts to understanding this dependence for the function $\textsc{Sum} \colon \Field^{N} \to \Field$ given by $\textsc{Sum}(B) \DefineEqual \sum_{i=1}^{N} B_{i}$. This has been studied before in \cite{JumaKRS09}: if $\SSCSubset = \Bits$ and $\SCDegree = 2$ then the algebraic query complexity of $\textsc{Sum}$ is exactly $N$.

For our purposes, however, it is not enough to commit to a single field element. Rather, we need to commit to the evaluation of a polynomial $Q \colon \Field^{\SCVars} \to \Field$ of degree $\SCDegree_{Q}$, which we do as follows. Let $K$ be a subset of $\Field$ of size $\SCDegree_{Q}+1$. The prover samples, for each $\vec{\alpha} \in K^{\SCVars}$, a random string $B^{\vec{\alpha}} \in \Field^{N}$ such that $\textsc{Sum}(B^{\vec{\alpha}}) = Q(\vec{\alpha})$. The prover views these strings as a function $B \colon K^{\SCVars} \times \SSCSubset^{\SSCVars} \to \Field$, and takes a low-degree extension $\LD{B} \colon \Field^{\SCVars} \times \Field^{\SSCVars} \to \Field$. The polynomial $\LD{B}(\vec{\VariableX}, \vec{\VariableY})$ has degree $d_{Q}$ in $\vec{\VariableX}$ and $d$ in $\vec{\VariableY}$; this is a commitment to $Q$ because $\sum_{\vec{\beta} \in \SSCSubset^{\SSCVars}} \LD{B}(\vec{\VariableX}, \vec{\beta})$ is a degree-$d_{Q}$ polynomial that agrees with $Q$ on $K^{\SCVars}$, and therefore equals $Q$.

Once again we will decommit to $Q(\vec{\alpha})$ using the sumcheck protocol, and so for binding we need $d$ to be small. For hiding, as in the single-element case, if $d$ is too small then a few queries to $\LD{\mathcal{B}}$ can yield information about $Q$. Moreover, it could be the case that the verifier can leverage the fact that $\LD{\mathcal{B}}$ is a \emph{joint} low-degree extension to learn some linear combination of evaluations of $Q$. We must exclude these possibilities in order to obtain our zero knowledge guarantees.

This question amounts to a generalization of algebraic query complexity where, given a list of strings $B_{1}, \ldots, B_{M}$, we determine how many queries we need to make to their \emph{joint} low-degree extension $\LD{B}$ to determine any nontrivial linear combination $\sum_{i=1}^{M} c_{i} \cdot \textsc{Sum}(B_{i})$. We will show that the `generalized' algebraic query complexity of $\textsc{Sum}$ is exactly $N$ provided $\SCDegree \geq 2(\SetCardinality{\SSCSubset}-1)$ (which is also the case for the standard algebraic query complexity).

In the remainder of the section we state our results in a form equivalent to the above that is more useful to us. Given an arbitrary polynomial $\StrongRandPoly \in \PolynomialRingIndOneXY{\Field}{\SCVars}{\VariableX}{\SSCVars}{\VariableY}{\SCDegree}{\SCDegree'}$, we ask how many queries are required to determine any nontrivial linear combination of $\sum_{\vec{y} \in \SSCSubset^{\SSCVars}} \StrongRandPoly(\vec{\alpha}, \vec{y})$ for $\vec{\alpha} \in \Field^{\SCVars}$. The following theorem is more general: it states that not only do we require many queries to determine \emph{any} linear combination, but that the number of queries grows linearly with the number of independent combinations that we wish to learn.

\begin{theorem}[algebraic query complexity of polynomial summation]
\label{thm:sum-query-lower-bound}
Let $\Field$ be a field, $\SCVars, \SSCVars, \SCDegree, \SCDegree' \in \Naturals$, and $\SSCSubset, K, L$ be finite subsets of $\Field$ such that $K \subseteq L$, $\SCDegree' \geq \SetCardinality{\SSCSubset} - 2$, and $\SetCardinality{K} = \SCDegree+1$. If $S \subseteq \Field^{\SCVars+\SSCVars}$ is such that there exist matrices $C \in \Field^{L^{m} \times \ell}$ and $D \in \Field^{S \times \ell}$ such that for all $\StrongRandPoly \in \PolynomialRingIndOneXY{\Field}{\SCVars}{\VariableX}{\SSCVars}{\VariableY}{\SCDegree}{\SCDegree'}$ and all $i \in \{1, \ldots, \ell\}$
	\begin{equation*}
	\sum_{\vec{\alpha} \in L^{\SCVars}} C_{\vec{\alpha},i} \sum_{\vec{y} \in \SSCSubset^{\SSCVars}} \StrongRandPoly(\vec{\alpha}, \vec{y}) = \sum_{\vec{q} \in S} D_{\vec{q},i} \StrongRandPoly(\vec{q}) \enspace,
	\end{equation*}
	then $\SetCardinality{S} \geq \rank(BC) \cdot (\min\{\SCDegree' - \SetCardinality{\SSCSubset} + 2, \SetCardinality{\SSCSubset}\})^{\SSCVars}$, where $B \in \Field^{K^{\SCVars} \times L^{\SCVars}}$ is such that column $\vec{\alpha}$ of $B$ represents $\StrongRandPoly(\vec{\alpha})$ in the basis $(\StrongRandPoly(\vec{\beta}))_{\vec{\beta} \in K^{\SCVars}}$.
\end{theorem}

We describe a special case of the above theorem that is necessary for our zero knowledge results, and then give an equivalent formulation in terms of random variables that we use in later sections. (Essentially, the linear structure of the problem implies that `worst-case' statements are equivalent to `average-case' statements.)

\begin{corollary}
\label{cor:special-case-query}
Let $\Field$ be a finite field, $\SSCSubset$ be a subset of $\Field$, and $\SCDegree,\SCDegree' \in \Naturals$ with $\SCDegree' \geq 2(\SetCardinality{\SSCSubset}-1)$. If $S \subseteq \Field^{\SCVars+\SSCVars}$ is such that there exist $(c_{\vec{\alpha}})_{\vec{\alpha} \in \Field^{\SCVars}}$ and $(d_{\vec{\beta}})_{\vec{\beta} \in \Field^{\SCVars+\SSCVars}}$ such that
\begin{itemize}[nolistsep]

  \item for all $\StrongRandPoly \in \PolynomialRingIndOneXY{\Field}{\SCVars}{\VariableX}{\SSCVars}{\VariableY}{\SCDegree}{\SCDegree'}$ it holds that
$
\sum_{\vec{\alpha} \in \Field^{\SCVars}} c_{\vec{\alpha}} \sum_{\vec{y} \in \SSCSubset^{\SSCVars}} \StrongRandPoly(\vec{\alpha}, \vec{y})
= \sum_{\vec{q} \in S} d_{\vec{q}} \StrongRandPoly(\vec{q})
$
and

  \item there exists $\StrongRandPoly' \in \PolynomialRingIndOneXY{\Field}{\SCVars}{\VariableX}{\SSCVars}{\VariableY}{\SCDegree}{\SCDegree'}$ such that $\sum_{\vec{\alpha} \in \Field^{\SCVars}} c_{\vec{\alpha}} \sum_{\vec{y} \in \SSCSubset^{\SSCVars}} \StrongRandPoly'(\vec{\alpha}, \vec{y}) \neq 0$,

\end{itemize}  
then $\SetCardinality{S} \geq \SetCardinality{\SSCSubset}^{\SSCVars}$.
\end{corollary}

\begin{corollary}[equivalent statement of \corref{cor:special-case-query}]
\label{cor:partial-sum-indep-vars}
Let $\Field$ be a finite field, $\SSCSubset$ be a subset of $\Field$, and $\SCDegree,\SCDegree' \in \Naturals$ with $\SCDegree' \geq 2(\SetCardinality{\SSCSubset}-1)$.  Let $Q$ be a subset of $\Field^{\SCVars+\SSCVars}$ with $\SetCardinality{Q} < \SetCardinality{\SSCSubset}^{\SSCVars}$ and let $\StrongRandPoly$ be uniformly random in $ \PolynomialRingIndOneXY{\Field}{\SCVars}{\VariableX}{\SSCVars}{\VariableY}{\SCDegree}{\SCDegree'}$. The ensembles $\big(\sum_{\vec{y} \in \SSCSubset^{\SSCVars}} Z(\vec{\alpha},\vec{y})\big)_{\vec{\alpha} \in \Field^{\SCVars}}$ and $\big(\StrongRandPoly(\vec{q})\big)_{\vec{q} \in Q}$ are independent.
\end{corollary}

The proofs of these results, and derivations of corresponding upper bounds, are provided in \appref{sec:query-complexity-appendix}.

\doclearpage
\section{Zero knowledge sumcheck from algebraic query lower bounds}
\label{sec:strong-zk-sumcheck}

We leverage lower bounds on the algebraic query complexity of polynomial summation (\secref{sec:algebraic-query-complexity}) to obtain an analogue of the sumcheck protocol with a strong zero knowledge guarantee, which we use in the applications that we consider.

The sumcheck protocol \cite{LundFKN92} is an Interactive Proof for claims of the form $\sum_{\vec{x} \in \SCSubset^{\SCVars}} \SCPoly(\vec{x}) = a$, where $\SCSubset$ is a subset of a finite field $\Field$, $\SCPoly$ is an $\SCVars$-variate polynomial over $\Field$ of individual degree at most $\SCDegree$, and $a$ is an element of $\Field$. The sumcheck protocol is \emph{not} zero knowledge (conjecturally).

Prior work \cite{BenSassonCFGRS16} obtains a sumcheck protocol, in the Interactive PCP model, with a certain zero knowledge guarantee. In that protocol, the prover first sends a proof oracle that consists of the evaluation of a random $\SCVars$-variate polynomial $\RandPoly$ of individual degree at most $\SCDegree$; after that, the prover and the verifier run the (standard) sumcheck protocol on a new polynomial obtained from $\SCPoly$ and $\RandPoly$. The purpose of $\RandPoly$ is to `mask' the partial sums, which are the intermediate values sent by the prover during the sumcheck protocol.

The zero knowledge guarantee in \cite{BenSassonCFGRS16} is the following: \emph{any verifier that makes $q$ queries to $\RandPoly$ learns at most $q$ evaluations of $\SCPoly$}. This guarantee suffices to obtain a zero knowledge protocol for $\sharpP$ (the application in \cite{BenSassonCFGRS16}) because the verifier can evaluate $\SCPoly$ efficiently at any point (as $\SCPoly$ is merely an arithmetization of a 3SAT formula).

We achieve a much stronger guarantee: \emph{any verifier that makes polynomially-many queries to $\RandPoly$ learns at most a single evaluation of $\SCPoly$} (that, moreover, lies within a chosen subset $\SoundnessSet^{\SCVars}$ of $\Field^{\SCVars}$). Our applications require this guarantee because we use the sumcheck simulator as a sub-simulator in a larger protocol, where $\SCPoly$ is a randomized low-degree extension of some function that is hard to compute for the verifier. The randomization introduces bounded independence, which makes a small number of queries easy to simulate.

The main idea to achieve zero knowledge as above is the following. Rather than sending the masking polynomial $\RandPoly$ directly, the prover sends a (perfectly-hiding and statistically-binding) commitment to it in the form of a random $(\SCVars+\SSCVars)$-variate polynomial $\StrongRandPoly$. The `real' mask is recovered by summing out $k$ variables: $\RandPoly(\vec{\VariableX}) \DefineEqual \sum_{\vec{\beta} \in \SSCSubset^{\SSCVars}} \StrongRandPoly(\vec{\VariableX}, \vec{\beta})$. Our lower bounds on the algebraic query complexity of polynomial summation (\secref{sec:algebraic-query-complexity}) imply that any $q$ queries to $\StrongRandPoly$, with $q < \SetCardinality{\SSCSubset}^{\SSCVars}$, yield \emph{no information} about $\RandPoly$. The prover, however, can elect to decommit to $\RandPoly(\vec{c})$ for a single point $\vec{c} \in \SoundnessSet^{\SCVars}$ chosen by the verifier. This is achieved using the zero knowledge sumcheck protocol of \cite{BenSassonCFGRS16} as a subroutine: the prover sends $w \DefineEqual \RandPoly(\vec{c})$ and then proves that $w = \sum_{\vec{\beta} \in \SSCSubset^{\SSCVars}} \StrongRandPoly(\vec{c}, \vec{\beta})$.

The protocol thus proceeds as follows. Given a security parameter $\SubsetSize \in \Naturals$, the prover sends the evaluations of two polynomials $\StrongRandPoly \in \PolynomialRingIndOneXY{\Field}{\SCVars}{\VariableX}{\SSCVars}{\VariableY}{\SCDegree}{2\SubsetSize}$ and $\AuxRandPoly \in \PolynomialRingIndOne{\Field}{\SSCVars}{\VariableY}{2\SubsetSize}$ as proof oracles. The verifier checks that both of these evaluations are close to low-degree, and uses self-correction to make for querying them. The prover sends two field elements $z_{1}$ and $z_{2}$, which are (allegedly) the summations of $\StrongRandPoly$ and $\AuxRandPoly$ over $\SCSubset^{\SCVars} \times \SSCSubset^{\SSCVars}$ and $\SSCSubset^{\SSCVars}$, respectively. The verifier replies with a random challenge $\rho \in \Field \setminus \{0\}$. The prover and the verifier then engage in the standard (not zero knowledge) sumcheck protocol on the claim ``$\sum_{\vec{\alpha} \in \SCSubset^{\SCVars}} \rho \SCPoly(\vec{\alpha}) + \RandPoly(\vec{\alpha}) = \rho \SCSum + z_{1}$''. This reduces the correctness of this claim to checking a claim of the form ``$\rho \SCPoly(\vec{c}) + \RandPoly(\vec{c}) = b$'' for some $\vec{c} \in \SoundnessSet^{\SCVars}$ and $b \in \Field$; the prover then decommits to $w \DefineEqual \RandPoly(\vec{c})$ as above. In sum, the verifier deduces that, with high probability, the claim ``$\rho \SCPoly(\vec{c}) = b - w$'' is true if and only if the original claim was.

If the verifier could evaluate $\SCPoly$ then the verifier could simply check the aforementioned claim and either accept or reject. We do not give the verifier access to $\SCPoly$ and, instead, we follow \cite{Meir13} and phrase sumcheck as a \emph{reduction} from a claim about a sum of a polynomial over a large product space to a claim about the evaluation of that polynomial at a single point. This view of the sumcheck protocol is useful later on when designing more complex protocols, which employ sumcheck as a subprotocol. The completeness and soundness definitions below are thus modified according to this viewpoint, where the verifier simply outputs the claim at the end.

Our protocol will be sound relative to a \emph{promise variant} of the sumcheck protocol, which we now define.

\begin{definition}
\label{def:sumcheck-relation}
The sumcheck relation and its promise variant are defined as follows.
\begin{itemize}

\item The \defemph{sumcheck relation} is the relation $\SCRelation$ of instance-witness pairs $\big( (\Field,\SCVars,\SCDegree,\SCSubset,\SCSum) , \SCPoly \big)$ such that:
\begin{itemize}[nolistsep]
  \item $\Field$ is a finite field, $\SCSubset$ is a subset of $\Field$, $\SCSum$ is an element of $\Field$, and $\SCVars,\SCDegree$ are positive integers with $\frac{\SCVars\SCDegree}{\SetCardinality{\Field}} < \frac{1}{2}$;
  \item $\SCPoly$ is a polynomial in $\PolynomialRingIndOne{\Field}{\SCVars}{\VariableX}{\SCDegree}$ and sums to $\SCSum$ on $\SCSubset^{\SCVars}$.
\end{itemize}

  \item The \defemph{sumcheck promise relation} is the pair of relations $(\SCRelation^{\yes},\SCRelation^{\no})$ where $\SCRelation^{\yes} \DefineEqual \SCRelation$ and $\SCRelation^{\no}$ are the pairs $\big( (\Field,\SCVars,\SCDegree,\SCSubset,\SCSum) , \SCPoly \big)$ such that $(\Field,\SCVars,\SCDegree,\SCSubset,\SCSum)$ is as above and $\SCPoly$ is in $\PolynomialRingIndOne{\Field}{\SCVars}{\VariableX}{\SCDegree}$ but does \sunderline{not} \mbox{sum to $\SCSum$ on $\SCSubset^{\SCVars}$.}

\end{itemize}
\end{definition}

\begin{remark}
In the case where the verifier can easily determine that $\SCPoly$ is low-degree (e.g., $\SCPoly$ is given as an arithmetic circuit), a protocol for the promise relation can be used to check the plain relation. In our setting, the verifier cannot even access $\SCPoly$, and so the promise is necessary.
\end{remark}

The definition below captures our zero knowledge goal for the sumcheck promise relation, in the Interactive PCP model. The key aspect of this definition is that the simulator is only allowed to make a \emph{single} query to the summand polynomial $\SCPoly$; in contrast, the definition of \cite{BenSassonCFGRS16} allows the simulator to make as many queries to $\SCPoly$ as the malicious verifier makes to the proof oracle. (Another aspect, motivated by the simulator's limitation, is that we now have to explicitly consider a bound $\SCStrength$ on a malicious verifier's queries.) This \emph{strong} form of zero knowledge, achieved by severely restricting the simulator's access to $\SCPoly$, is crucial for achieving the results in our paper.

\begin{definition}
\label{def:strong-sumcheck-ipcp}
A \defemph{$\SCStrength$-strong perfect zero knowledge Interactive PCP system for sumcheck} with soundness error $\SoundnessError$ is a pair of interactive algorithms $(\Prover,\Verifier)$ that satisfies the following properties.
\begin{itemize}

  \item \textsc{Completeness.}
  For every $\big( (\SCInput) , \SCPoly \big) \in \SCRelation^{\yes}$,
  $\Verifier(\SCInput)$ when interacting with
  $\Prover^{\SCPoly}(\SCInput)$
  outputs a claim of the form ``$F(\vec{\gamma}) = a$'' (with $\vec{\gamma} \in \Field^{\SCVars}$ and $a \in \Field$) that is true with probability $1$.

  \item \textsc{Soundness.}
  For every $\big( (\SCInput) , \SCPoly \big) \in \SCRelation^{\no}$ and malicious prover $\Malicious{\Prover}$,
  $\Verifier(\SCInput)$
  when interacting with $\Malicious{\Prover}$
  outputs a claim of the form ``$F(\vec{\gamma}) = a$'' (with $\vec{\gamma} \in \Field^{\SCVars}$ and $a \in \Field$) that is true with probability at most $\SoundnessError$.

  \item \textsc{Zero knowledge.}
  There exists a straightline simulator $\Simulator$ such that, for every instance-witness pair $\big( (\SCInput) , \SCPoly \big) \in \SCRelation^{\yes}$ and $\SCStrength$-query malicious verifier $\Malicious{\Verifier}$, the following two distributions are equal
\begin{equation*}
\Simulator^{\Malicious{\Verifier},\SCPoly}(\SCInput)
\quad\text{and}\quad
\IPCPView{\Prover^{\SCPoly}(\SCInput)}{\Malicious{\Verifier}}
\enspace.
\end{equation*}
Moreover, the simulator $\Simulator$ makes \sunderline{only a single query} to $\SCPoly$ (at a location that possibly depends on $\Malicious{\Verifier}$ and its random choices) and runs in time $\poly(\log \SetCardinality{\Field}, \SCVars, \SCDegree, \SetCardinality{\SCSubset}, \QueryComplexity_{\Malicious{\Verifier}})$, where $\QueryComplexity_{\Malicious{\Verifier}}$ is $\Malicious{\Verifier}$'s query complexity.
\end{itemize}
\end{definition}

The main result of this section, stated below, is a construction that efficiently fulfills the definition above.

\begin{theorem}[Strong PZK Sumcheck]
\label{thm:strong-sumcheck-ipcp}
For every positive integer $\SubsetSize$ with $\SubsetSize \leq \SetCardinality{\Field}$, positive integer $\SSCVars$, and subset $\SoundnessSet$ of $\Field$, there exists a $\SubsetSize^{\SSCVars}$-strong perfect zero knowledge Interactive PCP system $(\Prover,\Verifier)$ for sumcheck with soundness error $\SoundnessError = O(\frac{(\SCVars+\SSCVars) \cdot (\SCDegree + \SubsetSize)}{\SetCardinality{\SoundnessSet}})$ and the following efficiency parameters.
\begin{itemize}

  \item \emph{Oracle round:}
  $\Prover$ sends an oracle proof string $\Proof$, consisting of the evaluation tables of polynomials $\StrongRandPoly \in \PolynomialRingIndOneXY{\Field}{\SCVars}{\VariableX}{\SSCVars}{\VariableY}{\SCDegree}{2\SubsetSize}$ and $\AuxRandPoly \in \PolynomialRingIndOne{\Field}{\SSCVars}{\VariableY}{2\SubsetSize}$ drawn uniformly at random.

  \item \emph{Interactive proof:}
  after the oracle round, $\Prover$ and $\Verifier$ engage in an ($\SCVars+\SSCVars+1$)-round interactive proof; across the interaction, the verifier sends to the prover $\SCVars + \SSCVars + 1$ field elements, while the prover sends to the verifier $O((\SCVars+\SSCVars) \cdot \SCDegree)$ field elements. (In particular, the interaction is public-coin.)

  \item \emph{Queries:}
  after the interactive proof, $\Verifier$ non-adaptively queries $\Proof$ at $\poly(\log \SetCardinality{\Field}, \SCVars, \SCDegree)$ locations.

  \item \emph{Space and time:}
  \begin{itemize}[nolistsep]
    \item $\Prover$ runs in time $\SetCardinality{\Field}^{O(\SCVars+\SSCVars)}$ and space $\poly(\log \SetCardinality{\Field}, \SCDegree^{\SCVars}, \SubsetSize^{\SSCVars}, \SetCardinality{\SCSubset})$, and
    \item $\Verifier$ runs in time $\poly(\log \SetCardinality{\Field}, \SCVars, \SCDegree, \SetCardinality{\SCSubset}, \SSCVars, \SubsetSize)$ and space $O((\SCVars + \SSCVars) \log \SetCardinality{\Field})$.
  \end{itemize}
\end{itemize}
In addition, there is a simulator $\Simulator$ witnessing perfect zero knowledge for $(\Prover, \Verifier)$ such that $\Simulator$'s single query to the summand polynomial $\SCPoly$ belongs to the set $\SoundnessSet^{\SCVars}$, and $S$ runs in time $\poly(\SetCardinality{H}, \SCVars, \SCDegree, \SSCVars, \SubsetSize, \log \SetCardinality{\Field}) \cdot \QueryComplexity_{\Malicious{\Verifier}}^{3}$.
\end{theorem}

\begin{remark}
With two-way access to the random tape, the prover can be made to run in space $\poly(\log \SetCardinality{\Field}, \SCDegree, \SCVars, \SubsetSize, \SSCVars, \SetCardinality{\SCSubset})$.
\end{remark}

We divide the proof in two steps. First (\secref{sec:pzk-step-1}), we exhibit a protocol with the above properties in a hybrid model in which the prover and verifier have access to random low-degree polynomials. Second (\secref{sec:pzk-step-2}), we use low-degree testing and self-correction to `compile' this protocol into an Interactive PCP.

\subsection{Step 1}
\label{sec:pzk-step-1}

We construct a public-coin Interactive Proof for sumcheck that achieves zero knowledge in a model where the prover and verifier have access to certain low-degree polynomials. In the soundness case, these may be arbitrary; in the zero knowledge case, these are random and depend only on the size parameters of the instance.

\begin{construction}
\label{con:strong-sumcheck-ip}
Let $\SSCSubset$ be any subset of $\Field$ of size $\SubsetSize$. In the Interactive Proof system $\pair{\Strong{\IPSCProver}}{\Strong{\IPSCVerifier}}$:
\begin{itemize}

  \item $\Strong{\IPSCProver}$ and $\Strong{\IPSCVerifier}$ receive a sumcheck instance $(\SCInput)$ as common input;
  
  \item \mbox{$\Strong{\IPSCProver}$ and $\Strong{\IPSCVerifier}$ receive polynomials $\StrongRandPoly \in \PolynomialRingIndOneXY{\Field}{\SCVars}{\VariableX}{\SSCVars}{\VariableY}{\SCDegree}{2\SubsetSize}$ and $\AuxRandPoly \in \PolynomialRingIndOne{\Field}{\SSCVars}{\VariableY}{2\SubsetSize}$ as oracles;}
  
  \item $\Strong{\IPSCProver}$ additionally receives a summand polynomial $\SCPoly \in \PolynomialRingIndOne{\Field}{\SCVars}{\VariableX}{\SCDegree}$ as an oracle.

\end{itemize}
The interaction between $\Strong{\IPSCProver}$ and $\Strong{\IPSCVerifier}$ proceeds as follows:
\begin{enumerate}

  \item \label{step:sum-challenge} $\Strong{\IPSCProver}$ sends two elements in $\Field$ to $\Strong{\IPSCVerifier}$: $z_{1} \DefineEqual \sum_{\vec{\alpha} \in \SCSubset^{\SCVars}}\sum_{\vec{\beta} \in \SSCSubset^{\SSCVars}} \StrongRandPoly(\vec{\alpha},\vec{\beta})$ and $z_{2} \DefineEqual \sum_{\vec{\beta} \in \SSCSubset^{\SSCVars}} \AuxRandPoly(\vec{\beta})$.
  
  \item $\Strong{\IPSCVerifier}$ draws a random element $\rho_{1}$ in $\Field \setminus \{0\}$ and sends it to $\Strong{\IPSCProver}$.

  \item $\Strong{\IPSCProver}$ and $\Strong{\IPSCVerifier}$ run the sumcheck IP \cite{LundFKN92} on the statement ``$\sum_{\vec{\alpha} \in \SCSubset^{\SCVars}} \MaskedPoly(\vec{\alpha})=\rho_{1}\SCSum+z_{1}$'' where
\begin{equation*}
\MaskedPoly(\VariableX_{1},\dots,\VariableX_{\SCVars})
\DefineEqual
\rho_{1} \SCPoly(\VariableX_{1},\dots,\VariableX_{\SCVars})
+
\sum_{\vec{\beta} \in \SSCSubset^{\SSCVars}} \StrongRandPoly(\VariableX_{1},\dots,\VariableX_{\SCVars}, \vec{\beta})
\enspace,
\end{equation*}
with $\Strong{\IPSCProver}$ playing the role of the prover and $\Strong{\IPSCVerifier}$ that of the verifier, and the following modification.

For $i=1,\dots,\SCVars$, in the $i$-th round, $\Strong{\IPSCVerifier}$ samples its random element $c_{i}$ from the set $\SoundnessSet$ rather than from all of $\Field$; if $\Strong{\IPSCProver}$ ever receives $c_{i} \in \Field \setminus \SoundnessSet$, it immediately aborts. In particular, in the $\SCVars$-th round, $\Strong{\IPSCProver}$ sends a polynomial $g_{\SCVars}(\VariableX_{\SCVars}) \DefineEqual \rho_{1} \SCPoly(c_{1}, \dots, c_{\SCVars-1}, \VariableX_{\SCVars}) + \sum_{\vec{\beta} \in \SSCSubset^{\SSCVars}} \StrongRandPoly(c_{1}, \dots, c_{\SCVars-1}, \VariableX_{\SCVars}, \vec{\beta})$ for some $c_{1}, \dots, c_{\SCVars-1} \in \SoundnessSet$.

  \item $\Strong{\IPSCVerifier}$ sends $c_{\SCVars} \in \SoundnessSet$ to $\Strong{\IPSCProver}$.
  
  \item $\Strong{\IPSCProver}$ sends the element $w \DefineEqual \sum_{\vec{\beta} \in \SSCSubset^{\SSCVars}} \StrongRandPoly(\vec{c}, \vec{\beta})$ to $\Strong{\IPSCVerifier}$, where $\vec{c} \DefineEqual (c_{1}, \dots, c_{\SCVars})$.
  
  \item $\Strong{\IPSCVerifier}$ draws a random element $\rho_{2}$ in $\Field \setminus \{0\}$ and sends it to $\Strong{\IPSCProver}$.

  \item $\Strong{\IPSCProver}$ and $\Strong{\IPSCVerifier}$ engage in the sumcheck IP \cite{LundFKN92} on the claim ``$\sum_{\vec{\beta} \in \SSCSubset^{\SSCVars}} \rho_{2} \StrongRandPoly(\vec{c}, \vec{\beta}) + \AuxRandPoly(\vec{\beta}) = \rho_{2} w + z_{2}$''.

  \item $\Strong{\IPSCVerifier}$ outputs the claim ``$\SCPoly(\vec{c}) = \frac{g_{\SCVars}(c_{\SCVars}) - w}{\rho_{1}}$''.

\end{enumerate}
\end{construction}

We prove the following lemma about the construction above.

\begin{lemma}
\label{lem:strong-zk-sumcheck}
The IP system $\pair{\Strong{\IPSCProver}}{\Strong{\IPSCVerifier}}$ satisfies the following properties.
\begin{itemize}

  \item \textsc{Completeness.}
  For every instance-witness pair $\big( (\SCInput) , \SCPoly \big) \in \SCRelation^{\yes}$, polynomial $\StrongRandPoly \in \PolynomialRingIndOneXY{\Field}{\SCVars}{\VariableX}{\SSCVars}{\VariableY}{\SCDegree}{2\SubsetSize}$, and polynomial $\AuxRandPoly \in \PolynomialRingIndOne{\Field}{\SSCVars}{\VariableY}{2\SubsetSize}$, $\Strong{\IPSCVerifier}^{\StrongRandPoly,\AuxRandPoly}(\SCInput)$ when interacting with $\Strong{\IPSCProver}^{\SCPoly,\StrongRandPoly,\AuxRandPoly}(\SCInput)$ outputs a claim of the form ``$\SCPoly(\vec{\gamma}) = a$'' (with $\vec{\gamma} \in \Field^{\SCVars}$ and $a \in \Field$) that is true with probability $1$.

  \item \textsc{Soundness.}
  For every instance-witness pair $\big( (\SCInput) , \SCPoly \big) \in \SCRelation^{\no}$, polynomial $\StrongRandPoly \in \PolynomialRingIndOneXY{\Field}{\SCVars}{\VariableX}{\SSCVars}{\VariableY}{\SCDegree}{2\SubsetSize}$, polynomial $\AuxRandPoly \in \PolynomialRingIndOne{\Field}{\SSCVars}{\VariableY}{2\SubsetSize}$, and malicious prover $\Malicious{\Prover}$, $\Strong{\IPSCVerifier}^{\StrongRandPoly,\AuxRandPoly}(\SCInput)$ when interacting with $\Malicious{\Prover}$ outputs a claim of the form ``$\SCPoly(\vec{\gamma}) = a$'' (with $\vec{\gamma} \in \Field^{\SCVars}$ and $a \in \Field$) that is true with probability at most $\frac{\SCVars\SCDegree}{\SetCardinality{\SoundnessSet}} + \frac{\SSCVars \cdot 2\SubsetSize + 2}{\SetCardinality{\Field}-1}$.

  \item \textsc{Zero knowledge.}
  There exists a straightline simulator $\Strong{\IPSCSimulator}$ such that, for every instance-witness pair $\big( (\SCInput) , \SCPoly \big) \in \SCRelation^{\yes}$ and $\SubsetSize^{\SSCVars}$-query malicious verifier $\Malicious{\Verifier}$, the following two distributions are equal
\begin{equation*}
\Strong{\IPSCSimulator}^{\Malicious{\Verifier},\SCPoly}(\SCInput)
\quad\text{and}\quad
\IPCPView{\Strong{\IPSCProver}^{\SCPoly,\StrongRandPoly,\AuxRandPoly}(\SCInput)}{\Malicious{\Verifier}^{\StrongRandPoly,\AuxRandPoly}}
\enspace,
\end{equation*}
where $\StrongRandPoly$ is uniformly random in $\PolynomialRingIndOneXY{\Field}{\SCVars}{\VariableX}{\SSCVars}{\VariableY}{\SCDegree}{2\SubsetSize}$ and $\AuxRandPoly$ is uniformly random in $\PolynomialRingIndOne{\Field}{\SSCVars}{\VariableY}{2\SubsetSize}$. Moreover:
\begin{itemize}
  \item $\Strong{\IPSCSimulator}$ makes a single query to $\SCPoly$ at a point in $\SoundnessSet^{\SCVars}$;
  \item $\Strong{\IPSCSimulator}$ runs in time
\begin{equation*}
(\SCVars + \SSCVars)((\SCDegree+\SubsetSize)\QueryComplexity_{\Malicious{\Verifier}}\SetCardinality{\SCSubset} + (\SCDegree+\SubsetSize)^{3}\QueryComplexity_{\Malicious{\Verifier}}^{3}) \cdot \poly(\log\SetCardinality{\Field}) = \poly(\SetCardinality{H}, \SCVars, \SCDegree, \SSCVars, \SubsetSize, \log \SetCardinality{\Field}) \cdot \QueryComplexity_{\Malicious{\Verifier}}^{3}
\end{equation*}
where $\QueryComplexity_{\Malicious{\Verifier}}$ is $\Malicious{\Verifier}$'s query complexity;
  \item $\Strong{\IPSCSimulator}$'s behavior does not depend on $\SCSum$ until after the simulated $\Malicious{\Verifier}$ sends its first message.
\end{itemize}
\end{itemize}
\end{lemma}

\begin{proof}
Completeness is clear from the protocol description and the completeness property of sumcheck. Soundness follows from the fact that, if $((\SCInput),\SCPoly) \in \SCRelation^{\no}$, we can argue as follows:
\begin{itemize}

  \item For every polynomial $\StrongRandPoly \in \PolynomialRingIndOneXY{\Field}{\SCVars}{\VariableX}{\SSCVars}{\VariableY}{\SCDegree}{2\SubsetSize}$, with probability $1 - \frac{1}{\SetCardinality{\Field}-1}$ over the choice of $\rho_{1}$, $\sum_{\vec{\alpha} \in \SCSubset^{\SCVars}} \MaskedPoly(\vec{\alpha}) \neq \rho_{1}\SCSum + z_{1}$, i.e., the sumcheck claim is false.
  
  \item Therefore, by the soundness guarantee of sumcheck, with probability at least $1 - \SCVars \SCDegree/\SetCardinality{\SoundnessSet}$, either the verifier rejects or $\rho_{1} \SCPoly(\vec{c}) + \sum_{\vec{\beta} \in \SSCSubset^{\SSCVars}} \StrongRandPoly(\vec{c}, \vec{\beta}) \neq g_{\SCVars}(c_{\SCVars})$.
  
  \item Finally, we distinguish two cases depending on $\Malicious{\Prover}$:
  \begin{itemize}[nolistsep]
    \item If $\Malicious{\Prover}$ sends $w \neq \sum_{\vec{\beta} \in \SSCSubset^{\SSCVars}} \StrongRandPoly(\vec{c}, \vec{\beta})$, then $\sum_{\vec{\beta} \in \SSCSubset^{\SSCVars}} \rho_{2} \StrongRandPoly(\vec{c}, \vec{\beta}) + \AuxRandPoly(\vec{\beta}) \neq \rho_{2} w + z_{2}$ with probability $1 - \frac{1}{\SetCardinality{\Field}-1}$ over the choice of $\rho_{2}$, for any choice of $A$. In this case, by the soundness guarantee of the sumcheck protocol the verifier rejects with probability at least $1 - \frac{\SSCVars\cdot2\SubsetSize}{\SetCardinality{\Field}}$.
    \item If $\Malicious{\Prover}$ sends $w = \sum_{\vec{\beta} \in \SSCSubset^{\SSCVars}} \StrongRandPoly(\vec{c}, \vec{\beta})$, then $\SCPoly(\vec{c}) \neq \frac{g_{\SCVars}(c_{\SCVars}) - w}{\rho_{1}}$ with probability $1$.
  \end{itemize}
\end{itemize}
Taking a union bound on the above cases yields the claimed soundness error.

To show the (perfect) zero knowledge guarantee, we need to construct a suitably-efficient straightline simulator that perfectly simulates the view of any malicious verifier $\Malicious{\Verifier}$. We first construct an \emph{inefficient} simulator $\Strong{\SlowSimulator}$, and prove that its output follows the desired distribution; afterwards, we explain how the simulator can be made efficient.

\begin{mdframed}
{\small
The simulator $\Strong{\SlowSimulator}$, given straightline access to $\Malicious{\Verifier}$ and oracle access to $\SCPoly$, works as follows:

\begin{enumerate}

  \item Draw $\Simulated{\StrongRandPoly} \in \PolynomialRingIndOneXY{\Field}{\SCVars}{\VariableX}{\SSCVars}{\VariableY}{\SCDegree}{2\SubsetSize}$ and $\Simulated{\AuxRandPoly} \in \PolynomialRingIndOne{\Field}{\SSCVars}{\VariableY}{\SCDegree}$ uniformly at random.
  
  \item Begin simulating $\Malicious{\Verifier}$. Its queries to $\StrongRandPoly$ and $\AuxRandPoly$ are answered according to $\Simulated{\StrongRandPoly}$ and $\Simulated{\AuxRandPoly}$ respectively.
  
  \item Send $\Simulated{z}^{1} \DefineEqual \sum_{\vec{\alpha} \in \SCSubset^{\SCVars}} \sum_{\vec{\beta} \in \SSCSubset^{\SSCVars}} \StrongRandPoly(\vec{\alpha},\vec{\beta})$ and $\Simulated{z}^{2} \DefineEqual \sum_{\vec{\beta} \in \SSCSubset^{\SSCVars}} \AuxRandPoly(\vec{\beta})$ to $\Malicious{\Verifier}$.
  
  \item \label{step:strong-zksc-first-sumcheck} Receive $\tilde{\rho}_{1}$. Draw $\Simulated{\MaskedPoly} \in \PolynomialRingIndOne{\Field}{\SCVars}{\VariableX}{\SCDegree}$ uniformly at random conditioned on $\sum_{\vec{\alpha} \in \SCSubset^{\SCVars}} \Simulated{\MaskedPoly}(\vec{\alpha}) = \tilde{\rho}_{1} \SCSum + \Simulated{z}^{1}$, then engage in the sumcheck protocol on the claim ``$\sum_{\vec{\alpha} \in \SCSubset^{\SCVars}} \Simulated{\MaskedPoly}(\vec{\alpha}) = \tilde{\rho}_{1} \SCSum + \Simulated{z}^{1}$''. If in any round $\Malicious{\Verifier}$ sends $c_i \not\in \SoundnessSet$ as a challenge, abort.
  
  \item \label{step:strong-zksc-redraw} Let $\vec{c} \in \SoundnessSet^{\SCVars}$ be the point chosen by $\Malicious{\Verifier}$ in the sumcheck protocol above. Query $\SCPoly(\vec{c})$, and draw $\Simulated{\StrongRandPoly}' \in \PolynomialRingIndOneXY{\Field}{\SCVars}{\VariableX}{\SSCVars}{\VariableY}{\SCDegree}{2\SubsetSize}$ uniformly at random conditioned on
  \begin{itemize}[nolistsep]
    \item $\sum_{\vec{\alpha} \in \SCSubset^{\SCVars}} \sum_{\vec{\beta} \in \SSCSubset^{\SSCVars}} \Simulated{\StrongRandPoly}'(\vec{\alpha}, \vec{\beta}) = \Simulated{z}^{1}$,
    \item $\sum_{\vec{\beta} \in \SSCSubset^{\SSCVars}} \Simulated{\StrongRandPoly}'(\vec{c}, \vec{\beta}) = \Simulated{w}$, where $\Simulated{w} \DefineEqual \Simulated{\MaskedPoly}(\vec{c}) - \tilde{\rho}_{1} \SCPoly(\vec{c})$, and
    \item $\Simulated{\StrongRandPoly}'(\vec{\gamma}) = \Simulated{\StrongRandPoly}(\vec{\gamma})$ for all previous queries $\vec{\gamma}$ to $\StrongRandPoly$.
  \end{itemize}    
  From this point on, answer all queries to $\StrongRandPoly$ with $\Simulated{\StrongRandPoly}'$.

  \item Receive $\tilde{\rho}_{2}$ from $\Malicious{\Verifier}$. Draw $\Simulated{\MaskedPoly}' \in \PolynomialRingIndOne{\Field}{\SCVars}{\VariableX}{\SCDegree}$ uniformly at random conditioned on
  \begin{itemize}[nolistsep]
    \item $\sum_{\vec{\beta} \in \SSCSubset^{\SSCVars}} \Simulated{\MaskedPoly}'(\vec{\beta}) = \tilde{\rho}_{2} \Simulated{w} + \Simulated{z}^{2}$, and
    \item $\Simulated{\MaskedPoly}'(\vec{\gamma}) = \tilde{\rho}_{2} \Simulated{\StrongRandPoly}'(\vec{c}, \vec{\gamma}) + \Simulated{\AuxRandPoly}(\vec{\gamma})$ for all previous queries $\vec{\gamma}$ to $\AuxRandPoly$.
  \end{itemize}
  From this point on, answer all queries to $\AuxRandPoly(\vec{\gamma})$ with $\Simulated{\MaskedPoly}'(\vec{\gamma}) - \tilde{\rho}_{2} \Simulated{\StrongRandPoly}'(\vec{c}, \vec{\gamma})$.
  
  \item Engage in the sumcheck protocol on the claim ``$\sum_{\vec{\beta} \in \SSCSubset^{\SSCVars}} \Simulated{\MaskedPoly}'(\vec{\beta}) = \tilde{\rho}_{2} \Simulated{w} + \Simulated{z}^{2}$''.
  
  \item Output the view of the simulated $\Malicious{\Verifier}$.
\end{enumerate}
}
\end{mdframed}

Let $\MaskedPoly(\vec{\VariableX}) \DefineEqual \tilde{\rho}_{1} \SCPoly(\vec{\VariableX}) + \sum_{\vec{\beta} \in \SSCSubset^{\SSCVars}} \StrongRandPoly(\vec{\VariableX}, \vec{\beta})$, and $\MaskedPoly'(\vec{\VariableY}) \DefineEqual \tilde{\rho}_{2} \StrongRandPoly(\vec{c}, \vec{\VariableY}) + \AuxRandPoly(\vec{\VariableY})$. Observe that there exists a (deterministic) function $\FView{\cdot}$ such that
\begin{equation*}
\IPCPView{\Strong{\IPSCProver}^{\SCPoly,\StrongRandPoly,\AuxRandPoly}}{\Malicious{\Verifier}^{\SCPoly,\StrongRandPoly,\AuxRandPoly}} = \FView{\MaskedPoly, \MaskedPoly', \SCPoly, \StrongRandPoly, \StrongRandPoly, \Randomness}
\quad\text{and}\quad
\Strong{\SlowSimulator}^{\Malicious{\Verifier},\SCPoly} = \FView{\Simulated{\MaskedPoly}, \Simulated{\MaskedPoly}', \SCPoly, \Simulated{\StrongRandPoly}, \Simulated{\StrongRandPoly}', \Randomness} \enspace,
\end{equation*}
where the random variable $\Randomness$ is $\Malicious{\IOPVerifier}$'s private randomness. Indeed,
\begin{itemize}[nolistsep]

\item $\Malicious{\Verifier}$'s queries to $\StrongRandPoly$ up to \stepref{step:strong-zksc-redraw} are answered by $\StrongRandPoly$ and $\Simulated{\StrongRandPoly}$ respectively, and after \stepref{step:strong-zksc-redraw} by $\StrongRandPoly$ and $\Simulated{\StrongRandPoly}'$ respectively;

\item $\Malicious{\Verifier}$'s queries to $\AuxRandPoly(\cdot)$ are answered by $\MaskedPoly'(\cdot) - \tilde{\rho}_{2} \StrongRandPoly(\vec{c}, \cdot)$ and $\Simulated{\MaskedPoly}'(\cdot) - \tilde{\rho}_{2} \Simulated{\StrongRandPoly}'(\vec{c}, \cdot)$ respectively;

\item the messages between $\Strong{\IPSCProver}$ and $\Malicious{\Verifier}$ in the first sumcheck are a sumcheck on $\MaskedPoly$ and $\Simulated{\MaskedPoly}$ respectively;

\item $w = \MaskedPoly(\vec{c}) - \tilde{\rho}_{1} \SCPoly(\vec{c})$ in one case and $\Simulated{w} = \Simulated{\MaskedPoly}(\vec{c}) - \tilde{\rho}_{1} \SCPoly(\vec{c})$ in the other case; and

\item the messages between $\Strong{\IPSCProver}$ and $\Malicious{\Verifier}$ in the second sumcheck are a sumcheck on $\MaskedPoly'$ and $\Simulated{\MaskedPoly}'$ respectively.

\end{itemize}
We now argue that the outputs of $\FView{\cdot}$ in the two cases are identically distributed:
\begin{equation*}
\FView{\MaskedPoly, \MaskedPoly', \SCPoly, \StrongRandPoly, \StrongRandPoly, \Randomness}
\sim
\FView{\Simulated{\MaskedPoly}, \Simulated{\MaskedPoly}', \SCPoly, \Simulated{\StrongRandPoly}, \Simulated{\StrongRandPoly}', \Randomness}
\enspace.
\end{equation*}
We do so via several `hybrids', the first of which considers the distribution of $\MaskedPoly$. For every fixed $B \in \PolynomialRingIndOne{\Field}{\SCVars}{\VariableX}{\SCDegree}$, consider the following probability value, where $\StrongRandPoly \in \PolynomialRingIndOneXY{\Field}{\SCVars}{\VariableX}{\SSCVars}{\VariableY}{\SCDegree}{2\SubsetSize}$ is uniformly random, and $U \subseteq \Field^{\SCVars+\SSCVars} \times \Field$ is a set of query-answer pairs:
\begin{align*}
p(B|U) &\DefineEqual \Pr_{Z} \big[
\MaskedPoly(\vec{\alpha}) = B(\vec{\alpha})
\quad \forall\, \vec{\alpha} \in \Field^{\SCVars} 
\;
\big\vert
\;
\StrongRandPoly(\vec{\gamma}) = \omega
\quad \forall\, (\vec{\gamma},\omega) \in U 
\big] \\
&= \Pr_{Z} \left[
\sum_{\vec{\beta} \in \SSCSubset^{\SSCVars}} \StrongRandPoly(\vec{\alpha}, \vec{\beta}) = B(\vec{\alpha}) - \tilde{\rho}_{1} \SCPoly(\vec{\alpha})
\quad \forall\, \vec{\alpha} \in \Field^{\SCVars} 
\;
\middle\vert
\;
\StrongRandPoly(\vec{\gamma}) = \omega
\quad \forall\, (\vec{\gamma},\omega) \in U
\right]
\end{align*}
The probability value $p(B|U)$ describes the distribution of $\MaskedPoly$ in a real execution, when $U$ is the set of query-answer pairs to $\StrongRandPoly$ made by $\Malicious{\Verifier}$. By our lower bounds on the algebraic query complexity of polynomial summation (\corref{cor:partial-sum-indep-vars}), whenever $\SetCardinality{U} < \SubsetSize^{\SSCVars}$,
\begin{equation*}
p(B|U) = \Pr_{Z} \left[
\sum_{\vec{\beta} \in \SSCSubset^{\SSCVars}} \StrongRandPoly(\vec{\alpha}, \vec{\beta}) = B(\vec{\alpha}) - \tilde{\rho}_{1} \SCPoly(\vec{\alpha})
\quad \forall\, \vec{\alpha} \in \Field^{\SCVars} \right]
=
1/\SetCardinality{\PolynomialRingIndOne{\Field}{\SCVars}{\VariableX}{\SCDegree}}.
\end{equation*}
since $\StrongRandPoly$ is uniformly random. Thus $\MaskedPoly$ is uniformly random, and hence is identically distributed to $\Simulated{\MaskedPoly}$. Thus:
\begin{equation*}
\FView{\mhl{\MaskedPoly}, \MaskedPoly', \SCPoly, \StrongRandPoly, \StrongRandPoly, \Randomness}
\sim
\FView{\mhl{\Simulated{\MaskedPoly}}, \MaskedPoly', \SCPoly, \StrongRandPoly, \StrongRandPoly, \Randomness}
\enspace.
\end{equation*}

Next, $\MaskedPoly'(\vec{\VariableY}) = \tilde{\rho}_{2} \StrongRandPoly(\vec{c}, \vec{\VariableY}) + \AuxRandPoly(\vec{\VariableY})$ is uniformly random such that $\sum_{\vec{\beta} \in \SSCSubset^{\SSCVars}} \MaskedPoly'(\vec{\beta}) = \tilde{\rho}_{2} w + z_{2}$, while $\Simulated{\MaskedPoly'}$ is uniformly random such that $\sum_{\vec{\beta} \in \SSCSubset^{\SSCVars}} \Simulated{\MaskedPoly}'(\vec{\beta}) = \tilde{\rho}_{2} \Simulated{w} + \Simulated{z}^{2}$ and $\Simulated{\MaskedPoly}'(\vec{\gamma}_i) = \tilde{\rho}_{2} \Simulated{\StrongRandPoly}'(\vec{c}, \vec{\gamma}_i) + \Simulated{\AuxRandPoly}(\vec{\gamma}_i)$ for $\vec{\gamma}_{1}, \ldots \vec{\gamma}_{k}$ adversarially chosen. Since $\Simulated{\AuxRandPoly}$ is itself uniformly random such that $\sum_{\vec{\beta} \in \SSCSubset^{\SSCVars}} \Simulated{\AuxRandPoly}(\vec{\beta}) = 0$, this gives that $(\MaskedPoly', \StrongRandPoly) \sim (\Simulated{\MaskedPoly}', \Simulated{\StrongRandPoly}')$, as $\StrongRandPoly$ and $\Simulated{\StrongRandPoly}'$ are both uniformly random summing to zero. Moreover, queries to $\Simulated{\StrongRandPoly}'$ are independent of $\Simulated{\MaskedPoly}$ by \corref{cor:partial-sum-indep-vars}. Thus,
\begin{equation*}
\FView{\Simulated{\MaskedPoly}, \mhl{\MaskedPoly'}, \SCPoly, \mhl{\StrongRandPoly}, \mhl{\StrongRandPoly}, \Randomness}
\sim
\FView{\Simulated{\MaskedPoly}, \mhl{\Simulated{\MaskedPoly}'}, \SCPoly, \mhl{\Simulated{\StrongRandPoly}'}, \mhl{\Simulated{\StrongRandPoly}'}, \Randomness}
\enspace.
\end{equation*}

Finally we examine $\Simulated{\StrongRandPoly}$, which is used only to answer queries to $\StrongRandPoly$ until \stepref{step:strong-zksc-redraw}. Note that $\Simulated{\StrongRandPoly}'$ is drawn identically to $\Simulated{\StrongRandPoly}$ except for the additional condition that $\sum_{\vec{\beta} \in \SSCSubset^{\SSCVars}} \Simulated{\StrongRandPoly}'(\vec{c}, \vec{\beta}) = \Simulated{v}$, and that $\Malicious{\Verifier}$'s queries to $\StrongRandPoly$ remain consistent. By \corref{cor:partial-sum-indep-vars}, the summation condition is independent of the answers to $\Malicious{\Verifier}$'s queries to $\StrongRandPoly$. So $\Malicious{\Verifier}$'s queries to $\Simulated{\StrongRandPoly}$ are hence identically distributed to the same queries to $\Simulated{\StrongRandPoly}'$, and thus
\begin{equation*}
\FView{\Simulated{\MaskedPoly}, \Simulated{\MaskedPoly}', \SCPoly, \mhl{\Simulated{\StrongRandPoly}'}, \Simulated{\StrongRandPoly}', \Randomness}
\sim
\FView{\Simulated{\MaskedPoly}, \Simulated{\MaskedPoly}', \SCPoly, \mhl{\Simulated{\StrongRandPoly}}, \Simulated{\StrongRandPoly}', \Randomness}
\enspace.
\end{equation*}
This concludes the argument for the correctness of the inefficient simulator $\Strong{\SlowSimulator}$.

To complete the proof of zero knowledge, we note that $\Strong{\SlowSimulator}$ can be transformed into an efficient simulator $\Strong{\IPSCSimulator}$ by using succinct constraint detection for the Reed--Muller code extended with partial sums \cite{BenSassonCFGRS16}: more precisely, we can use the algorithm of \corref{cor:efficient-poly-simulator} to answer both point and sum queries to $\StrongRandPoly$, $\AuxRandPoly$, and $\MaskedPoly$, in a stateful way, maintaining corresponding tables $\AnsTable{\Simulated{\StrongRandPoly}}$, $\AnsTable{\Simulated{\AuxRandPoly}}$, and $\AnsTable{\Simulated{\MaskedPoly}}$.
\end{proof}

\subsection{Step 2}
\label{sec:pzk-step-2}

The Interactive Proof described and analyzed in \secref{sec:pzk-step-1} assumes that the prover and verifier have access to certain low-degree polynomials. We now use low-degree testing and self-correction to compile that Interactive Proof into an Interactive PCP, where the prover sends to the verifier evaluations of these polynomials as part of the proof oracle. This will conclude the proof of \thmref{thm:strong-sumcheck-ipcp}.

\begin{proof}[Proof of \thmref{thm:strong-sumcheck-ipcp}]
Construct an IPCP system $(\Prover,\Verifier)$ for sumcheck as follows:
\begin{itemize}
		
  \item The prover $\Prover$, given input $(\SCInput)$ and oracle access to $\SCPoly$, samples polynomials $\StrongRandPoly \in \PolynomialRingIndOneXY{\Field}{\SCVars}{\VariableX}{\SSCVars}{\VariableY}{\SCDegree}{2\SubsetSize}$, $\AuxRandPoly \in \PolynomialRingIndOne{\Field}{\SSCVars}{\VariableY}{2\SubsetSize}$ uniformly at random, and sends their evaluations to the verifier $\Verifier$; then $\Prover$ simulates $\Strong{\IPSCProver}^{\SCPoly,\StrongRandPoly,\AuxRandPoly}(\Field,\SCVars,\SCDegree,\SCSubset,\SCSum)$.
		
  \item The verifier $\Verifier$, after receiving a proof string $\Proof = (\StrongRandPoly, \AuxRandPoly)$, simulates $\Strong{\IPSCVerifier}^{\SCPoly,\Proof}(\SCInput)$ up to $\IPSCVerifier$'s two queries $\vec{\alpha}_{1} \in \Field^{\SCVars+\SSCVars}, \vec{\alpha}_{2} \in \Field^{\SSCVars}$ to $\StrongRandPoly,\AuxRandPoly$ respectively (which occur after the interaction), which $\Verifier$ does not answer directly but instead answers as follows. First, $\Verifier$ checks that $\StrongRandPoly$ is $\varrho$-close to the evaluation of a polynomial in $\PolynomialRingIndOneXY{\Field}{\SCVars}{\VariableX}{\SSCVars}{\VariableY}{\SCDegree}{2\SubsetSize}$ by performing an individual-degree test with proximity parameter $\varrho \DefineEqual \frac{1}{8}$ and soundness error $\epsilon \DefineEqual \frac{\SCVars\SCDegree}{\SetCardinality{\Field}}$ \cite{GoldreichS06,GurR15}; then, $\Verifier$ computes $\StrongRandPoly(\vec{\alpha}_{1})$ via self-correction with soundness error $\epsilon$ \cite{RubinfeldS96,AroraS03}, and replies with this value. The procedures above are repeated to determine $\AuxRandPoly(\vec{\alpha}_{2})$. Both procedures require $\poly(\log \SetCardinality{\Field}, \SCVars, \SCDegree)$ queries and time. Finally, $\Verifier$ rejects if $\IPSCVerifier$ rejects or the individual degree test rejects.
		
\end{itemize}
Completeness and perfect zero knowledge of $(\Prover,\Verifier)$ are inherited, in a straightforward way, from those of $(\IPSCProver,\IPSCVerifier)$. We now argue soundness. Consider an instance-witness pair $\big( (\SCInput) , \SCPoly \big) \in \SCRelation^{\no}$ and a malicious prover $\Malicious{\Prover}$, and denote by $\Malicious{\Proof} = (\Malicious{\StrongRandPoly}, \Malicious{\AuxRandPoly})$ the proof string sent by $\Malicious{\Prover}$. We distinguish between the following two cases.
\begin{itemize}
		
  \item \emph{Case 1: $\Malicious{\StrongRandPoly}$ is $\varrho$-far from evaluations of polynomials in $\PolynomialRingIndOneXY{\Field}{\SCVars}{\VariableX}{\SSCVars}{\VariableY}{\SCDegree}{2\SubsetSize}$ or $\Malicious{\AuxRandPoly}$ is $\varrho$-far from evaluations of polynomials in $\PolynomialRingIndOne{\Field}{\SSCVars}{\VariableY}{2\SubsetSize}$.}
  
  In this case, the low-degree test accepts with probability at most $\epsilon$.
		
  \item \emph{Case 2: $\Malicious{\StrongRandPoly}$ is $\varrho$-close to evaluations of polynomials in $\PolynomialRingIndOneXY{\Field}{\SCVars}{\VariableX}{\SSCVars}{\VariableY}{\SCDegree}{2\SubsetSize}$ and $\Malicious{\AuxRandPoly}$ is $\varrho$-close to evaluations of polynomials in $\PolynomialRingIndOne{\Field}{\SSCVars}{\VariableY}{2\SubsetSize}$.}
		
  In this case, let $\Malicious{\StrongRandPoly}'$ be the unique polynomial in $\PolynomialRingIndOneXY{\Field}{\SCVars}{\VariableX}{\SSCVars}{\VariableY}{\SCDegree}{2\SubsetSize}$ whose evaluation is $\varrho$-close to $\Malicious{\StrongRandPoly}$; this polynomial exists because $\varrho$ is less than the unique decoding radius (of the corresponding Reed--Muller code), which equals $\frac{1}{2} (1 - \frac{\SCDegree}{\SetCardinality{\Field}})^{\SCVars} (1 - \frac{2\SubsetSize}{\SetCardinality{\Field}})^{\SSCVars}$, and is at least $\frac{1}{4}$ by the assumption that $\frac{\SCVars\SCDegree}{\SetCardinality{\Field}} + \frac{\SSCVars \cdot 2\SubsetSize}{\SetCardinality{\Field}} < \frac{1}{2}$. $\Malicious{\AuxRandPoly}'$ is defined analogously. By the soundness of $(\IPSCProver,\IPSCVerifier)$, the probability that $\IPSCVerifier^{\SCPoly,\Malicious{\StrongRandPoly}',\Malicious{\AuxRandPoly}'}$ accepts is at most $\frac{\SCVars\SCDegree}{\SetCardinality{\SoundnessSet}} + \frac{\SSCVars \cdot 2\SubsetSize + 2}{\SetCardinality{\Field}}$ (see \lemref{lem:strong-zk-sumcheck}). However $\Verifier$ only has access to $\Malicious{\Proof}$, and uses self-correction on it to compute $\Malicious{\StrongRandPoly}', \Malicious{\AuxRandPoly}'$ at the location $\vec{\alpha} \in \Field^{\SCVars}$ required by $\IPSCVerifier$; the probability that the returned values are not correct is at most $2\epsilon$. Hence, by a union bound, $\Verifier$ accepts with probability at most $\frac{\SCVars\SCDegree}{\SetCardinality{\SoundnessSet}} + \frac{\SSCVars \cdot 2\SubsetSize + 2}{\SetCardinality{\Field}} + 2\epsilon$.
		
\end{itemize}
Overall, we deduce that $\Verifier$ accepts with probability at most $\max \{\epsilon \,,\, \frac{\SCVars\SCDegree}{\SetCardinality{\SoundnessSet}} + \frac{\SSCVars \cdot 2\SubsetSize + 2}{\SetCardinality{\Field}} + 2\epsilon \} \leq 6 \frac{(\SCVars+\SSCVars) \cdot (\SCDegree+\SubsetSize)}{\SetCardinality{\SoundnessSet}}$.
\end{proof}

\doclearpage
\section{Zero knowledge for non-deterministic exponential time}
\label{sec:zk-nexp}

In this section we use the zero knowledge sumcheck protocol developed in \secref{sec:strong-zk-sumcheck} (along with the \cite{BenSassonCFGRS16} protocol) to derive a zero knowledge analogue of the \cite{BabaiFL91,BabaiFLS91} protocol for $\NEXP$.%
\footnote{This section is written so that the proof can be understood independently of subsequent sections of the paper. Using the framework of sum-product circuits developed in \secref{sec:sum-product-protocols}, we can simplify this proof; see \appref{sec:zk-nexp-spc}.}%
Recall that in this protocol, the prover first sends a low-degree extension of a $\NEXP$ witness, and then engages in the \cite{LundFKN92} sumcheck protocol on a polynomial related to the instance. To make this zero knowledge, the prover will first take a \emph{randomized} low-degree extension $R$ of the witness (which provides some bounded independence). The oracle contains a commitment to $R$: the prover draws a polynomial uniformly at random subject to the condition that `summing out' a few of its variables yields $R$, and places its evaluation in the oracle.

The prover and verifier then engage in the zero knowledge sumcheck detailed in \secref{sec:strong-zk-sumcheck} on the \cite{BabaiFLS91} polynomial. This ensures that the verifier learns nothing through the interaction except for a single evaluation of the summand polynomial, which corresponds to learning a constant number of evaluations of the randomized witness. Bounded independence ensures that these evaluations do not leak any information. The prover provides these evaluations to the verifier, who will then check their correctness by engaging in an instance of the \cite{BenSassonCFGRS16} sumcheck protocol for each evaluation. Note that here we are satisfied with the weaker guarantee provided by the \cite{BenSassonCFGRS16} protocol because the simulator is able to simulate any polynomial number of queries to the commitment.

Following \cite{BabaiFLS91}, the arithmetization encodes bit strings as elements in $\SPSubset^{m}$ for some $\SPSubset$ of size $\poly(\BitSize{B})$, rather than with $\SPSubset=\Bits$ as in \cite{BabaiFL91}, for greater efficiency.

\medskip
We start by defining the \emph{oracle 3-satisfiability problem}, which is the $\NEXP$-complete problem used by \cite{BabaiFL91} to construct two-prover interactive proofs for $\NEXP$.

\begin{definition}[$\OracleSATRelation$]
	\label{def:O3SAT}
	The \defemph{oracle 3-satisfiability relation}, denoted $\OracleSATRelation$, consists of all instance-witness pairs $(\Instance,\Witness)=\big((r, s, B), A\big)$, where $r, s$ are positive integers, $B \colon \Bits^{r + 3s + 3} \to \Bits$ is a boolean formula, and $A \colon \Bits^{s} \to \Bits$ is a function, that satisfy the following condition:
	\begin{equation*}
	\forall\,z \in \Bits^{r},\;
	\forall\,b_{1}, b_{2}, b_{3} \in \Bits^{s},\;
	B\big(z, b_{1}, b_{2}, b_{3}, A(b_{1}), A(b_{2}), A(b_{3})\big) = 1
	\enspace.
	\end{equation*}
\end{definition}

\begin{theorem}[PZK IPCP for $\NEXP$]
	\label{thm:pzk-for-nexp}
	For every query bound function $\SCStrength(n)$, the $\NEXP$-complete relation $\OracleSATRelation$ has a (public coin and non-adaptive) Interactive PCP that is perfect zero knowledge against all $\SCStrength$-query malicious verifiers. In more detail:
	{\small
		\begin{equation*}
		\OracleSATRelation \in
		\PZKIPCP
		\left[
		\begin{tabular}{rl}
		soundness error: & $1/2$ \\
		round complexity: & $O(r + s + \log \SCStrength)$ \\[1mm]
		proof length: & $\poly(2^{\BitSize{B}}, \SCStrength)$ \\
		query complexity: & $\poly(\BitSize{B} + \log \SCStrength)$ \\[1mm]
		prover time: & $\poly(2^{\SetCardinality{B}}, \SCStrength)$ \\
		verifier time: &  $\poly(\BitSize{B} + \log \SCStrength)$ \\[1mm]
		verifier space: & $O(\BitSize{B} + \log \SCStrength)$ \\[1mm]
		simulator overhead: & $\poly(\BitSize{B} + \log \SCStrength) \cdot \QueryComplexity_{\Malicious{\Verifier}}^{3}$
		\end{tabular}
		\right]
		\enspace.
		\end{equation*}}
\end{theorem}

Note that the prover running time given above assumes that the prover is given a witness as auxiliary input.

\begin{proof}
	Let $\Field$ be an extension field of $\Field_{2}$. Let $\LD{B} \colon \Field^{m} \to \Field$ be the `direct' arithmetization of the negation of $B$: rewrite $B$ by using ANDs and NOTs; negate its output; replace each $\mathrm{AND}(a,b)$ with $a \cdot b$ and $\mathrm{NOT}(a)$ with $1-a$. For every $\vec{x} \in \Bits^{r + 3s + 3}$, $\LD{B}(\vec{x}) = 0$ if $B(\vec{x})$ is true, and $\LD{B}(\vec{x}) = 1$ if $B(\vec{x})$ is false. Note that $\LD{B}$ is computable in time $\poly(\BitSize{B})$ and has total degree $O(\BitSize{B})$.
	
	Observe that $(r, s, B) \in \OracleSATRelation$ if and only if there exists a multilinear function $\LD{A} \colon \Field^{s} \to \Field$ that is boolean on $\Bits^{s}$ such that $\LD{B}(\vec{z}, \vec{b}_{1}, \vec{b}_{2}, \vec{b}_{3}, \LD{A}(\vec{b}_{1}), \LD{A}(\vec{b}_{2}), \LD{A}(\vec{b}_{3})) = 0$ for all $\vec{z} \in \Bits^{r}$, $\vec{b}_{1}, \vec{b}_{2}, \vec{b}_{3} \in \Bits^{s}$. The requirement that $\LD{A}$ is boolean on $\Bits^{s}$ can be encoded by $2^{s}$ constraints: $\LD{A}(\vec{b})(1 - \LD{A}(\vec{b})) = 0$ for every $\vec{b} \in \Bits^{s}$.
	
	These constraints can be expressed as follows:
	\begin{align*}
	\left\{
	g_{1}(\vec{\alpha}) \DefineEqual \LD{B}(\vec{z}, \vec{b}_{1}, \vec{b}_{2}, \vec{b}_{3}, \LD{A}(\vec{b}_{1}), \LD{A}(\vec{b}_{2}), \LD{A}(\vec{b}_{3})) = 0 \right\}_{\vec{z} \in \Bits^{r},\, \vec{b}_{i} \in \Bits^{s}} \\
	\left\{
	g_{2}(\vec{\beta}) \DefineEqual \LD{A}(\vec{b})(1 - \LD{A}(\vec{b})) = 0 \right\}_{\vec{b} \in \Bits^{s}}
	\quad\quad\quad\quad\quad\quad\quad
	\end{align*}
	
	Let $F$ be the polynomial over $\Field$ given by
	\begin{equation*}
	F(\vec{\VariableX}, \vec{\VariableY})
	\DefineEqual
	\sum_{\vec{\alpha} \in \Bits^{r + 3s}}
	\left(
	g_{1}(\vec{\alpha}) \vec{\VariableX}^{\vec{\alpha}} + g_{2}(\vec{\alpha}_{[s]}) \vec{\VariableY}^{\vec{\alpha}}
	\right)
	\enspace,
	\end{equation*}
	where $\vec{\VariableX}^{\vec{\alpha}} \DefineEqual \VariableX_{1}^{\alpha_{1}} \cdots \VariableX_{\ell}^{\alpha_{\ell}}$ for $\vec{\alpha} \in \Bits^{\ell}$, and $\vec{\alpha}_{[s]}$ are the first $s$ coordinates in $\vec{\alpha}$.
	
	Note that $F$ is the zero polynomial if and only if all the above equations hold. Since $F$ is a polynomial of total degree $r + 3s$, if $F$ is not the zero polynomial then it is zero on at most an $\frac{r + 3s}{\SetCardinality{\Field}}$ fraction of points in $\Field^{2(r + 3s)}$.
	
	For $\alpha_{i} \in \Bits$ it holds that $\VariableX_{i}^{\alpha_{i}} = 1 + (\VariableX_{i} - 1)\alpha_{i}$, so we can also write
	\begin{equation*}
	F(\vec{\VariableX}, \vec{\VariableY})
	= \sum_{\vec{\alpha} \in \Bits^{r + 3s}}
	\left(
	g_{1}(\vec{\alpha})       \cdot \prod_{i=1}^{r+3s} (1 + (\VariableX_{i} - 1)\alpha_{i})
	+  g_{2}(\vec{\alpha}_{[s]}) \cdot \prod_{i=1}^{r+3s} (1 + (\VariableY_{i} - 1)\alpha_{i})
	\right) =: \sum_{\vec{\alpha} \in \Bits^{r + 3s}} f(\vec{\VariableX},\vec{\VariableY},\vec{\alpha})
	\enspace.
	\end{equation*}
	
	Let $\SCSubset$ be a subfield of $\Field$; define $m_{1} \DefineEqual r/\log \SetCardinality{\SPSubset}$ and $m_{2} \DefineEqual s/\log \SetCardinality{\SPSubset}$ (assuming without loss of generality that both are integers). For $i \in \{1, 2\}$, let $\gamma_{i} \colon \SPSubset^{m_{i}} \to \Bits^{m_{i} \log \SetCardinality{\SPSubset}}$ be the lexicographic order on $\SPSubset^{m_{i}}$. The low-degree extension $\LD{\gamma}_{i}$ of $\gamma_{i}$ is computable by an arithmetic circuit constructible in time $\poly(\SetCardinality{\SPSubset}, m_{i}, \log \SetCardinality{\Field})$ \cite[Claim 4.2]{GoldwasserKR15}. Let $\gamma \colon \SPSubset^{m_{1} + 3m_{2}} \to \Bits^{r + 3s}$ be such that $\gamma(\vec{\alpha}, \vec{\beta}_{1}, \vec{\beta}_{2}, \vec{\beta}_{3}) = (\gamma_{1}(\vec{\alpha}), \gamma_{2}(\vec{\beta}_{1}), \gamma_{2}(\vec{\beta}_{2}), \gamma_{2}(\vec{\beta}_{3}))$ for all $\vec{\alpha} \in \SPSubset^{m_{1}}$, $\vec{\beta}_{1}, \vec{\beta}_{2}, \vec{\beta}_{3} \in \SPSubset^{m_{2}}$; let $\LD{\gamma} \colon \Field^{m_{1} + 3m_{2}} \to \Field^{r + 3s}$ be its low-degree extension.
	
	We can use the above notation to write $F$ equivalently as
	\begin{align*}
	F(\vec{\VariableX}, \vec{\VariableY}) =
	\sum_{\substack{\vec{\alpha} \in \SCSubset^{m_{1}} \\
			\vec{\beta}_{1}, \vec{\beta}_{2}, \vec{\beta}_{3} \in \SCSubset^{m_{2}}}}
	&g_{1}(\LD{\gamma}(\vec{\alpha}, \vec{\beta}_{1}, \vec{\beta}_{2}, \vec{\beta}_{3}))
	\prod_{i=1}^{r+3s} (1 + (\VariableX_{i} - 1)\LD{\gamma}(\vec{\alpha}, \vec{\beta}_{1}, \vec{\beta}_{2}, \vec{\beta}_{3})_{i}) \\
	&+g_{2}(\LD{\gamma}_{2}(\vec{\beta}_{1}))
	\prod_{i=1}^{r+3s} (1 + (\VariableY_{i} - 1)\LD{\gamma}(\vec{\alpha}, \vec{\beta}_{1}, \vec{\beta}_{2}, \vec{\beta}_{3})_{i}) \enspace.
	\end{align*}
	
	We are now ready to specify the protocol. Let $\SSCVars \DefineEqual \lceil \log \SCStrength / \log \SetCardinality{\SCSubset} \rceil$.
	\begin{enumerate}
		\item The prover draws a polynomial $\StrongRandPoly$ uniformly at random from $\PolynomialRingIndOneXY{\Field}{m_{2}}{\VariableX}{\SSCVars}{\VariableY}{\SetCardinality{\SCSubset}+2}{2\SetCardinality{\SCSubset}}$, subject to the condition that $\sum_{\vec{\beta} \in \SSCSubset^{\SSCVars}} \StrongRandPoly(\vec{\alpha},\vec{\beta}) = A(\gamma_{2}(\vec{\alpha}))$ for all $\vec{\alpha} \in \SCSubset^{m_{2}}$. It then generates an oracle $\Proof_{0}$ for the $\SetCardinality{\SCSubset}^{\SSCVars}$-strong zero knowledge sumcheck protocol (\secref{sec:strong-zk-sumcheck}) on input $(\Field, m_{1}+3m_{2}, \deg(f), \SCSubset, 0)$ and oracles $\Proof_{1},\Proof_{2},\Proof_{3}$ for the \cite{BenSassonCFGRS16} zero knowledge sumcheck protocol on input $(\Field, \SSCVars, 2\SetCardinality{\SCSubset}, \SCSubset, \cdot)$. (Recall that in both zero knowledge sumchecks, the oracle message does not depend on the claim itself.) The prover sends an oracle which is the concatenation of the evaluation of $\StrongRandPoly$ with $(\Proof_{0}, \Proof_{1}, \Proof_{2}, \Proof_{3})$.
		
		\item The verifier chooses $\vec{x}, \vec{y} \in \Field^{r+3s}$ uniformly at random and sends them to the prover. The prover and verifier engage in the zero knowledge sumcheck protocol of \secref{sec:strong-zk-sumcheck} on the claim ``$F(\vec{x}, \vec{y}) = 0$'' with $\SoundnessSet = \Field \setminus \SCSubset$ using $\Proof_{1}$ as the oracle message. This reduces the claim to checking that $f(\vec{x},\vec{y},\vec{c},\vec{c}'_{1},\vec{c}'_{2},\vec{c}'_{3}) = a$ for uniformly random $\vec{c} \in (\Field \setminus \SCSubset)^{m_{1}}$, $\vec{c}'_{1}, \vec{c}'_{2}, \vec{c}'_{3} \in (\Field \setminus \SCSubset)^{m_{2}}$ and some $a \in \Field$ provided by the prover.
		
		\item The prover provides $h_{i} \DefineEqual A(\gamma_{2}(\vec{c}'_{i}))$ for each $i \in \{1,2,3\}$. The verifier substitutes these values into the expression for $f$ to check the above claims, and rejects if they do not hold.
		
		\item The prover and verifier engage in the zero knowledge sumcheck protocol of \cite{BenSassonCFGRS16} on the claims ``$\sum_{\vec{\beta} \in \SCSubset^{\SSCVars}} \StrongRandPoly(\vec{\alpha},\vec{\beta}) = h_{i}$'' for each $i \in \{1,2,3\}$, using $\Proof_{i}$ as the oracle message.
		
		\item The verifier checks that $\StrongRandPoly$ is low-degree (with proximity parameter $\varrho \DefineEqual \frac{1}{8}$ and soundness error $\SoundnessError \DefineEqual \frac{1}{\SetCardinality{\Field}}$), and uses self-correction (with soundness error $\SoundnessError$) to query it at the points required by the \cite{BenSassonCFGRS16} protocol above.
	\end{enumerate}
	
	\parhead{Completeness}
	If $((r, s, B), A) \in \OracleSATRelation$ then $F(\vec{\VariableX}, \vec{\VariableY})$ is the zero polynomial; hence $F(\vec{x}, \vec{y}) = 0$ for all $\vec{x}, \vec{y} \in \Field^{r + 3s}$. Completeness follows from the completeness of the zero knowledge sumcheck protocols.
	
	\parhead{Soundness}
	Suppose that $(r, s, B) \notin \Language(\OracleSATRelation)$ and let $(\Malicious{\StrongRandPoly}, \Malicious{\Proof}_{0}, \Malicious{\Proof}_{1}, \Malicious{\Proof}_{2}, \Malicious{\Proof}_{3})$ be the oracle message. If $\Malicious{\StrongRandPoly}$ is $\varrho$-far from an evaluation of a polynomial in $\PolynomialRingIndOneXY{\Field}{m_{2}}{\VariableX}{\SSCVars}{\VariableY}{\SetCardinality{\SCSubset}+2}{\SubsetSize}$ then the verifier rejects with probability at least $1 - \SoundnessError$. Otherwise, there exists a unique polynomial $\StrongRandPoly \in \PolynomialRingIndOneXY{\Field}{m_{2}}{\VariableX}{\SSCVars}{\VariableY}{\SetCardinality{\SCSubset}+2}{\SubsetSize}$ whose evaluation is $\varrho$-close to $\Malicious{\StrongRandPoly}$. Let $\Malicious{A} \DefineEqual \sum_{\vec{\beta} \in \SCSubset^{\SSCVars}} \StrongRandPoly(\vec{\VariableX},\vec{\beta})$, which we think of as playing the role of $\LD{A}(\gamma_{2}(\cdot))$ in $F$.
	
	If $(r, s, B) \notin \Language(\OracleSATRelation)$ then there is no choice of $\LD{A}$ such that $F(\vec{\VariableX}, \vec{\VariableY})$ is the zero polynomial. Thus, $F(\vec{x}, \vec{y}) = 0$ with probability at most $(r + 3s)/\SetCardinality{\Field}$ over the choice of $\vec{x}, \vec{y}$. By the soundness of the zero knowledge sumcheck protocol (\thmref{thm:strong-sumcheck-ipcp}), the verifier outputs a false claim ``$f(\vec{x},\vec{y},\vec{\alpha}) = a$'' with probability at least $1 - O((m_{1}+m_{2}+\SSCVars)\SetCardinality{\SCSubset})/(\SetCardinality{\Field} - \SetCardinality{\SCSubset}))$. If substituting $h_{i}$ for $\LD{A}(\gamma_{2}(\vec{c}'_{i}))$ in $f$ does not yield $a$, then the verifier rejects. Otherwise, it must be the case that for at least one $i \in \{1,2,3\}$, $\Malicious{A}(\vec{c}'_{i}) \neq h_{i}$. By the soundness of the \cite{BenSassonCFGRS16} sumcheck protocol, the verifier rejects with probability at least $1 - O(\frac{\SSCVars \SetCardinality{\SCSubset}}{\SetCardinality{\Field}})$. Taking a union bound, the verifier rejects with probability at least $1 - O((m_{1}+m_{2}+\SSCVars)\SetCardinality{\SCSubset}/\SetCardinality{\Field}) = 1 - O((r+s+\log \SCStrength) \SetCardinality{\SCSubset}/\SetCardinality{\Field})$.
	
	\parhead{Zero knowledge}
	Perfect zero knowledge for this protocol is witnessed by the following simulator.
\begin{mdframed}
	{\small
		
		\begin{enumerate}[nolistsep]
			
			\item Draw a uniformly random polynomial $\Simulated{\StrongRandPoly} \in \PolynomialRingIndOneXY{\Field}{m_{2}}{\VariableX}{\SSCVars}{\VariableY}{\SetCardinality{\SCSubset}+2}{2\SetCardinality{\SCSubset}}$. 
			
			\item Run the $\SetCardinality{\SCSubset}^{k}$-strong ZK sumcheck simulator on input $(\Field, m_{1}+3m_{2}, \deg(f), \SCSubset, 0)$, and use it to answer queries to $\Proof_{0}$ throughout. In parallel, run three copies of the simulator for the \cite{BenSassonCFGRS16} sumcheck on input $(\Field, k, 2\SetCardinality{\SCSubset}, \SCSubset, \cdot)$, and use them to answer queries to $\Proof_{1}, \Proof_{2}, \Proof_{3}$ respectively. Recall that the behavior of each simulator does not depend on the claim being proven until after the first simulated message, so we can choose these later.
				
			\item Receive $\vec{x}, \vec{y} \in \Field^{r+3s}$ from $\Malicious{\Verifier}$. 
			
			\item Simulate the strong ZK sumcheck protocol on the claim ``$F(\vec{x},\vec{y}) = 0$''. The subsimulator will query $f$ at a single location $\vec{c} \in (\Field - \SCSubset)^{r+3s}$. Reply with the value $f(\vec{x},\vec{y},\vec{c})$, for $\vec{c} = (\vec{c}_{0}, \vec{c}_{1}, \vec{c}_{2}, \vec{c}_{3}) \in (\Field \setminus \SCSubset)^{r+3s}$. To compute this requires values $\LD{A}(\LD{\gamma}(\vec{c}_{i}))$ for $i \in \{1,2,3\}$; we substitute each of these with $\Simulated{h}^{i} \in \Field$ drawn uniformly at random (except: if $\vec{c}_{i} = \vec{c}_{j}$ for $i \neq j$ then fix $\Simulated{h}^{i} = \Simulated{h}^{j}$).
			
			\item For $i \in \{1,2,3\}$, simulate the \cite{BenSassonCFGRS16} sumcheck protocol on the claim ``$\sum_{\vec{\beta} \in \SCSubset^{\SSCVars}} \StrongRandPoly(\vec{\alpha},\vec{\beta}) = \Simulated{h}^{i}$''. Whenever the subsimulator queries $\StrongRandPoly$, answer using $\Simulated{\StrongRandPoly}$.
			
		\end{enumerate}
	}%
\end{mdframed}

The verifier's view consists of its interaction with $\Prover$ during the four sumchecks, and its queries to the oracle. The \secref{sec:strong-zk-sumcheck} zero knowledge sumcheck subsimulator guarantees that the queries to $\Proof_{0}$ and the first sumcheck are perfectly simulated given a single query to $f$ at the point $\vec{c} \in (\Field \setminus \SCSubset)^{r+3s}$ chosen by $\Malicious{\Verifier}$. Since $\LD{A}'(\vec{\VariableX}) = \sum_{\vec{\beta} \in \SCSubset^{\SSCVars}} \StrongRandPoly(\vec{\VariableX},\vec{\beta}) \in \PolynomialRingIndOne{\Field}{\SCVars}{\VariableX}{\SetCardinality{\SCSubset}+2}$, the evaluation of $\LD{A}$ at any $3$ points outside of $\SCSubset^{\SCVars}$ does not determine its value at any point in $\SCSubset^{\SCVars}$. In particular, this means that the values $h_{i}$ sent by the prover in the original protocol are independently uniformly random in $\Field$ (except if $\vec{c}_{i} = \vec{c}_{j}$ for $i \neq j$ as above). Thus the $\Simulated{h}^{i}$ are identically distributed to the $h_{i}$, and therefore both the prover message and the simulator's query are perfectly simulated.

The \cite{BenSassonCFGRS16} sumcheck simulator ensures that the view of the verifier in the rest of the sumchecks is perfectly simulated given $q_{\Malicious{\Verifier}}$ queries to $\StrongRandPoly$, where $q_{\Malicious{\Verifier}}$ is the number of queries the verifier makes across all $\Proof_{i}$, $i \in \{1,2,3\}$. Hence the number of `queries' the simulator makes to $\Simulated{\StrongRandPoly}$ is strictly less than $\SCStrength$ (because $\Malicious{\Verifier}$ is $\SCStrength$-query). By \corref{cor:partial-sum-indep-vars}, any set of strictly less than $\SCStrength$ queries to $\StrongRandPoly$ is independent of $\LD{A}'$, and so the answers are identically distributed to the answers to those queries if they were made to a uniformly random polynomial, which is the distribution of $\Simulated{\StrongRandPoly}$.

Clearly drawing a uniformly random polynomial in $\Simulated{\StrongRandPoly} \in \PolynomialRingIndOneXY{\Field}{m_{2}}{\VariableX}{\SSCVars}{\VariableY}{\SetCardinality{\SCSubset}+2}{2\SetCardinality{\SCSubset}}$ is not something we can do in polynomial time. However, we can instead use the algorithm of \corref{cor:efficient-poly-simulator} to draw $\StrongRandPoly$ (a simple modification allows us to handle different degrees in $\vec{\VariableX}, \vec{\VariableY}$, or we could simply set the degree bound for both to be $2\SetCardinality{\SCSubset}$; the proof still goes through). The running time of the simulator is then $\poly(m_{1}, m_{2}, k, \SetCardinality{\SCSubset}, \log \SetCardinality{\Field})$.

It remains to choose $\Field$ and $\SCSubset$. We set $\SetCardinality{\SCSubset} = \poly(r + s + \log b)$ and $\SetCardinality{\Field} = \poly(\SetCardinality{\SCSubset})$ large enough that the soundness error is $o(1)$. The running time of the verifier is then $\poly(\BitSize{B}, \log \SCStrength)$, as is the running time of the simulator. The proof length is $\Field^{O(m_{1}+m_{2}+k)} = 2^{O(r+s)} \cdot \poly(\SCStrength)$.
\end{proof}

\doclearpage
\section{Delegating sum-product computations}
\label{sec:sum-product-protocols}

We define \emph{sum-product circuits}, a type of computation involving alternations of
\begin{inparaenum}[(i)]
  \item summing polynomials over hypercubes, and
  \item combining polynomials via low-degree arithmetic circuits.
\end{inparaenum}
Computing the output of a sum-product circuit is (conjecturally) hard (indeed, we will show that it is $\PSPACE$-complete), but we show how to efficiently delegate such computations via an Interactive Proof.

We proceed in three steps. First, we provide intuition for why it is natural to consider sum-product alternations (\secref{sec:sum-product-intuition}). Then, we define sum-product \emph{formulas}, which are a special case (in a way that is analogous to how boolean formulas specialize boolean circuits) and show how to delegate their evaluation (\secref{sec:sum-product-formulas}). Finally, we define sum-product \emph{circuits} and show how to delegate their evaluation (\secref{sec:sum-product-circuits}).

In later sections, we additionally achieve zero knowledge via an Interactive PCP (\secref{sec:pzk-for-sum-products}), and explain how to `program' sum-product circuits so that:
their evaluation captures $\PSPACE$ (\secref{sec:zk-pspace}) or, more generally, low-depth circuit computations (\secref{sec:zk-gkr}); and
their satisfaction captures $\NEXP$ (\secref{sec:zk-nexp}).

\subsection{Intuition for definition}
\label{sec:sum-product-intuition}

We provide intuition for why it is natural to consider sum-product alternations. Let $\Field$ be a finite field, $H$ a subset of $\Field$, and $m$ a positive integer.%
\footnote{Throughout, we assume that $\SetCardinality{\Field}$ is at least a constant fraction larger than $\SetCardinality{H}$. In fact, we will typically take $\SetCardinality{\Field}$ to be larger than this (e.g., at least polynomial in $\SetCardinality{H}$) to achieve good soundness.}
The sumcheck protocol (\secref{sec:sumcheck-protocol}) supports checking claims of the form ``$a = \sum_{\vec{\beta} \in \SPSubset^{\SPVars}} P(\vec{\beta})$'' for a given field element $a \in \Field$ and low-degree $\SPVars$-variate polynomial $P$ over $\Field$, if the verifier can efficiently evaluate $P$ at any point (e.g., the verifier has a small arithmetic circuit for $P$, or the verifier has oracle access to $P$, or others).
\begin{center}
Can the verifier still check the claim even if $P$ is an expression involving other polynomials?
\end{center}
Suppose that $P(\vec{\VariableX})$ (allegedly) equals $C(\vec{\VariableX},P_{1}(\vec{\VariableX}),\dots,P_{t}(\vec{\VariableX}))$ for some low-degree $t$-variate `combiner' polynomial $C$ and low-degree $m$-variate polynomials $P_{1},\dots,P_{t}$, and suppose that the verifier has small arithmetic circuits for all these polynomials. In this case the verifier can still efficiently evaluate $P$ at any given point, and the sumcheck protocol directly applies. However, now suppose instead that each polynomial $P_{i}(\vec{\VariableX})$ \emph{itself} (allegedly) equals $\sum_{\vec{\gamma} \in \SPSubset^{\SPVars}} C(\vec{\VariableX},P_{i,1}(\vec{\VariableX},\vec{\gamma}),\dots,P_{i,t}(\vec{\VariableX},\vec{\gamma}))$ for some low-degree $2m$-variate polynomials $P_{i,1},\dots,P_{i,t}$. Now the sumcheck protocol \emph{does not} directly apply, due to the \emph{alternation} of sums and products. What to do?

\parhead{Sum-product expressions, and protocols for them}
More generally (and informally), we call $P \colon \Field^{m} \to \Field$ an $m$-variate \emph{sum-product expression} if
\begin{inparaenum}[(i)]
  \item $P$ is a low-degree (individual degree less than $\SetCardinality{H}$) arithmetic circuit, or
  \item $P(\vec{\VariableX})$ equals $\sum_{\vec{\beta} \in H^{m}} C(\vec{\VariableX}, \vec{\beta}, P_{1}(\vec{\VariableX},\vec{\beta}),\dots,P_{t}(\vec{\VariableX},\vec{\beta}))$ where $C$ is a low-degree `combiner' arithmetic circuit and $P_{1},\dots,P_{t}$ are $2m$-variate sum-product expressions.
\end{inparaenum}

By building on ideas of \cite{Shamir92,Shen92,GoldwasserKR15}, we can \emph{still} use the sumcheck protocol, now as a subroutine of a larger Interactive Proof, to verify claims of the form ``$a = P(\vec{\omega})$'' for a given $a \in \Field$, sum-product expression $P$, and $\vec{\omega} \in \Field^{m}$, as we now sketch --- and thereby handle sum-product alternations.

If $P$ is an arithmetic circuit, then the verifier can check the claim directly by evaluating $P$ at $\vec{\omega}$. Otherwise, proceed as follows. Define $\LD{P}$ to be the low-degree extension of $P$ (see \secref{sec:basic-notations}):
\begin{equation*}
\LD{P}(\VariableX)
= \sum_{\vec{\alpha} \in H^{m}}
    \Lagrange{\SPSubset^{n}}(\vec{\VariableX}, \vec{\alpha})
    \sum_{\vec{\beta} \in H^{m}}
      C(\vec{\alpha},\vec{\beta},P_{1}(\vec{\alpha},\vec{\beta}),\dots,P_{t}(\vec{\alpha},\vec{\beta}))
\enspace.
\end{equation*}
Recall that $\Lagrange{H^{m}}(\vec{\VariableX}, \vec{\VariableY})$ is the unique $m$-variate polynomial, of degree less than $\SetCardinality{H}$, such that, for all $(\vec{\alpha},\vec{\beta}) \in H^{m} \times H^{m}$, $\Lagrange{H^{m}}(\vec{\alpha},\vec{\beta})$ equals $1$ when $\vec{\alpha}=\vec{\beta}$ and equals $0$ otherwise.

The prover and verifier run the sumcheck protocol on the claim ``$a = \LD{P}(\vec{\omega})$'' and obtain a new claim
\begin{equation*}
\text{``}
a' =
\Lagrange{\SPSubset^{n}}(\vec{\omega}, \vec{r}_{2})
\cdot C(\vec{r}_{1},\vec{r}_{2}, P_{1}(\vec{r}_{1},\vec{r}_{2}),\dots,P_{t}(\vec{r}_{1},\vec{r}_{2}))
\text{''}
\end{equation*}
for some $a' \in \Field$ derived from the prover's messages and $\vec{r}_{1},\vec{r}_{2} \in \Field^{m}$ drawn uniformly at random by the verifier. The prover then sends $h_{1},\dots,h_{t}$ and the verifier checks that $a' = \Lagrange{\SPSubset^{n}}(\vec{\omega}, \vec{r}_{2}) \cdot C(h_{1},\dots,h_{t})$. (Note that this expression involves only low-degree polynomials.) The verifier then recursively checks, for $i=1,\dots,t$, that ``$a_{i} = P_{i}(\vec{r}_{1},\vec{r}_{2})$'', relying on the fact that each $P_{i}$ is itself a sum-product expression.

The reason for taking the low-degree extension $\LD{P}$ of $P$ is to prevent a degree blowup for intermediate claims, and is also used in the GKR protocol \cite{GoldwasserKR15} as well as Shen's protocol \cite{Shen92} (known as \emph{degree reduction} there). In particular, depending on the form of the combiner $C$, the degree of $P(\vec{X})$ can be somewhat larger than that of the $P_{i}$ subexpressions. Even a factor $2$ increase in the degree would, after $k$ rounds, lead to a factor $2^{k}$ increase overall, which for modest $k$ would make the communication complexity of the sumcheck protocol superpolynomial. The degree reduction step ensures that the degrees of the intermediate claims do not increase.

\parhead{Towards sum-product formulas}
The above informal discussion motivates the formulation of tree-like computations that combine values of previous hypercube sums by way of functions of bounded degree --- we call these \emph{sum-product formulas} (in analogy to boolean formulas that are also tree-like computations). Our definition also features crucial degrees of freedom, which make `programming' these formulas more efficient, that we now discuss.

First, we allow each internal vertex $v$ in the tree to be labeled with a potentially different combiner arithmetic circuit $\SPPoly[v]$ of small (total) degree. An input $\SPInput$ to the sum-product formula then consists of labeling each leaf vertex $v$ with a polynomial $\SPLeaf[v]$, potentially represented as an arithmetic circuit, of small (individual) degree, and the edges in the tree determine how to `evaluate' a vertex, as follows. The \emph{value} of a vertex $v$ on input $\SPInput$ equals the circuit $\SPLeaf[v]$ if $v$ is a leaf vertex, or equals $\sum_{\vec{\beta} \in \SPSubset^{\SPVars}}
\SPPoly[v]
\big(
\SPValueL{\SPInput}{u_{1}}
(\vec{\VariableX},\vec{\beta})
\dots,
\SPValueL{\SPInput}{u_{\aodeg}}
(\vec{\VariableX},\vec{\beta})
\big)$
if $v$ is an internal vertex, where $u_{1}, \ldots, u_{\aodeg}$ are the children of $v$. The value of the formula on $\SPInput$ is the value of the root on $\SPInput$.

Second, we allow flexible `arity' in the sums: each edge $e$ is labeled with finite sets of positive integers $\SPFreeProjection[e]$ and $\SPSumProjection[e]$ that determine which `free variables' and `summation variables' are passed on to the child corresponding to $e$. In other words, now the recursion looks like
$
\sum_{\vec{\beta} \in \SPSubset^{\SPVars[v]}}
\SPPoly[v]
\big(
\SPValueL{\SPInput}{u_{1}}
(\Restrict{\vec{\VariableX}}{\SPFreeProjection[e_{1}]},
\Restrict{\vec{\beta}}{\SPSumProjection[e_{1}]}),
\dots,
\SPValueL{\SPInput}{u_{\aodeg}}
(\Restrict{\vec{\VariableX}}{\SPFreeProjection[e_{\aodeg}]},
\Restrict{\vec{\beta}}{\SPSumProjection[e_{\aodeg}]})
\big)$
where $e_{1} = (v, u_{1}), \dots, e_{\aodeg} = (v, u_{\aodeg})$ and $\SPVars[v] \DefineEqual \max (\SPSumProjection[e_{1}] \cup \cdots \cup \SPSumProjection[e_{\TreeDegree[\Tree]{v}}])$.

In \secref{sec:sum-product-formulas} we provide the formal definition of sum-product formulas, and also describe how to outsource computations about them. We provide this only as a simpler stepping stone towards the next definition.

\parhead{Sum-product circuits: re-using sub-computations}
A boolean formula is limited in that it cannot re-use sub-computations; a sum-product formula is similarly limited. Thus, in analogy to boolean circuits, we consider \emph{sum-product circuits}, in which sub-computations can be re-used according to an underlying directed acyclic graph (in fact, we will need a multi-graph), rather than a tree. Of course, one can always reduce sum-product circuits to sum-product formulas by `opening up' the graph into a tree --- but in the worst case this incurs an exponential blowup in the resulting tree. We use a standard trick to modify the protocol for sum-product formulas so to support merging multiple sub-claim computations at a vertex into one claim (regardless of the in-degree of the vertex), which avoids this explosion. This is also a necessary step in GKR's protocol \cite{GoldwasserKR15} (though we implement it differently for compatibility with our zero-knowledge protocols).

In \secref{sec:sum-product-circuits} we provide a formal definition of sum-product circuits, and then describe how to extend the ideas discussed so far to also support outsourcing computations about sum-product circuits. In the rest of the paper we only use (and must use) sum-product circuits, as sum-product formulas are not expressive enough for our purposes.

\subsection{Sum-product formulas}
\label{sec:sum-product-formulas}

The purpose of this section is to
\begin{inparaenum}[(i)]
  \item introduce \emph{sum-product formulas}, and
  \item give proof systems for two computational problems about these, \emph{evaluation} and \emph{satisfaction}.
\end{inparaenum}

\subsubsection{Formal definition}
\label{sec:sum-product-formulas-definition}

As with a boolean formula, the `topology' of a a sum-product formula is a \emph{tree}. A (rooted) tree $\Tree=(\VertexSet,\EdgeSet)$ is an acyclic connected graph in which edges are directed away from a distinguished vertex, known as the root of $\Tree$ and denoted $\Root[\Tree]$; the sinks of the graph are known as the leaves of $\Tree$ while all other vertices are known as internal vertices. The \emph{depth} of a vertex $v$, denoted $\TreeDepth[\Tree]{v}$, is the number of edges on the path from $\Root[\Tree]$ to $v$ (thus $\Root[\Tree]$  has depth $0$). The depth of $\Tree$, denoted $\TreeDepth[]{\Tree}$, is the maximum depth of any vertex $v$ in $\VertexSet$. The width of $\Tree$, denoted $\Width{\Tree}$, is the maximum number of vertices at any depth: $\max_{i=1}^{\TreeDepth[]{\Tree}} \SetCardinality{\{ v \in \VertexSet : \TreeDepth[\Tree]{v} = i \}}$. The \emph{out-degree} of a vertex $v$ is denoted $\TreeDegree[\Tree]{v}$ and equals the number of children of $v$; the \emph{in-degree} is $1$ for all vertices except the root.

As outlined in \secref{sec:sum-product-intuition}, we eventually consider trees in which each internal vertex specifies a function that is recursively defined in terms of its children's functions. The number of inputs to these functions varies from vertex to vertex, and we specify the \emph{arity} of these functions via certain edge labels. Namely, each edge $e = (u, v)$ is labeled by two ``projections'' $\SPFreeProjection[e]$ and $\SPSumProjection[e]$ that, respectively, specify which free variables of $u$ (the $\vec{\VariableX}$ part in \eqnref{eqn:vertex-value-simplified} below) and summation variables of $u$ (the $\vec{\beta}$ part in \eqnref{eqn:vertex-value-simplified} below) are passed on to $v$. In order for these projections to yield a well-defined notion of arity, they must satisfy certain consistency properties, and this motivates the following definition.

\begin{definition}
\label{def:ari-tree}
A tuple $\Tree = (\VertexSet, \EdgeSet, \SPFreeProjection, \SPSumProjection)$ is an \defemph{ari-tree} if $(\VertexSet, \EdgeSet)$ is a tree and both $\SPFreeProjection$ and $\SPSumProjection$ label every edge $e$ in $\EdgeSet$ with finite sets of positive integers $\SPFreeProjection[e]$ and $\SPSumProjection[e]$ that satisfy the following property. For every vertex $v$ in $\VertexSet$, there exists a (unique) non-negative integer $\SPArity{v}$ such that:
\begin{inparaenum}[(1)]
  \item if $v$ is the root then $\SPArity{v} = 0$, otherwise $\SPArity{v} = \SetCardinality{\SPFreeProjection[e]} + \SetCardinality{\SPSumProjection[e]}$ where $e$ is $v$'s (unique) incoming edge;
  \item $\SPFreeProjection[e_{1}], \dots, \SPFreeProjection[e_{\aodeg}] \subseteq \{1, \dots, \SPArity{v}\}$, where $e_{1}, \dots, e_{\aodeg}$ are $v$'s outgoing edges.
\end{inparaenum}
\end{definition}

For convenience, we denote by $\SPMaxArity(\Tree)$ the maximum of $\SPArity{v}$ across all vertices $v$ in the vertex set $\VertexSet$ of $\Tree$. Moreover, for every vertex $v$, we define $\SPVars[v] \DefineEqual \max (\SPSumProjection[e_{1}] \cup \cdots \cup \SPSumProjection[e_{\aodeg}])$ so that $\SPSumProjection[e_{1}],\dots,\SPSumProjection[e_{\aodeg}] \subseteq \{1, \dots, \SPVars[v]\}$.

We are now ready to define a sum-product formula $\SPFormula$, an input $\SPInput$ for $\SPFormula$, and how to evaluate $\SPFormula$ on $\SPInput$.

\begin{definition}
\label{def:sum-product-formula}
A \defemph{sum-product formula} is a tuple $\SPFormula = \SPFormulaTuple$ where:
$\Field$ is a finite field,
$\SPSubset$ is a subset of $\Field$ (represented as a list of field elements), $\InternalVertexSetDegree,\SPLeafDegree$ are positive integers (represented in unary) with $\SPLeafDegree \geq \SetCardinality{\SPSubset}$,
$\Tree = (\VertexSet,\EdgeSet,\SPFreeProjection,\SPSumProjection)$ is an ari-tree, and
$\SPPoly{}$ labels each internal vertex $v$ of $\Tree$ with an arithmetic circuit $\SPPoly[v](\vec{\VariableX},\vec{\VariableY},\vec{\VariableZ}) \colon \Field^{\SPArity{v}} \times  \Field^{\SPVars[v]} \times \Field^{\TreeDegree[\Tree]{v}} \to \Field$ of \underline{total} degree at most $\InternalVertexSetDegree$.
An \defemph{input} $\SPInput$ for $\SPFormula$ labels each leaf vertex $v$ of $\Tree$ with a polynomial $\SPLeaf[v] \colon \Field^{\SPArity{v}} \to \Field$ of \underline{individual} degree at most $\SPLeafDegree$.
The \defemph{value} of $\SPFormula$ on an input $\SPInput$ is denoted $\SPValueL{\SPInput}{\SPFormula}$ and equals $\SPValueL{\SPInput}{\Root[\Tree]}$, which we define below.

\begin{mdframed}
The value of a vertex $v$ of $\Tree$ on an input $\SPInput$ is denoted $\SPValueL{\SPInput}{v}$ and is recursively defined as follows. If $v$ is a leaf vertex, then $\SPValueL{\SPInput}{v}$ equals the polynomial $\SPLeaf[v]$. If $v$ is an internal vertex, then $\SPValueL{\SPInput}{v}$ is the $\SPArity{v}$-variate polynomial over $\Field$ defined by the following expression:
\begin{equation}
\label{eqn:vertex-value-simplified}
\sum_{\vec{\beta} \in \SPSubset^{\SPVars[v]}}
\SPPoly[v]
\big(
\vec{\VariableX},
\vec{\beta},
\SPValueL{\SPInput}{u_{1}}
(\Restrict{\vec{\VariableX}}{\SPFreeProjection[e_{1}]},
\Restrict{\vec{\beta}}{\SPSumProjection[e_{1}]}),
\dots,
\SPValueL{\SPInput}{u_{\aodeg}}
(\Restrict{\vec{\VariableX}}{\SPFreeProjection[e_{\aodeg}]},
\Restrict{\vec{\beta}}{\SPSumProjection[e_{\aodeg}]})
\big)
\enspace,
\end{equation}
where $\aodeg \DefineEqual \TreeDegree[\Tree]{v}$, and $e_{1} = (v, u_{1}), \dots, e_{\aodeg} = (v, u_{\aodeg})$ are the outgoing edges of $v$. In particular, $\SPValueL{\SPInput}{\SPFormula} = \SPValueL{\SPInput}{\Root[\Tree]}$ is a constant in $\Field$.
\end{mdframed}
\end{definition}

Given a sum-product formula we can ask two types of computational problems:
  (\emph{evaluation}) does a given input lead to a given output?
  (\emph{satisfaction}) does there exist an input that leads to a given output?
We now define each of these.

\begin{definition}[SPFE problem]
\label{def:formula-evaluation}
The \defemph{sum-product formula evaluation problem} is the following: given a sum-product formula $\SPFormula$, value $\SPOutput$, and input $\SPInput$ (given as a mapping from each leaf vertex $v$ of $\SPFormula$'s ari-tree $\Tree$ to an arithmetic circuit computing the polynomial $\SPLeaf[v]$), determine if $\SPValueL{\SPInput}{\SPFormula}=\SPOutput$. This problem induces the language
\begin{equation*}
\SPFELanguage
\DefineEqual
\big\{
(\SPFormula, \SPOutput, \SPInput)
\text{ s.t. }
\SPValueL{\SPInput}{\SPFormula} = \SPOutput
\big\}
\enspace.
\end{equation*}
(When $\SPInput$ is given as above, $\MaxSpace{\SPInput}$ denotes the maximum space required to evaluate any circuit in $\SPInput$.)
\end{definition}

\begin{definition}[SPFS problem]
\label{def:formula-satisfaction}
The \defemph{sum-product formula satisfaction problem} is the following: given a sum-product formula $\SPFormula$, partial mapping of leaf vertices to arithmetic circuits $\SPInput$, and value $\SPOutput$, determine if there exists a mapping $\SPAuxInput$ from the leaf vertices not in the domain of $\SPInput$ to polynomials s.t.\ $\SPValueL{\SPInput ,\SPAuxInput}{\SPFormula}=\SPOutput$. This problem induces the relation
\begin{equation*}
\SPFSRelation
\DefineEqual
\big\{
\big((\SPFormula,\SPOutput,\SPInput),\SPAuxInput\big)
\text{ s.t. }
\SPValueL{\SPInput,\SPAuxInput}{\SPFormula} = \SPOutput
\big\}
\enspace.
\end{equation*}
(We refer to $\SPInput$ as the explicit input and $\SPAuxInput$ as the auxiliary input.)
\end{definition}

\subsubsection{Delegating sum-product formula evaluation problems}
\label{sec:sum-product-formulas-evaluation}

We give an Interactive Proof to delegate sum-product formula \emph{evaluation} problems.

\begin{theorem}[IP for SPFE]
\label{thm:ip-for-SPFE}
There exists a public-coin Interactive Proof for the language $\SPFELanguage$. In more detail:
\begin{equation*}
\SPFELanguage \in
\mathbf{AM}
\left[
\begin{tabular}{rl}
soundness error: & $O(\InternalVertexSetDegree\SPLeafDegree \cdot \SPMaxArity(\Tree) \cdot \SetCardinality{\VertexSet(\Tree)}/\SetCardinality{\Field})$ \\
round complexity: & $O( \TreeDepth[]{\Tree} \cdot \SPMaxArity(\Tree))$ \\[1mm]
prover time: & $\poly(\BitSize{\SPFormula}, \BitSize{\SPInput}, \SetCardinality{\SPSubset}^{\SPMaxArity(\Tree)})$ \\
verifier time: & $\poly(\BitSize{\SPFormula}) + O(\BitSize{\SPInput})$ \\[1mm]
verifier space: & $O(\SPMaxArity(\Tree) \cdot \Width{\Tree} \cdot \log \SetCardinality{\Field} + \log \BitSize{\SPFormula} + \MaxSpace{\SPInput})$
\end{tabular}
\right]
\enspace.
\end{equation*}
\end{theorem}

Before describing the Interactive Proof system, we define for every vertex $v$ in the ari-tree (of a sum-product formula) a function $\LDSPValueL{\SPInput}{v}$ based on the function $\SPValueL{\SPInput}{v}$, as follows. If $v$ is a leaf vertex, then $\LDSPValueL{\SPInput}{v}$ equals the polynomial $\SPLeaf[v]$. If instead $v$ is an internal vertex, then $\LDSPValueL{\SPInput}{v}$ is the low-degree extension of the evaluation of $\SPValueL{\SPInput}{v}$ on $\SPSubset^{\SPArity{v}}$:
\begin{equation}
\label{eq:lde-of-v}
\LDSPValueL{\SPInput}{v}(\vec{\VariableX})
\DefineEqual
\sum_{\vec{\alpha} \in \SPSubset^{\SPArity{v}}}
  \Lagrange{\SPSubset^{\SPArity{v}}}(\vec{\VariableX}, \vec{\alpha})
  \sum_{\vec{\beta} \in \SPSubset^{\SPVars[v]}}
    \SPPoly[v]\big(
	  \vec{\alpha}, \vec{\beta},
      \SPValueL{\SPInput}{u_{1}}(\Restrict{\vec{\alpha}}{\SPFreeProjection[e_{1}]},
       \Restrict{\vec{\beta}}{\SPSumProjection[e_{1}]}),
      \dots,
      \SPValueL{\SPInput}{u_{\aodeg}}(\Restrict{\vec{\alpha}}{\SPFreeProjection[e_{\aodeg}]},
      \Restrict{\vec{\beta}}{\SPSumProjection[e_{\aodeg}]})
    \big) 
\enspace.
\end{equation}
Since $\LDSPValueL{\SPInput}{v}$ agrees with $\SPValueL{\SPInput}{v}$ on $\SPSubset^{\SPArity{v}}$, we can equivalently define $\LDSPValueL{\SPInput}{v}$ in terms of the $\LDSPValueL{\SPInput}{u_{j}}$ rather than the $\SPValueL{\SPInput}{u_{j}}$:
\begin{align*}
\LDSPValueL{\SPInput}{v}(\vec{\VariableX})
\DefineEqual
& \sum_{\vec{\alpha} \in \SPSubset^{\SPArity{v}}}
\Lagrange{\SPSubset^{\SPArity{v}}}(\vec{\VariableX}, \vec{\alpha})
\sum_{\vec{\beta} \in \SPSubset^{\SPVars[v]}}
\SPPoly[v]\big(
\vec{\alpha}, \vec{\beta},
\LDSPValueL{\SPInput}{u_{1}}(\Restrict{\vec{\alpha}}{\SPFreeProjection[e_{1}]},
\Restrict{\vec{\beta}}{\SPSumProjection[e_{1}]}),
\dots,
\LDSPValueL{\SPInput}{u_{\aodeg}}(\Restrict{\vec{\alpha}}{\SPFreeProjection[e_{\aodeg}]},
\Restrict{\vec{\beta}}{\SPSumProjection[e_{\aodeg}]})
\big) \\
=&\sum_{\vec{\alpha} \in \SPSubset^{\SPArity{v}}}
\sum_{\vec{\beta} \in \SPSubset^{\SPVars[v]}}
\Lagrange{\SPSubset^{\SPArity{v}}}(\vec{\VariableX}, \vec{\alpha})
\cdot
\SPPoly[v]\big(
\vec{\alpha}, \vec{\beta},
\LDSPValueL{\SPInput}{u_{1}}(\Restrict{\vec{\alpha}}{\SPFreeProjection[e_{1}]},
\Restrict{\vec{\beta}}{\SPSumProjection[e_{1}]}),
\dots,
\LDSPValueL{\SPInput}{u_{\aodeg}}(\Restrict{\vec{\alpha}}{\SPFreeProjection[e_{\aodeg}]},
\Restrict{\vec{\beta}}{\SPSumProjection[e_{\aodeg}]})
\big)
\enspace.
\end{align*}
Note that the summand in the last line above is a polynomial, and its individual degree in $(\vec{\alpha},\vec{\beta})$ is at most $\SetCardinality{\SPSubset} + \InternalVertexSetDegree \cdot \max\{\SPLeafDegree ,\SetCardinality{\SPSubset}\} \leq 2 \InternalVertexSetDegree \SPLeafDegree$. Indeed, $\LDSPValueL{\SPInput}{u_{i}}(\Restrict{\vec{\alpha}}{\SPFreeProjection[e_{i}]},
\Restrict{\vec{\beta}}{\SPSumProjection[e_{i}]})$ has individual degree at most $\SPLeafDegree$ if $u_{i}$ is a leaf, and individual degree at most $\SetCardinality{\SPSubset}$ otherwise (as $\LDSPValueL{\SPInput}{v}$ has individual degree at most $\SetCardinality{\SPSubset}$ in $\vec{\VariableX}$); $\SPPoly[v]$ computes a polynomial of total degree at most $\InternalVertexSetDegree$; and $\Lagrange{\SPSubset^{\SPArity{v}}}(\vec{\VariableX}, \vec{\alpha})$ has individual degree at most $\SetCardinality{\SPSubset}$ in $\vec{\alpha}$. We use this degree bound below.

\begin{proof}
The prover and verifier receive a SPFE instance $(\SPFormula, \SPOutput, \SPInput)$ as input. They both associate, for each vertex $v$ of its ari-tree $\Tree$, a label $(\vec{\gamma}_{v}, a_{v})$ with $\vec{\gamma}_{v} \in \Field^{\SPArity{v}}$ and $a_{v} \in \Field$; for the root, this label equals $(\EmptyVector, \SPOutput)$, while for all other vertices this label is defined during the protocol. The prover and verifier then interact as follows.
\begin{enumerate}

  \item For every internal vertex $v$ of $\Tree$ taken in (any) topological order, letting $\aodeg \DefineEqual \TreeDegree[\Tree]{v}$:
  \begin{enumerate}[nolistsep]
    
    \item The prover and verifier invoke the sumcheck protocol \cite{LundFKN92,Shamir92} on the claim ``$\LDSPValueL{\SPInput}{v}(\vec{\gamma}_{v}) = a_{v}$''. By the end of this subprotocol, the verifier has chosen $\vec{c}_{1} \in \Field^{\SPArity{v}}$ and $\vec{c}_{2} \in \Field^{\SPVars[v]}$ uniformly at random, and has derived from the prover's messages a value $b \in \Field$ that allegedly satisfies the following equality:
\begin{equation}
\label{eqn:SPFE-claim}
b \DefineEqual \Lagrange{\SPSubset^{\SPArity{v}}}(\vec{\gamma}_{v}, \vec{c}_{1})
\cdot
\SPPoly[v]\big(
\vec{c}_{1}, \vec{c}_{2},
\LDSPValueL{\SPInput}{u_{1}}(\Restrict{\vec{c}_{1}}{\SPFreeProjection[e_{1}]},
\Restrict{\vec{c}_{2}}{\SPSumProjection[e_{1}]}),
\dots,
\LDSPValueL{\SPInput}{u_{\aodeg}}(\Restrict{\vec{c}_{1}}{\SPFreeProjection[e_{\aodeg}]},
\Restrict{\vec{c}_{2}}{\SPSumProjection[e_{\aodeg}]})
\big) \enspace,
\end{equation}
where $e_{1} = (v, u_{1}), \dots, e_{\aodeg} = (v, u_{\aodeg})$ are $v$'s outgoing edges.

    \item The prover sends $h_{1} \DefineEqual \LDSPValueL{\SPInput}{u_{1}}(\Restrict{\vec{c}_{1}}{\SPFreeProjection[e_{1}]},
    \Restrict{\vec{c}_{2}}{\SPSumProjection[e_{1}]}), \dots, h_{\aodeg} \DefineEqual \LDSPValueL{\SPInput}{u_{\aodeg}}(\Restrict{\vec{c}_{1}}{\SPFreeProjection[e_{\aodeg}]},
    \Restrict{\vec{c}_{2}}{\SPSumProjection[e_{\aodeg}]}) \in \Field$, and the verifier checks that
\begin{equation}
\label{eqn:SPFE-check}
b = \Lagrange{\SPSubset^{\SPArity{v}}}(\vec{\gamma}_{v}, \vec{c}_{1})
\cdot
\SPPoly[v]\big(
\vec{c}_{1}, \vec{c}_{2},
h_{1},
\dots,
h_{\aodeg}
\big) \enspace.
\end{equation}

   \item For $j = 1,\dots, \aodeg$, the verifier sets $(\vec{\gamma}_{u_{j}},a_{u_{j}}) \DefineEqual ((\Restrict{\vec{c}_{1}}{\SPFreeProjection[e_{j}]},
   \Restrict{\vec{c}_{2}}{\SPSumProjection[e_{j}]}), h_{j})$.

  \end{enumerate}

  \item For every leaf vertex $v$ of $\Tree$, the verifier checks that $\LDSPValueL{\SPInput}{v}(\vec{\gamma}_{v}) = a_{v}$, i.e., that $\SPLeaf[v](\vec{\gamma}_{v}) = a_{v}$.

\end{enumerate}
While the above description considers sequential invocations of the sumcheck protocol, these can be run in parallel in $\TreeDepth[]{\Tree}$ phases: first the root (which has depth $0$), then all vertices of depth $1$, then all vertices of depth $2$, and so on until all vertices of depth $\TreeDepth[]{\Tree}-1$. Each such phase requires $O(\SPMaxArity(\Tree))$ rounds, so that the number of rounds is now $O(\TreeDepth[]{\Tree} \cdot \SPMaxArity(\Tree))$, as claimed. The claimed running times for the prover and verifier follow immediately from the above description. The claimed space bound follows from the observation that in phase $i$, if we also check the leaf vertices at depth $i$ during this phase, then the verifier may discard $(\vec{\gamma}_{v}, a_{v})$ for all vertices $v$ with $\TreeDepth[\Tree]{v} < i-1$; and that the value of $\Lagrange{\SPSubset^{\SPArity{v}}}$ can be computed in space $O(\log(\SPArity{v}) + \log \SetCardinality{\Field})$. We are left to argue the claimed soundness error.

If for some internal vertex $v$ it holds that $\SPValueL{\SPInput}{v}(\gamma_{v}) \neq a_{v}$, then the soundness property of the sumcheck protocol implies that either the verifier rejects or \equnref{eqn:SPFE-claim} holds with probability at most $2\InternalVertexSetDegree \SPLeafDegree \cdot \SPArity{v}/\SetCardinality{\Field}$. In this latter case, either \equnref{eqn:SPFE-claim} fails to hold and the verifier rejects, or there exists $j \in \{1, \dots, \TreeDegree[\Tree]{v}\}$ such that $e_j \neq \LDSPValueL{\SPInput}{u_{j}}(\vec{c}_{1},\vec{c}_{2})$, which means that there exists a vertex $u$ in the next layer (in fact, $u = u_{j}$ suffices) for which $\LDSPValueL{\SPInput}{u}(\vec{\gamma}_{u}) \neq a_{u}$. If $u$ is a leaf vertex then the verifier will reject when considering $u$; otherwise we repeat the above argument. Taking a union bound over the internal vertices of $\Tree$ yields the claimed soundness error.
\end{proof}

\subsubsection{Delegating sum-product formula satisfaction problems}
\label{sec:sum-product-formulas-satisfaction}

We give an Interactive PCP to delegate sum-product formula \emph{satisfaction} problems, via a simple extension of the Interactive Proof for \emph{evaluation} problems in the previous section. Similarly to \cite{Shamir92,Shen92,GoldwasserKR15}, the verifier only needs to access the formula's input at a few locations, at the end of the protocol; thus the prover can simply send the input as a proof oracle, and the verifier can query it (via suitable low-degree testing and self-correction of polynomials).

\begin{theorem}[IPCP for SPFS]
\label{thm:ipcp-for-SPFS}
There exists a (public-coin and non-adaptive) Interactive PCP for the relation $\SPFSRelation$. In more detail:
\begin{equation*}
\SPFSRelation \in
\mathbf{IPCP}
\left[
\begin{tabular}{rl}
soundness error: & $O(\InternalVertexSetDegree\SPLeafDegree \cdot \SPMaxArity(\Tree) \cdot \SetCardinality{\VertexSet(\Tree)}/\SetCardinality{\Field})$ \\
round complexity: & $O( \TreeDepth[]{\Tree} \cdot \SPMaxArity(\Tree))$ \\[1mm]
proof length: & $O(\SetCardinality{\VertexSet(\Tree)} \cdot \SetCardinality{\Field}^{\SPMaxArity(\Tree)})$ \\
query complexity: & $\SetCardinality{\VertexSet(\Tree)} \cdot \poly(\log \SetCardinality{\Field}, \SPMaxArity(\Tree), \SPLeafDegree)$ \\[1mm]
prover time: & $\poly(\BitSize{\SPFormula}, \BitSize{\SPInput}, \BitSize{\SPAuxInput}, \SetCardinality{\SPSubset}^{\SPMaxArity(\Tree)})$ \\
verifier time: & $\poly(\BitSize{\SPFormula}, \BitSize{\SPInput})$ \\[1mm]
verifier space: & $O(\SPMaxArity(\Tree) \cdot \Width{\Tree} \cdot \log \SetCardinality{\Field} + \log \BitSize{\SPFormula} + \MaxSpace{\SPInput})$
\end{tabular}
\right]
\enspace.
\end{equation*}
\end{theorem}

\begin{proof}[Proof sketch]
The prover and verifier receive a SPFS instance $(\SPFormula, \SPOutput, \SPInput)$ as input, and the prover additionally receives an auxiliary input $\SPAuxInput$ for $\SPFormula$ that is a valid witness for $(\SPFormula, \SPOutput, \SPInput)$.
\begin{itemize}

  \item \textbf{Oracle.}
  The prover sends to the verifier the proof string $\Proof \DefineEqual \SPFormula$, where each polynomial $\SPAuxLeaf[v] \colon \Field^{\SPArity{v}} \to \Field$ is represented by its evaluation table over the whole domain.
  
  \item \textbf{Interaction.}
  The prover and verifier engage in an Interactive Proof for the claim ``$(\SPFormula, \SPOutput, (\SPInput, \SPAuxInput)) \in\SPFELanguage$'' using the protocol from the proof of \thmref{thm:ip-for-SPFE} above. The verifier must access $\SPAuxInput$ only at the end of the protocol, and at few locations: for each leaf vertex $v$ of $\Tree$ where $v$ is not in the domain of $\SPInput$, the verifier needs the value of $\SPAuxLeaf[v]$ at a single location $\vec{\gamma}_{v}$. Thus, the verifier tests that each $\SPAuxLeaf[v]$ is close to the evaluation of a polynomial of suitable degree \cite{GoldreichS06,GurR15}, and then uses self-correction to read each $\SPAuxLeaf[v](\vec{\gamma}_{v})$ \cite{RubinfeldS96,AroraS03}.
  
\end{itemize}
Setting parameters for low-degree testing and self-correction appropriately (for the case of individual-degree multi-variate polynomials) yields the parameters claimed in the theorem statement.
\end{proof}

\subsection{Sum-product circuits}
\label{sec:sum-product-circuits}

The purpose of this section is to
\begin{inparaenum}[(i)]
  \item introduce \emph{sum-product circuits}, and
  \item give proof systems for two computational problems about these, \emph{evaluation} and \emph{satisfaction}.
\end{inparaenum}

\subsubsection{Formal definition}
\label{sec:sum-product-circuits-definition}

As with a boolean circuit, the `topology' of a a sum-product circuit is a \emph{directed acyclic multi-graph}: a tuple $\Graph = (\VertexSet, \EdgeSet)$ where $\EdgeSet$ is a multi-set of directed edges in $\VertexSet \times \VertexSet$ with no directed cycles. We assume that there is a \emph{single} vertex $\Root[\Graph] \in \VertexSet$ with in-degree zero, known as the root. The vertices with out-degree zero are known as the leaves, while all other vertices are known as internal vertices. We also assume that, for every vertex $v$, all directed paths from the root $\Root[\Graph]$ to $v$ have the same length, which we denote $\GraphDepth[\Graph]{v}$. The depth of $\Graph$, denoted $\GraphDepth[]{\Graph}$, is the maximum depth of any vertex $v$ in $\VertexSet$. The width of $\Graph$, denoted $\Width{\Graph}$, is the maximum number of vertices at any depth: $\max_{i=1}^{\GraphDepth[]{\Graph}} \SetCardinality{\{ v \in \VertexSet : \GraphDepth[\Graph]{v} = i \}}$. The in-degree and out-degree of a vertex $v$ are denoted by $\GraphInDegree[\Graph]{v}$ and $\GraphOutDegree[\Graph]{v}$; we also define $\GraphMaxInDegree{\Graph} \DefineEqual \max_{v \in \VertexSet} \GraphInDegree[\Graph]{v}$.

\begin{definition}
\label{def:ari-graph}
A tuple $\Graph = (\VertexSet, \EdgeSet, \SPFreeProjection, \SPSumProjection)$ is an \defemph{ari-graph} if $(\VertexSet, \EdgeSet)$ is a directed acyclic multi-graph and both $\SPFreeProjection$ and $\SPSumProjection$ label every edge $e$ in $\EdgeSet$ with finite sets of positive integers $\SPFreeProjection[e]$ and $\SPSumProjection[e]$ that satisfy the following property. For every vertex $v$ in $\VertexSet$, there exists a (unique) non-negative integer $\SPArity{v}$ such that:
\begin{inparaenum}[(1)]
  \item if $v$ is the root then $\SPArity{v} = 0$, otherwise $\SPArity{v} = \SetCardinality{\SPFreeProjection[e_{1}]} + \SetCardinality{\SPSumProjection[e_{1}]} = \cdots = \SetCardinality{\SPFreeProjection[e_{\GraphInDegree[\Graph]{v}}]} + \SetCardinality{\SPSumProjection[e_{\GraphInDegree[\Graph]{v}}]}$ where $e_{1}, \dots, e_{\GraphInDegree[\Graph]{v}}$ are $v$'s incoming edges;
  \item $\SPFreeProjection[e_{1}], \dots, \SPFreeProjection[e_{\GraphOutDegree[\Graph]{v}}] \subseteq \{1, \dots, \SPArity{v}\}$, where $e_{1}, \dots, e_{\GraphOutDegree[\Graph]{v}}$ are $v$'s outgoing edges.
\end{inparaenum}
\end{definition}

For convenience, we denote by $\SPMaxArity(\Graph)$ the maximum of $\SPArity{v}$ across all vertices $v$ in the vertex set $\VertexSet$ of $\Graph$. Moreover, for every vertex $v$, we define $\SPVars[v] \DefineEqual \max(\SPSumProjection[e_{1}] \cup \cdots \cup \SPSumProjection[e_{\GraphOutDegree[\Graph]{v}}])$ so that $\SPSumProjection[e_{1}],\dots,\SPSumProjection[e_{\GraphOutDegree[\Graph]{v}}] \subseteq \{1, \dots, \SPVars[v]\}$.

We are now ready to define a sum-product circuit $\SPCircuit$, an input $\SPInput$ for $\SPCircuit$, and how to evaluate $\SPCircuit$ on $\SPInput$.

\begin{definition}
\label{def:sum-product-circuit}
A \defemph{sum-product circuit} is a tuple $\SPCircuit = \SPCircuitTuple$ where:
$\Field$ is a finite field,
$\SPSubset$ is a subset of $\Field$ (represented as a list of field elements),
$\InternalVertexSetDegree,\SPLeafDegree$ are positive integers (represented in unary) with $\SPLeafDegree \geq \SetCardinality{\SPSubset}$,
$\Graph=(\VertexSet, \EdgeSet, \SPFreeProjection, \SPSumProjection)$ is an ari-graph, and
$\SPPoly{}$ labels each internal vertex $v$ of $\Graph$ with an arithmetic circuit $\SPPoly[v](\vec{\VariableX},\vec{\VariableY},\vec{\VariableZ}) \colon \Field^{\SPArity{v}} \times  \Field^{\SPVars[v]} \times \Field^{\TreeDegree[\Tree]{v}} \to \Field$ of \underline{total} degree at most $\InternalVertexSetDegree$.
An \defemph{input} $\SPInput$ for $\SPCircuit$ labels each leaf vertex $v$ of $\Graph$ with a polynomial $\SPLeaf[v] \colon \Field^{\GraphDepth[\Graph]{v} \cdot \SPVars} \to \Field$ of \underline{individual} degree at most $\SPLeafDegree$.
The \defemph{value} of $\SPCircuit$ on an input $\SPInput$ is denoted $\SPValueL{\SPInput}{\SPCircuit}$ and equals $\SPValueL{\SPInput}{\Root[\Graph]}$, which we define below.

\begin{mdframed}
The value of a vertex $v$ of $\Graph$ on an input $\SPInput$ is denoted $\SPValueL{\SPInput}{v}$ and is recursively defined as follows. If $v$ is a leaf vertex, then $\SPValueL{\SPInput}{v}$ equals the polynomial $\SPLeaf[v]$. If $v$ is an internal vertex, then $\SPValueL{\SPInput}{v}$ is the $\SPArity{v}$-variate polynomial over $\Field$ defined by the following expression:
\begin{equation*}
\sum_{\vec{\beta} \in \SPSubset^{\SPVars[v]}}
\SPPoly[v]
\big(
\vec{\VariableX}, \vec{\beta},
\SPValueL{\SPInput}{u_{1}}
(\Restrict{\vec{\VariableX}}{\SPFreeProjection[e_{1}]},
\Restrict{\vec{\beta}}{\SPSumProjection[e_{1}]}),
\dots,
\SPValueL{\SPInput}{u_{\aodeg}}
(\Restrict{\vec{\VariableX}}{\SPFreeProjection[e_{\aodeg}]},
\Restrict{\vec{\beta}}{\SPSumProjection[e_{\aodeg}]})
\big)
\enspace,
\end{equation*}
where $\aodeg := \GraphOutDegree[\Graph]{v}$, and $e_{1} = (v, u_{1}), \dots, e_{\aodeg} = (v, u_{\aodeg})$ are the outgoing edges of $v$ (with multiplicity). In particular, $\SPValueL{\SPInput}{\SPCircuit} = \SPValueL{\SPInput}{\Root[\Graph]}$ is a constant in $\Field$.
\end{mdframed}
\end{definition}

Given a sum-product circuit we can ask two types of computational problems:
  (\emph{evaluation}) does a given input lead to a given output?
  (\emph{satisfaction}) does there exist an input that leads to a given output?
We now define each of these.

\begin{definition}[SPCE problem]
\label{def:circuit-evaluation}
The \defemph{sum-product circuit evaluation problem} is the following: given a sum-product circuit $\SPCircuit$, value $\SPOutput$, and input $\SPInput$ (given as a mapping from each leaf vertex $v$ of $\SPCircuit$'s ari-graph $\Graph$ to an arithmetic circuit computing the polynomial $\SPLeaf[v]$), determine if $\SPValueL{\SPInput}{\SPCircuit}=\SPOutput$. This problem induces the language
\begin{equation*}
\SPCELanguage
\DefineEqual
\big\{
(\SPCircuit, \SPOutput, \SPInput)
\text{ s.t. }
\SPValueL{\SPInput}{\SPCircuit} = \SPOutput
\big\}
\enspace.
\end{equation*}
(When $\SPInput$ is given as above, $\MaxSpace{\SPInput}$ denotes the maximum space required to evaluate any circuit in $\SPInput$.)
\end{definition}

\begin{definition}[SPCS problem]
\label{def:circuit-satisfaction}
The \defemph{sum-product circuit satisfaction problem} is the following: given a sum-product circuit $\SPCircuit$, partial mapping of leaf vertices to arithmetic circuits $\SPInput$, and value $\SPOutput$, determine if there exists a mapping $\SPAuxInput$ from the leaf vertices not in the domain of $\SPInput$ to polynomials s.t.\ $\SPValueL{\SPInput ,\SPAuxInput}{\SPCircuit}=\SPOutput$. This problem induces the relation
\begin{equation*}
\SPCSRelation
\DefineEqual
\big\{
\big((\SPCircuit,\SPOutput,\SPInput),\SPAuxInput\big)
\text{ s.t. }
\SPValueL{\SPInput,\SPAuxInput}{\SPCircuit} = \SPOutput
\big\}
\enspace.
\end{equation*}
(We refer to $\SPInput$ as the explicit input and $\SPAuxInput$ as the auxiliary input.)
\end{definition}

\subsubsection{Delegating sum-product circuit evaluation problems}
\label{sec:sum-product-circuits-evaluation}

We give an Interactive Proof to delegate sum-product formula \emph{evaluation} problems.

\begin{theorem}[IP for SPCE]
\label{thm:ip-for-SPCE}
There exists a public-coin Interactive Proof for the language $\SPCELanguage$. In more detail:
\begin{equation*}
\SPCELanguage \in
\mathbf{AM}
\left[
\begin{tabular}{rl}
soundness error: & $O(\InternalVertexSetDegree\SPLeafDegree \cdot \SPMaxArity(\Graph) \cdot \SetCardinality{\VertexSet(\Graph)}/\SetCardinality{\Field})$ \\
round complexity: & $O( \GraphDepth[]{\Graph} \cdot \SPMaxArity(\Graph))$ \\[1mm]
prover time: & $\poly(\BitSize{\SPCircuit}, \BitSize{\SPInput}, \SetCardinality{\SPSubset}^{\SPMaxArity(\Graph)})$ \\
verifier time: & $\poly(\BitSize{\SPCircuit}) + O(\GraphMaxInDegree{\Graph} \cdot \BitSize{\SPInput})$ \\[1mm]
verifier space: & $O(\SPMaxArity(\Graph) \cdot \Width{\Graph} \cdot \GraphMaxInDegree{\Graph} \cdot \log \SetCardinality{\Field} + \log \BitSize{\SPCircuit} + \MaxSpace{\SPInput})$
\end{tabular}
\right]
\enspace.
\end{equation*}
\end{theorem}

\begin{proof}[Proof sketch]
The prover and verifier receive as input a SPCE instance $(\SPCircuit, \SPOutput, \SPInput)$. They both associate with each vertex $v$ of its ari-graph $\Graph$ a \emph{set} of labels $\Labels{v}$; for the root, this set contains only the pair $(\EmptyVector, \SPOutput)$, while for all other vertices this set is initially empty and will be populated with at most $\GraphInDegree[\Graph]{v}$ pairs during the protocol. The prover and verifier then interact as follows.
\begin{enumerate}

  \item For every internal vertex $v$ of $\Graph$ taken in (any) topological order, letting $\aodeg \DefineEqual \GraphOutDegree[\Graph]{v}$:
  \begin{enumerate}[nolistsep]
    
    \item For every $(\vec{\gamma}_j, a_j)$ in $\Labels{v}$, the verifier samples a random $\alpha_j \in \Field$ and sends it to the prover.
    
    \item The prover and verifier invoke the sumcheck protocol on the following claim:
\begin{equation*}
\text{``}\quad
  \sum_{j=1}^{\SetCardinality{\Labels{v}}} \alpha_j \LDSPValueL{\SPInput}{v}(\vec{\gamma}_j)
= \sum_{j=1}^{\SetCardinality{\Labels{v}}} \alpha_j a_j
\quad\text{''.}
\end{equation*}
By the end of this subprotocol, the verifier has chosen $\vec{c}_{1} \in \Field^{\SPArity{v}}$ and $\vec{c}_{2} \in \Field^{\SPVars[v]}$ uniformly at random, and has derived from the prover's messages a value $b \in \Field$ that allegedly satisfies the following equality:
\begin{equation}
\label{eqn:SPCE-claim}
b =
\sum_{j=1}^{\SetCardinality{\Labels{v}}}
\alpha_j
\Big(
  \Lagrange{\SPSubset^{\SPArity{v}}}(\vec{\gamma_j}, \vec{c}_{1})
  \cdot
\SPPoly[v]\big(
\vec{c}_{1}, \vec{c}_{2},
\LDSPValueL{\SPInput}{u_{1}}(\Restrict{\vec{c}_{1}}{\SPFreeProjection[e_{1}]},
\Restrict{\vec{c}_{2}}{\SPSumProjection[e_{1}]}),
\dots,
\LDSPValueL{\SPInput}{u_{\aodeg}}(\Restrict{\vec{c}_{1}}{\SPFreeProjection[e_{\aodeg}]},
\Restrict{\vec{c}_{2}}{\SPSumProjection[e_{\aodeg}]})
\big)
\Big) \enspace,
\end{equation}
where $e_{1} = (v, u_{1}), \dots, e_{\aodeg} = (v, u_{\aodeg})$ are the outgoing edges of $v$ (with multiplicity).

    \item \label{step:sp-verifier-check} The prover sends $h_{1} \DefineEqual \LDSPValueL{\SPInput}{u_{1}}(\Restrict{\vec{c}_{1}}{\SPFreeProjection[e_{1}]},
    \Restrict{\vec{c}_{2}}{\SPSumProjection[e_{1}]}), \dots, h_{\aodeg} \DefineEqual \LDSPValueL{\SPInput}{u_{\aodeg}}(\Restrict{\vec{c}_{1}}{\SPFreeProjection[e_{\aodeg}]},
    \Restrict{\vec{c}_{2}}{\SPSumProjection[e_{\aodeg}]}) \in \Field$, and the verifier checks that
\begin{equation}
\label{eqn:SPCE-check}
b =
\sum_{j=1}^{\SetCardinality{\Labels{v}}}
\alpha_j
\Big(
  \Lagrange{\SPSubset^{\SPArity{v}}}(\vec{\gamma_j}, \vec{c}_{1})
  \cdot
  \SPPoly[v]\big(\vec{c}_{1}, \vec{c}_{2}, h_{1}, \dots, h_{\aodeg}\big)
\Big) \enspace.
\end{equation}

    \item For every $j = 1, \dots, \aodeg$, the verifier adds the label $((\Restrict{\vec{c}_{1}}{\SPFreeProjection[e_{j}]}, \Restrict{\vec{c}_{2}}{\SPSumProjection[e_{j}]}), h_{j})$ to $\Labels{u_{j}}$.

  \end{enumerate}

\item For every leaf vertex $v$ of $\Graph$, and for every $(\vec{\gamma}, a) \in \Labels{v}$, the verifier checks that $\LDSPValueL{\SPInput}{v}(\vec{\gamma}) = a$, i.e., that $\SPLeaf[v](\vec{\gamma}) = a$. (Note that every set $\Labels{v}$ has had $\GraphInDegree[\Graph]{v}$ labels added to it. The size of $\Labels{v}$ is then at most $\GraphInDegree[\Graph]{v}$, with the `strictly less' case occurring if the verifier happens to have added the same label twice.)

\end{enumerate}
While the above description considers sequential invocations of the sumcheck protocol, these can be run in parallel in $\GraphDepth[]{\Graph}$ phases: first the root (which has depth $0$), then all vertices of depth $1$, then all vertices of depth $2$, and so on until all vertices of depth $\GraphDepth[]{\Graph}-1$. Each such phase requires $O(\SPMaxArity(\Graph))$ rounds, so that the number of rounds is now $O(\GraphDepth[]{\Graph} \cdot \SPMaxArity(\Graph))$, as claimed. The claimed running times for the prover and verifier follow immediately from the above description. The claimed space bound can be attained by discarding all labels at previous levels before moving to the next level, and checking leaf vertices at the same time as the internal vertices at the same depth. We are left to argue the claimed soundness error.

Suppose that, when considering some internal vertex $v$ of $\Graph$ in the protocol above, there exists $(\vec{\gamma}, a) \in \Labels{v}$ such that $\LDSPValueL{\SPInput}{v}(\vec{\gamma}) \neq a$. Then, with probability at least $1 - 1/\SetCardinality{\Field}$, it holds that $\sum_{j=1}^{\SetCardinality{\Labels{v}}} \alpha_j \LDSPValueL{\SPInput}{v}(\vec{\gamma}_j) \neq \sum_{j=1}^{\SetCardinality{\Labels{v}}} \alpha_j a_j$, which means that the prover and verifier invoke the sumcheck protocol on a false claim. By the soundness of the sumcheck protocol, either the verifier rejects or \equnref{eqn:SPCE-claim} holds with probability at most $2\InternalVertexSetDegree\SPLeafDegree \cdot \SPArity{v}/\SetCardinality{\Field}$. In this latter case, either \equnref{eqn:SPCE-check} fails to hold and the verifier rejects, or there exists $j \in \{1, \dots, \GraphOutDegree[\Graph]{v}\}$ such that $h_j \neq \LDSPValueL{\SPInput}{u_{j}}(\Restrict{\vec{c}_{1}}{\SPFreeProjection[e_{j}]},
\Restrict{\vec{c}_{2}}{\SPSumProjection[e_{j}]})$, which means that there exists a vertex $u$ in the next layer (in particular, $u = u_{j}$) for which there exists $(\vec{\gamma}', a') \in \Labels{u}$ such that $\LDSPValueL{\SPInput}{u}(\vec{\gamma}') \neq a'$. If $u$ is a leaf vertex then the verifier will reject when considering $u$; otherwise we repeat the above argument. Taking a union bound over the internal vertices of $\Graph$ yields the claimed soundness error.
\end{proof}

\subsubsection{Delegating sum-product circuit satisfaction problems}
\label{sec:sum-product-circuits-satisfaction}

We give an Interactive PCP to delegate sum-product circuit \emph{satisfaction} problems.

\begin{theorem}[IPCP for SPCS]
\label{thm:ipcp-for-SPCS}
There exists a (public-coin and non-adaptive) Interactive PCP for the relation $\SPCSRelation$. In more detail:
\begin{equation*}
\SPCSRelation \in
\mathbf{IPCP}
\left[
\begin{tabular}{rl}
soundness error: & $O(\InternalVertexSetDegree\SPLeafDegree \cdot \SPMaxArity(\Graph) \cdot \SetCardinality{\EdgeSet(\Graph)}/\SetCardinality{\Field})$ \\
round complexity: & $O( \GraphDepth[]{\Graph} \cdot \SPMaxArity(\Graph))$ \\[1mm]
proof length: & $O(\SetCardinality{\VertexSet(\Graph)} \cdot \SetCardinality{\Field}^{\SPMaxArity(\Graph)})$ \\
query complexity: & $\SetCardinality{\VertexSet(\Graph)} \cdot \poly(\log \SetCardinality{\Field}, \SPMaxArity(\Graph), \SPLeafDegree)$ \\[1mm]
prover time: & $\poly(\BitSize{\SPCircuit}, \BitSize{\SPInput}, \BitSize{\SPAuxInput}, \SetCardinality{\SPSubset}^{\SPMaxArity(\Graph)})$ \\
verifier time: & $\poly(\BitSize{\SPCircuit}, \BitSize{\SPInput})$ \\[1mm]
verifier space: & $O(\SPMaxArity(\Graph) \cdot \Width{\Graph} \cdot \GraphMaxInDegree{\Graph} \cdot \log \SetCardinality{\Field} + \log \BitSize{\SPCircuit} + \MaxSpace{\SPInput})$
\end{tabular}
\right]
\enspace.
\end{equation*}
\end{theorem}

\begin{proof}
The protocol is analogous to that in the proof of \thmref{thm:ipcp-for-SPFS}, which considers the relation $\SPFSRelation$ (sum-product formula satisfaction) rather than $\SPCSRelation$ (sum-product circuit satisfaction). Specifically, we only need to replace the Interactive Proof for the language $\SPFELanguage$ (sum-product formula evaluation) with the Interactive Proof for the language $\SPCELanguage$ (sum-product circuit evaluation) that we gave in the proof of \thmref{thm:ip-for-SPCE}. We omit the details, except for one technicality that we now describe.

The strategy described in the above paragraph eventually leads the verifier to read, via self-correction, at most $\GraphInDegree[\Graph]{v}$ values of $\SPAuxLeaf[v]$ for every leaf vertex $v$ of $\Graph$ not in the domain of $\SPInput$. Overall, the verifier reads at most $\sum_{v} \GraphInDegree[\Graph]{v} \leq \SetCardinality{\EdgeSet(\Graph)}$ values via self-correction, which corresponds to $\SetCardinality{\EdgeSet(\Graph)} \cdot \poly(\log \SetCardinality{\Field} + \SPMaxArity(\Graph) + \SPLeafDegree)$ actual queries, and a soundness error of $O(\InternalVertexSetDegree\SPLeafDegree \cdot \SPMaxArity(\Graph) \cdot \SetCardinality{\VertexSet(\Graph)}/\SetCardinality{\Field})$. To obtain the stated query complexity, we reduce the number of values read via self-correction to a single value per leaf vertex, via the following standard trick.

Let $t_1, \dots, t_{\GraphInDegree[\Graph]{v}} \in \Field$ be arbitrary distinct values known to both the prover and verifier, and let $\Curve{v} \colon \Field \to \Field^{\SPArity{v}}$ be the unique polynomial of degree less than $\GraphInDegree[\Graph]{v}$ such that $\Curve{v}(t_{i}) = \vec{\gamma}_{i}$ for every $(\vec{\gamma}_{i}, a_{i}) \in \Labels{v}$. The prover sends $\ComposedCurve{v} \DefineEqual (\SPAuxLeaf[v] \circ \Curve{v}) \colon \Field \to \Field$ to the verifier (as a list of at most $\SPArity{v} \cdot \GraphInDegree[\Graph]{v} \cdot \SPLeafDegree$ coefficients), who checks that $\ComposedCurve{v}(t_{i}) = a_{i}$ for every $i \in \{1, \dots, \GraphInDegree[\Graph]{v}\}$. The verifier then picks $t \in \Field$ uniformly at random and checks that $\SPAuxLeaf[v](\Curve{v}(t)) = \ComposedCurve{v}(t)$; this involves obtaining, via self-correction, the value of $\SPAuxLeaf[v]$ at $\Curve{v}(t)$. Soundness is maintained because if the prover sends $\ComposedCurve{v}' \neq \ComposedCurve{v}$ then the probability that $\SPAuxLeaf[v](\Curve{v}(t)) = \ComposedCurve{v}'(t)$ for uniformly random $t \in \Field$ is at most $\SPArity{v} \cdot \GraphInDegree[\Graph]{v} \cdot \SPLeafDegree/\SetCardinality{\Field}$. The overall soundness error is then, by a union bound, at most $O(\InternalVertexSetDegree\SPLeafDegree \cdot \SPMaxArity(\Graph) \cdot \SetCardinality{\VertexSet(\Graph)}/\SetCardinality{\Field}) + \sum_{v} \SPArity{v} \cdot \GraphInDegree[\Graph]{v} \cdot \SPLeafDegree/\SetCardinality{\Field} = O(\InternalVertexSetDegree\SPLeafDegree \cdot \SPMaxArity(\Graph) \cdot \SetCardinality{\EdgeSet(\Graph)}/\SetCardinality{\Field})$. The additional space required for this test is $O(\SPMaxArity(\Graph) \cdot \GraphMaxInDegree{\Graph} \cdot \log \SetCardinality{\Field} + \log \BitSize{\SPCircuit})$, so the space bound is unaffected.
\end{proof}

\doclearpage
\section{Zero knowledge sum-product protocols}
\label{sec:pzk-for-sum-products}

The protocols for delegating sum-product computations described in \secref{sec:sum-product-protocols} are not zero knowledge. We show how to delegate, in the Interactive PCP model, sum-product circuit evaluation problems (\secref{sec:pzk-sum-product-evaluation}) and satisfaction problems (\secref{sec:pzk-sum-product-satisfaction}). As a special case, we also obtain the same for sum-product formulas.

\subsection{The case of sum-product evaluation}
\label{sec:pzk-sum-product-evaluation}

The purpose of this section is to show that the language $\SPCELanguage$ (consisting of sum-product circuit evaluation problems, see \defref{def:circuit-evaluation}) has perfect zero knowledge Interactive PCPs:

\begin{theorem}[PZK IPCP for $\SPCELanguage$]
\label{thm:pzk-for-SPCE}
For every query bound function $\SCStrength(n)$, the language $\SPCELanguage$ has a (public-coin and non-adaptive) Interactive PCP that is perfect zero knowledge against all $\SCStrength$-query malicious verifiers. In more detail, letting $\alpha \DefineEqual \log \SCStrength / \log \SetCardinality{\SPSubset}$:
{\small
\begin{equation*}
\SPCELanguage \in
\PZKIPCP
\left[
\begin{tabular}{rl}
soundness error: & $O(\InternalVertexSetDegree\SPLeafDegree \cdot \GraphMaxInDegree{\Graph} \cdot (\SPMaxArity(\Graph) + \alpha) \cdot \SetCardinality{\VertexSet(\Graph)}/\SetCardinality{\Field})$ \\
round complexity: & $O( \GraphDepth[]{\Graph} \cdot (\SPMaxArity(\Graph) + \alpha))$ \\[1mm]
proof length: & $O(\SetCardinality{\VertexSet(\Graph)} \cdot \SetCardinality{\Field}^{\SPMaxArity(\Graph) + \alpha})$ \\
query complexity: & $\SetCardinality{\VertexSet(\Graph)} \cdot \poly(\log \SetCardinality{\Field}, \SPMaxArity(\Graph), \alpha, \InternalVertexSetDegree, \SPLeafDegree, \GraphMaxInDegree{\Graph})$ \\[1mm]
prover time: & $\poly(\BitSize{\SPCircuit}, \BitSize{\SPInput}, \SetCardinality{\Field}^{\SPMaxArity(\Graph) + \alpha})$ \\
verifier time: &  $\poly(\BitSize{\SPCircuit}, \alpha) + O(\GraphMaxInDegree{\Graph} \cdot \BitSize{\SPInput})$ \\[1mm]
verifier space: & $O((\SPMaxArity(\Graph) + \alpha) \cdot \Width{\Graph} \cdot \GraphMaxInDegree{\Graph} \cdot \log \SetCardinality{\Field} + \log \BitSize{\SPCircuit} + \MaxSpace{\SPInput})$ \\[1mm]
simulator overhead: & $\poly(\BitSize{\SPCircuit}, \alpha) \cdot (\BitSize{\SPInput} +  \QueryComplexity_{\Malicious{\Verifier}}^{3})$.
\end{tabular}
\right]
\enspace.
\end{equation*}}
\end{theorem}

We introduce some notation before the proof. Given a subset $\SPSubset$ of $\Field$ and a positive integer $\SPVars$, we denote by $\ZeroPoly{\SPSubset^{\SPVars}}$ the $\SPVars$-variate polynomial $\prod_{i=1}^{\SPVars} \prod_{\alpha \in \SPSubset} (\VariableX_{i} - \alpha)$. Note that $\ZeroPoly{\SPSubset^{\SPVars}}$ is zero on $\SPSubset^{\SPVars}$ and nonzero on $(\Field - \SPSubset)^{\SPVars}$, and can be evaluated in $\poly(\SetCardinality{\SPSubset}+\SPVars)$ field operations and space $O(\log \SPVars + \log \SetCardinality{\Field})$.

To guide us to the proof of the above theorem, it is instructive to look at why the protocol of \thmref{thm:ip-for-SPCE} may not be zero knowledge. We identify two potential sources of leakage: the first is the values $\LDSPValueL{\SPInput}{v}(\vec{c})$ for internal vertices $v$ which the prover sends in \stepref{step:sp-verifier-check} of the protocol; the second is the partial sums which are leaked by the sumcheck subprotocol itself.

We resolve the first issue by replacing, for each internal vertex $v$, the low-degree extension $\LDSPValueL{\SPInput}{v}$ in the SPCE protocol with a \emph{randomized} low-degree extension $\RLDSPValueL{\SPInput}{v}$, which is $\GraphInDegree[\Graph]{v}$-wise independent outside of $\SPSubset^{\SPArity{v}}$. More precisely, for any list of distinct query points $\vec{\gamma}_{1}, \ldots, \vec{\gamma}_{\GraphInDegree[\Graph]{v}} \in (\Field - \SPSubset)^{\SPArity{v}}$, $(\RLDSPValueL{\SPInput}{v}(\vec{\gamma}_{i}))_{i=1}^{\GraphInDegree[\Graph]{v}}$ is uniformly random in $\Field^{\ell}$. Given such a low-degree extension, it suffices to ensure that the verifier may only learn its evaluations inside $(\Field - \SPSubset)^{\SPVars}$, and then only on at most $\GraphInDegree[\Graph]{v}$ distinct points; then we can simulate all of these queries with uniformly random field elements.

Given a sum-product circuit $\SPCircuit = \SPCircuitTuple$, we define for every vertex $v$ in its ari-graph $\Graph=(\VertexSet, \EdgeSet, \SPFreeProjection, \SPSumProjection)$ a \emph{random variable} $\RLDSPValueL{\SPInput}{v}$ based on $\SPValueL{\SPInput}{v}$, as follows. Let $\SubsetSize \DefineEqual 2\InternalVertexSetDegree \cdot (\SPLeafDegree + \GraphMaxInDegree{\Graph}) = \Omega(\SetCardinality{\SPSubset})$. If $v$ is a leaf vertex, then $\RLDSPValueL{\SPInput}{v}$ simply equals $\SPLeaf[v]$ with probability $1$. If instead $v$ is an internal vertex, then $\RLDSPValueL{\SPInput}{v}$ is a ``randomized'' low-degree extension of $\SPValueL{\SPInput}{v}$:
\begin{equation*}
\RLDSPValueL{\SPInput}{v}(\vec{\VariableX})
\DefineEqual
\LDSPValueL{\SPInput}{v}(\vec{\VariableX})
+
\ZeroPoly{\SPSubset^{\SPArity{v}}}(\vec{\VariableX})
\sum_{\vec{\gamma} \in \SSCSubset^{\SSCVars}}
  \RandPoly_{v}(\vec{\VariableX}, \vec{\gamma})
\enspace,
\end{equation*}
where $\LDSPValueL{\SPInput}{v}$ is the \emph{fixed} low-degree extension of $\SPValueL{\SPInput}{v}$ that we used in \secref{sec:sum-product-protocols} (see \equnref{eq:lde-of-v}), $k$ is a security parameter, $\SSCSubset$ is an arbitrary subset of $\Field$ of size $\SubsetSize$ with $0 \in \SSCSubset$, $\RandPoly_{v}$ is uniformly random in $\PolynomialRingIndOneXY{\Field}{\SPArity{v}}{\VariableX}{\SSCVars}{\VariableY}{\GraphInDegree[\Graph]{v}}{2\SubsetSize}$, and $\GraphInDegree[\Graph]{v}$ is the in-degree of $v$ in $\Graph$. (In the protocol, the verifier receives in the oracle message the evaluation table of independently drawn $\RandPoly_{v}$, for every internal vertex $v$ of $\Graph$.)

Since $\RLDSPValueL{\SPInput}{v}$ agrees with $\SPValueL{\SPInput}{v}$ on $\SPSubset^{\SPArity{v}}$, we can equivalently write (using the definition of $\LDSPValueL{\SPInput}{v}$):
\begin{align*}
\RLDSPValueL{\SPInput}{v}(\vec{\VariableX})
\DefineEqual
&
\left(
\sum_{\vec{\alpha} \in \SPSubset^{\SPArity{v}}}
\sum_{\vec{\beta} \in \SPSubset^{\SPVars[v]}}
    \Lagrange{\SPSubset^{\SPArity{v}}}(\vec{\VariableX}, \vec{\alpha})
    \cdot
\SPPoly[v]\big(
\vec{\alpha},\vec{\beta},
\RLDSPValueL{\SPInput}{u_{1}}(\Restrict{\vec{\alpha}}{\SPFreeProjection[e_{1}]},
\Restrict{\vec{\beta}}{\SPSumProjection[e_{1}]}),
\dots,
\RLDSPValueL{\SPInput}{u_{\aodeg}}(\Restrict{\vec{\alpha}}{\SPFreeProjection[e_{\aodeg}]},
\Restrict{\vec{\beta}}{\SPSumProjection[e_{\aodeg}]})
\big)
\right) \\
&\hspace{8cm}
+
\ZeroPoly{\SPSubset^{\SPArity{v}}}(\vec{\VariableX})
\sum_{\vec{\gamma} \in \SSCSubset^{\SSCVars}}
  \RandPoly_{v}(\vec{\VariableX}, \vec{\gamma})
  \\
=& 
\sum_{\vec{\alpha} \in \SPSubset^{\SPArity{v}}}
\sum_{\vec{\beta} \in \SPSubset^{\SPVars[v]}}
\sum_{\vec{\gamma} \in \SSCSubset^{\SSCVars}}
\Big(
\Lagrange{\SSCSubset^{\SSCVars}}(\vec{0}, \vec{\gamma})
\cdot
\Lagrange{\SPSubset^{\SPArity{v}}}(\vec{\VariableX}, \vec{\alpha}) 
\cdot
\SPPoly[v]\big(
\vec{\alpha},\vec{\beta},
\RLDSPValueL{\SPInput}{u_{1}}(\Restrict{\vec{\alpha}}{\SPFreeProjection[e_{1}]},
\Restrict{\vec{\beta}}{\SPSumProjection[e_{1}]}),
\dots,
\RLDSPValueL{\SPInput}{u_{\aodeg}}(\Restrict{\vec{\alpha}}{\SPFreeProjection[e_{\aodeg}]},
\Restrict{\vec{\beta}}{\SPSumProjection[e_{\aodeg}]})
\big) \\
&\hspace{7cm}
+ \Lagrange{\SCSubset^{\SPArity{v}+\SPVars[v]}}((\vec{\alpha},\vec{\beta}),\vec{0})
\cdot
\ZeroPoly{\SPSubset^{\SPArity{v}}}(\vec{\VariableX}) 
\cdot
\RandPoly_{v}(\vec{\VariableX}, \vec{\gamma})
\Big)
\enspace.
\end{align*}
Note that the individual degree of $\RLDSPValueL{\SPInput}{v}(\vec{\VariableX})$ is exactly $\SetCardinality{\SPSubset} + \GraphInDegree[\Graph]{v}$. The individual degree of the summand in the last line (in $\vec{\alpha},\vec{\beta},\vec{\gamma}$) is at most $\max\{2\SubsetSize, \SetCardinality{\SPSubset} + \GraphInDegree[\Graph]{v} + \InternalVertexSetDegree \cdot \max\{\SPLeafDegree, \SetCardinality{\SPSubset} + \max_{1 \leq i \leq \GraphOutDegree[\Graph]{v}} \GraphInDegree[\Graph]{u_{i}} \} \} \leq 2\SubsetSize$.

Observe first that $\RandPoly_{v}$ is a perfectly-hiding commitment to the random polynomial $S_{v}(\vec{\VariableX}) := \sum_{\vec{\gamma} \in \SSCSubset^{\SSCVars}} \RandPoly_{v}(\vec{\VariableX}, \vec{\gamma})$, and so $S_{v}(\vec{x})$ itself is uniformly random even conditioned on strictly fewer than $\SubsetSize^{k}$ queries to $\RandPoly_{v}$. Then since $\ZeroPoly{\SPSubset^{\SPArity{v}}}(\vec{\VariableX})$ is non-zero in $(\Field - \SPSubset)^{\SPArity{v}}$ and the individual degree of $S_{v}$ is $\GraphInDegree[\Graph]{v}$, the required independence property holds. If we ensure that the verifier chooses its challenges in the sumcheck protocol from $\Field - \SPSubset$ rather than all of $\Field$ (i.e. the prover aborts otherwise), then the values sent to the verifier in the SPCE protocol will indeed be uniformly random.

We resolve the second issue by replacing the sumcheck subprotocol with the zero knowledge IPCP for sumcheck defined in \secref{sec:strong-zk-sumcheck}. This requires sending $O(\SetCardinality{\VertexSet(\Graph)})$ proofs, which we can concatenate together with the $\RandPoly_{v}$ into a single oracle. Note that here we will require the full strength of the zero knowledge guarantee which we obtain in \thmref{thm:strong-sumcheck-ipcp} (as opposed to the weaker guarantee of \cite{BenSassonCFGRS16}), because the simulator is not able to make an arbitrary polynomial number of queries to $\RLDSPValueL{\SPInput}{v}$ for any internal vertex $v$. Instead the number of queries must be bounded by $\GraphInDegree[\Graph]{v}$ for each $v$, so that these queries can be simulated by choosing uniformly random field elements; the zero knowledge guarantee of \thmref{thm:strong-sumcheck-ipcp} allows us to do exactly that.

We are now ready to put everything together.

\begin{proof}
Fix $k := \lceil \log \SCStrength / \log \SubsetSize \rceil$. The prover and verifier receive as input a SPCE instance $(\SPCircuit, \SPOutput, \SPInput)$. We first describe the oracle message that is first sent to the verifier, and then describe the subsequent interaction between the prover and verifier.
\begin{itemize}

\item \textbf{Oracle.}
The prover sends to the verifier a proof string $\Proof$ that contains, for each internal vertex $v$ of $\Graph$:
\begin{itemize}[nolistsep]
  \item the evaluation table of the polynomial $
\RandPoly_{v}
\in
  \PolynomialRingIndOneXY{\Field}{\SPArity{v}}{\VariableX}{\SSCVars}{\VariableY}{\GraphInDegree[\Graph]{v}}{2\SubsetSize}
$
drawn independently and uniformly at random;
  \item a proof string $\Proof_{v}$ which is the oracle sent in a $\SubsetSize^{k}$-strong zero knowledge sumcheck protocol on input $(\Field, \SPArity{v} + \SPVars[v] + k, \SubsetSize, \SPSubset, \cdot)$. \footnote{Recall that we do not need to specify $\SCSum$ until later on in the protocol, and it will depend on the verifier's random choices.}
\end{itemize}

  \item \textbf{Interaction.}
  The prover and verifier associate, for each vertex $v$ of $\Graph$, a \emph{set} of labels $\Labels{v}$; for the root, this set contains only the pair $(\EmptyVector, \SPOutput)$, while for all other vertices this set is initially empty and will be populated during the protocol. The prover and verifier then interact as follows.
\begin{enumerate}

  \item \label{step:sp-main-loop}
  \mbox{For every internal vertex $v$ of $\Graph$ taken in (any) topological order, letting $\aodeg \DefineEqual \GraphOutDegree[\Graph]{v}$:}
  \begin{enumerate}
  
      \item For every $(\vec{\gamma}_j, a_j)$ in $\Labels{v}$, the verifier samples a random $\alpha_j \in \Field$ and sends it to the prover.
      
      \item The prover and verifier invoke a $\SubsetSize^{\SSCVars}$-strong perfect zero knowledge Interactive Probabilistically Checkable Proof system for sumcheck (see \secref{sec:strong-zk-sumcheck}) on the claim
\begin{equation*}
\text{``}\quad
  \sum_{j=1}^{\SetCardinality{\Labels{v}}} \alpha_j \RLDSPValueL{\SPInput}{v}(\vec{\gamma}_j)
= \sum_{j=1}^{\SetCardinality{\Labels{v}}} \alpha_j a_j
\quad\text{''}
\end{equation*}
using $\Proof_{v}$ as the oracle, and with $\SoundnessSet \DefineEqual \Field \setminus \SCSubset$. By the end of this subprotocol, the verifier has chosen $\vec{c}_{1} \in \Field^{\SPArity{v}}$, $\vec{c}_{2} \in \Field^{\SPVars}$, and $\vec{c}_{3} \in \Field^{\SSCVars}$ uniformly at random, and has derived from the prover's messages a value $b \in \Field$ that allegedly satisfies the following equality:
\begin{align}
\label{eqn:pzk-SPCE-claim}
b =
\sum_{j=1}^{\SetCardinality{\Labels{v}}}
\alpha_j
\Big(
&\Lagrange{\SSCSubset^{\SSCVars}}(\vec{0}, \vec{c}_{3})
\cdot
\Lagrange{\SPSubset^{\SPArity{v}}}(\vec{\gamma}_{j}, \vec{c}_{1}) 
\cdot
\SPPoly[v]\big(
\vec{c}_{1}, \vec{c}_{2},
\RLDSPValueL{\SPInput}{u_{1}}(\Restrict{\vec{c}_{1}}{\SPFreeProjection[e_{1}]},
\Restrict{\vec{c}_{2}}{\SPSumProjection[e_{1}]}),
\dots,
\RLDSPValueL{\SPInput}{u_{\aodeg}}(\Restrict{\vec{c}_{1}}{\SPFreeProjection[e_{\aodeg}]},
\Restrict{\vec{c}_{2}}{\SPSumProjection[e_{\aodeg}]})
\big)\\
&+ \Lagrange{\SCSubset^{\SPArity{v} + \SPVars[v]}}((\vec{c}_{1},\vec{c}_{2}),\vec{0})
\cdot
\ZeroPoly{\SPSubset^{\SPArity{v}}}(\vec{\gamma}_{j}) 
\cdot
\RandPoly_{v}(\vec{\gamma}_{j}, \vec{c}_{3})
\Big) \nonumber
\enspace,
\end{align}
where $e_{1} = (v, u_{1}), \dots, e_{\aodeg} = (v, u_{\aodeg})$ are the outgoing edges of $v$ (with multiplicity).
    
    \item The prover sends $\vec{h}_{v} \DefineEqual (h_{1},\dots,h_{\aodeg})$ where $h_{1} \DefineEqual \RLDSPValueL{\SPInput}{u_{1}}(\Restrict{\vec{c}_{1}}{\SPFreeProjection[e_{1}]},
    \Restrict{\vec{c}_{2}}{\SPSumProjection[e_{1}]}),
    \dots, h_{\aodeg} \DefineEqual \RLDSPValueL{\SPInput}{u_{\aodeg}}(\Restrict{\vec{c}_{1}}{\SPFreeProjection[e_{\aodeg}]},
    \Restrict{\vec{c}_{2}}{\SPSumProjection[e_{\aodeg}]})$. For every $j = 1, \dots, \aodeg$, the verifier adds the label $((\Restrict{\vec{c}_{1}}{\SPFreeProjection[e_{j}]},\Restrict{\vec{c}_{2}}{\SPSumProjection[e_{j}]}), h_{j})$ to $\Labels{u_{j}}$.

\end{enumerate}

  \item For every internal vertex $v$ of $\Graph$, the verifier checks that
  \begin{equation}
  \label{eqn:pzk-SPCE-check}
  b =
  \sum_{j=1}^{\SetCardinality{\Labels{v}}}
  \alpha_{j}
  \Big(
  \Lagrange{\SSCSubset^{\SSCVars}}(\vec{0}, \vec{c}_{3})
  \cdot
  \Lagrange{\SPSubset^{\SPArity{v}}}(\vec{\gamma}_{j}, \vec{c}_{1}) 
  \cdot
  \SPPoly[v]\big(\vec{c}_{1}, \vec{c}_{2}, \vec{h}_{v}\big)
  + \Lagrange{\SCSubset^{\SPArity{v} + \SPVars[v]}}((\vec{c}_{1},\vec{c}_{2}),\vec{0})
  \cdot
  \ZeroPoly{\SPSubset^{\SPArity{v}}}(\vec{\gamma}_{j}) 
  \cdot
  \RandPoly_{v}(\vec{\gamma}_{j}, \vec{c}_{3})
  \Big)
  \enspace.
  \end{equation}
  For this, the prover sends $\big(\RandPoly_{v}(\vec{\gamma}_{j}, \vec{c}_{3})\big)_{j=1}^{\SetCardinality{\Labels{v}}}$, which the verifier uses to compute the above expression. The verifier checks that these values are correct using a standard interpolation trick and a single query to $\RandPoly_{v}$. To make this query, the verifier
  \begin{inparaenum}[(i)]
  	\item tests that $\RandPoly_{v}$ is close to (the evaluation of) a polynomial in $\PolynomialRingIndOneXY{\Field}{\SPArity{v}}{\VariableX}{\SSCVars}{\VariableY}{\GraphInDegree[\Graph]{v}}{2\SubsetSize}$;
  	\item uses self-correction to make the required query.
  \end{inparaenum}

  \item \label{step:sp-final-check}
  For every leaf vertex $v$ of $\Graph$, and for every $(\vec{\gamma}, a) \in \Labels{v}$, the verifier checks that $\RLDSPValueL{\SPInput}{v}(\vec{\gamma}) = a$, i.e., that $\SPLeaf[v](\vec{\gamma}) = a$.
  
\end{enumerate}

\end{itemize}
\parhead{Efficiency}
The protocol runs a $\SubsetSize^{\SSCVars}$-strong zero knowledge sumcheck protocol at most $\SetCardinality{\VertexSet}$ times on polynomials of at most $\SPMaxArity(\Graph) + \SSCVars$ variables. Inspection of the protocol shows that we can execute \stepref{step:sp-main-loop} in parallel for all vertices at the same depth, so the number of rounds is at most $\GraphDepth{\Graph}(\SPMaxArity(\Graph) + 2\SSCVars + 2)$. The running time of the verifier is clearly polynomial in $\BitSize{\SPCircuit}$. The number of queries to the $\RandPoly_{v}$ is at most $\SetCardinality{\VertexSet(\Graph)} \cdot \poly(\log \SetCardinality{\Field} + k + \SubsetSize)$, and the number of queries to the $\Proof_{v}$ is at most $\SetCardinality{\VertexSet(\Graph)} \cdot \poly(\log \SetCardinality{\Field} + \SPMaxArity(\Graph) + \SSCVars + \SubsetSize)$.

\parhead{Completeness}
Perfect completeness is clear from the protocol description and the perfect completeness of the zero-knowledge sumcheck protocol.
	
\parhead{Soundness}
Let $v_{1},v_{2},\dots$ be any topological order of the internal vertices of $\Graph$, and define let $V_{\geq i}$ be the union of $\{v_{j} : j \geq i\}$ and the leaf vertices of $\Graph$. First we argue that if before iteration $i$ of \stepref{step:sp-main-loop} there exists $v \in V_{\geq i}$ such that $\Labels{v}$ contains $(\vec{\gamma}, a)$ with $\RLDSPValueL{\SPInput}{v}(\vec{\gamma}) \neq a$, then after this iteration with high probability either
\begin{inparaenum}[(a)]
  \item the verifier rejects, or
  \item there exists $w \in V_{\geq i+1}$, such that there is $(\vec{\gamma}',b) \in \Labels{w}$ with $\RLDSPValueL{\SPInput}{w}(\vec{\gamma}') \neq b$.
\end{inparaenum}
If $v \in V_{\geq i+1}$ then there is nothing to do, so we may assume that $v = v_{i}$. The probability that $\sum_{j=1}^{\SetCardinality{\Labels{v_{i}}}} \alpha_j \RLDSPValueL{\SPInput}{v_{i}}(\vec{\gamma}_j) = \sum_{j=1}^{\SetCardinality{\Labels{v_{i}}}} \alpha_j a_j$ is $1/\SetCardinality{\Field}$. So with probability $1 - 1/\SetCardinality{\Field}$ we run the zero knowledge sumcheck protocol on a false claim about a polynomial with $(\SPArity{v} + \SPVars[v] + \SSCVars)$ variables and individual degree at most $2\SubsetSize$, which means that with probability at least $1 - O((\SPArity{v} + \SPVars[v] + \SSCVars) \cdot \SubsetSize/(\SetCardinality{\Field}-\SetCardinality{\SPSubset}))$ either the verifier rejects or outputs the \emph{false} claim
\begin{align*}
b =
\sum_{j=1}^{\SetCardinality{\Labels{v}}}
\big(
\alpha_{j}
&\Lagrange{\SSCSubset^{\SSCVars}}(\vec{0}, \vec{c}_{3})
\cdot
\Lagrange{\SPSubset^{\SPArity{v}}}(\vec{\gamma}_{j}, \vec{c}_{1}) 
\cdot
\SPPoly[v]\big(
\vec{c}_{1}, \vec{c}_{2},
\RLDSPValueL{\SPInput}{u_{1}}(\Restrict{\vec{c}_{1}}{\SPFreeProjection[e_{1}]},
\Restrict{\vec{c}_{2}}{\SPSumProjection[e_{1}]}),
\dots,
\RLDSPValueL{\SPInput}{u_{\aodeg}}(\Restrict{\vec{c}_{1}}{\SPFreeProjection[e_{\aodeg}]},
\Restrict{\vec{c}_{2}}{\SPSumProjection[e_{\aodeg}]})
\big)\\
&+ \Lagrange{\SCSubset^{\SPArity{v} + \SPVars[v]}}((\vec{c}_{1},\vec{c}_{2}),\vec{0})
\cdot
\ZeroPoly{\SPSubset^{\SPArity{v}}}(\vec{\gamma}_{j}) 
\cdot
\RandPoly_{v}(\vec{\gamma}_{j}, \vec{c}_{3})
\big)
\enspace.
\end{align*}

The verifier receives values $(r_{j})_{j=1}^{\SetCardinality{\Labels{v}}}$, which it substitutes for $\RandPoly_{v}(\vec{\gamma}_{j}, \vec{c}_{3})$ in \equnref{eqn:pzk-SPCE-check}. If this expression does not evaluate to $b$, then the verifier rejects. If $\RandPoly_{v}$ is far from any polynomial in $\PolynomialRingIndOneXY{\Field}{\SPArity{v}}{\VariableX}{\SSCVars}{\VariableY}{\GraphInDegree[\Graph]{v}}{2}$, then the verifier rejects with high probability. If there exists some $j$ such that $r_{j} \neq \RandPoly_{v}(\vec{\gamma}_{j}, \vec{c}_{3})$, then by the soundness of the polynomial interpolation test, the verifier will reject with probability at least $1-O((\SPArity{v} + k) \cdot \SubsetSize/\SetCardinality{\Field})$. Otherwise, it must be the case that $\RLDSPValueL{\SPInput}{u_{j}}(\vec{c}) \neq h_{j}$ where $u_{j}$ is one of the $\GraphOutDegree[\Graph]{v_{i}}$ children of $v_{i}$. Since $u_{j} \in V_{\geq i+1}$, with high probability after this iteration either the verifier has already rejected or there is $w \in V_{\geq i+1}$ and $(\vec{\gamma}', a') \in \Labels{w}$ with $\RLDSPValueL{\SPInput}{w}(\vec{\gamma}') \neq a'$.

The above implies soundness in a straightforward way: if $\SPValueL{\SPInput}{\SPCircuit} \neq \SPOutput$ then the condition is satisfied before the first iteration, and in each iteration with high probability either the condition is maintained or the verifier rejects. Thus with high probability either the verifier rejects or the invariant holds before the last iteration; after the last iteration, the verifier will reject with high probability because only leaf vertices are left (and claims about them are checked directly). More precisely, by a union bound over all the internal vertices, and setting the parameters of the proximity test appropriately, if $\SPValueL{\SPInput}{\SPCircuit} \neq \SPOutput$ then the verfier accepts with probability $O(\SetCardinality{\VertexSet(\Graph)}(\SPMaxArity(\Graph) + \SSCVars)\cdot \SubsetSize/(\SetCardinality{\Field}-\SetCardinality{\SPSubset}))$.

\parhead{Zero knowledge}
We prove that the protocol has perfect zero knowledge by exhibiting a polynomial-time simulator that perfectly samples the view of any malicious verifier. We will assume that the simulator maintains the label sets $\Labels{v}$ in the same way as the honest verifier, but for clarity we will not state this in its description.
	
\begin{mdframed}
{\small
			
\begin{enumerate}[nolistsep]

\item \label{step:sp-draw-masks}
For every internal vertex $v$ of $\Graph$, sample $\Simulated{\RandPoly}^{v} \in \PolynomialRingIndOneXY{\Field}{\SPArity{v}}{\VariableX}{\SSCVars}{\VariableY}{\GraphInDegree[\Graph]{v}}{2\SubsetSize}$ uniformly at random. Use $\Simulated{\RandPoly}^{v}$ to answer queries to $\RandPoly_{v}$.
				
\item For every internal vertex $v$ of $\Graph$, run the $\SubsetSize^{k}$-strong ZK sumcheck simulator on input $(\Field, \SPArity{v} + \SPVars[v] + k, \SubsetSize, \SPSubset, \cdot)$, and use it to answer queries to $\Proof_{v}$ throughout. Recall that the behavior of each simulator does not depend on the claim being proven until after the first simulated message, so we can choose these later.

\item \mbox{For every internal vertex $v$ of $\Graph$ taken in (any) topological order, letting $\aodeg \DefineEqual \GraphOutDegree[\Graph]{v}$:}
\begin{enumerate}[nolistsep]

\item Receive $\Malicious{\alpha}_{1}, \dots, \Malicious{\alpha}_{\SetCardinality{\Labels{v}}}$ from the verifier.

\item Using the subsimulator for $v$, simulate the strong ZK sumcheck protocol on the claim
\begin{equation*}
\text{``}\quad
  \sum_{j=1}^{\SetCardinality{\Labels{v}}} \Malicious{\alpha}_j \RLDSPValueL{\SPInput}{v}(\vec{\gamma}_j)
= \sum_{j=1}^{\SetCardinality{\Labels{v}}} \Malicious{\alpha}_j a_j
\quad\text{''.}
\end{equation*}
The subsimulator will query the oracle $\SCPoly$ at a single location $\vec{c} = (\vec{c}_{1}, \vec{c}_{2}, \vec{c}_{3})$ with $\vec{c}_{1} \in \Field^{\SPArity{v}}$, $\vec{c}_{2} \in \Field^{\SPVars[v]}$, and $\vec{c}_{3} \in \Field^{\SSCVars}$. Reply with the value
\begin{align*}
\sum_{j=1}^{\SetCardinality{\Labels{v}}}
\alpha_{j}
\Big(
&\Lagrange{\SSCSubset^{\SSCVars}}(\vec{0}, \vec{c}_{3})
\cdot
\Lagrange{\SPSubset^{\SPArity{v}}}(\vec{\gamma}_{j}, \vec{c}_{1}) 
\cdot
\SPPoly[v](\vec{c}_{1}, \vec{c}_{2}, \Simulated{h}^{1}, \dots, \Simulated{h}^{\aodeg}) \\
&+ \Lagrange{\SCSubset^{\SPArity{v} + \SPVars[v]}}((\vec{c}_{1},\vec{c}_{2}),\vec{0})
\cdot
\ZeroPoly{\SPSubset^{\SPArity{v}}}(\vec{\gamma}_{j}) 
\cdot
\RandPoly_{v}(\vec{\gamma}_{j}, \vec{c}_{3})
\Big)
\enspace,
\end{align*}
where $\Simulated{h}^1, \dots, \Simulated{h}^{\aodeg} \in \Field$ are chosen as follows. Let $e_{1} = (v, u_{1}), \dots, e_{\aodeg} = (v, u_{\aodeg})$ be the outgoing edges of $v$ (with multiplicity). For every $k \in \{1, \dots, \aodeg\}$, letting $\vec{c}_{k} \DefineEqual (\Restrict{\vec{c}_{1}}{\SPFreeProjection[e_{k}]},\Restrict{\vec{c}_{2}}{\SPSumProjection[e_{k}]})$:
\begin{enumerate}[nolistsep]
  \item if $u_{k}$ is a leaf vertex, then $\Simulated{h}^{k} \DefineEqual \SPLeaf[u_{k}](\vec{c}_{k})$.
  \item if $u_{k}$ is an internal vertex and $(\vec{c}_{k}, h) \in \Labels{u_{k}}$ for some $h \in \Field$, $\Simulated{h}^{k} \DefineEqual h$.
  \item if $u_{k}$ is an internal vertex and $(\vec{c}_{k}, h) \notin \Labels{u_{k}}$ for all $h \in \Field$, sample $\Simulated{h}^{k}$ at random and add $(\vec{c}_{k}, \Simulated{h}^{k})$ to $\Labels{u_{k}}$.
\end{enumerate}

    \item Send $\Simulated{h}^{1}, \dots, \Simulated{h}^{\aodeg}$ to the verifier.
    
  \end{enumerate}
\end{enumerate}
}%
\end{mdframed}
	
	The view of the verifier in a real execution is composed of the messages from the prover during each sumcheck protocol, the values $\RLDSPValueL{\SPInput}{v}(\vec{\gamma})$ for every internal vertex $v$ in $\VertexSet$ and $(\vec{\gamma}, a) \in \Labels{v}$, and the verifier's queries to the oracles $\RandPoly_{v}$ and $\Proof_{v}$ for every internal vertex $v$ in $\VertexSet$.
	
	Any $\SCStrength$-query malicious verifier $\Malicious{\Verifier}$ may query any $\RandPoly_{v}$ at strictly fewer than $\SCStrength$ points. By \corref{cor:partial-sum-indep-vars}, $S_{v}(\vec{x}) := \sum_{\vec{y} \in \SSCSubset^{\SSCVars}} \RandPoly_{v}(\vec{x}, \vec{y})$ is uniformly random in $\PolynomialRingIndOne{\Field}{\SPArity{v}}{\VariableX}{\GraphInDegree[\Graph]{v}}$, even conditioned on the values of the fewer than $\SubsetSize^{\SSCVars} \leq \SCStrength$ queries made by $\Malicious{\Verifier}$ to $\RandPoly_{v}$. Therefore any string $\big(\RLDSPValueL{\SPInput}{v}(\vec{\gamma}_{1}), \dots, \RLDSPValueL{\SPInput}{v}(\vec{\gamma}_{\ell})\big)$ for $\ell \leq \GraphInDegree[\Graph]{v}$ and distinct $\vec{\gamma}_{1}, \dots, \vec{\gamma}_{\ell} \in (\Field - \SPSubset)^{\SPArity{v}}$ is identically distributed to a uniformly random string in $\Field^{\ell}$. Observe that for every internal vertex $v$ in $\VertexSet$, $\SetCardinality{\Labels{v}} \leq \GraphInDegree[\Graph]{v}$, and the prover sends $\RLDSPValueL{\SPInput}{v}(\vec{\gamma})$ for each $(\vec{\gamma}, a) \in \Labels{v}$, so all of these values are uniformly random in $\Field$, and thus identically distributed to the values $\Simulated{h}^i$ the simulator sends.
	
	Clearly $\Simulated{\RandPoly}^{v}$ and $\RandPoly_{v}$ are identically distributed for every internal vertex $v$ in $\VertexSet$. The perfect zero knowledge property of the strong ZK sumcheck simulator guarantees that the simulation of $\Proof_{v}$ and the messages sent during the sumcheck protocol is perfect given a single query to the oracle $\SCPoly$. We simulate this query by substituting values $\Simulated{h}^{i}$ in place of $\RLDSPValueL{\SPInput}{u_{i}}(\vec{c}_{i})$, which we argue above are identically distributed.
	
	It follows from the description and the efficiency of the subsimulator that the simulator runs in time $\poly(\BitSize{\SPCircuit}) \cdot \BitSize{\SPInput} + \SetCardinality{\VertexSet(\Graph)} (\SPMaxArity(\Graph) + k)(\SubsetSize \QueryComplexity_{\Malicious{\Verifier}} \SetCardinality{\SPSubset} + \SubsetSize^{3} \QueryComplexity_{\Malicious{\Verifier}}^{3}) \cdot \poly(\log \SetCardinality{\Field})$, provided we use the algorithm of \corref{cor:efficient-poly-simulator} for \stepref{step:sp-draw-masks}.
\end{proof}

\subsection{The case of sum-product satisfaction}
\label{sec:pzk-sum-product-satisfaction}

To make the sum-product circuit satisfaction protocol zero knowledge, in addition to the above considerations we must also avoid leaking information about the witness $\SPInput$. This we achieve using similar techniques to those seen previously: for each leaf vertex $w$, rather than directly sending the auxiliary input polynomial $\SPAuxLeaf[w](\vec{\VariableX})$, the prover sends a polynomial $\SPAuxLeaf[w]'(\vec{\VariableX},\vec{\VariableY})$ chosen randomly from the set of low-degree polynomials summing over $\SPSubset^{k}$ to $\SPAuxLeaf[w](\vec{\VariableX})$ for some appropriately-chosen $k$. This acts as a perfectly hiding commitment to the witness.

We show that we can efficiently construct a sum-product circuit $\SPCircuit'$ which outputs $\SPOutput$ on input $(\SPInput,\SPAuxLeaf[w]')$ if and only if the original sum-product circuit outputs $\SPOutput$ on input $(\SPInput,\SPAuxLeaf[w])$. We then obtain our zero-knowledge protocol by running the protocol of \thmref{thm:pzk-for-SPCE} on $\SPCircuit'$, implementing the queries to the input as queries to the proof (with self-correction).

\begin{theorem}[PZK IPCP for $\SPCSRelation$]
	\label{thm:pzk-for-SPCS}
	For every query bound function $\SCStrength(n)$, the relation $\SPCSRelation$ has a (public-coin and non-adaptive) Interactive PCP that is perfect zero knowledge against all $\SCStrength$-query malicious verifiers. In more detail, letting $\alpha \DefineEqual \log \SCStrength / \log \SetCardinality{\SPSubset}$:
	{\small
		\begin{equation*}
		\SPCSRelation \in
		\PZKIPCP
		\left[
		\begin{tabular}{rl}
		soundness error: & $O(\InternalVertexSetDegree\SPLeafDegree \cdot \GraphMaxInDegree{\Graph} \cdot (\SPMaxArity(\Graph) + \alpha) \cdot \SetCardinality{\VertexSet(\Graph)}/\SetCardinality{\Field})$ \\
		round complexity: & $O( \GraphDepth[]{\Graph} \cdot (\SPMaxArity(\Graph) + \alpha))$ \\[1mm]
		proof length: & $O(\SetCardinality{\VertexSet(\Graph)} \cdot \SetCardinality{\Field}^{\SPMaxArity(\Graph) + 2\alpha})$ \\
		query complexity: & $\SetCardinality{\VertexSet(\Graph)} \cdot \poly(\log \SetCardinality{\Field}, \SPMaxArity(\Graph), \alpha, \InternalVertexSetDegree, \SPLeafDegree, \GraphMaxInDegree{\Graph})$ \\[1mm]
		prover time: & $\poly(\BitSize{\SPCircuit}, \BitSize{\SPInput}, \BitSize{\SPAuxInput}, \SetCardinality{\Field}^{\SPMaxArity(\Graph) + \alpha})$ \\
		verifier time: &  $\poly(\BitSize{\SPCircuit}, \BitSize{\SPInput})$ \\[1mm]
		verifier space: & $O((\SPMaxArity(\Graph) + \alpha) \cdot \Width{\Graph} \cdot \GraphMaxInDegree{\Graph} \cdot \log \SetCardinality{\Field} + \log \BitSize{\SPCircuit} + \MaxSpace{\SPInput})$ \\
		simulator overhead: & $\poly(\BitSize{\SPCircuit}, \alpha) \cdot (\BitSize{\SPInput} +  \QueryComplexity_{\Malicious{\Verifier}}^{3})$.
		\end{tabular}
		\right]
		\enspace.
		\end{equation*}}
\end{theorem}

\begin{proof}
Fix $k \DefineEqual \log \SCStrength / \log \SetCardinality{\SPSubset}$. Let $\SPCircuit = \SPCircuitTuple$ be a sum-product circuit and let $\Graph = (\VertexSet, \EdgeSet, \SPFreeProjection, \SPSumProjection)$ be its ari-graph. Let $\SPInput$ be a partial mapping of the leaf vertices of $\Graph$ to arithmetic circuits. We construct the new sum-product circuit $\SPCircuit' = (\Field, \SPSubset, \InternalVertexSetDegree, \max\{\SPLeafDegree,2\SetCardinality{\SPSubset}\}, \Graph', \SPPoly')$ where:
\begin{itemize}

  \item the new ari-graph $\Graph' = (\VertexSet', \EdgeSet', \SPFreeProjection', \SPSumProjection')$ extends $\Graph$ by adding a new vertex $v_{w}$ and a new edge $(w, v_{w})$ for each leaf vertex $w$ not in the domain of $\SPInput$ as follows
\begin{align*}
\VertexSet' &\DefineEqual \VertexSet \cup \{ v_{w} : w \in \VertexSet \text{ is a leaf} \} \enspace, \\
\EdgeSet' &\DefineEqual \EdgeSet \cup \{ (w, v_{w}) : w \in \VertexSet \text{ is a leaf} \} \enspace, \\
\SPFreeProjection[e]' &\DefineEqual
  \begin{cases}
    \SPFreeProjection[e] & \text{ if $e \in \EdgeSet$} \\
    \{1, \dots, \SPArity{w}\} & \text{ if $e = (w, v_{w})$ and $w$ is a leaf in $\VertexSet$}
  \end{cases} \enspace, \\
\SPSumProjection[e]' &\DefineEqual
  \begin{cases}
    \SPSumProjection[e] & \text{ if $e \in \EdgeSet$} \\
    \{1, \dots, k\} & \text{ if $e = (w, v_{w})$ and $w$ is a leaf in $\VertexSet$}
  \end{cases} \enspace.
\end{align*}
Note that the vertex and edge sets at most double in size, $\GraphDepth{\Graph'} = \GraphDepth{\Graph} + 1$, $\SPArity{v_{w}} = \SPArity{w} + k$ and $\SPVars[v_{w}] = k$.

  \item $\SPPoly'$ is such that $\SPPoly[v]' := \SPPoly[v]$ for every internal vertex $v$ in $\VertexSet$, and $\SPPoly[v]'$ is the univariate polynomial $\VariableX$ for every leaf vertex $v$ in $\VertexSet$.

\end{itemize}
Before describing the protocol, we prove a claim that relates the two sum-product circuits $\SPCircuit$ and $\SPCircuit'$.
	
\begin{uclaim}
For every $\SPOutput \in \Field$, $(\SPCircuit, \SPOutput, \SPInput) \in \Language(\SPCSRelation)$ if and only if $(\SPCircuit', \SPOutput, (\SPInput, \SPAuxInput')) \in \SPCELanguage$ for some input $\SPAuxInput'$ for $\SPCircuit'$ such that $\SPAuxLeaf[v_{w}]' \in \PolynomialRingIndOneXY{\Field}{\SPArity{w}}{\VariableX}{k}{\VariableY}{\SPLeafDegree}{2\SetCardinality{\SPSubset}}$ for every leaf $w$ in $\VertexSet$.
\end{uclaim}

\begin{proof}[Proof of claim]
First suppose that $(\SPCircuit, \SPOutput, \SPInput) \in \Language(\SPCSRelation)$, so there exists an auxiliary input $\SPAuxInput$ for $\SPCircuit$ such that $\SPValueL{\SPInput,\SPAuxInput}{\SPCircuit} = \SPOutput$. Define $\SPAuxInput'$ to be the auxiliary input for $\SPCircuit'$ such that $\SPAuxLeaf[v_{w}]'(\vec{\VariableX}, \vec{\VariableY}) \DefineEqual \Lagrange{\SPSubset^{k}}(\vec{\VariableY}, \vec{0}) \cdot \SPLeaf[w](\vec{\VariableX})$ for every leaf $w$ in $\VertexSet$ in the domain of $\SPAuxInput$. Note that for each such leaf $w$ in $\VertexSet$, $\SPAuxLeaf[v_{w}]'$ respects the desired degree bounds and, moreover,
$
  \SPValueL{\SPInput,\SPAuxInput'}{w}(\vec{\VariableX})
= \sum_{\vec{\beta} \in \SPSubset^{k}} \Lagrange{\SPSubset^{k}}(\vec{\beta}, \vec{0}) \cdot \SPAuxLeaf[w](\vec{\VariableX})
= \SPAuxLeaf[w](\vec{\VariableX})
= \SPValueL{\SPInput,\SPAuxInput}{w}(\vec{\VariableX})
$.
By construction of $\Graph'$, we deduce that $\SPValueL{\SPInput, \SPAuxInput'}{\SPCircuit'} = \SPOutput$, so that $(\SPCircuit', \SPOutput, (\SPInput, \SPAuxInput')) \in \SPCELanguage$.

Next suppose that $(\SPCircuit', \SPOutput, (\SPInput,\SPAuxInput')) \in \SPCELanguage$ for some input $\SPAuxInput'$ for $\SPCircuit'$ such that $\SPAuxLeaf[v_{w}]' \in \PolynomialRingIndOneXY{\Field}{\SPArity{w}}{\VariableX}{k}{\VariableY}{\SPLeafDegree}{2\SetCardinality{\SPSubset}}$ for every leaf $w$ in $\VertexSet$ not in the domain of $\SPInput$. Define $\SPAuxInput$ to be the input for $\SPCircuit$ such that $\SPAuxLeaf[w](\vec{\VariableX}) := \sum_{\beta \in \SPSubset^{k}} \SPAuxLeaf[v_{w}]'(\vec{\VariableX}, \vec{\beta})$ for each leaf $w$ in $\VertexSet$. Note that for each leaf $w$ in $\VertexSet$ not in the domain of $\SPInput$, $\SPAuxLeaf[w]$ has individual degree at most $\SPLeafDegree$ and, moreover,
$
  \SPValueL{\SPInput,\SPAuxInput}{w}(\vec{\VariableX})
= \sum_{\beta \in \SPSubset^{k}} \SPAuxLeaf[v_{w}]'(\vec{\VariableX}, \vec{\beta})
= \SPValueL{\SPInput,\SPAuxInput'}{w}(\vec{\VariableX})
$.
By construction of $\Graph'$, this implies that $\SPValueL{\SPInput,\SPAuxInput'}{\SPCircuit} = \SPOutput$, so that $(\SPCircuit, \SPOutput, \SPInput) \in \Language(\SPCSRelation)$.
\end{proof}
	
The protocol proceeds as follows. The prover and verifier receive a SPCS instance $(\SPCircuit, \SPOutput, \SPInput)$ as input, and the prover additionally receives an auxiliary input $\SPAuxInput$ for $\SPCircuit$ that is a valid witness for $(\SPCircuit, \SPOutput, \SPInput)$. Both use $\SPCircuit$ to construct the new sum-product circuit $\SPCircuit'$ from the above claim; in addition, the prover uses $\SPAuxInput$ to sample an input $\SPAuxInput'$ for $\SPCircuit'$ by choosing, for each leaf $w$ in the domain of $\SPAuxInput$, a polynomial $\SPAuxLeaf[v_{w}]' \in \PolynomialRingIndOneXY{\Field}{\SPArity{w}}{\VariableX}{k}{\VariableY}{\SPLeafDegree}{2\SetCardinality{\SPSubset}}$ uniformly at random conditioned on $\sum_{\vec{\beta} \in \SPSubset^{k}} \SPAuxLeaf[v_{w}]'(\vec{\alpha}, \vec{\beta}) = \SPAuxLeaf[w](\vec{\alpha})$ for all $\vec{\alpha} \in \Field^{\SPArity{w}}$. The prover and verifier then engage in the zero knowledge Interactive Probabilistically Checkable Proof for the language $\SPCELanguage$ (see \thmref{thm:pzk-for-SPCE}) on input $(\SPCircuit', \SPOutput, (\SPInput,\SPAuxInput'))$, with the prover appending the auxiliary input $\SPAuxInput'$ to the proof oracle. At the end of the protocol the verifier needs to make a single query $\vec{\gamma}_{v_{w}}$ to $\SPAuxLeaf[v_{w}]'$ for each leaf vertex $w$ in the domain of $\SPLeaf$, so the verifier tests that $\SPAuxLeaf[v_{w}]'$ is close to the evaluation of a polynomial in $\PolynomialRingIndOneXY{\Field}{\SPArity{w}}{\VariableX}{k}{\VariableY}{\SPLeafDegree}{2\SetCardinality{\SPSubset}}$ and then uses self-correction to read $\SPAuxLeaf[v_{w}]'(\vec{\gamma}_{v_{w}})$.

Completeness is straightforward to argue; we only discuss soundness and then zero knowledge.
	
Suppose that $(\SPCircuit, \SPOutput, \SPInput) \notin \Language(\SPCSRelation)$, and let $\SPAuxLeafM$ be the input for $\SPCircuit'$ that the prover has appended to the proof oracle. If there exists a leaf vertex $w$ in the domain of $\SPAuxLeafM$ such that $\SPAuxLeafM[v_{w}]$ is more than $\delta$-far from a polynomial in $\PolynomialRingIndOneXY{\Field}{\SPArity{w}}{\VariableX}{k}{\VariableY}{\SPLeafDegree}{2\SetCardinality{\SPSubset}}$, then the verifier accepts with probability at most $\epsilon$. So suppose that this is not the case, and let $\SPAuxLeafM[v_{w}]'$ be the unique polynomial in $\PolynomialRingIndOneXY{\Field}{\SPArity{w}}{\VariableX}{k}{\VariableY}{\SPLeafDegree}{2\SetCardinality{\SPSubset}}$ that is $\delta$-close to $\SPAuxLeafM[v_{w}]$. Using self-correction, the verifier obtains $\SPAuxLeafM[v_{w}]'(\vec{\gamma}_{v_{w}})$ with probability at least $1 - \epsilon$. By a union bound, the probability that the verifier learns all the values correctly is at least $1-\epsilon \SetCardinality{\VertexSet(\Graph)}$. By the claim, $(\SPCircuit, \SPOutput, \SPInput) \notin \Language(\SPCSRelation)$ implies that $(\SPCircuit', \SPOutput, (\SPInput, \SPAuxLeafM')) \notin \SPCELanguage$, so if all of the verifier's queries are answered according to $\SPAuxLeafM'$, the soundness of the protocol for $\SPCELanguage$ implies that the verifier accepts with probability at most $O(\InternalVertexSetDegree\SPLeafDegree \cdot \GraphMaxInDegree{\Graph} \cdot (\SPMaxArity(\Graph) + \log \SCStrength) \cdot \SetCardinality{\VertexSet(\Graph)}/\SetCardinality{\Field})$. Setting $\epsilon$ and $\delta$ appropriately yields the claimed soundness error.
	
Perfect zero knowledge follows from \corref{cor:partial-sum-indep-vars}: if a verifier makes fewer than $\SetCardinality{\SPSubset}^{k}$ queries to the proof, then each $\SPAuxLeaf[v]'$ is indistinguishable from a random polynomial. The simulator for this protocol simply runs the simulator from \thmref{thm:pzk-for-SPCE} on $(\SPCircuit', \SPOutput, (\SPInput, \SPAuxInput'))$, and uses the algorithm of \corref{cor:efficient-poly-simulator} to simulate its queries to each $\SPAuxLeaf[v]'$.
\end{proof}

\doclearpage
\section{Zero knowledge for polynomial space}
\label{sec:zk-pspace}

We present an efficient reduction from problems decidable in polynomial space ($\PSPACE$) to sum-product circuit evaluation problems ($\SPCELanguage$, \defref{def:circuit-evaluation}). The reduction yields perfect zero knowledge IPCPs for $\PSPACE$, via our construction of perfect zero knowledge IPCPs for $\SPCELanguage$ (\thmref{thm:pzk-for-SPCE}), resolving an open problem of \cite{BenSassonCFGRS16}.

\medskip
Our starting point is the language of true quantified boolean formulas (TQBFs), which is $\PSPACE$-complete:

\begin{definition}
\label{def:tqbf}
Let $\TQBFLanguage$ be the language of quantified boolean formulas $\QFormula = \Quantifier_{1} x_{1}  \cdots \Quantifier_{\NumVars} x_{\NumVars} \Formula(x_{1},\dots,x_{\NumVars})$, with $\Quantifier_{i} \in \{\forall,\exists\}$, that evaluate to true. We denote by $\NumVars$ the number of variables and by $\NumClauses$ the number of clauses in $\Formula$.
\end{definition}

The theorem below is a `zero knowledge analogue' of Shamir's protocol for $\TQBFLanguage$ \cite{Shamir92}, sharing all its key features except that it is an Interactive PCP rather than an Interactive Proof. We prove that $\TQBFLanguage$ has a public-coin Interactive PCP that is perfect zero knowledge, with exponential proof length and polynomial query complexity; also, like Shamir's protocol, the number of rounds is $O(\NumVars^{2})$, the prover runs in space $\poly(\NumClauses)$ and the verifier in time $\poly(\NumClauses)$.

\begin{theorem}[PZK IPCP for $\PSPACE$]
\label{thm:pzk-for-TQBF}
For every query bound function $\QueryBound(n)$, the $\PSPACE$-complete language $\TQBFLanguage$ has a (public-coin and non-adaptive) Interactive PCP that is perfect zero knowledge against all $\QueryBound$-query malicious verifiers. In more detail:
{\small
\begin{equation*}
\TQBFLanguage \in
\PZKIPCP
\left[
\begin{tabular}{rl}
soundness error: & $1/2$ \\
round complexity: & $O(\NumVars \cdot (\NumVars + \log \QueryBound))$ \\[1mm]
proof length: & $(\BitSize{\QFormula} + \log \QueryBound)^{O(\NumVars + \log \QueryBound)}$ \\
query complexity: & $\poly(\BitSize{\QFormula}, \log \QueryBound)$ \\[1mm]
prover time: & $(\BitSize{\QFormula} + \log \QueryBound)^{O(\NumVars + \log \QueryBound)}$ \\
verifier time: &  $\poly(\NumVars, \log \QueryBound) + O(\NumClauses)$ \\[1mm]
verifier space: & $O((\NumVars + \log \QueryBound) \cdot \log \BitSize{\QFormula})$ \\[1mm]
simulator overhead: & $\poly(\NumVars + \log \QueryBound) \cdot \QueryComplexity_{\Malicious{\Verifier}}^{3}$.
\end{tabular}
\right]
\enspace.
\end{equation*}}
\end{theorem}

We first state and prove the reduction from $\TQBFLanguage$ to $\SPCELanguage$ (\lemref{lem:qbf-to-SPCE} below), and then prove the theorem. The reduction is based on Shen's arithmetization of QBFs \cite{Shen92} and (at least implicitly) his idea of degree reduction. (In particular, we do not rely on Shamir's restriction to `simple' QBFs \cite{Shamir92}.) We use Shen's arithmetization in order to ensure that the arithmetized expression only takes boolean values.

\begin{lemma}[$\TQBFLanguage \to \SPCELanguage$]
\label{lem:qbf-to-SPCE}
There exists a polynomial-time function $f$ such that, for every quantified boolean formula $\QFormula$ and prime $p$, $f(\QFormula, p) \in \SPCELanguage$ if and only if $\QFormula$ is true. Moreover, if $\QFormula$ has $\NumVars$ variables and $\NumClauses$ clauses then $f(\QFormula, p)$ is a sum-product circuit instance $\SPCircuit = \SPCircuitTuple$ with $\SetCardinality{\Field} = p$, $\SetCardinality{\SPSubset}=\Theta(1)$, $\InternalVertexSetDegree=\Theta(1)$, $\SPLeafDegree = \Theta(c)$, $\SetCardinality{\VertexSet(\Graph)}=\Theta(\NumVars)$, $\SPMaxArity(\Graph) = \Theta(\NumVars)$, $\GraphMaxInDegree{\Graph} = \Theta(1)$, $\Width{\Graph} = \Theta(1)$, $\MaxSpace{\SPInput} = O(\log \BitSize{\QFormula})$.
\end{lemma}

\begin{proof}
Let $\QFormula$ be a quantified boolean formula that, without loss of generality, has the following `regular' form:
\begin{equation*}
\QFormula =
\forall x_{1} \exists x_{2}
\cdots
\forall x_{\NumVars-1} \exists x_{\NumVars}\,
  \Formula(x_{1}, x_{2}, \dots, x_{\NumVars-1}, x_{\NumVars})
\enspace,
\end{equation*}
where $\Formula$ is a 3-CNF formula with $\NumClauses$ clauses, and $\NumVars$ is even. This regular form is achievable with only constant multiplicative overheads in the number of variables and clauses. We arithmetize the formula and quantifiers.
\begin{itemize}

  \item The arithmetization of a 3-CNF formula $\Formula$ with variables $z_{1},\dots,z_{\NumVars}$ and clauses $\Clause_{1},\dots,\Clause_{\NumClauses}$ is the polynomial $\LD{\Formula}$ of total degree at most $3\NumClauses$ given by:
$
\LD{\Formula}(\VariableZ_{1},\dots,\VariableZ_{\NumVars})
\DefineEqual
  \prod_{\Clause \in \Formula}
  (1-
      \prod_{\{i : z_{i} \in \Clause \}} (1 - \VariableZ_{i})
      \cdot
      \prod_{\{i : \bar z_{i} \in \Clause \}} \VariableZ_{i}
   )
$. Note that $\LD{\Formula}(x_{1},\dots,x_{\NumVars}) = \Formula(x_{1},\dots,x_{\NumVars})$ for all boolean $x_{1},\dots,x_{\NumVars}$ (i.e., they agree on the boolean hypercube).

  \item The arithmetization of the two quantifiers is as follows: $\forall x \, \Formula(x)$ maps to $\prod_{x \in \Bits} \LD{\Formula}(x)$ and $\exists x \, \Formula(x)$ maps to $\coprod_{x \in \Bits} \LD{\Formula}(x) \DefineEqual \big(1 - (1 - \LD{\Formula}(0))(1 - \LD{\Formula}(1))\big)$. These arithmetic expressions have the same value as the boolean expressions (over any field); in particular they are $0$ or $1$.

\end{itemize}
For any prime $p$, consider the following construction.
\begin{itemize}

  \item Use $\NumVars,\NumClauses,p$ to construct a sum-product circuit instance $\SPCircuit = \SPCircuitTuple$ where
\begin{equation*}
\Field \DefineEqual \Field_{p}\,;\;
\SPSubset \DefineEqual \Bits\,;\;
\InternalVertexSetDegree \DefineEqual 4\,;\;
\SPLeafDegree \DefineEqual 3 \NumClauses\,;\;
\end{equation*}
the ari-graph $\Graph \DefineEqual (\VertexSet, \EdgeSet)$ is defined as follows
\begin{align*}
\VertexSet &\DefineEqual
	 \{ v_{i} \}_{i \in \{0,\dots,\NumVars\}} \\
\EdgeSet &\DefineEqual
	 \{ e_{i}, e_{i}' = (v_{i}, v_{i+1}) \}_{i \in \{0,\dots,\NumVars-1\}} \\
\SPFreeProjection[e] &\DefineEqual
\{1, \dots, i\} \quad \text{ if $e = e_{i}$ or $e=e_{i}'$ for some $i \in \{0,\dots,\NumVars-1\}$,} \\
\SPSumProjection[e] &\DefineEqual \begin{cases}
	\{1\}   & \text{ if $e = e_{i}$ for some $i \in \{0,\dots,\NumVars-1\}$ or} \\
	\{2\}   & \text{ if $e = e'_{i}$ for some $i \in \{0,\dots,\NumVars-1\},$} \\
\end{cases} \\
\SPPoly[v_{i}] &\DefineEqual \begin{cases}
	(1 - \VariableX_{i}) \VariableX_{i} \cdot \VariableZ \VariableZ' & \text{ if $i \in \{0,\dots,\NumVars-1\}$ is even or} \\
	(1 - \VariableX_{i}) \VariableX_{i} \cdot \big(1-(1-\VariableZ)(1-\VariableZ')\big) & \text{ if $i \in \{0,\dots,\NumVars-1\}$ is odd,} \\
\end{cases}
\end{align*}
where $\VariableZ$, $\VariableZ'$ in both cases correspond to the edges $e_{i}, e'_{i}$.

  \item Construct an input $\SPInput$ for $\SPCircuit$ that maps $v_{\NumVars}$ to the polynomial $\LD{\Formula}$ in $\PolynomialRingIndOne{\Field_{p}}{\NumVars}{\VariableX}{3\NumClauses}$ that equals the arithmetization of $\Formula$.

\end{itemize}
See \figref{fig:shamir-2}, \figref{fig:shamir-4}, \figref{fig:shamir-n} for diagrams of this sum-product circuit and input for $2$, $4$, $\NumVars$ variables respectively.

We claim that for every even $i$ in $\{0,\dots,\NumVars\}$ and $x_{1}, \ldots, x_{i} \in \Bits$:
\begin{equation}
\label{eqn:tqbf-sum-product}
  \SPValueL{\SPInput}{v_{i}}(x_{1}, x_{2} \dots, x_{i-1}, x_{i})
= \prod_{x_{i+1} \in \Bits} \coprod_{x_{i+2} \in \Bits}
  \cdots
  \prod_{x_{\NumVars-1} \in \Bits} \coprod_{x_{\NumVars} \in \Bits}
    \LD{\Formula}(x_{1}, x_{2}, \dots, x_{\NumVars-1}, x_{\NumVars})
\enspace.
\end{equation}
We argue the equality by induction on $i$. When $i = \NumVars$, \equnref{eqn:tqbf-sum-product} is $\SPValueL{\SPInput}{v_{n}} = \LD{\Formula}$, which holds by definition; next, we assume the equality for $i+2$ and prove it for $i$. By construction, for every $x_{1}, \ldots, x_{i} \in \Bits$,
\begin{align*}
  \SPValueL{\SPInput}{v_{i+1}}(x_{1}, \dots, x_{i})
&= \quad \sum_{\mathclap{x_{i+1},x_{i+1}' \in \Bits}} \;
    (1 - x_{i+1}) x_{i+1}'
    \cdot
    (1 - (1 - \SPValueL{\SPInput}{v_{x_{i+2}}}(x_{1}, \dots, x_{i}, x_{i+1}))(1 - \SPValueL{\SPInput}{v_{x_{i+2}}}(x_{1}, \dots, x_{i}, x_{i+1}'))) \\
&= (1 - (1 - \SPValueL{\SPInput}{v_{x_{i+2}}}(x_{1}, \dots, x_{i}, 0))(1 - \SPValueL{\SPInput}{v_{x_{i+2}}}(x_{1}, \dots, x_{i}, 1))) \\
&= \coprod_{x_{i+2} \in \Bits}
\cdots
\prod_{x_{\NumVars-1} \in \Bits} \coprod_{x_{\NumVars} \in \Bits}
\LD{\Formula}(x_{1}, x_{2}, \dots, x_{\NumVars-1}, x_{\NumVars})
\enspace,
\end{align*}
where the second equality follows by the inductive assumption. Then, by a similar argument,
\begin{align*}
   \SPValueL{\SPInput}{v_{i}}(x_{1}, \dots, x_{i})
&= \SPValueL{\SPInput}{v_{i+1}}(x_{1}, \dots, x_{i}, 0)
     \cdot
   \SPValueL{\SPInput}{v_{i}}(x_{1}, \dots, x_{i}, 1) \\
&= \prod_{x_{i+1} \in \Bits} \coprod_{x_{i+2} \in \Bits}
\cdots
\prod_{x_{\NumVars-1} \in \Bits} \coprod_{x_{\NumVars} \in \Bits}
\LD{\Formula}(x_{1}, x_{2}, \dots, x_{\NumVars-1}, x_{\NumVars})
\enspace,
\end{align*}
which completes the induction.

We now turn to the reduction's completeness and soundness. The key fact, which immediately follows from \equnref{eqn:tqbf-sum-product} and the arithmetization's properties, is that the quantified boolean formula $\QFormula$ is true if and only if $\SPValueL{\SPInput}{\SPCircuit} = \SPValueL{\SPInput}{v_{x_{1}}}$ is $1$ (over any field $\Field$). Thus, regardless of the choice of prime $p$, $\QFormula$ is true if and only if $(\SPCircuit, 1, \SPInput) \in \SPCELanguage$.
\end{proof}

\begin{proof}[Proof of \thmref{thm:pzk-for-TQBF}]
The prover and verifier receive as input a quantified boolean formula $\QFormula$. They agree on a deterministic procedure to find a prime in the range $[\NumClauses \NumVars^{3} \log \QueryBound, 2\NumClauses \NumVars^{3} \log \QueryBound]$; this can easily be done in polynomial time by exhaustive search and brute-force primality testing. The prover and the verifier then engage in the perfect zero knowledge Interactive PCP for $\SPCELanguage$ (\thmref{thm:pzk-for-SPCE}) on the input $f(\QFormula, p)$. The soundness error of the protocol is $O(\InternalVertexSetDegree\SPLeafDegree \cdot \GraphMaxInDegree{\Graph} \cdot (\SPMaxArity(\Graph) + \log \QueryBound) \cdot \SetCardinality{\VertexSet(\Graph)}/\SetCardinality{\Field})=O(\NumClauses \NumVars^{2}\log \QueryBound / p)$; the size of $p$ ensures that the soundness error is $O(1/\NumVars) < 1/2$ for sufficiently large $\NumVars$.
\end{proof}

\begin{figure}[t!]
\centering
\begin{minipage}[b]{0.33\textwidth}
  \centering
\includegraphics[width=0.45\textwidth]{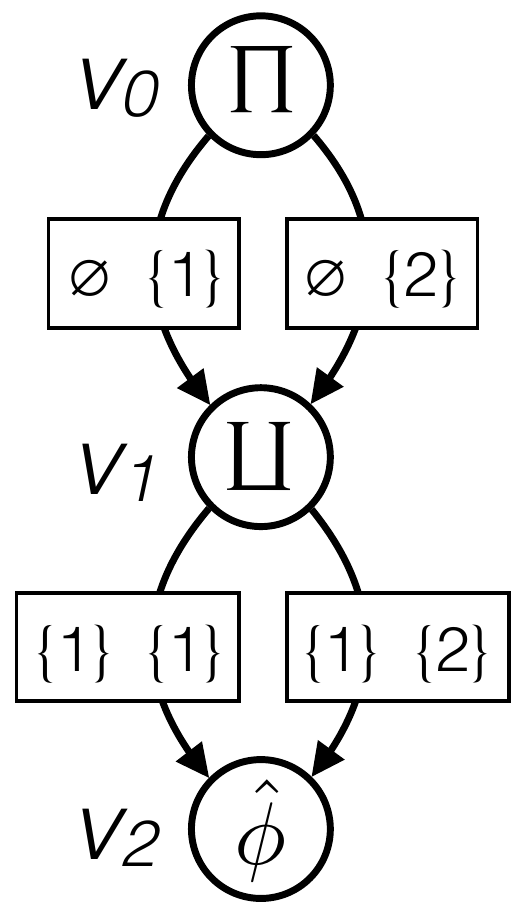}
\caption{Sum-product circuit and its input for QBF with $2$ variables.}
\label{fig:shamir-2}
\end{minipage}%
\begin{minipage}[b]{0.33\textwidth}
  \centering
\includegraphics[width=0.5\textwidth]{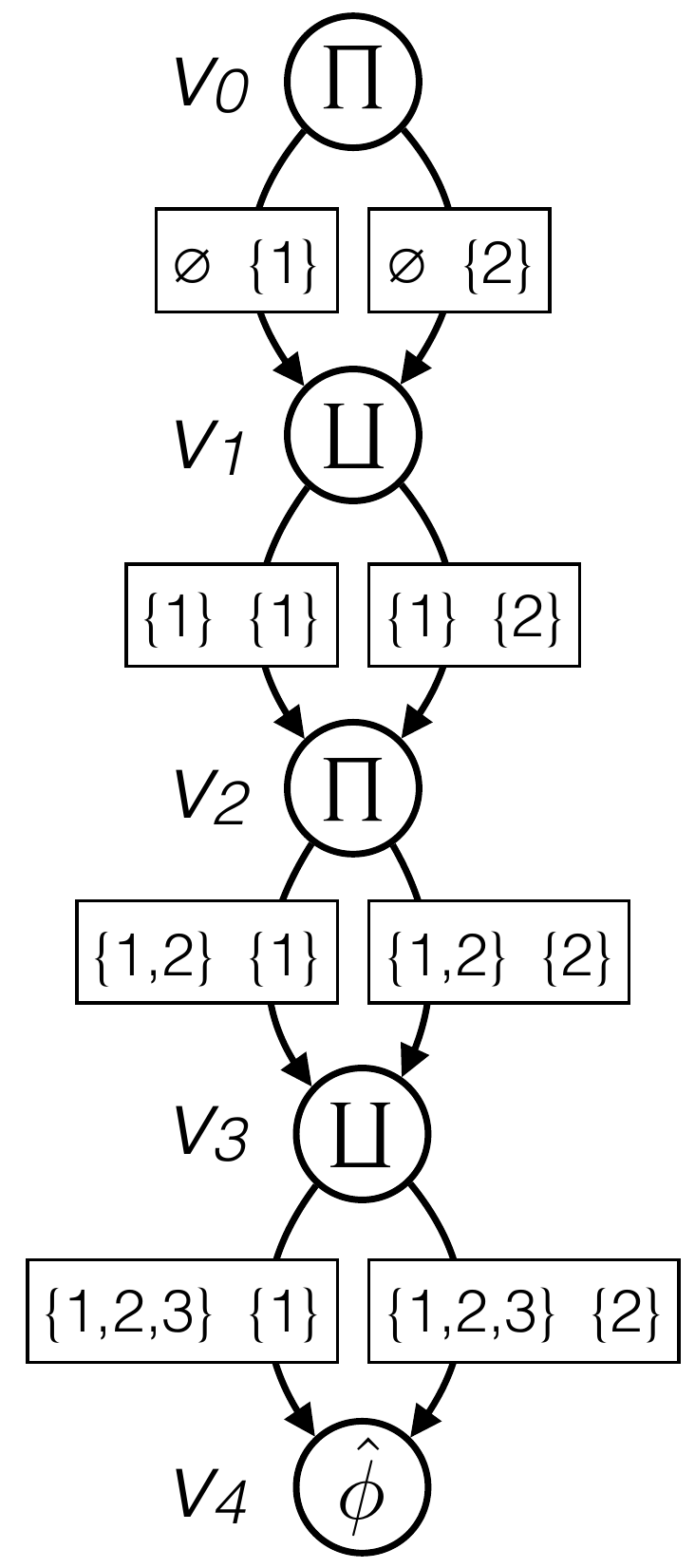}
\caption{Sum-product circuit and its input for QBF with $4$ variables.}
\label{fig:shamir-4}
\end{minipage}%
\begin{minipage}[b]{0.33\textwidth}
  \centering
\includegraphics[width=0.5\textwidth]{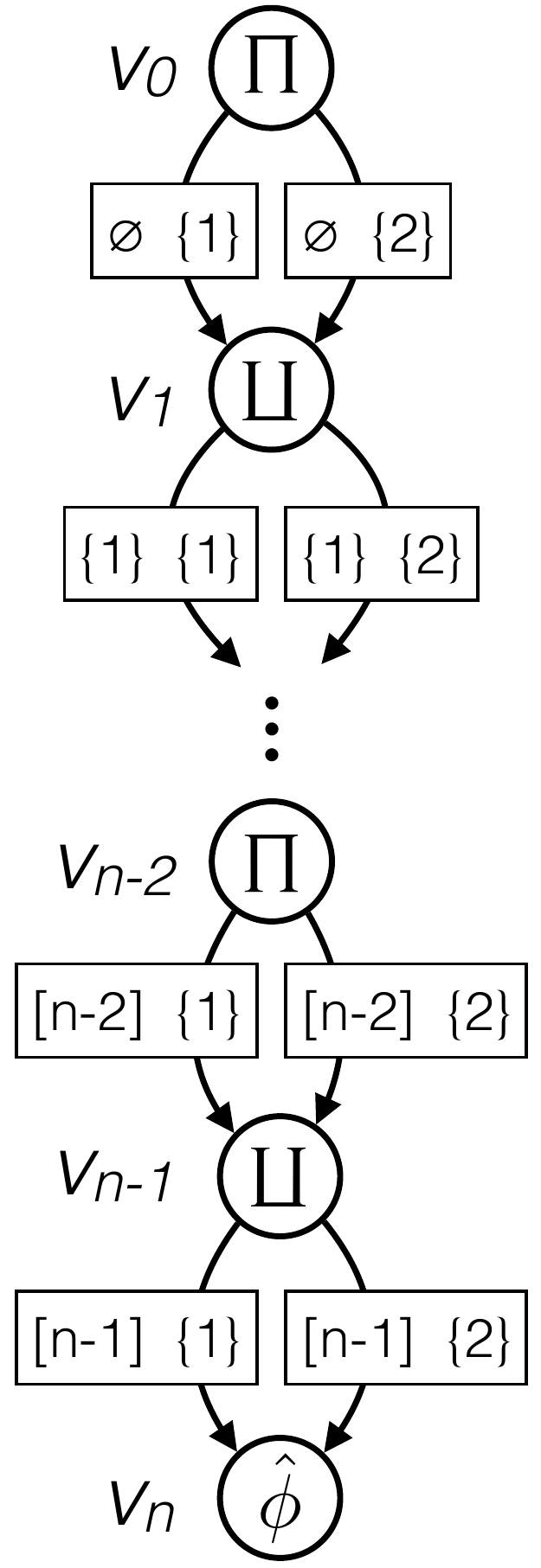}
\caption{Sum-product circuit and its input for QBF with $\NumVars$ variables.}
\label{fig:shamir-n}
\end{minipage}
\end{figure}

\doclearpage
\section{Zero knowledge for the evaluation of low-depth circuits}
\label{sec:zk-gkr}

We present a `zero knowledge analogue' of GKR's protocol for circuit evaluation \cite{GoldwasserKR15}, except that the protocol is an Interactive PCP rather than an Interactive Proof. We achieve this by exhibiting an efficient reduction from circuit evaluation (for certain circuits) to sum-product circuit evaluation problems ($\SPCELanguage$, \defref{def:circuit-evaluation}). The reduction then yields the desired zero knowledge protocol (\thmref{thm:gkr-space-uniform} below), via our construction of perfect zero knowledge Interactive PCPs for $\SPCELanguage$ (\thmref{thm:pzk-for-SPCE}). We now provide context, state the theorem, and then describe its proof.

Goldwasser, Kalai, and Rothblum \cite{GoldwasserKR15} give \emph{interactive proofs for muggles} (also known as \emph{doubly-efficient interactive proofs} \cite{ReingoldRR16}) for low-depth circuits that are sufficiently uniform.

Namely, given a language $\Language$ decidable by a family of $O(\log \Size(n))$-space uniform boolean circuits of size $\Size(n)$ and depth $\Depth(n)$, \cite{GoldwasserKR15} gives a public-coin Interactive Proof for $\Language$ where the prover runs in time $\poly(\Size(n))$ and the verifier runs in time $(n + \Depth(n)) \cdot \poly(\log \Size(n))$ and space $O(\log \Size(n))$; the number of rounds and communication complexity are $\Depth(n) \cdot \poly(\log \Size(n))$. Moreover, if the verifier is given oracle access to the low-degree extension of the circuit's input, the verifier runs only in time $\Depth(n) \cdot \poly(\log \Size(n))$; this \emph{sub-linear} running time enables applications to Interactive Proofs \emph{of Proximity} \cite{RothblumVW13}.

The theorem below provides similar features, but in addition also provides perfect zero knowledge, via an Interactive PCP instead of an Interactive Proof.

\begin{theorem}[PZK IPCP for low-depth uniform boolean circuits]
\label{thm:gkr-space-uniform}
Let $\Language$ be a language decidable by a family of $O(\log \Size(n))$-space uniform boolean circuits of size $\Size(n)$ and depth $\Depth(n)$, and let $\SCStrength(n)$ be a query bound function. Then $\Language$ has a (public-coin and non-adaptive) Interactive PCP that is perfect zero knowledge against all $\SCStrength$-query malicious verifiers. In more detail:
{\small
\begin{equation*}
\Language \in
\PZKIPCP
\left[
\begin{tabular}{rl}
soundness error: & $1/2$ \\
round complexity: & $\Depth(n) \cdot \log \SCStrength(n) \cdot \poly(\log \Size(n))$ \\[1mm]
proof length: & $\Depth(n) \cdot \poly(\Size(n), \SCStrength(n))$ \\
query complexity: & $\Depth(n) \cdot \poly(\log \Size(n), \log \SCStrength(n))$ \\[1mm]
prover time: & $\Depth(n) \cdot \poly(\Size(n), \SCStrength(n))$ \\
verifier time: &  $n \cdot \poly(\Depth(n), \log \Size(n), \log \SCStrength(n))$ \\[1mm]
verifier space: & $O(\log \Size(n) + \log \SCStrength(n)/\log\log\Size(n))$ \\[1mm]
simulator overhead: & $\poly(\Depth(n), \log \Size(n), \log \SCStrength(n)) \cdot (n +  \QueryComplexity_{\Malicious{\Verifier}}^{3})$.
\end{tabular}
\right]
\enspace.
\end{equation*}}
Moreover, if the verifier (resp. simulator) is given oracle access to the low-degree extension of the circuit's input, the verifier runs in time $\Depth(n) \cdot \poly(\log \Size(n), \log \SCStrength(n))$, and the simulator overhead is $\poly(\Depth(n), \log \Size(n), \log \SCStrength(n)) \cdot \QueryComplexity_{\Malicious{\Verifier}}^{3}$.
\end{theorem}

Our proof of the theorem follows the structure in \cite{GoldwasserKR15}. After recalling some notions for circuits (\secref{sec:layered-arithmetic-circuits}) and extending some of our definitions (\secref{sec:extended-notions}), we proceed in three steps:
\begin{itemize}

  \item \emph{Step 1} (\secref{sec:sumproduct-barebones}).
  We reduce the evaluation of a given layered arithmetic circuit to the evaluation of a sum-product circuit, when given oracles for low-degree extensions for the circuit's wiring predicates. One can view this step as casting the \emph{barebones protocol} of \cite[Section 3]{GoldwasserKR15} in the framework of sum-product circuits.
  
  \item \emph{Step 2} (\secref{sec:sumproduct-space-uniform}).
  We reduce small-space Turing machine computations to sum-product circuit evaluations.
  
  \item \emph{Step 3} (\secref{sec:proof-of-gkr-space-uniform}). We combine the previous two steps to prove the theorem, using the fact that small-space Turing machines can evaluate low-degree extensions of the wiring predicates for sufficiently uniform boolean circuits.

\end{itemize} 
Before discussing each of the above steps, we discuss consequences for Turing machine computations.

As in \cite{GoldwasserKR15}, we can derive from the above theorem a useful corollary for Turing machine computations, using their reduction from Turing machines to boolean circuits.

\begin{corollary}[PZK IPCP for Turing machines]
\label{cor:gkr-uniform}
Let $\Language$ be a language decidable by a (deterministic) Turing machine in space $s(n) = \Omega(\log n)$, and let $\SCStrength(n)$ be a query bound function. Then $\Language$ has a (public-coin and non-adaptive) Interactive PCP that is perfect zero knowledge against all $\SCStrength$-query malicious verifiers. In more detail:
{\small
\begin{equation*}
\Language \in
\PZKIPCP
\left[
\begin{tabular}{rl}
soundness error: & $1/2$ \\
round complexity: & $\poly(s(n)) \cdot \log \SCStrength(n)$ \\[1mm]
proof length: & $\poly(2^{s(n)}, \SCStrength(n))$ \\
query complexity: & $\poly(s(n), \SCStrength(n))$ \\[1mm]
prover time: & $\poly(2^{s(n)}, \SCStrength(n))$ \\
verifier time: &  $n \cdot \poly(s(n), \log \SCStrength(n))$ \\[1mm]
verifier space: & $O(s(n) + \log \SCStrength(n)/\log s(n))$
\end{tabular}
\right]
\enspace.
\end{equation*}}
\end{corollary}

\begin{proof}
By \cite[Lemma 4.1]{GoldwasserKR15}, a language $\Language$ decidable by a Turing machine in time $t(n)$ and space $s(n)$ is also decidable by a $O(s(n))$-space uniform circuit family of size $\poly(t(n) 2^{s(n)}) = \poly(2^{s(n)})$ and depth $\poly(s(n))$. Applying \thmref{thm:gkr-space-uniform} to this circuit family yields the corollary.
\end{proof}

In each of the results above, one can set the query bound $\SCStrength(n)$ to be superpolynomial in $n$ (say, $\SCStrength(n) \DefineEqual n^{O(\log \log n)}$) to obtain Interactive PCPs that are perfect zero knowledge against all polynomial-time malicious verifiers, without affecting the efficiency of the honest verifier.

\subsection{Notations for layered arithmetic circuits}
\label{sec:layered-arithmetic-circuits}

We briefly recall some definitions and observations from \cite{GoldwasserKR15}.

\begin{definition}
\label{def:layered-arithmetic-circuits}
A \emph{layered arithmetic circuit} $\Circuit \colon \Field^{n} \to \Field$ of depth $\Depth$ and size $\Size$ (with $n \leq \Size$) is an arithmetic circuit, with fan-in $2$, arranged into $\Depth+1$ layers: the output layer (layer $0$) has a single gate; layers $1,\dots,\Depth-1$ have $\Size$ gates each; and the input layer (layer $\Depth$) has $n$ gates. For $i \in \{0,\dots,\Depth-1\}$, each gate in layer $i$ has two inputs, which are gates in layer $i+1$; the $n$ gates in layer $\Depth$ are $\Circuit$'s inputs; the single gate in layer $0$ is $\Circuit$'s output.
\end{definition}

For $i \in \{1,\dots,\Depth - 1\}$, we denote by $\Layer{i} \colon \Field^{n} \times [\Size] \to \Field$ the function such that $\Layer{i}(\vec{x}, j) = v$ if and only if the $j$-th gate of the $i$-th layer has value $v$ when $\Circuit$'s input is $\vec{x} \in \Field^{n}$. Moreover, we denote by $\Layer{\Depth} \colon \Field^{n} \times [n] \to \Field$ the function $\Layer{\Depth}(\vec{x}, j) \DefineEqual x_{j}$, and by $\Layer{\Depth} \colon \Field^{n} \to \Field$ the function $\Layer{0}(\vec{x}) \DefineEqual \Circuit(x)$.

Let $\SPSubset \subseteq \Field$ and $\GKRVars,\GKRVarsInput \in \Naturals$ be such that $\SetCardinality{\SPSubset} \geq 2$, $\GKRVars \geq \lceil \log \Size / \log \SetCardinality{\SPSubset} \rceil$, and $\GKRVarsInput \geq \lceil \log n / \log \SetCardinality{\SPSubset} \rceil$.

For $i \in \{1,\dots,\Depth - 1\}$, we can equivalently view $\Layer{i}$ as a function from $\Field^{n} \times \SPSubset^{\GKRVars}$ to $\Field$, by taking an arbitrary order $\alpha$ on $\SPSubset^{\GKRVars}$, and letting $\Layer{i}(\vec{x}, \vec{z}) = 0$ for every $\vec{z} \in \SPSubset^{\GKRVars}$ with $\alpha(z) > \Size$. The following equation then relates $\Layer{i-1}$ to $\Layer{i}$:
\begin{equation*}
\Layer{i-1}(\vec{x}, \vec{z})
  = \sum_{\vec{\omega}_{1},\vec{\omega}_{2} \in \SPSubset^{\GKRVars}}
      \Add{i}(\vec{z}, \vec{\omega}_{1}, \vec{\omega}_{2})
         \cdot \big(\Layer{i}(\vec{x}, \vec{\omega}_{1})
                  + \Layer{i}(\vec{x}, \vec{\omega}_{2})\big)
	+ \Mult{i}(\vec{z}, \vec{\omega}_{1}, \vec{\omega}_{2})
	     \cdot \big(\Layer{i}(\vec{x}, \vec{\omega}_{1})
	          \cdot \Layer{i}(\vec{x}, \vec{\omega}_{2})\big)
\end{equation*}
where $\Add{i} \colon \SPSubset^{3\GKRVars} \to \Field$ (resp., $\Mult{i} \colon \SPSubset^{3\GKRVars} \to \Field$) is the predicate such that, for every $(\vec{a}, \vec{b}, \vec{c}) \in \SPSubset^{3\GKRVars}$, $\Add{i}(\vec{a}, \vec{b}, \vec{c})$ (resp., $\Mult{i}(\vec{a}, \vec{b}, \vec{c})$) equals $1$ if the $\alpha(\vec{a})$-th gate of layer $i-1$ is an addition (resp., multiplication) gate whose inputs are gates $\alpha(\vec{b}) \leq \alpha(\vec{c})$ of layer $i$, or $0$ otherwise. The situation for the input layer and output layer is somewhat different.
\begin{itemize}

  \item The input layer (layer $\Depth$) has $n$ gates (rather than $\Size$), so we can equivalently view $\Layer{\Depth}$ as a function from $\Field^{n} \times \SPSubset^{\GKRVarsInput}$ to $\Field$, by taking an arbitrary order $\alpha'$ on $\SPSubset^{\GKRVarsInput}$, and letting $\Layer{\Depth}(\vec{x}, \vec{z}) = 0$ for every $\vec{z} \in \SPSubset^{\GKRVarsInput}$ with $\alpha'(z) > n$. Naturally, we need to adjust the expression for $\Layer{\Depth-1}$ and the definitions of $\Add{\Depth}, \Mult{\Depth}$ appropriately.

  \item The output layer (layer $0$) has a single gate (rather than $\Size$), so and we can write
\begin{equation*}
\Layer{0}(\vec{x}) = \sum_{\vec{\omega}_{1},\vec{\omega}_{2} \in \SPSubset^{\GKRVars}} \Add{1}(\vec{0}, \vec{\omega}_{1}, \vec{\omega}_{2}) \cdot \big(\Layer{1}(\vec{x}, \vec{\omega}_{1}) + \Layer{1}(\vec{x}, \vec{\omega}_{2})\big)
+ \Mult{1}(\vec{0}, \vec{\omega}_{1}, \vec{\omega}_{2}) \cdot \big(\Layer{1}(\vec{x}, \vec{\omega}_{1}) \cdot \Layer{1}(\vec{x}, \vec{\omega}_{2})\big) \enspace.
\end{equation*}
\end{itemize}

\begin{remark}
For $i \in \{1,\dots,\Depth\}$, we view $\Layer{i}$ as a function not only of the gate number $j$ (represented as $\vec{z} \in \SPSubset^{\GKRVars}$) but also of the circuit's input $\vec{x}$, and we view $\Layer{0}$ as a function of the input only. In contrast, \cite{GoldwasserKR15} defines $\Layer{0},\dots,\Layer{\Depth}$ with $\vec{x}$ `hard-coded'. We require the additional flexibility ($\vec{x}$ is an input) to compose sum-product circuits below.
\end{remark}

\subsection{Sum-product subcircuits and oracle inputs}
\label{sec:extended-notions}

Later on we will need to assemble different sum-product circuits into one such circuit, which requires relaxing the definition of an ari-graph to allow the root to have positive arity. This is stated formally below (difference highlighted).

\begin{definition}[extends \defref{def:ari-graph}]
\label{def:ari-graph-2}
A tuple $\Graph = (\VertexSet, \EdgeSet, \SPFreeProjection, \SPSumProjection)$ is an \defemph{ari-graph} if $(\VertexSet, \EdgeSet)$ is a directed acyclic multi-graph and both $\SPFreeProjection$ and $\SPSumProjection$ label every edge $e$ in $\EdgeSet$ with finite sets of positive integers $\SPFreeProjection[e]$ and $\SPSumProjection[e]$ that satisfy the following property. For every vertex $v$ in $\VertexSet$, there exists a (unique) non-negative integer $\SPArity{v}$ such that, if $e_{1}, \dots, e_{t}$ are $v$'s outgoing edges:
\begin{inparaenum}[(1)]
  \item \hl{if $v$ is the root then $\SPArity{v} = \max(\SPSumProjection[e_{1}] \cup \ldots \cup \SPSumProjection[e_{t}])$}, otherwise $\SPArity{v} = \SetCardinality{\SPFreeProjection[e_{1}]} + \SetCardinality{\SPSumProjection[e_{1}]} = \cdots = \SetCardinality{\SPFreeProjection[e_{t}]} + \SetCardinality{\SPSumProjection[e_{t}]}$ where $e_{1}, \dots, e_{\GraphInDegree[\Graph]{v}}$ are $v$'s incoming edges;
  \item $\SPFreeProjection[e_{1}], \dots, \SPFreeProjection[e_{t}] \subseteq \{1, \dots, \SPArity{v}\}$.
\end{inparaenum}
\end{definition}

For consistency, we retain the original definition of sum-product circuits, where the root must have arity zero, and so its value is a constant in $\Field$. Instead, we now define sum-product \emph{sub}circuits, where the root may have positive arity.

\begin{definition}
\label{def:sum-product-subcircuit}
A \defemph{sum-product subcircuit} $\SPCircuit$ is a tuple $\SPCircuitTuple$ whose definition is identical to that of a sum-product circuit (see \defref{def:sum-product-circuit}), except that the root of $\Graph$ may have positive arity. In particular, $\SPValueL{\SPInput}{\SPCircuit}$ may be a non-constant polynomial (having the same arity as the root of $\Graph$).
\end{definition}

It is not difficult to see that, provided the arities match, we can replace a leaf of a sum-product circuit $\SPCircuit$ with a sum-product \emph{subcircuit} $\SPCircuit_{0}$ computing the same function (over the appropriate subdomain) without affecting $\SPCircuit$'s output.

Finally, in the discussions below, we also need a way to talk about oracles in the context of a sum-product (sub)circuit. The structure of the protocol is such that the calls to the oracle will be made at the leaves of $\SPCircuit$, and so we associate the oracle with the circuit's input.

\begin{definition}
\label{def:oracle-input}
Given a sub-product (sub)circuit $\SPCircuit$ and a list of oracles $O = \{O_{i}\}_{i \in [\ell]}$ with $O_{i} \colon \Field^{k} \to \Field$ for each $i \in [\ell]$, we say that $\SPInput^{O}$ is an \defemph{oracle input} for $\SPCircuit$ if it labels the leaves of $\SPCircuit$ with oracle circuits: arithmetic circuits that may include oracle gates. An oracle gate is a gate labeled with the name of an oracle $O_{i}$; it has $k$ inputs and a single output defined as the evaluation of $O_{i}$ on its inputs.
\end{definition}

\subsection{Sum-product subcircuits for layered arithmetic circuits}
\label{sec:sumproduct-barebones}

We show that the evaluation problem for a given layered arithmetic circuit can be reduced to a sum-product subcircuit, when given oracles for low-degree extensions of the wiring predicates.
The reduction essentially consists of casting the \emph{barebones protocol} of \cite[Section 3]{GoldwasserKR15} as (the evaluation of) a sum-product subcircuit on an oracle input. The barebones protocol is an Interactive Proof that enables a verifier to check a statement of the form ``$\Circuit(\vec{x}) = y$'', where $\Circuit$ is a layered arithmetic circuit of size $\Size$ and depth $\Depth$, in time $n \cdot \poly(\Depth, \log \Size)$ and space $O(\log \Size)$, provided that the verifier has oracle access to low-degree extensions of the wiring predicates $\{\Add{i}, \Mult{i}\}_{i \in \{1,\dots,\Depth\}}$ for the circuit $\Circuit$. The arithmetization of $\Circuit$ underlying that protocol provides the basic intuition for how to `program' a sum-product subcircuit to encode this computation.

\begin{lemma}[sum-product subcircuits for layered arithmetic circuits]
\label{lem:gkr-barebones}
Let $\Field$ be a finite field, $\Circuit \colon \Field^{n} \to \Field$ a layered arithmetic circuit of depth $\Depth$ and size $\Size$, $\SPSubset \subseteq \Field$, $\GKRVars, \GKRVarsInput \in \Naturals$ and $\WireOracles = (\LDAdd{i}, \LDMult{i})_{i \in \{1,\dots,\Depth\}}$ any degree-$\delta$ extensions of the wiring predicates $\{\Add{i}, \Mult{i}\}_{i \in \{1,\dots,\Depth\}}$, with $\delta \leq \polylog(\Size)$. Then there exists a sum-product subcircuit $\SPCircuit$, constructible in time $n \cdot \poly(\Depth, \log \Size)$ and space $O(\log \Size)$, and oracle input $\SPInput^{\WireOracles}$ constructible in time $\poly(\Depth, \log \Size)$ and space $O(\log \Size)$, such that $\SPValueL{\SPInput^{\WireOracles}}{\SPCircuit}(\vec{x}) = \Circuit(\vec{x})$, for all $\vec{x} \in \Field^{n}$.
	
	Moreover, $\SPCircuit = \SPCircuitTuple$ where $\SetCardinality{\SPSubset} = \poly(\Depth, \log \Size)$, $\InternalVertexSetDegree = \Theta(1)$, $\SPLeafDegree = \poly(\Depth, \log \Size)$, $\SetCardinality{\VertexSet(\Graph)} = \Theta(\Depth)$, $\SPMaxArity(\Graph) = O(n + \log \Size/\log \SetCardinality{\Field})$, $\GraphMaxInDegree{\Graph} = \Theta(1)$, $\Width{\Graph} = \Theta(1)$; and $\MaxSpace{\SPInput} = O(\log \Size)$.
\end{lemma}

\begin{proof}
Let $\SPSubset \subseteq \Field$ be such that $\SetCardinality{\Field}^{\Omega(1)} \leq \SetCardinality{\SPSubset} \leq \poly(\Depth, \log \Size)$; $\GKRVars \in \Naturals$ be such that $\Size \leq \SetCardinality{\SPSubset}^{\GKRVars} \leq \poly(\Size)$; $\GKRVarsInput \in \Naturals$ be such that $n \leq \SetCardinality{\SPSubset}^{\GKRVarsInput} \leq n \cdot \poly(\Depth, \log \Size)$.

First, we use $\Depth, \GKRVars, \GKRVarsInput$ to construct an ari-graph $\Graph = (\VertexSet, \EdgeSet, \SPFreeProjection, \SPSumProjection)$ as follows:
\begin{align*}
	\VertexSet &\DefineEqual \{ v_{i} \}_{i \in \{0,\dots,\Depth\}} \cup \{ u_{\Add{i}}, u_{\Mult{i}} \}_{i \in \{1,\dots,\Depth\}} \\
	\EdgeSet &\DefineEqual \{ e_{i}, e'_{i} = (v_{i}, v_{i+1}) \}_{i \in \{0,\dots,\Depth-1\}} \cup \{ (v_{i}, u_{f_{i+1}}) : f \in \{\Add{}, \Mult{}\} \}_{i \in \{0,\dots,\Depth-1\}} \\
	\SPFreeProjection[e] &\DefineEqual \begin{cases}
		\varnothing & \text{ if $e = (v_{0}, u_{f_{1}})$ for $f \in \{\Add{}, \Mult{}\}$, or} \\
		n + [\GKRVars] & \text{ if $e = (v_{i}, u_{f_{i+1}})$ for $f \in \{\Add{}, \Mult{}\}$ and $i \in \{1,\dots,\Depth-1\}$,} \\
		[n] & \text{ otherwise;}
	\end{cases} \\
	\SPSumProjection[e] &\DefineEqual \begin{cases}
		[2\GKRVars] & \text{ if $e = (v_{i}, u_{f_{i+1}})$ for $f \in \{\Add{}, \Mult{}\}$ and $i \in \{0,\dots,\Depth-2\}$,} \\
		[\GKRVars] & \text{ if $e = e_{i}$ for some $i \in \{0,\dots,\Depth-2\}$,} \\
		\GKRVars+[\GKRVars] & \text{ if $e = e'_{i}$ for some $i \in \{0,\dots,\Depth-2\}$,} \\
		[2\GKRVarsInput] & \text{ if $e = (v_{\Depth-1}, u_{f_{\Depth}})$ for $f \in \{\Add{}, \Mult{}\}$,} \\
		[\GKRVarsInput] & \text{ if $e = e_{\Depth-1}$, or} \\
		\GKRVarsInput+[\GKRVarsInput] & \text{ if $e = e'_{\Depth-1}$.}
	\end{cases}
\end{align*}
Note that $\EdgeSet$ is a multiset and contains, for every $i \in \{0, \dots, \Depth-1\}$, two distinct edges, $e_{i}$ and $e'_{i}$, from $v_{i}$ to $v_{i+1}$ (with different projection labels). The root of $\Graph$ is $v_{0}$, and its leaves are $v_{\Depth}$ and $\{u_{\Add{i}}, u_{\Mult{i}}\}_{i \in \{1,\dots,\Depth\}}$.

Next, we construct the sum-product circuit $\SPCircuit = \SPCircuitTuple$ where $\InternalVertexSetDegree \DefineEqual 3$, $\SPLeafDegree \DefineEqual \max(\delta, \SetCardinality{\SPSubset})$, and $\SPPoly$ labels internal vertices of $\Graph$ as follows. For every $i \in \{0,\dots,\Depth-1\}$, $\SPPoly[v_{i}] \DefineEqual \VariableX_{1} \cdot (\VariableY_{1} + \VariableY_{2}) + \VariableX_{2} \cdot \VariableY_{1} \cdot \VariableY_{2}$, where $\VariableY_{1}, \VariableY_{2}$ correspond to $e_{i}, e_{i}'$ respectively, and $\VariableX_{1}, \VariableX_{2}$ correspond to $(v_{i}, u_{\Add{i}}), (v_{i}, u_{\Mult{i}})$ respectively.

Finally, we construct the input $\SPInput$ for $\SPCircuit$ as follows:
\begin{itemize}[nolistsep]
  \item for all $f \in \{\Add{}, \Mult{}\}$, $\SPLeaf[u_{f_{1}}]$ is the oracle circuit that outputs $\LD{f}_{1}(\vec{0}, \vec{\omega}_{1}, \vec{\omega}_{2})$ on input $(\vec{\omega}_{1}, \vec{\omega}_{2})$;
  \item for all $i \in \{2,\dots,\Depth\}$ and $f \in \{\Add{}, \Mult{}\}$, $\SPLeaf[u_{f_{i}}]$ is the oracle circuit that outputs $\LD{f}_{i}(\vec{z}, \vec{\omega}_{1}, \vec{\omega}_{2})$ on input $(\vec{z}, \vec{\omega}_{1}, \vec{\omega}_{2})$;
  \item $\SPLeaf[v_{\Depth}]$ is the polynomial $\Layer{\Depth}(\vec{x}, \vec{z}) \DefineEqual \sum_{\vec{\beta} \in \SPSubset^{\GKRVarsInput}} \Lagrange{\SPSubset^{\GKRVarsInput}}(\vec{z},\vec{\beta}) x_{\alpha'(\vec{\beta})}$ (with $x_{i} \DefineEqual 0$ for $i > n$), which has individual degree less than $\SetCardinality{\SPSubset}$, that is the low-degree extension of $\Layer{\Depth}$ and can be computed in time $\SetCardinality{\SPSubset}^{\GKRVarsInput} \cdot \poly(\SetCardinality{\SPSubset}, \GKRVarsInput) \leq n \cdot \poly(\Depth, \log \Size)$ and space $O(\GKRVarsInput \cdot \log \SetCardinality{\SPSubset}) \leq \log n + \polylog(\Depth, \log \Size) \leq O(\log \Size)$.
\end{itemize}
See \figref{fig:gkr-1}, \figref{fig:gkr-2}, \figref{fig:gkr-d} for diagrams of this sum-product circuit and input for depth $1$, $2$, $\Depth$ respectively.

We are left to argue correctness of the reduction: it is easy to see that $\SPValueL{\SPInput^{\WireOracles}}{v_{\Depth}}(\vec{x},\vec{z}) = \Layer{D}(\vec{x},\vec{z})$ for every $\vec{x} \in \Field^{n}$ and $\vec{z} \in \SPSubset^{\GKRVarsInput}$; hence, for every $i \in \{1,\dots,\Depth-1\}$, $\SPValueL{\SPInput^{\WireOracles}}{v_{i}}(\vec{x}, \vec{z}) = \Layer{i}(\vec{x}, \vec{z})$ for every $\vec{x} \in \Field^{n}$ and $\vec{z} \in \SPSubset^{\GKRVars}$. We conclude that $\SPValueL{\SPInput^{\WireOracles}}{\SPCircuit} = \SPValueL{\SPInput^{\WireOracles}}{v_{0}} = \Layer{0}$ (as functions in $\vec{x}$), as claimed.
\end{proof}

\begin{figure}[t!]
\centering
\begin{minipage}[b]{0.33\textwidth}
  \centering
\includegraphics[width=0.74\textwidth]{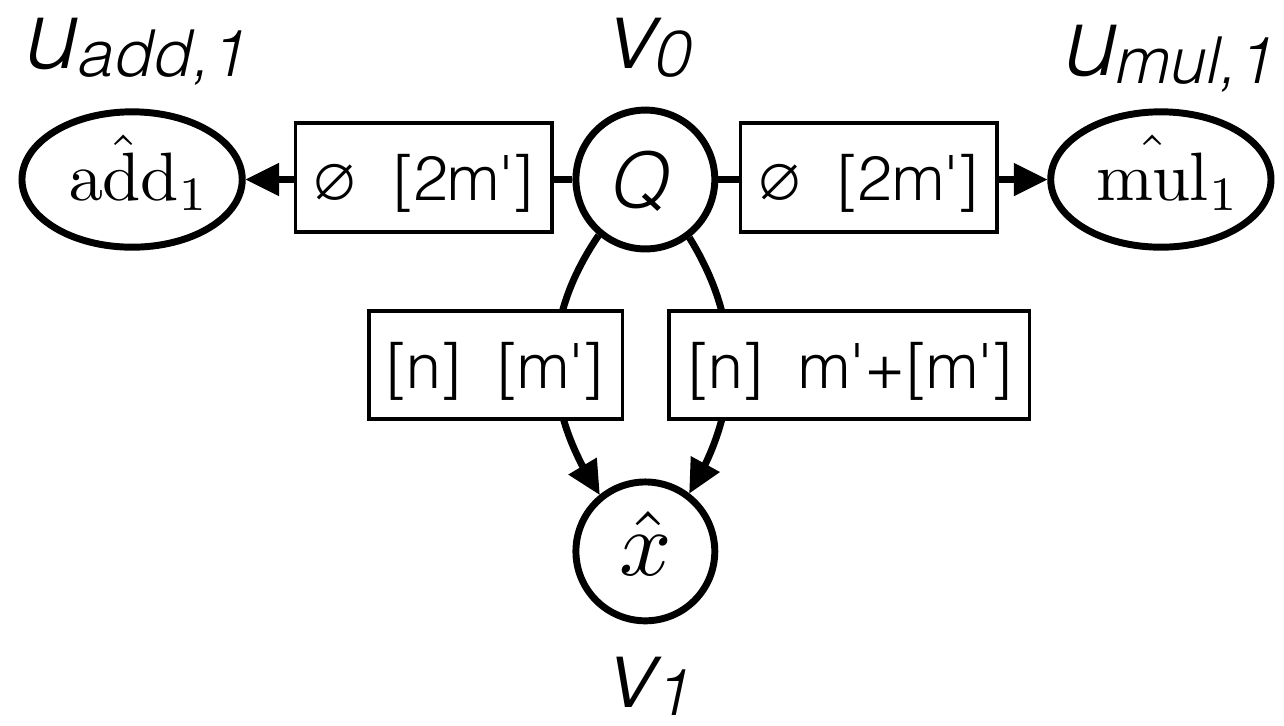}
\caption{Sum-product circuit and its input for layered circuit of depth $1$.}
\label{fig:gkr-1}
\end{minipage}%
\begin{minipage}[b]{0.33\textwidth}
  \centering
\includegraphics[width=0.74\textwidth]{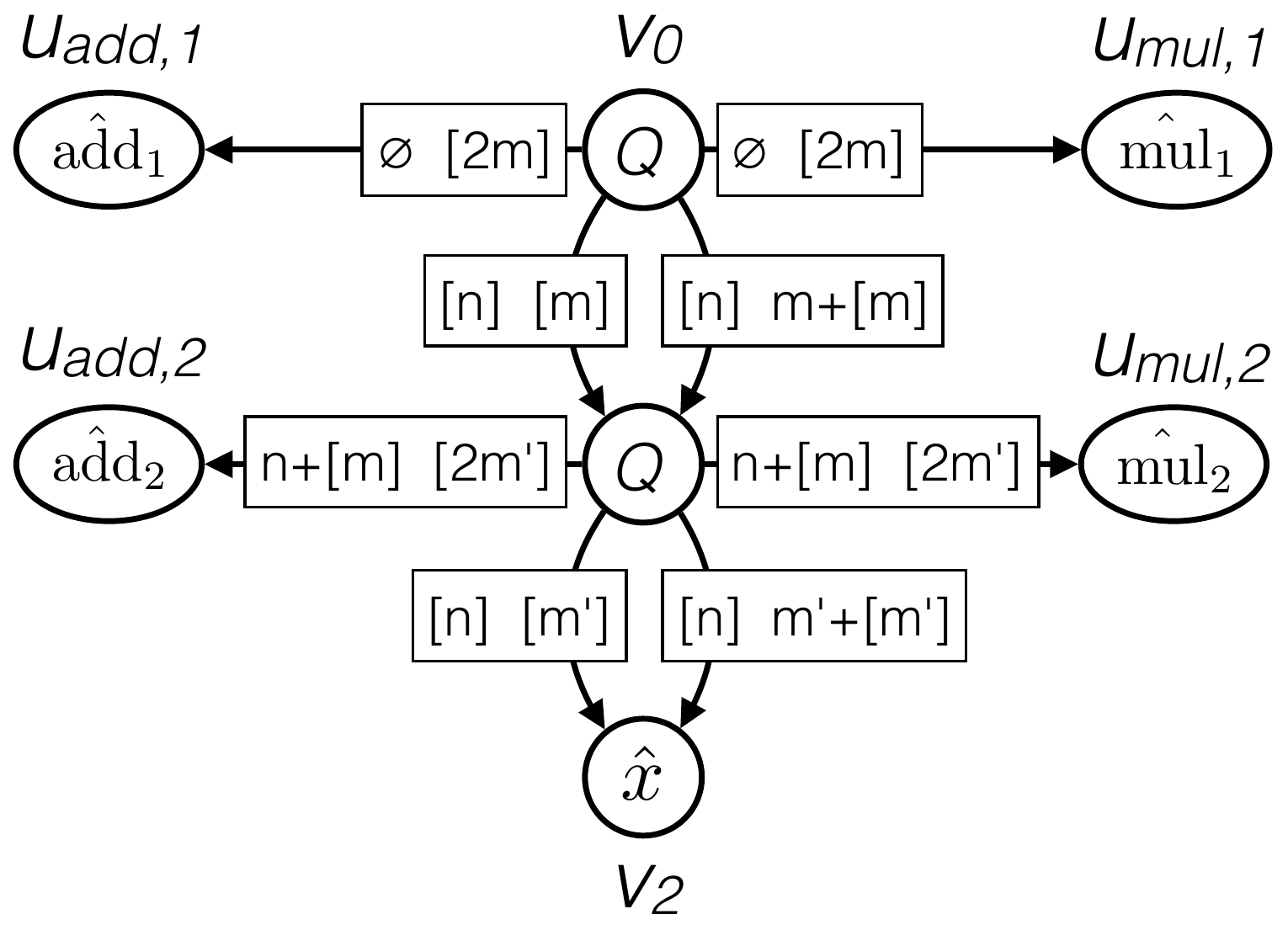}
\caption{Sum-product circuit and its input for layered circuit of depth $2$.}
\label{fig:gkr-2}
\end{minipage}%
\begin{minipage}[b]{0.33\textwidth}
  \centering
\includegraphics[width=0.74\textwidth]{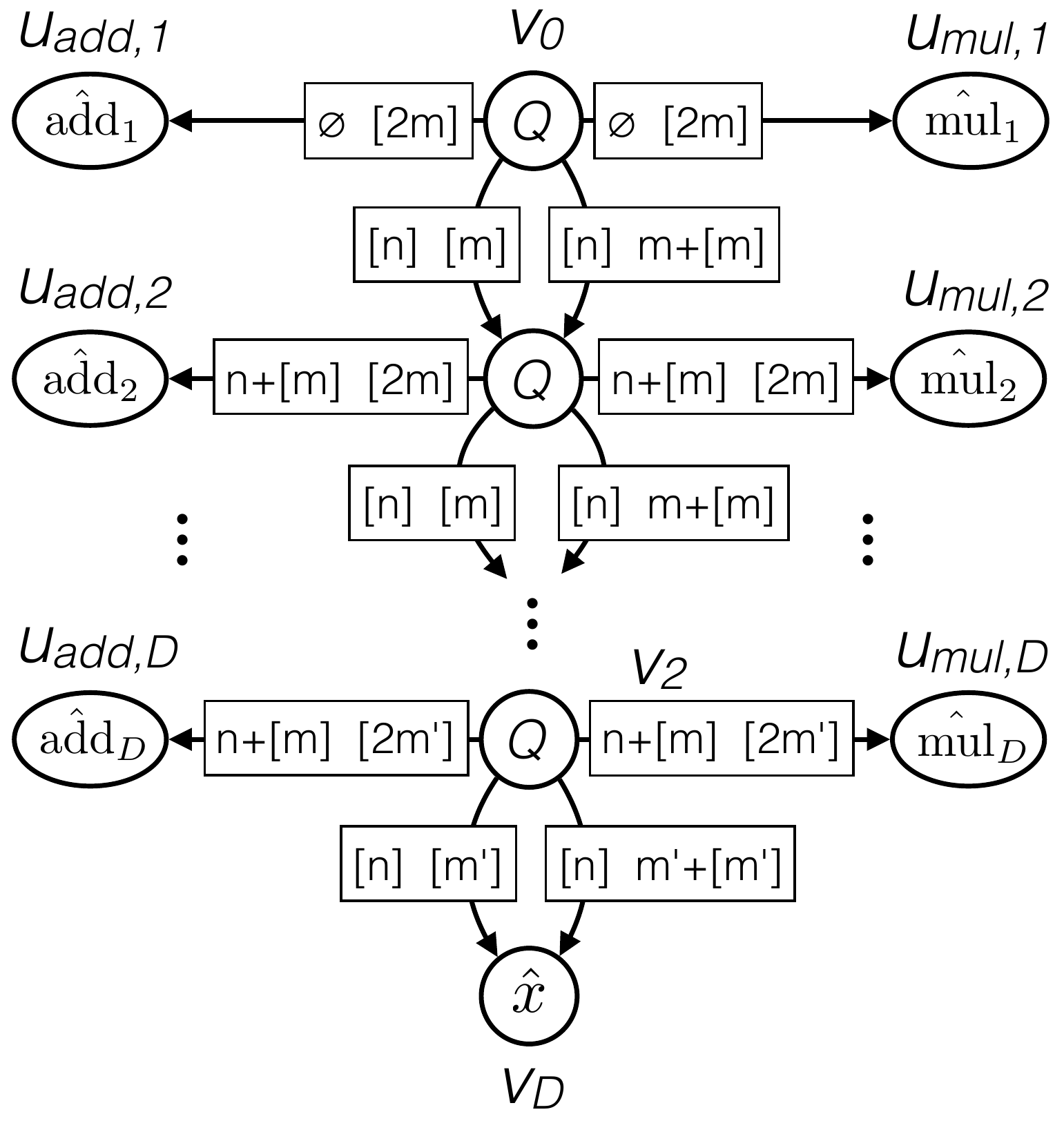}
\caption{Sum-product circuit and its input for layered circuit of depth $\Depth$.}
\label{fig:gkr-d}
\end{minipage}
\end{figure}

\subsection{Sum-product subcircuits for small-space Turing machines}
\label{sec:sumproduct-space-uniform}

We show that small-space Turing machine computations can be reduced to the evaluation of sum-product subcircuits whose size is polynomially related to the space bound. This theorem is the analogue of \cite[Theorem 4.4]{GoldwasserKR15} for our framework: we use \lemref{lem:gkr-barebones} and then instantiate the oracle input with explicit arithmetic circuits.

To prove \thmref{thm:gkr-space-uniform} we need to construct a sum-product subcircuit whose value is a function $f \colon \Field^{n} \to \Field$ computable by a Turing machine in small space. However, if we simply apply \lemref{lem:gkr-barebones} to a circuit given by \cite[Lemma 4.3]{GoldwasserKR15}, then we only obtain sum-product subcircuits for \emph{boolean} functions computable in small space. We therefore apply \cite[Lemma 4.3]{GoldwasserKR15} to the language $\Language_{f} \DefineEqual \{ (\vec{x}, a) : f(\vec{x}) = a \}$, which by \lemref{lem:gkr-barebones} yields a subcircuit whose value is a function $f' \colon \Field^{n} \times \Field \to \Bits$, the indicator function of $\Language_{f}$. Then it holds that $\sum_{a \in \Field} a \cdot f'(x, a) = f(x)$, and this summation can be implemented by adding an extra vertex to the subcircuit.

\begin{lemma}[sum-product subcircuits for small-space Turing machines]
\label{lem:gkr-subcircuit}
Let $s \colon \Naturals \to \Naturals$ be a space function with $s(n) = \Omega(\log n)$, $\SPSubset$ an extension field of $\Field_{2}$, and $\Field$ an extension field of $\SPSubset$ with $\SetCardinality{\SPSubset} = \poly(s(n))$ and $\SetCardinality{\Field} = \poly(\SetCardinality{\SPSubset})$. Let $f \colon \Field^{n} \to \Field$ be a function computable by a non-deterministic Turing machine $\Machine$ in space $s(n)$ in the following sense: $\Machine$ accepts $(\vec{x}, a) \in \Field^{n} \times \Field$ if and only if $f(\vec{x}) = a$. Then there exist a sum-product subcircuit $\SPCircuit$ and input $\SPInput$ for $\SPCircuit$ such that $\SPValueL{\SPInput}{\SPCircuit}(\vec{x}) = f(\vec{x})$ for all $x \in \Field^{n}$; $\SPCircuit$ is constructible in time $n \cdot \poly(s(n))$ and space $O(s(n))$, and $\SPInput$ is constructible in time $n \cdot \poly(s(n))$ and space $O(s(n))$.
	
Moreover, $\SPCircuit = \SPCircuitTuple$ where $\InternalVertexSetDegree = \Theta(1)$, $\SPLeafDegree = \poly(s(n))$, $\SetCardinality{\VertexSet(\Graph)} = O(s^{2}(n))$, $\SPMaxArity(\Graph) = O(n + s(n)/\log \SetCardinality{\Field})$, $\GraphMaxInDegree{\Graph} = \Theta(1)$, $\Width{\Graph} = \Theta(1)$; and $\MaxSpace{\SPInput} = O(s(n))$.
\end{lemma}

\begin{proof}
Fix an input length $n$, and let $k \DefineEqual \log \SetCardinality{\Field}/\log \SetCardinality{\SPSubset}$ ($k$ is an integer because $\Field$ is an extension field of $\SPSubset$). Let $\gamma \colon \SPSubset^{k} \to \Field$ be the trivial isomorphism, and let $\Language \DefineEqual \{(\vec{x}, \vec{z}) : f(\vec{x}) = \gamma(\vec{z}) \}$. Note that $\Language$ is computable by a non-deterministic Turing machine in space $O(s(n))$ because
\begin{inparaenum}[(i)]
  \item $f$ is and
  \item $\gamma(\vec{z})$ can be implemented at the bit level as the identity function.
\end{inparaenum}
 By \cite[Lemma 4.3]{GoldwasserKR15}, $\Language$ can be computed by a layered arithmetic circuit $\Circuit$ over $\Field$ of size $\Size(n) = \poly(2^{s(n)})$ and depth $\Depth(n) = O(s^{2}(n))$; moreover, one can generate arithmetic circuits for and evaluate some low-degree extensions $\LDAdd{i}, \LDMult{i}$ of the wiring predicates in time $\poly(s(n))$ and space $O(\log s(n))$. The degree of these extensions is $\poly(s(n)) = \polylog(\Size(n))$, independent of $\Field$ (which is crucial for soundness).

Note that $\SPSubset$ is of size $\poly(s(n)) = \poly(\Depth(n), \log \Size(n))$. By \lemref{lem:gkr-barebones}, we can construct a sum-product subcircuit $\SPCircuit' = \SPCircuitTuple$ and oracle input $\SPInput^{\WireOracles}$ such that $\SPValueL{\SPInput^{\WireOracles}}{\SPCircuit'}(\vec{x}, \vec{z}) = \Circuit(\vec{x},\vec{z})$ for every $\vec{x} \in \Field^{n}$ and $\vec{z} \in \SPSubset^{k}$. We replace the oracle gates in $\SPInput^{F}$ with the explicit circuits for $\LDAdd{i}, \LDMult{i}$ to obtain a non-oracle input $\SPInput'$.

Now observe that
\begin{equation*}
f(\vec{x})
= \sum_{\vec{z} \in \SPSubset^{k}} \Circuit(\vec{x}, \vec{z}) \cdot \gamma(\vec{z})
= \sum_{\vec{z} \in \SPSubset^{k}} \SPValueL{\SPInput}{\SPCircuit'}(\vec{x}, \vec{z}) \cdot \LD{\gamma}(\vec{z})
\end{equation*}
where $\LD{\gamma} \colon \Field^{k} \to \Field$ is the unique $\Field$-linear function such that $\LD{\gamma}(\vec{z}) = \gamma(\vec{z})$ for $\vec{z} \in \SPSubset^{k}$.  Note that $\LD{\gamma}$ is computable by an arithmetic circuit of size $O(k)$, and can be evaluated in space $O(k + \log \SetCardinality{\Field}) = O(\log \SetCardinality{\Field})$.

The sum-product subcircuit $\SPCircuit$ is obtained from $\SPCircuit'$ by adding a vertex $v$ to $\Graph$ labeled with $\VariableX \cdot \VariableY$ and a leaf vertex $u_{\gamma}$, and two edges $e \DefineEqual (v, \Root[\Graph])$ and $e' \DefineEqual (v, u_{\gamma})$ with labels $(\SPFreeProjection[e],\SPSumProjection[e]) \DefineEqual ([n],[k])$ and $(\SPFreeProjection[e'],\SPSumProjection[e']) \DefineEqual (\varnothing,[k])$. The input $\SPInput$ is obtained from $\SPInput'$ by mapping $u_{\gamma}$ to the arithmetic circuit computing $\LD{\gamma}$.
\end{proof}

\subsection{Proof of \thmref{thm:gkr-space-uniform}}
\label{sec:proof-of-gkr-space-uniform}

Let $\Language$ be a language decidable by a family of $O(\log \Size(n))$-space uniform boolean circuits of size $\Size(n)$ and depth $\Depth(n)$, and let $\SCStrength(n)$ be a query bound function. The prover and verifier receive an input $\vec{x} \in \Bits^{n}$. Let $\Circuit$ be the boolean circuit (i.e., arithmetic circuit over $\Field_{2}$) for this input size; we assume that $\Circuit$ is layered with fan-in $2$, as in \defref{def:layered-arithmetic-circuits}. (Any circuit can be efficiently converted into this form, with only a quadratic increase in size and a $\log \Size(n)$-factor increase in depth, which does not affect the theorem's statement.)
	
Let $\SPSubset$ be an extension field of $\Field_{2}$ and $\Field$ an extension field of $\SPSubset$ such that $\SetCardinality{\Field} = \poly(\Depth(n), \log \Size(n))$ and $\SetCardinality{\SPSubset} = \SetCardinality{\Field}^{\Omega(1)}$. Define $\GKRVars \DefineEqual \lceil \log \Size(n) / \log \SetCardinality{\SPSubset} \rceil$ and $\GKRVarsInput \DefineEqual \lceil \log n / \log \SetCardinality{\SPSubset} \rceil$. Throughout we fix $\InternalVertexSetDegree = \Theta(1)$ and $\SPLeafDegree = \poly(\Depth(n), \log \Size(n))$, as these settings suffice for all the (sub)circuits that we construct.

\parhead{Part A: efficiently computable sum-product subcircuits for low-degree extensions of the wiring predicates}
By \cite[Claim 4.6]{GoldwasserKR15}, for every $i \in \{1,\dots,\Depth(n)\}$, the low-degree extensions $\LDAdd{i}, \LDMult{i} \colon \Field^{\GKRVars} \to \Field$ of $\Circuit$'s wiring predicates can be computed by a Turing machine in space $O(\log \Size(n))$. The machine takes as input $1^{n}$, $\SPSubset$, $\Field$, $\GKRVars$, $\GKRVarsInput$, $i$, $f \in \{\Add{},\Mult{}\}$, and $(\vec{a},\vec{b},\vec{c}) \in \Field^{3\GKRVars}$, and outputs $\LD{f}_{i}(\vec{a}, \vec{b}, \vec{c})$. Here we consider all inputs besides $(\vec{a},\vec{b},\vec{c})$ as hard-coded in the machine. By \lemref{lem:gkr-subcircuit}, there exist sum-product subcircuits $\SPCircuit_{f_{i}} = (\Field, \SPSubset, \InternalVertexSetDegree, \SPLeafDegree, \Graph_{f_{i}}, \SPPoly[f_{i}])$ and inputs $\SPInput_{f_{i}}$ (for every $f \in \Add{}, \Mult{}$ and $i \in \{1,\dots,\Depth\}$) such that $\SPValueL{\SPInput_{f_{i}}}{\SPCircuit_{f_{i}}}(\vec{a}, \vec{b}, \vec{c}) = \LD{f}_{i}(\vec{a}, \vec{b}, \vec{c})$ for all $(\vec{a},\vec{b},\vec{c}) \in \Field^{3\GKRVars}$.

For every $f \in \Add{}, \Mult{}$ and $i \in \{1,\dots,\Depth\}$, we have $\SetCardinality{\VertexSet(\Graph_{f_{i}})} = \poly(\log \Size(n))$, $\SPMaxArity(\Graph_{f_{i}}) = O(\log \Size(n)/\log \SetCardinality{\Field})$, $\GraphMaxInDegree{\Graph_{f_{i}}} = \Theta(1)$, $\Width{\Graph_{f_{i}}} = \Theta(1)$, and $\MaxSpace{\SPInput_{f_{i}}} = O(\log \Size(n))$.

\parhead{Part B: sum-product circuit with oracle input for $\Circuit(\vec{x})$}
We now invoke \lemref{lem:gkr-barebones} on $\Circuit$ to obtain a sum-product subcircuit $\SPCircuit' = (\Field, \SPSubset, \InternalVertexSetDegree, \SPLeafDegree, \Graph', \SPPoly')$ and oracle input $\SPInput^{\WireOracles}$. For the theorem, we need to transform this into a sum-product \emph{circuit} (i.e., where the root has arity $0$). To do this, we will modify the projections so that the input $\vec{x}$ is no longer `carried down', and then `hard-code' it into $\SPInput^{\WireOracles}$.

Let $\SPFreeProjection'$ be given by defining, for every $e \in \EdgeSet'$,
\begin{equation*}
\SPFreeProjection[e]' \DefineEqual
\begin{cases}
  \{1, \dots, \GKRVars\} & \text{ if $e = (v_{i}, u_{f_{i+1}})$ for $f \in \{\Add{}, \Mult{}\}$ and $i \in \{1,\dots,\Depth-1\}$ or} \\
  \varnothing & \text{ otherwise.}
\end{cases}
\end{equation*}
Observe that in $\Graph'' \DefineEqual (\VertexSet(\Graph'), \EdgeSet(\Graph'), \SPFreeProjection', \SPSumProjection(\Graph'))$, the root $v_{0}$ has $\SPArity{v_{0}} = 0$, and so $\Graph''$ is an ari-graph. Let the input $\SPInput'^{\WireOracles}$ be identical to $\SPInput^{\WireOracles}$ except that $\SPLeaf[v_{\Depth}]'^{\WireOracles}(\vec{\VariableZ}) \DefineEqual \SPLeaf[v_{\Depth}]^{\WireOracles}(\vec{x}, \vec{\VariableZ})$; that is, we `hard-code' the input $\vec{x}$ to the circuit $\Circuit$ into the input $\SPInput'^{\WireOracles}$ of the sum-product circuit (as in the original GKR protocol). This can only reduce the degree of $\SPLeaf[v_{\Depth}]'^{\WireOracles}$.
Now let $\SPCircuit'' \DefineEqual (\Field, \SPSubset, \InternalVertexSetDegree, \SPLeafDegree, \Graph'', \SPPoly')$; it holds that that $\SPCircuit''$ is a sum-product circuit and $\SPValueL{\SPInput'^{\WireOracles}}{\SPCircuit''} = \SPValueL{\SPInput^{\WireOracles}}{\SPCircuit'}(\vec{x})$.

We obtain the following parameters: $\SetCardinality{\VertexSet(\Graph'')} = \Theta(\Depth(n))$, $\SPMaxArity(\Graph'') = O(\log \Size(n)/\log \SetCardinality{\Field})$, $\GraphMaxInDegree{\Graph''} = \Theta(1)$, $\Width{\Graph''} = \Theta(1)$; and $\MaxSpace{\SPInput'^{\WireOracles}} = O(\log \Size(n))$.

\parhead{Part C: composing sum-product subcircuits}
The final sum-product circuit $\SPCircuit = \SPCircuitTuple$ is constructed as follows. The ari-graph $\Graph$ is obtained from $\Graph''$ by replacing each leaf node $u_{f_{i}}$ with the ari-graph $\Graph_{f_{i}}$, every $f \in \Add{}, \Mult{}$ and $i \in \{1,\dots,\Depth\}$. The vertex label $\SPPoly$ is the union of the vertex labels $\SPPoly'$ and $\{\SPPoly[f_{i}]\}_{f \in \{\Add{}, \Mult{}\}, i \in \{1,\dots,\Depth\}}$.

The input $\SPInput$ for $\SPCircuit$ is equals the union of $\{\SPInput_{f_{i}}\}_{f \in \{\Add{}, \Mult{}\}, i \in \{1,\dots,\Depth\}}$ and the mapping $\SPLeaf[v_{\Depth}] \DefineEqual \SPLeaf[v_{\Depth}]'^{\WireOracles}$ (for $v_{\Depth} \in \Graph''$).

Given the above definitions, one can verify that $\SPValueL{\SPInput}{\SPCircuit} = \Circuit(\vec{x})$.

Moreover, by inspection: $\SetCardinality{\VertexSet(\Graph)} = \Depth(n) \cdot \poly(\log \Size(n))$, $\SPMaxArity(\Graph) = O(\log \Size(n) / \log \SetCardinality{\Field})$, $\GraphMaxInDegree{\Graph} = \Theta(1)$, $\Width{\Graph} = \poly(\log \Size(n))$, and $\MaxSpace{\SPInput} = O(\log \Size(n))$.

\parhead{Part D: invoke PZK IPCP for sum-product circuit evaluation}
All oracles have now been instantiated, we can now rely on our construction of perfect zero knowledge Interactive PCPs for sum-product circuit evaluation problems. More precisely, the prover and verifier construct the sum-product circuit $\SPCircuit$ and its input $\SPInput$ as above, in time $n \cdot \poly(\Depth(n), \log \Size(n))$ and space $O(\log \Size(n))$. (While this amount of space is not enough to store the whole circuit at once, the verifier can construct each part as needed.) The prover and verifier then engage in the protocol of \thmref{thm:pzk-for-SPCE} on the input $(\SPCircuit, 1, \SPInput)$. Correctness comes from the fact that $(\SPCircuit, 1, \SPInput) \in \SPCELanguage$ if and only if $\Circuit(\vec{x}) = 1$.

Plugging the parameters above into \thmref{thm:pzk-for-SPCE} yields the parameters claimed by \thmref{thm:gkr-space-uniform}, with the exception of the verifier space bound. A direct computation only shows that the verifier runs in space $\poly(\log \Size(n))$, because $\Width{\Graph} = \poly(\log \Size(n))$. Yet, since $\SPCircuit$ is `treelike' in the sense that the subcircuits can be handled independently, the evaluation protocol need only consider a constant number of vertices at any one time, at the expense of a $\poly(\log \Size(n))$-factor increase in the round complexity. Doing so yields the claimed $O(\log \Size(n))$ space complexity of the verifier.

\dtignore{
\doclearpage
\section{Zero knowledge for the satisfaction of low-depth circuits}
\label{sec:zk-gkr-np}

\subsection{Intro text}
In fact we obtain two kinds of zero knowledge guarantee, depending on whether we consider the circuit evaluation or circuit satisfaction problems. We provide results for both.

\parhead{Circuit satisfaction with standard zero knowledge}
One can also view the protocol of \cite{GoldwasserKR15} as an IPCP for delegating certain \emph{non-deterministic} computations: languages in $\NP$ whose witnesses can be checked in log-space uniform $\mathbf{NC}$. Namely, the prover first sends to the verifier the witness, and then the prover and verifier engage in the delegation protocol for the computation that checks the witness. If the witness is smaller than the circuit, then the verifier has saved in efficiency compared to directly checking the witness itself.

For this setting the standard (non-fine-grained) definition of zero knowledge suffices, and we prove the following.

\begin{theorem}[Informal version of \thmref{thm:gkr-space-uniform-np}]
	Languages in $\NP$ whose witness are checkable in log-space uniform $\mathbf{NC}$ have perfect zero knowledge Interactive PCPs, whose verifier runs in quasilinear time and logarithmic space.
\end{theorem}

\subsubsection{Zero knowledge from black-box cryptography}

All of our information-theoretic zero knowledge protocols in the IPCP model can be compiled into corresponding cryptographic proofs via transformations that preserve zero knowledge, while only making a \emph{black-box} use of cryptography. We already mentioned some compilation options in \secref{sec:our-focus}, and here we mention yet another one. As an example, from \thmref{thm:gkr-space-uniform-np} we can derive an analogue of the succinct zero knowledge proofs in \cite[Theorem 1.7]{GoldwasserKR15} by only making a black-box use of one-way functions.

\begin{theorem}
	Assume the existence of one-way functions. Let $\Language$ be an $\NP$ language whose witnesses are checkable in log-space uniform $\mathbf{NC}$. Then $\Language$ has a zero knowledge proof, which only makes black-box use of a one-way function, in which the verifier runs in quasilinear time and logarithmic space.
\end{theorem}

\subsection{Body text}

We can also consider the GKR protocol as a way to delegate certain \emph{nondeterministic} languages. This is considered in \cite[Section 5]{GoldwasserKR15}, in the context of zero knowledge proofs whose communication complexity is polynomial in the witness size (rather than the instance size).

The above result for low-depth circuits can be extended to nondeterministic computation by dividing the input to the circuit into an explicit input $\vec{x}$ and auxiliary input (witness) $\vec{z}$, which become $\SPInput$ and $\SPAuxInput$ respectively after transforming the circuit into a sum-product circuit. Note that what we actually obtain is a reduction from the satisfiability problem for the corresponding arithmetic circuit over $\Field$, which may be satisfiable even if the original boolean circuit was not. To remedy this, we also add gates which ensure that the auxiliary variables have values in $\Bits$. Applying \thmref{thm:pzk-for-SPCS} then yields the following theorem.

\begin{theorem}[PZK IPCP for $\NP$ relations checkable by low-depth uniform boolean circuits]
	\label{thm:gkr-space-uniform-np}
	Let $\Language$ be a language in $\NP$ checkable by a family of $O(\log \Size(n))$-space uniform boolean circuits of size $\Size(n)$ and depth $\Depth(n)$, and let $\SCStrength(n)$ be a query bound function. Then $\Language$ has a (public-coin and non-adaptive) Interactive PCP that is perfect zero knowledge against all $\SCStrength$-query malicious verifiers. In more detail:
	{\small
		\begin{equation*}
		\Language \in
		\PZKIPCP
		\left[
		\begin{tabular}{rl}
		soundness error: & $1/2$ \\
		round complexity: & $\Depth(n) \cdot \log \SCStrength(n) \cdot \poly(\log \Size(n))$ \\[1mm]
		proof length: & $\Depth(n) \cdot \poly(\Size(n), \SCStrength(n))$ \\
		query complexity: & $\Depth(n) \cdot \poly(\log \Size(n), \log \SCStrength(n))$ \\[1mm]
		prover time: & $\Depth(n) \cdot \poly(\Size(n), \SCStrength(n))$ \\
		verifier time: &  $n \cdot \poly(\Depth(n), \log \Size(n), \log \SCStrength(n))$ \\[1mm]
		verifier space: & $O(\log \Size(n) + \log \SCStrength(n)/\log\log\Size(n))$ \\[1mm]
		simulator overhead: & $\poly(\Depth(n), \log \Size(n), \log \SCStrength(n)) \cdot (n +  \QueryComplexity_{\Malicious{\Verifier}}^{3})$.
		\end{tabular}
		\right]
		\enspace.
		\end{equation*}}
	Moreover, if the verifier (resp. simulator) is given oracle access to the low-degree extension of the circuit's input, the verifier runs in time $\Depth(n) \cdot \poly(\log \Size(n), \log \SCStrength(n))$, and the simulator overhead is $\poly(\Depth(n), \log \Size(n), \log \SCStrength(n)) \cdot \QueryComplexity_{\Malicious{\Verifier}}^{3}$.
\end{theorem}
}

\doclearpage
\section*{Acknowledgments}

The authors would like to thank Eli Ben-Sasson for numerous enlightening conversations in early stages of this work. The authors also thank Thomas Vidick for suggesting that we use our techniques to obtain a zero knowledge analogue of \cite{BabaiFL91,BabaiFLS91}, as captured by \thmref{thm:pzk-for-nexp}. The authors thank Tom Gur for comments on the writeup.

\appendix

\doclearpage
\section{Algebraic query complexity of polynomial summation: details}
\label{sec:query-complexity-appendix}

In this section we prove the results whose statements appear in \secref{sec:algebraic-query-complexity}.

\subsection{Proof of \thmref{thm:sum-query-lower-bound}}

	First, since $\StrongRandPoly$ has individual degree at most $\SCDegree$ in $\vec{\VariableX}$, we can rewrite any such linear combination in the following way:
	\begin{equation*}
	\sum_{\vec{\alpha} \in L^{\SCVars}} C_{\vec{\alpha},i} \sum_{\vec{y} \in \SSCSubset^{\SSCVars}} \StrongRandPoly(\vec{\alpha}, \vec{y})
	= \sum_{\vec{\alpha} \in L^{\SCVars}} C_{\vec{\alpha},i} \sum_{\vec{\beta} \in K^{\SCVars}} b_{\vec{\beta},\vec{\alpha}} \sum_{\vec{y} \in \SSCSubset^{\SSCVars}} \StrongRandPoly(\vec{\alpha}, \vec{y})
	= \sum_{\vec{\alpha} \in K^{\SCVars}} C'_{\vec{\alpha},i} \sum_{\vec{y} \in \SSCSubset^{\SSCVars}} \StrongRandPoly(\vec{\alpha}, \vec{y})
	= \sum_{\vec{q} \in S} D_{\vec{q},i} \StrongRandPoly(\vec{q}) \enspace,
	\end{equation*}
	where $C' \DefineEqual BC$. If $\SCDegree' = \SetCardinality{\SSCSubset} - 2$ then the bound is trivial. Otherwise, let $H$ be some arbitrary subset of $G$ of size $\min\{\SCDegree' - \SetCardinality{\SSCSubset} + 2, \SetCardinality{\SSCSubset}\}$. Let $P_{0} \subseteq \PolynomialRingIndOneXY{\Field}{\SCVars}{\VariableX}{\SSCVars}{\VariableY}{\SCDegree}{\SetCardinality{\SCSubset}-1}$ be such that for all $p \in P_{0}$ and for all $\vec{q} \in S$, $p(\vec{q}) = 0$. Since this is at most $S$ linear constraints, $P_{0}$ has dimension at least $(\SCDegree+1)^{\SCVars}\SetCardinality{\SCSubset}^{\SSCVars}-\SetCardinality{S}$. Let $B_{0} \in \Field^{n \times (\SCDegree+1)^{\SCVars} \SetCardinality{\SCSubset}^{\SSCVars}}$ be a matrix whose rows form a basis for the vector space $\{\big(p(\vec{\alpha},\vec{\beta})\big)_{\vec{\alpha} \in K^{\SCVars}, \vec{\beta} \in \SCSubset^{\SSCVars}} : p \in P_{0}\}$ of evaluations of polynomials in $P_{0}$ on $K^{\SCVars} \times \SCSubset^{\SSCVars}$; we have $n \geq (\SCDegree+1)^{\SCVars}\SetCardinality{\SCSubset}^{\SSCVars}-\SetCardinality{S}$. By an averaging argument there exists $\vec{\beta}_{0} \in \SCSubset^{k}$ such that the submatrix $B_{\vec{\beta}_{0}}$ consisting of columns $(\vec{\alpha},\vec{\beta}_{0})$ of $B_{0}$ for each $\vec{\alpha} \in K^{\SCVars}$ has rank at least $(\SCDegree+1)^{\SCVars}-\SetCardinality{S}/\SetCardinality{\SCSubset}^{\SSCVars}$.
	
	Let $q \in \PolynomialRingIndOne{\Field}{\SSCVars}{\VariableY}{\SetCardinality{\SSCSubset}-1}$ be the polynomial such that $q(\vec{\beta}_{0}) = 1$, and $q(\vec{y}) = 0$ for all $\vec{y} \in \SSCSubset^{\SSCVars} - \{\vec{\beta}_{0}\}$. For arbitrary $p \in P_{0}$, let $Z(\vec{\VariableX}, \vec{\VariableY}) \DefineEqual q(\vec{\VariableY}) p(\vec{\VariableX},\vec{\VariableY}) \in \PolynomialRingIndOneXY{\Field}{\SCVars}{\VariableX}{\SSCVars}{\VariableY}{\SCDegree}{\SetCardinality{\SCSubset}+\SetCardinality{\SSCSubset}-2}$. Observe that our choice of $\SCSubset$ ensures that the degree of $\StrongRandPoly$ in $\vec{\VariableY}$ is at most $\SCDegree'$. Then for all $i \in \{1, \ldots, \ell\}$, it holds that
	\begin{equation*}
	\sum_{\vec{\alpha} \in K^{\SCVars}} C'_{\vec{\alpha},i} \sum_{\vec{y} \in \SSCSubset^{\SSCVars}} Z(\vec{\alpha},\vec{y})
	= \sum_{\vec{\alpha} \in K^{\SCVars}} C'_{\vec{\alpha},i} \cdot p(\vec{\alpha}, \vec{\beta}_{0})
	= \sum_{\vec{q} \in S} D_{\vec{q},i} \cdot Z(\vec{\alpha},\vec{y}) = 0 \enspace.
	\end{equation*}
	Thus the column space of $C'$ is contained in the null space of $B_{\vec{\beta}_{0}}$, and so the null space of $B_{\vec{\beta}_{0}}$ has rank at least $\rank(C')$. Hence $(d+1)^{m} - \rank(C') \geq \rank(B_{\vec{\beta}_{0}}) \geq (d+1)^{m} - \SetCardinality{S}/\SetCardinality{\SCSubset}^{\SSCVars}$, so $\SetCardinality{S} \geq \rank(C') \cdot \SetCardinality{\SCSubset}^{\SSCVars}$, which yields the theorem. \qed

\subsection{Proof of \corref{cor:partial-sum-indep-vars}}

We will need a simple fact from linear algebra: that `linear independence equals statistical independence'. That is, if we sample an element from a vector space and examine some subsets of its entries, these distributions are independent if and only if there does not exist a linear dependence between the induced subspaces. The formal statement of the claim is as follows; its proof is deferred to the end of this subsection.
	
\begin{claim}
	\label{claim:linear-algebra-fact}
	Let $\Field$ be a finite field and $\Domain$ a finite set. Let $V \subseteq \Field^{\Domain}$ be an $\Field$-vector space, and let $\vec{v}$ be a random variable which is uniform over $V$. For any subdomains $S, S' \subseteq \Domain$, the restrictions $\Restrict{\vec{v}}{S}$ and $\Restrict{\vec{v}}{S'}$ are statistically dependent if and only if there exist constants $(c_{i})_{i \in S}$ and $(d_{i})_{i \in S'}$ such that:
	\begin{itemize}
		\item there exists $\vec{w} \in V$ such that $\sum_{i \in S} c_{i} w_{i} \neq 0$, and
		\item for all $\vec{w} \in V$, $\sum_{i \in S} c_{i} w_{i} = \sum_{i \in S'} d_{i} w_{i}$.
	\end{itemize}
\end{claim}
Now observe that 
\begin{equation*}
\Big\{ \Big(\big(\StrongRandPoly(\vec{\gamma})\big)_{\vec{\gamma} \in \Field^{\SCVars+\SSCVars}}, \big(\sum_{\vec{y} \in \SSCSubset^{\SSCVars}} \StrongRandPoly(\vec{\alpha},\vec{y})\big)_{\vec{\alpha} \in \Field^{\SCVars}}\Big) : \StrongRandPoly \in \PolynomialRingIndOneXY{\Field}{\SCVars}{\VariableX}{\SSCVars}{\VariableY}{\SCDegree}{\SCDegree'} \Big\}
\end{equation*}
is an $\Field$-vector space with domain $\Field^{\SCVars+\SSCVars} \cup \Field^{\SCVars}$. Consider subdomains $\Field^{\SCVars}$ and $S$. Since $\SetCardinality{S} < \SetCardinality{\SSCSubset}^{\SSCVars}$, by \thmref{thm:sum-query-lower-bound} there exist no constants $(c_{\vec{\alpha}})_{\alpha \in \Field^{\SCVars}}$, $(d_{\vec{\gamma}})_{\vec{\gamma} \in S}$ such that the conditions of the claim hold. This completes the proof.

\begin{proof}[Proof of \clmref{claim:linear-algebra-fact}]
	For arbitrary $\vec{x} \in \Field^{S}, \vec{x}' \in \Field^{S'}$, we define the quantity
	\begin{equation*}
	p_{\vec{x},\vec{x}'} \DefineEqual \Pr_{\vec{v} \in V}
	\left[
	\Restrict{\vec{v}}{S} = \vec{x} \wedge \Restrict{\vec{v}}{S'} = \vec{x}'
	\right] \enspace.
	\end{equation*}
	Let $d \DefineEqual \dim(V)$, and let $B \in \Field^{D \times d}$ be a basis for $V$. Let $B_{S} \in \Field^{S \times d}$ be $B$ restricted to rows corresponding to elements of $S$, and let $B_{S'}$ be defined likewise. Finally, let $B_{S,S'} \in \Field^{(\SetCardinality{S}+\SetCardinality{S'})\times d}$ be the matrix whose rows are the rows of $B_{S}$, followed by the rows of $B_{S'}$. Then
	\begin{equation*}
	p_{\vec{x},\vec{x}'} = \Pr_{\vec{z} \in \Field^{d}}
	\left[
	B_{S,S'} \cdot \vec{z} = (\vec{x}, \vec{x}')
	\right] \enspace.
	\end{equation*}
	One can verify that, for any matrix $A \in \Field^{m \times n}$,
	\begin{equation*}
	\Pr_{\vec{z} \in \Field^{n}} [A\vec{z} = \vec{b}] = \begin{cases}
	\Field^{-\rank(A)} & \text{ if $\vec{b} \in \mathrm{colsp}(A)$, and} \\
	0 & \text{otherwise.}
	\end{cases}
	\end{equation*}
	The column space $\mathrm{colsp}(B_{S,S'}) \subseteq \mathrm{colsp}(B_{S}) \times \mathrm{colsp}(B_{S'})$, and equality holds if and only if $\rank(B_{S,S'}) = \rank(B_{S}) + \rank(B_{S'})$. It follows that $p_{\vec{x},\vec{x}'} = \Pr_{\vec{v} \in V}[\Restrict{\vec{v}}{S} = \vec{x}] \cdot \Pr_{\vec{v} \in V}[\Restrict{\vec{v}}{S'} = \vec{x}']$ if and only if $\rank(B_{S,S'}) = \rank(B_{S}) + \rank(B_{S'})$. By the rank-nullity theorem and the construction of $B_{S,S'}$, this latter condition holds if and only if $\mathrm{nul}(B_{S,S'}^{T}) \subseteq \mathrm{nul}(B_{S}^{T}) \times \mathrm{nul}(B_{S'}^{T})$. To conclude the proof, it remains only to observe that the condition in the claim is equivalent to the existence of vectors $\vec{c} \in \Field^{S}$, $\vec{d} \in \Field^{S'}$ such that $\vec{c} \notin \mathrm{nul}(B_{S}^{T})$ but $(\vec{c},-\vec{d}) \in \mathrm{nul}(B_{S,S'}^{T})$.
\end{proof}

\subsection{Upper bounds}
\label{sec:upper-bounds}

In this section we show that in certain cases the degree constraints in \thmref{thm:sum-query-lower-bound} are tight.

\subsubsection{The case of multilinear polynomials}
\label{sec:multilinear-polynomials}

The first result is for the case of multivariate polynomials over any finite field, where $\SCSubset \subseteq \Field$ is arbitrary. The proof is a simple extension of a proof due to \cite{JumaKRS09} for the case $\SCSubset = \Bits$.

\begin{theorem}[Multilinear Polynomials]
	\label{thm:multilinear-upper-bound}
	Let $\Field$ be a finite field, $\SCSubset$ a subset of $\Field$, and $\gamma \DefineEqual \sum_{\alpha \in \SCSubset} \alpha$. For every $\PolyA \in \PolynomialRingIndOne{\Field}{\SCVars}{\VariableX}{1}$ (i.e., for every $\SCVars$-variate multilinear polynomial $\PolyA$) it holds that
	\begin{equation*}
	\sum_{\vec{\alpha} \in \SCSubset^{\SCVars}} \PolyA(\vec{\alpha})
	=
	\begin{cases}
	\PolyA\big(\frac{\gamma}{\SetCardinality{\SCSubset}},\dots,\frac{\gamma}{\SetCardinality{\SCSubset}}\big) \cdot \SetCardinality{\SCSubset}^{\SCVars} & \text{ if } \Characteristic{\Field} \nmid \SetCardinality{\SCSubset} \\
	\kappa \cdot \gamma^{\SCVars} & \text{ if } \Characteristic{\Field} \mid \SetCardinality{\SCSubset}
	\end{cases}
	\enspace,
	\end{equation*}
	where $\kappa$ is the coefficient of $\VariableX_{1} \cdots \VariableX_{\SCVars}$ in $\PolyA$.
\end{theorem}

\begin{proof}
	First suppose that $\Characteristic{\Field}$ does not divide $\SetCardinality{\SCSubset}$. Let $\vec{\alpha}$ be uniformly random in $\SCSubset^{\SCVars}$; in particular, $\alpha_{i}$ and $\alpha_{j}$ are independent for $i \neq j$. For every monomial $m(\vec{\VariableX}) = \VariableX_{1}^{e_{1}} \cdots \VariableX_{\SCVars}^{e_{\SCVars}}$ with $e_{1},\dots,e_{\SCVars} \in \Bits$,
	\begin{equation*}
	\Expectation[M(\vec{\alpha})]
	= \Expectation[\alpha_{1}^{e_{1}} \cdots \alpha_{\SCVars}^{e_{\SCVars}}]
	= \Expectation[\alpha_{1}^{e_{1}}] \cdots \Expectation[\alpha_{\SCVars}^{e_{\SCVars}}]
	= \Expectation[\alpha_{1}]^{e_{1}} \cdots \Expectation[\alpha_{\SCVars}]^{e_{\SCVars}}
	= M(\Expectation[\alpha_{1}], \dots, \Expectation[\alpha_{\SCVars}])
	\enspace.
	\end{equation*}
	Since $\PolyA$ is a linear combination of monomials, $\Expectation[\PolyA(\vec{\alpha})] = \PolyA(\Expectation[\vec{\alpha}])$. Each $\alpha_{i}$ is uniformly random in $\SCSubset$, so $\Expectation[\alpha_{i}] = \frac{1}{\SetCardinality{\SCSubset}} \sum_{\alpha \in \SCSubset} \alpha = \frac{\gamma}{\SetCardinality{\SCSubset}}$, and thus $\PolyA(\Expectation[\vec{\alpha}]) = \PolyA(\frac{\gamma}{\SetCardinality{\SCSubset}},\dots,\frac{\gamma}{\SetCardinality{\SCSubset}})$, which implies that $\Expectation[\PolyA(\vec{\alpha})] = \PolyA(\frac{\gamma}{\SetCardinality{\SCSubset}},\dots,\frac{\gamma}{\SetCardinality{\SCSubset}})$. To deduce the claimed relation, it suffices to note that $\Expectation[\PolyA(\vec{\alpha})] = \frac{1}{\SetCardinality{\SCSubset}^{\SCVars}}\sum_{\vec{\alpha} \in \SCSubset^{\SCVars}} \PolyA(\vec{\alpha})$.
	
	Next suppose that $\Characteristic{\Field}$ divides $\SetCardinality{\SCSubset}$. For every monomial $m(\vec{\VariableX}) = \VariableX_{1}^{e_{1}} \cdots \VariableX_{\SCVars}^{e_{\SCVars}}$ with $e_{1},\dots,e_{\SCVars} \in \Bits$:
	\begin{itemize}
		
		\item if there exists $j \in [\SCVars]$ such that $e_j = 0$ then
		\begin{equation*}
		\sum_{\vec{\alpha} \in \SCSubset^{\SCVars}} M(\vec{\alpha})
		= \SetCardinality{\SCSubset}
		\sum_{\alpha_{1},\dots,\alpha_{j-1},\alpha_{j+1},\dots,\alpha_{\SCVars} \in \SCSubset}
		\alpha_{1}^{e_{1}} \cdots \alpha_{j-1}^{e_{j-1}} \alpha_{j+1}^{e_{j+1}} \cdots \alpha_{\SCVars}^{e_{\SCVars}}
		= 0
		\enspace.
		\end{equation*}
		
		\item if instead $e_{1} = \cdots = e_{\SCVars} = 1$ then
		\begin{equation*}
		\sum_{\vec{\alpha} \in \SCSubset^{\SCVars}} M(\vec{\alpha})
		= \sum_{\vec{\alpha} \in \SCSubset^{\SCVars}}\prod_{i=1}^{\SCVars} \alpha_{i}
		= \prod_{i=1}^{\SCVars} \sum_{\alpha_{i} \in \SCSubset} \alpha_{i}
		= \left(\sum_{\alpha \in \SCSubset} \alpha \right)^{\SCVars}
		\enspace. \qedhere
		\end{equation*}
	\end{itemize} 
\end{proof}

The following corollary shows that for prime fields of odd size, the value of $\sum_{\vec{\alpha} \in \SCSubset^{\SCVars}} \PolyA(\vec{\alpha})$ can be computed efficiently for any $\SCSubset \subseteq \Field$ using at most a single query to $\PolyA$.

\begin{corollary}
	\label{cor:prime-field-multilinear}
	Let $\Field$ be a prime field of odd size, $\SCSubset$ a subset of $\Field$, and $\gamma \DefineEqual \sum_{\alpha \in \SCSubset} \alpha$. For every $\PolyA \in \PolynomialRingIndOne{\Field}{\SCVars}{\VariableX}{1}$ (i.e., for every $\SCVars$-variate multilinear polynomial $\PolyA$) it holds that
	\begin{equation*}
	\sum_{\vec{\alpha} \in \SCSubset^{\SCVars}} \PolyA(\vec{\alpha})
	=
	\begin{cases}
	\PolyA\big(\frac{\gamma}{\SetCardinality{\SCSubset}},\dots,\frac{\gamma}{\SetCardinality{\SCSubset}}\big) \cdot \SetCardinality{\SCSubset}^{\SCVars} & \text{ if } \Characteristic{\Field} \nmid \SetCardinality{\SCSubset} \\
	0 & \text{ if } \Characteristic{\Field} \mid \SetCardinality{\SCSubset}
	\end{cases}
	\enspace.
	\end{equation*}
\end{corollary}
\begin{proof}
	\thmref{thm:multilinear-upper-bound} directly implies both cases. If $\Characteristic{\Field}$ does not divide $\SetCardinality{\SCSubset}$, then the claimed value is as in the theorem. If instead $\Characteristic{\Field}$ divides $\SetCardinality{\SCSubset}$, then it must be the case that $\SCSubset = \Field$, since $p \DefineEqual \Characteristic{\Field}$ equals $\SetCardinality{\Field}$; in this case, $\gamma = \sum_{\alpha \in \SCSubset} \alpha = (p-1)p/2$, which is divisible by $p$ since $2$ must divide $p-1$ (as $p$ is odd).
\end{proof}

\subsubsection{The case of subsets with group structure}
\label{sec:case-of-groups}

In this section we show that if $\SCSubset$ is assumed to have some group structure, then few queries may suffice even for polynomials of degree greater than one. In particular, the following result shows that a single query suffices for $\SCDegree \leq \SetCardinality{\SCSubset}$ when $\SCSubset$ is a multiplicative subgroup of $\Field$.

\begin{lemma}[Multiplicative Groups]
	\label{lem:multiplicative-group}
	Let $\Field$ be a field, $\RMSubDomain$ a finite multiplicative subgroup of $\Field$, and $\RMVars,\RMDegree$ positive integers with $\RMDegree < \SetCardinality{\RMSubDomain}$. For every $\PolyA \in \PolynomialRingIndOne{\Field}{\RMVars}{\VariableX}{\RMDegree}$,
	\begin{equation*}
	\sum_{\vec{\alpha} \in \RMSubDomain^{\RMVars}} \PolyA(\vec{\alpha})
	= \PolyA(0,\dots,0) \cdot \SetCardinality{\RMSubDomain}^{\RMVars}
	\enspace.
	\end{equation*}
\end{lemma}

\begin{remark}
	\label{rem:necessity-of-degree-bound}
	The hypothesis that $\RMDegree < \SetCardinality{\RMSubDomain}$ is necessary for the lemma, as we now explain. Choose $\RMSubDomain = \Multiplicative{\SubField}$, where $\SubField$ is a proper subfield of $\Field$, $\RMVars=1$, and $\RMDegree=\SetCardinality{\RMSubDomain}$. Consider the polynomial $\VariableX^{\SetCardinality{\RMSubDomain}}$, which has degree at least $\RMDegree$: $\VariableX^{\SetCardinality{\RMSubDomain}}$ vanishes on $0$; however, $\VariableX^{\SetCardinality{\RMSubDomain}}$ evaluates to $1$ everywhere on $\RMSubDomain$ so that its sum over $\RMSubDomain$ equals $\SetCardinality{\RMSubDomain} \neq 0$. (Note that if $\RMSubDomain$ is a multiplicative subgroup of $\Field$ then $\Characteristic{\Field} \nmid \SetCardinality{\RMSubDomain}$ because $\SetCardinality{\RMSubDomain}$ equals $\Characteristic{\Field}^{k}-1$ for some positive integer $k$.)
\end{remark}

\begin{proof}
	The proof is by induction on the number of variables $\RMVars$. The base case is when $\RMVars=1$, which we argue as follows. The group $\RMSubDomain$ is cyclic, because it is a (finite) multiplicative subgroup of a field; so let $\omega$ generate $\RMSubDomain$. Writing $\PolyA(\VariableX_{1})=\sum_{j=0}^{\RMDegree} \beta_{j} \VariableX_{1}^{j}$ for some $\beta_{0},\dots,\beta_{\RMDegree} \in \Field$, we have
	\begin{equation*}
	\sum_{\alpha_{1} \in \RMSubDomain} \PolyA(\alpha_{1})
	= \sum_{i=0}^{\SetCardinality{\RMSubDomain}-1} \PolyA(\omega^i)
	= \sum_{i=0}^{\SetCardinality{\RMSubDomain}-1} \sum_{j=0}^{\RMDegree} \beta_{j} \omega^{ij}
	= \sum_{j=0}^{\RMDegree} \beta_{j} \sum_{i=0}^{\SetCardinality{\RMSubDomain}-1} (\omega^{j})^{i}
	= \beta_{0} \SetCardinality{\RMSubDomain}
	= f(0) \SetCardinality{\RMSubDomain}
	\enspace,
	\end{equation*}
	which proves the base case. The second-to-last equality follows from the fact that for every $\gamma \in \RMSubDomain$,
	\begin{equation*}
	\sum_{i=0}^{\SetCardinality{\RMSubDomain}-1} \gamma^{i}
	=
	\begin{cases}
	\SetCardinality{\RMSubDomain} & \text{if $\gamma=1$} \\
	\frac{\gamma^{\SetCardinality{\RMSubDomain}}-1}{\gamma-1}=0 & \text{if $\gamma \neq 1$}
	\end{cases}
	\enspace.
	\end{equation*}
	
	For the inductive step, assume the statement for any number of variables less than $\RMVars$; we now prove that it holds for $\RMVars$ variables as well. Let $\PolyA_{\alpha}$ denote $\PolyA$ with the variable $\VariableX_{1}$ fixed to $\alpha$. Next, apply the inductive assumption below in the second equality (with $\RMVars-1$ variables) and last one (with $1$ variable), to obtain
	\begin{align*}
	\sum_{\vec{\alpha} \in \RMSubDomain^{\RMVars}} \PolyA(\alpha_{1},\dots,\alpha_{\RMVars})
	&= \sum_{\alpha_{1}\in \RMSubDomain} \sum_{(\alpha_{2},\dots,\alpha_{\RMVars})\in \RMSubDomain^{\RMVars-1}} \PolyA_{\alpha_{1}}(\alpha_{2},\dots,\alpha_{\RMVars}) \\
	&= \SetCardinality{\RMSubDomain}^{\RMVars-1} \sum_{\alpha_{1} \in \RMSubDomain} \PolyA_{\alpha_{1}}(0^{\RMVars-1}) \\
	&= \SetCardinality{\RMSubDomain}^{\RMVars-1} \sum_{\alpha_{1} \in \RMSubDomain} \PolyA(\alpha_{1},0,\dots,0) \\
	&= \SetCardinality{\RMSubDomain}^{\RMVars} \PolyA(0,\dots,0)
	\enspace,
	\end{align*}
	as claimed.
\end{proof}

\begin{lemma}[Additive Groups]
	\label{lem:additive-group}
	Let $\Field$ be a field, $\RMSubDomain$ a finite additive subgroup of $\Field$, and $\RMVars,\RMDegree$ positive integers with $\RMDegree < \SetCardinality{\RMSubDomain}$. For every $\PolyA \in \PolynomialRingIndOne{\Field}{\RMVars}{\VariableX}{\RMDegree}$,
	\begin{equation*}
	\sum_{\vec{\alpha} \in \RMSubDomain^{\RMVars}} \PolyA(\vec{\alpha})
	= \kappa \cdot a_{0}^{\RMVars} \enspace,
	\end{equation*}
	where $\kappa$ is the coefficient of $\VariableX_{1}^{\SetCardinality{\RMSubDomain}-1} \cdots \VariableX_{\RMVars}^{\SetCardinality{\RMSubDomain}-1}$ in $\PolyA$, and $a_{0}$ is the (formal) linear term of the subspace polynomial $\prod_{h \in \RMSubDomain} (\VariableX - h)$. In particular, if $\PolyA$ has total degree strictly less than $\RMVars (\SetCardinality{\RMSubDomain}-1)$, then the above sum evaluates to $0$.
\end{lemma}
\begin{proof}
	Without loss of generality, let $\RMDegree := \SetCardinality{\RMSubDomain} - 1$. The proof is by induction on the number of variables $\RMVars$. When $\RMVars = 1$, we have that $\PolyA(\VariableX) = \sum_{j=0}^{\RMDegree} \beta_{j} \VariableX^{j}$ for some $\beta_{0}, \ldots, \beta_{\RMDegree} \in \Field$. Then
	\begin{equation*}
	\sum_{\alpha \in \RMSubDomain} \PolyA(\alpha)
	= \sum_{\alpha \in \RMSubDomain} \sum_{j=0}^{\RMDegree} \beta_{j} \alpha^{j}
	= \sum_{j=0}^{\RMDegree} \beta_{j} \sum_{\alpha \in \RMSubDomain} \alpha^{j}
	= \beta_{\RMDegree} a_{0}
	\end{equation*}
	where the final equality follows by \cite[Theorem 1]{ByottC99}, and the fact that $d = \SetCardinality{\RMSubDomain} - 1$.
	
	For the inductive step, assume the statement for $\RMVars-1$ variables; we now prove that it holds for $\RMVars$ variables as well. Let $\PolyA_{\alpha}$ denote $\PolyA$ with the variable $\VariableX_{1}$ fixed to $\alpha$; we have $\PolyA_{\alpha}(\VariableX_{2}, \dots, \VariableX_{\RMVars}) = \sum_{\vec{e} \in \{0,\dots,\RMDegree\}^{\RMVars}} \beta_{\vec{e}} \cdot \alpha^{e_{1}} \VariableX_{2}^{e_{2}} \dots \VariableX_{\RMVars}^{e_{\RMVars}}$. Next, apply the inductive hypothesis below in the second equality (with $\RMVars-1$ variables) to obtain
	\begin{equation*}
	\sum_{\vec{\alpha} \in \RMSubDomain^{\RMVars}} \PolyA(\alpha_{1},\dots,\alpha_{\RMVars})
	= \sum_{\alpha_{1}\in \RMSubDomain} \sum_{(\alpha_{2},\dots,\alpha_{\RMVars})\in \RMSubDomain^{\RMVars-1}} \PolyA_{\alpha_{1}}(\alpha_{2},\dots,\alpha_{\RMVars})
	= \sum_{\alpha_{1}\in \RMSubDomain} a_{0}^{\RMVars-1} \kappa(\alpha_{1}) \enspace,
	\end{equation*}
	where $\kappa(\VariableX_{1}) := \sum_{j=0}^{\RMDegree} \beta_{(j, \RMDegree, \dots, \RMDegree)} \VariableX_{1}^{j}$. Applying the hypothesis again for $1$ variable yields
	\begin{equation*}
	\sum_{\alpha_{1}\in \RMSubDomain} a_{0}^{\RMVars-1} \kappa(\alpha_{1}) = a_{0}^{\RMVars} \cdot \beta_{(\RMDegree, \dots, \RMDegree)} \enspace,
	\end{equation*}
	and the claim follows.
\end{proof}

\dtignore{
\doclearpage
\newcommand{\ILS}{\Lambda}
\section{Impossibility of perfectly hiding efficient-prover interactive locking schemes}

Following \cite{GoyalIMS10}, we define an interactive locking scheme.

\begin{definition}
	An interactive locking scheme $\ILS = (\Prover, \Verifier)$ is an IPCP of the following form:
	
	\begin{itemize}
		\item The common input is of the form $1^{n}$ where $n$ is the security parameter.
		\item The prover $\Prover$ receives a private input $b \in \Bits$ and private randomness $r$.
		\item The protocol has two distinct phases:
		\begin{inparaenum}[(i)]
			\item the commitment phase, consisting of the oracle message $\Proof$ and some rounds of interaction between $\Prover$ and $\Verifier$, and
			\item the decommitment phase, consisting of some rounds of interaction between $\Prover$ and $\Verifier$.
		\end{inparaenum}
		The verifier $\Verifier$ may query the oracle $\Proof$ at any time during the protocol. After the decommitment phase, the verifier may either accept and output $b \in \Bits$ or reject.
		\item \textsc{Completeness. }
		For any $b \in \Bits$, $\Verifier$ accepts with probability $1$ when interacting with $\Prover$.
		\item \textsc{Binding. }
		$\ILS$ is $(1-\delta)$-binding if for any $\Malicious{\Prover}$, the probability over $\Verifier$'s randomness in the commitment phase that the commitment is $\delta$-ambiguous is at most $\delta$. The commitment is $\delta$-ambiguous if there exist malicious provers $\Malicious{\Prover}_{0}$, $\Malicious{\Prover}_{1}$ such that for each $i \in \Bits$, $\Verifier$ accepts and outputs $i$ when interacting with $\Malicious{\Prover}_{i}$ in the decommitment phase with probability at least $\delta$.
		\item \textsc{Hiding. }
		Let $\Malicious{\Verifier}$ be a malicious verifier, and let $X_{b}$ be the random variable corresponding to its view (i.e., oracle queries and prover messages) when interacting with $\Prover$ on input $b$. Then $\ILS$ is $(\epsilon,\QueryBound)$-hiding if the statistical distance between $X_{0}$ and $X_{1}$ is at most $\epsilon$, for all $\Malicious{\Verifier}$ making at most $\QueryBound$ queries to $\Proof$.
		\item \textsc{Efficiency. }
		We say $\ILS$ is efficient if $\Prover$ and $\Proof$ can be implemented by randomized polynomial-time Turing machines. In particular, $\Proof$ may be of exponential length, but there must be a stateless polynomial-time algorithm for computing the result of any query.
	\end{itemize}
	
\end{definition}

\begin{theorem}
	There exists no $\ILS$ which is $\Omega(1)$-binding, $(0,n^{\omega(1)})$-hiding and efficient.
\end{theorem}
\begin{proof}
	Let $\ILS = (\Prover, \Verifier)$ be a locking scheme which is $(1-\delta)$ binding and efficient. In particular, suppose both $\Prover$ and $\Verifier$ can be implemented by a randomized Turing machine in time $n^{c}$. Without loss of generality, we may assume that $\Verifier$ makes no queries to $\Proof$ during the commitment phase: it can instead ask the prover to provide these values and check them during the decommitment phase. We design a verifier making at most $\poly(n)$ queries which distinguishes the cases $b = 0$ and $b = 1$ with nonzero probability, which proves the statement.
	
	Let $\Malicious{\Verifier}$ be a malicious verifier defined as follows. $\Malicious{\Verifier}$ chooses $r' \in \Bits^{n^{c}}$. It then simulates $\Verifier$ interacting with $\Prover(0;r')$ $k$ times, and outputs $0$ if and only if every run outputs $0$; otherwise it outputs $1$.
	
	Suppose that $\Proof$ is generated by $P(0;r)$ for some $r \in \Bits^{n^{c}}$. Then with probability at least $2^{-n^{c}}$, $\Malicious{\Verifier}$ chooses $r'$ that is identical to the real randomness of the prover. By the completeness of the scheme, with probability at least $2^{-n^{c}}$, $\Malicious{\Verifier}$ outputs $0$.
	
	Now suppose that $\Proof$ is generated by $P(1;r)$ for some $r \in \Bits^{n^{c}}$. On each of the $k$ runs, $\Verifier$ outputs $0$ with probability at most $\delta$. Hence the probability that $\Verifier$ outputs $0$ is at most $\delta^{k} < 2^{-n^{c}}$ for $k > n^{c}/(\log 1/\delta)$. The number of queries $\Malicious{\Verifier}$ makes to $\Proof$ is $n^{2c}/(\log 1/\delta)$, which is at most $n^{\omega(1)}$ for $\delta = O(1)$.
\end{proof}
}

\doclearpage
\section{Proof of \thmref{thm:pzk-for-nexp} via sum-product circuits}
\label{sec:zk-nexp-spc}

In this section we re-prove \thmref{thm:pzk-for-nexp} using sum-product circuits.
In particular we reduce $\OracleSATRelation$ to $\SPCSRelation$ by constructing a sum-product circuit whose satisfaction encodes oracle 3-satisfiability. The reduction then yields \thmref{thm:pzk-for-nexp} by the perfect zero knowledge IPCP for $\SPCSRelation$ given in \thmref{thm:pzk-for-SPCS}.

\begin{lemma}[$\OracleSATRelation \to \SPCSRelation$]
	\label{lemma:o3sat-to-spcs}
	Let $\SPSubset$ be an extension field of $\Field_{2}$ with $\SetCardinality{H} \geq 8$, and $\Field$ an extension field of $\SPSubset$. There exist polynomial-time functions $f,g$ such that for every $r, s \in \Naturals$ and boolean formula $B \colon \Bits^{r + 3s + 3} \to \Bits$:
	\begin{enumerate}
		\item if $A \colon \Bits^{s} \to \Bits$ is such that $((r, s, B), A) \in \OracleSATRelation$ then
		\begin{equation*}
		\Pr_{\vec{x}, \vec{y} \gets \Field^{r + 3s}}\Big[
		\Big((f(r, s, B, \vec{x}, \vec{y}), 0, \EmptyVector), g(A)\Big) \in \SPCSRelation
		\Big]=1
		\enspace;
		\end{equation*}
		\item if $(r, s, B) \notin \Language(\OracleSATRelation)$ then
		\begin{equation*}
		\Pr_{\vec{x}, \vec{y} \gets \Field^{r + 3s}}\Big[
		\Big(f(r, s, B, \vec{x}, \vec{y}), 0, \EmptyVector\Big) \in \Language(\SPCSRelation)
		\Big]\leq \frac{r + 3s}{\SetCardinality{\Field}}
		\enspace.
		\end{equation*}
	\end{enumerate}
	Moreover, $f(r, s, B, \vec{x}, \vec{y})$ is a sum-product circuit $\SPCircuitTuple$ with $\InternalVertexSetDegree = O(\BitSize{B} \cdot \SetCardinality{\SPSubset})$, $\SPLeafDegree = O(\SetCardinality{\SPSubset})$, $\SetCardinality{\VertexSet(\Graph)} = \Theta(1)$, $\SPMaxArity(\Graph) = O((r + s)/\log \SetCardinality{\SPSubset})$, $\GraphMaxInDegree{\Graph} = \Theta(1)$, $\Width{\Graph} = \Theta(1)$.
\end{lemma}

\begin{proof}
	Recall that in \secref{sec:zk-nexp} we reduced checking $((B,r,s),A) \in \OracleSATRelation$ to checking whether the following expression is the zero polynomial, for some polynomials $g_{1}, g_{2}$ which depend on the low-degree extension of $A$.
	\begin{align*}
	F(\vec{\VariableX}, \vec{\VariableY}) =
	\sum_{\substack{\vec{\alpha} \in \SCSubset^{m_{1}} \\
			\vec{\beta}_{1}, \vec{\beta}_{2}, \vec{\beta}_{3} \in \SCSubset^{m_{2}}}}
	&g_{1}(\LD{\gamma}(\vec{\alpha}, \vec{\beta}_{1}, \vec{\beta}_{2}, \vec{\beta}_{3}))
	\prod_{i=1}^{r+3s} (1 + (\VariableX_{i} - 1)\LD{\gamma}(\vec{\alpha}, \vec{\beta}_{1}, \vec{\beta}_{2}, \vec{\beta}_{3})_{i}) \\
	&+g_{2}(\LD{\gamma}_{2}(\vec{\beta}_{1}))
	\prod_{i=1}^{r+3s} (1 + (\VariableY_{i} - 1)\LD{\gamma}(\vec{\alpha}, \vec{\beta}_{1}, \vec{\beta}_{2}, \vec{\beta}_{3})_{i}) \enspace.
	\end{align*}
	
	Fix $\vec{x}, \vec{y} \in \Field^{r + 3s}$. We construct a sum-product circuit $\SPCircuit = \SPCircuitTuple$ whose value with explicit input $\EmptyVector$ (the empty map) and auxiliary input $\LD{A}$ is $F(\vec{x}, \vec{y})$. Fix $\InternalVertexSetDegree \DefineEqual O(\BitSize{B} \cdot \SetCardinality{H})$ and $\SPLeafDegree \DefineEqual \SetCardinality{\SPSubset}$. Let $G = (\VertexSet, \EdgeSet, \SPSumProjection, \SPFreeProjection)$ be the ari-graph defined as follows:
	\begin{align*}
	\VertexSet &\DefineEqual \{ v, u \} \enspace, \\
	\EdgeSet &\DefineEqual \{ e_{1}, e_{2}, e_{3} = (v, u) \} \enspace, \\
	\SPFreeProjection[e] &\DefineEqual \varnothing \quad \text{ for all $e \in \EdgeSet$} \enspace, \\
	\SPSumProjection[e_{i}] &\DefineEqual m_{1} + (i-1)m_{2} + [m_{2}] \quad \text{ for $i \in \{1,2,3\}$} \enspace.
	\end{align*}
	The vertex label $\SPPoly$ is given by
	\begin{equation*}
	\SPPoly[v] \DefineEqual
	\left(
	\LD{B}(\vec{\VariableY}, \VariableZ_{1}, \VariableZ_{2}, \VariableZ_{3}) \cdot
	\prod_{i=1}^{r+3s+3}(1 + (x_{i} - 1))\LD{\gamma}(\vec{\VariableY})_{i}
	\right)
	+
	\left(
	\VariableZ_{1} (1-\VariableZ_{1}) \cdot
	\prod_{i=1}^{r+3s} (1 + (y_{i} - 1))\LD{\gamma}(\vec{\VariableY})_{i}
	\right)
	\enspace,
	\end{equation*}
	where the variables $Z_{1},Z_{2},Z_{3}$ correspond to the edges $e_{1},e_{2},e_{3}$ respectively.
	
	From the construction of the circuit, if $\SPAuxInput_{u} = \LD{A}(\LD{\gamma}_{2}(\cdot))$ then $\SPValueL{\SPInput, \SPAuxInput}{\SPCircuit} = F(\vec{x}, \vec{y})$. The stated parameters follow easily.
	
	If $((r, s, B), A) \in \OracleSATRelation$ then $F(\vec{\VariableX}, \vec{\VariableY})$ is the zero polynomial. Hence $F(\vec{x}, \vec{y}) = 0$ for all $\vec{x}, \vec{y} \in \Field^{r + 3s}$. We conclude that if $((r, s, B), A) \in \OracleSATRelation$ then $((\SPCircuit, 0, \SPInput), \SPAuxInput) \in \SPCSRelation$ with probability $1$ over the choice of $\vec{x}, \vec{y}$.
	
	If $(r, s, B) \notin \Language(\OracleSATRelation)$ then there is no choice of $\LD{A}$ such that $F(\vec{\VariableX}, \vec{\VariableY})$ is the zero polynomial. Thus, for any choice of $\LD{A}$, $F(\vec{x}, \vec{y}) = 0$ with probability at most $(r + 3s)/\SetCardinality{\Field}$ over the choice of $\vec{x}, \vec{y}$. We conclude that if $(r, s, B) \notin \Language(\OracleSATRelation)$ then $(\SPCircuit, 0, \SPInput) \in \Language(\SPCSRelation)$ with probability at most $(r + 3s)/\SetCardinality{\Field}$.
\end{proof}

\begin{proof}[Proof of \thmref{thm:pzk-for-nexp}]
	The protocol proceeds as follows. The prover and verifier determinstically (non-interactively) choose an extension field $\SPSubset$ of $\Field_{2}$ and an extension field $\Field$ of $\SPSubset$ such that $\SetCardinality{\Field} = \poly(\SetCardinality{\SPSubset})$ and $\SetCardinality{\SPSubset} = \poly(\BitSize{B}, \log \SCStrength)$. The verifier chooses $\vec{x}, \vec{y} \in \Field^{r + 3s}$ uniformly at random, and sends them to the prover; the prover and verifier construct the sum-product circuit $\SPCircuit \DefineEqual f(r, s, B, \vec{x}, \vec{y})$, and then engage in the protocol of \thmref{thm:pzk-for-SPCS} on the input $(\SPCircuit, 0, \EmptyVector)$, with auxiliary input $g(A)$.
	
	If $((B,r,s),A) \in \OracleSATRelation$ then $((\SPCircuit, 0, \EmptyVector), g(A)) \in \SPCSRelation$ with probability $1$, and so the verifier will accept with probability $1$. If $(B,r,s) \notin \Language(\OracleSATRelation)$ then $(\SPCircuit, 0, \EmptyVector) \in \Language(\SPCSRelation)$ with probability at most $(r + 3s)/\SetCardinality{\Field}$. Also, if $(\SPCircuit, 0, \EmptyVector) \not\in \Language(\SPCSRelation)$, then the verifier accepts with probability at most $O(\InternalVertexSetDegree\SPLeafDegree \cdot \GraphMaxInDegree{\Graph} \cdot (\SPMaxArity(\Graph) + \alpha) \cdot \SetCardinality{\VertexSet(\Graph)}/\SetCardinality{\Field}) = O(\frac{\SetCardinality{\SPSubset}^{2} \cdot\BitSize{B} \cdot (r+s+\log \SCStrength)}{\SetCardinality{\Field} \cdot \log \SetCardinality{\SPSubset}})$. By our choices of the cardinality of $\SPSubset$ and $\Field$, both of these probabilities are $O(1/\BitSize{B})$ so that, by a union bound, the probability that the verifier accepts is also $O(1/\BitSize{B})$, which is less than $1/2$.
	
	The protocol as stated is not an IPCP because there is a round of interaction before the oracle is sent. However, observe that the oracle does not depend on the choice of $\vec{x}, \vec{y}$, and therefore can be sent before this interaction.
\end{proof}

\doclearpage

\clearpage
\printbibliography

\end{document}